\documentclass[11pt]{amsart}
\usepackage[utf8]{inputenc}
\pdfoutput=1
\usepackage{amsmath,amssymb,amsthm,amsfonts}
\usepackage{mathtools}%

\usepackage[english]{babel}%
\usepackage{comment}%
\usepackage[unicode]{hyperref}%
\usepackage{bbm}%
\usepackage{bm}
\usepackage{mathrsfs}%

\usepackage{tikz,graphicx,color}
\usepackage{tikz-cd}%
\usepackage{tikz-3dplot}%
\usetikzlibrary{calc}%
\usetikzlibrary{arrows}%
\usetikzlibrary{shapes}%
\usetikzlibrary{patterns}%
\usetikzlibrary{positioning}%
\usetikzlibrary{arrows.meta}
\usetikzlibrary{decorations.markings}
\usetikzlibrary{knots}

\usepackage{epstopdf}%

\usepackage[arrow]{xy}%
\usepackage{diagbox}%
\usepackage{subfig}%
\usepackage{arcs}%
\usepackage{xcolor}%
\usepackage{xspace}

\usepackage[margin=1.0in]{geometry}%
\usepackage[subtle,tracking=normal,mathdisplays=tight]{savetrees}
\usepackage{microtype}

\usepackage{enumitem}
\usepackage{letltxmacro}
\usepackage{thmtools,etoolbox}
\makeatletter
\IfFormatAtLeastTF{2024-11-01}{%
 \renewcommand*\@addtoreset[2]{%
 \bgroup
 \edef\aliasctr@@truelist{\aliasctr@follow{#2}}%
 \let\@elt\relax
 \expandafter\@cons\aliasctr@@truelist{{#1}}%
 \egroup
 \expandafter\xdef\csname theH#1\endcsname{%
 \expandafter\noexpand\csname theH#2\endcsname.%
 \noexpand\the\noexpand\value{#1}}%
 }%
}{}
\makeatother

\def\myarabic#1{\normalfont(\roman{#1})}
\newlist{theoremlist}{enumerate}{1}
\setlist[theoremlist]{label=\myarabic{theoremlisti},ref={\myarabic{theoremlisti}},itemindent=0pt,labelindent=0pt,
 leftmargin=*,noitemsep}

\makeatletter
\renewcommand{\p@theoremlisti}{\perh@ps{\thetheorem}}
\protected\def\perh@ps#1#2{\textup{#1#2}}
\newcommand{\itemrefperh@ps}[2]{\textup{#2}}
\newcommand{\itemref}[1]{\begingroup\let\perh@ps\itemrefperh@ps\ref{#1}\endgroup}

\newcommand{\crefi}[2]{%
 \nameCref{#1}~\hyperref[#2]{%
 \ref*{#1}%
 \begingroup\let\perh@ps\itemrefperh@ps\ref*{#2}\endgroup%
 }%
}
\makeatother

\usepackage{nameref}
\usepackage[capitalize,nosort]{cleveref}

\usepackage{scalerel}
\usepackage{stackengine}

\usepackage{xr} %

\makeatletter
\patchcmd{\XR@@}
 {\set@curr@file{#3}}
 {\edef\@curr@file{#3}}
 {}
 {\PackageError{xr-local-fix}
 {Could not patch xr}
 {Check the installed xr version.}}

\def\XRauxsuffix{.aux}
\patchcmd{\XR@@}
 {\@curr@file.aux}
 {\@curr@file\XRauxsuffix}
 {}
 {\PackageError{xr-local-fix}
 {Suffix patch failed}
 {Check the installed xr version.}}

\def\extfile{loop_ampl}
\IfFileExists{\extfile.aux}{%
 \externaldocument[N-]{\extfile}%
 
}{%
 \def\extfile{loop_ampl_aux}
 \begingroup
 \def\XRauxsuffix{.tex}%
 \externaldocument[N-]{\extfile}[loop_ampl.pdf]%
 \endgroup
 
}
\makeatother

\newtheorem{theoremintro}{Theorem}
\newtheorem{theorem}{Theorem}[section]

\newtheorem{lemma}[theorem]{Lemma}
\newtheorem{proposition}[theorem]{Proposition}
\newtheorem{corollary}[theorem]{Corollary}

\newtheorem{problem}[theorem]{Problem}
\newtheorem{question}[theorem]{Question}

\theoremstyle{definition}
\newtheorem{remark}[theorem]{Remark}
\newtheorem{definition}[theorem]{Definition}
\newtheorem{example}[theorem]{Example}
\newtheorem{notation}[theorem]{Notation}

\crefname{figure}{Figure}{Figures}

\def\figref#1(#2){Figure~\hyperref[#1]{\ref*{#1}(#2)}}

\AtBeginEnvironment{theorem}{\crefalias{theoremlisti}{theorem}}
\AtBeginEnvironment{lemma}{\crefalias{theoremlisti}{lemma}}
\AtBeginEnvironment{proposition}{\crefalias{theoremlisti}{proposition}}
\AtBeginEnvironment{corollary}{\crefalias{theoremlisti}{corollary}}

\def\Acal{\mathcal{A}}\def\Bcal{\mathcal{B}}\def\Ccal{\mathcal{C}}\def\Hcal{\mathcal{H}}\def\Ical{\mathcal{I}}\def\Lcal{\mathcal{L}}\def\Mcal{\mathcal{M}}\def\Ncal{\mathcal{N}}\def\Ocal{\mathcal{O}}\def\Rcal{\mathcal{R}}\def\Tcal{\mathcal{T}}

\def\C{{\mathbb{C}}}
\def\R{{\mathbb{R}}}
\def\Z{{\mathbb{Z}}}

\def\<{{\langle}}
\def\>{{\rangle}}
\def\eps{{\epsilon}}

\def\id{\operatorname{id}}
\def\det{\operatorname{det}}
\def\Ker{\operatorname{Ker}}
\def\diag{\operatorname{diag}}
\def\rank{\operatorname{rank}}
\def\Span{\operatorname{Span}}
\def\wt{\operatorname{wt}}
\def\RR{{\mathbb R}}
\def\RP{{\RR\mathbb P}}

\def\xrasim{\xrightarrow{\sim}}

\def\GL{\operatorname{GL}}
\def\SL{\operatorname{SL}}
\def\Fl{\operatorname{Fl}}
\def\Mat{\operatorname{Mat}}
\def\Gr{\operatorname{Gr}}
\def\Grtnn{\Gr_{\ge 0}}
\def\Grtp{\Gr_{>0}}
\def\Pio{\Pi^\circ}
\def\Povtp_#1{\Pi_{#1}^{>0}}
\def\Ptp_#1{\Povtp_{#1}}
\def\Povtnn_#1{\Pi_{#1}^{\geq0}}
\def\BND{\Bcal}
\def\Boundkn{\BND(k,n)}
\def\Boundxx(#1,#2){\BND(#1,#2)}
\def\Bounda{\BND_{2,0}(k,n)}
\def\Boundb{\BND_{0,2}(k,n)}
\def\OG{\operatorname{OG}}
\def\OGtnn{\OG_{\ge0}}
\def\Boundkdk{\BND(k,2k)}

\numberwithin{equation}{section}

\makeatletter
\@namedef{subjclassname@2020}{\textup{2020} Mathematics Subject Classification}
\makeatother

\def\wtT{\wt_{|\Tcal|}}

\def\fX{f_{X}}
\def\sumw{\alphaT^\circ}
\def\sumb{\alphaT^\bullet}

\def\sumwP_#1{\sumw_{\Pcurve,#1}}
\def\sumbP_#1{\sumb_{\Pcurve,#1}}

\def\Apm{\apm}
\def\wind{\operatorname{wind}}
\def\windskip{\wind_{\hat1\hat2}}

\makeatletter
\let\save@mathaccent\mathaccent
\newcommand*\if@single[3]{%
 \setbox0\hbox{${\mathaccent"0362{#1}}^H$}%
 \setbox2\hbox{${\mathaccent"0362{\kern0pt#1}}^H$}%
 \ifdim\ht0=\ht2 #3\else #2\fi
 }
\newcommand*\rel@kern[1]{\kern#1\dimexpr\macc@kerna}
\newcommand*\widebar[1]{\@ifnextchar^{{\wide@bar{#1}{0}}}{\wide@bar{#1}{1}}}
\newcommand*\wide@bar[2]{\if@single{#1}{\wide@bar@{#1}{#2}{1}}{\wide@bar@{#1}{#2}{2}}}
\newcommand*\wide@bar@[3]{%
 \begingroup
 \def\mathaccent##1##2{%
 \let\mathaccent\save@mathaccent
 \if#32 \let\macc@nucleus\first@char \fi
 \setbox\z@\hbox{$\macc@style{\macc@nucleus}_{}$}%
 \setbox\tw@\hbox{$\macc@style{\macc@nucleus}{}_{}$}%
 \dimen@\wd\tw@
 \advance\dimen@-\wd\z@
 \divide\dimen@ 3
 \@tempdima\wd\tw@
 \advance\@tempdima-\scriptspace
 \divide\@tempdima 10
 \advance\dimen@-\@tempdima
 \ifdim\dimen@>\z@ \dimen@0pt\fi
 \rel@kern{0.6}\kern-\dimen@
 \if#31
 \overline{\rel@kern{-0.6}\kern\dimen@\macc@nucleus\rel@kern{0.4}\kern\dimen@}%
 \advance\dimen@0.4\dimexpr\macc@kerna
 \let\final@kern#2%
 \ifdim\dimen@<\z@ \let\final@kern1\fi
 \if\final@kern1 \kern-\dimen@\fi
 \else
 \overline{\rel@kern{-0.6}\kern\dimen@#1}%
 \fi
 }%
 \macc@depth\@ne
 \let\math@bgroup\@empty \let\math@egroup\macc@set@skewchar
 \mathsurround\z@ \frozen@everymath{\mathgroup\macc@group\relax}%
 \macc@set@skewchar\relax
 \let\mathaccentV\macc@nested@a
 \if#31
 \macc@nested@a\relax111{#1}%
 \else
 \futurelet\first@char\gobble@till@marker#1\endmarker
 \ifcat\noexpand\first@char A\else
 \fi
 \macc@nested@a\relax111{\first@char}%
 \fi
 \endgroup
}
\makeatother

\makeatletter
\DeclareRobustCommand{\cev}[1]{%
 \mathpalette\do@cev{#1}%
}
\newcommand{\do@cev}[2]{%
 \fix@cev{#1}{+}%
 \reflectbox{$\m@th#1\vec{\reflectbox{$\fix@cev{#1}{-}\m@th#1#2\fix@cev{#1}{+}$}}$}%
 \fix@cev{#1}{-}%
}
\newcommand{\fix@cev}[2]{%
 \ifx#1\displaystyle
 \mkern#23mu
 \else
 \ifx#1\textstyle
 \mkern#23mu
 \else
 \ifx#1\scriptstyle
 \mkern#22mu
 \else
 \mkern#22mu
 \fi
 \fi
 \fi
}

\makeatother

\newlength\arrowheight

\def\bdrypath{\vec\partial_{\E}}
\def\bdrypathf{\vec\partial_{\Ef}}
\def\bdrypathbot{\vec\partial_{\Ebot}}
\def\bdrypathp{\vec\partial_{\E'}}

\def\MatroidG{\Mcal_{\G}}

\def\fhat{f^\vee}

\def\Gvee{\G^\vee}

\def\xllp{\xd'_{\la,\lat}}
\def\Tllp{\Tcal'_{\la,\lat}}
\def\Tllcs{\Tcal_{\lacs,\latcs,C}}
\def\Iji{I_{\bm{(j,i]}}}

\def\StauOm{S(\tauOm)}
\def\StauOmbot{S(\tauOmbot)}
\def\StauOmbotin{S_{\insub}(\tauOmbot)}
\def\StauOmbotout{S_{\outsub}(\tauOmbot)}

\def\Ttaucoef{c^{\ijsepsup}_{\tau,T}}
\def\Ttaucoefxx#1#2{c^{#1\ssep #2}_{\tau,T}}
\def\fgsup{\ff\ssep\f}
\def\TtaucoefOmfgof{c^{\fgsup}_{\Om}\!\!}

\def\TtaucoefOmfgijof{c^{\fgijsup}_{\Om}\!\!\!}
\def\TtaucoefOmeeof{c^{\e_1,e_2}_{\Om}}

\def\BCFWG{\mathbf{\G}^{\BCFWop}}
\def\BCFWGkn{\BCFWG_{k,n}}
\def\BCFWGknL{\BCFWG_{k_L,n_L}}
\def\BCFWGknR{\BCFWG_{k_R,n_R}}
\def\BCFWGkdk{\BCFWG_{k,2k}}

\def\brat[#1|{[#1|}
\def\kett|#1]{|#1]}
\def\bralat[#1|{[#1|_{\lat}}
\def\ketlat|#1]{|#1]_{\lat}}
\def\Ray{{\mathfrak{R}}}
\def\ray{\Tcomp{\Ray}}
\def\rayTO{\Ray}
\def\rayT{\Tcomp{\Ray}}
\def\rayO{\Ocomp{\Ray}}

\def\raypointO_#1{\rayO(\r_{#1})}
\def\raypointTO_#1{\rayTO(\r_{#1})}

\def\Acurve{{\bm{a}}}
\def\AcurveT{\Tcomp{{\bm{a}}}}
\def\Bcurve{{\bm{b}}}
\def\BcurveT{\Tcomp{{\bm{b}}}}
\def\Ccurve{{\bm{c}}}
\def\AcurveL{\Acurve_L}
\def\AcurveTL{\AcurveT_L}
\def\BcurveL{\Bcurve_L}
\def\BcurveTL{\BcurveT_L}
\def\CcurveL{\Ccurve_L}
\def\AcurveR{\Acurve_R}
\def\AcurveTR{\AcurveT_R}
\def\BcurveR{\Bcurve_R}
\def\BcurveTR{\BcurveT_R}
\def\CcurveR{\Ccurve_R}
\def\Dcurve{{\bm{d}}}
\def\DcurveT{\Tcomp{\bm{d}}}

\def\tauT{(\tau,T)}
\def\tauTOm{(\tauOm,\TOm)}
\def\tauOm{\tau_{\Om}}
\def\TOm{T_{\Om}}
\def\tauTOmbot{(\tauOmbot,\TOmbot)}
\def\tauOmbot{\tau_{\Ombot}}
\def\TOmbot{T_{\Ombot}}
\def\tauTOmf{(\tauOmf,\TOmf)}
\def\tauOmf{\tau_{\Omf}}
\def\TOmf{T_{\Omf}}
\def\Gtnn{\GL^{\geq0}_n(\R)}
\def\Otp{\operatorname{O}^{>0}_{2k}(\R)}
\def\Otnn{\operatorname{O}^{\geq0}_{2k}(\R)}

\def\RWint{\Rg^\circ_{\intop}}
\def\RBint{\Rg^\bullet_{\intop}}

\def\helmin{k_{\min}}

\def\rcrit{\outacc{r}}

\def\G{\Gamma}
\def\GD{\Gamma^\ast}
\def\GDp{\Gamma^{\prime\ast}}
\def\vertex{v}
\def\v{\vertex}
\def\vv{u}

\def\fletter{g}
\def\ffletter{f}

\def\face{\fletter^\ast}
\def\ff{\ffletter^\ast}

\def\f{\face}

\def\ffout{\outacc{\ffletter}^{\ast}}

\def\bdryaccent#1{#1^{\partial}}
\def\bdvletter{u}
\def\bdv{\bdvletter^{\partial}}
\def\bdvp{\bdvletter^{\partial\prime}}
\def\bdvx{\tilde \bdvletter^{\partial}}
\def\bdwx{\tilde \wv^{\partial}}

\def\bdvf{\fouracc{\bdvletter}^{\partial}}
\def\bdef{\fouracc{\e}^{\partial}}

\def\bde{\bdryaccent\e}
\def\bdep{\e^{\partial\prime}}
\def\bdf{\ffletter^{\partial\ast}}

\def\Verts{\mathbf{V}}
\def\Faces{\Verts^\ast}
\def\intop{\mathtt{int}}
\def\Vint{\Verts_{\intop}}
\def\Vbd{\Verts_\partial}
\def\Fint{\Faces_{\intop}}
\def\Fbd{\Faces_\partial}
\def\bv{b}
\def\wv{w}
\def\e{e}
\def\east{e^\ast}
\def\E{\mathbf{E}}
\def\East{\mathbf{E}^\ast}
\def\Eint{\E_{\intop}}
\def\Ebd{\E_\partial}
\def\alt{\operatorname{alt}}
\def\altp{\alt^\perp}
\def\lalats{\lalatbf_n}
\def\la{\lambda}
\def\lat{\tilde\lambda}
\def\lalat{(\la,\lat)}
\def\lalak{\lalatbf_{k,n}^+}
\def\lalakdk{\lalatbf_{k,2k}^+}
\def\lalakx_#1{\lalatbf_{#1}^+}
\def\TRIPLES{\lalatbf\mkern-1mu\bm{C}_{k,n}^+}
\def\laCf{\labf\mkern-1mu\bm{C}_{k,n,f}^+}
\def\Meas{\operatorname{Meas}}
\def\BV{\Verts^{\bullet}}
\def\WV{\Verts^{\circ}}
\def\BVint{\BV_{\intop}}
\def\WVint{\WV_{\intop}}
\def\BVbd{\BV_{\partial}}
\def\WVbd{\WV_{\partial}}
\def\Fw{F^\circ}
\def\Fb{\tilde F^\bullet}
\def\I{\mathbf{i}}
\def\gauge{g}
\def\pt{p}
\def\Mandop{S}
\def\Mand_#1(#2){\Mandop_{#2}(#1)}
\def\brxx<#1,#2>{\<#1\,#2\>}
\def\brla<#1,#2>{\brxx<#1,#2>_{\la}}
\def\brlae<#1,#2>{\brxx<#1,#2>_{\laeps}}
\def\brlapr<#1,#2>{\brxx<#1,#2>_{\la'}}

\def\brlat[#1,#2]{[#1\,#2]_{\lat}}
\def\La{\Lambda}
\def\Lat{\tilde\Lambda}
\def\PhiLL{\Phi_{\La,\Lat}}
\def\Pip{\Pi^{>0}}

\def\Grnda{\Grndxx20}
\def\Grndb{\Grndxx02}

\def\twonondeg{$2^\partial$-nondegenerate\xspace}
\def\Qla{Q_{\la}}
\def\Cast{\C^\times}
\def\Rast{\R^\times}
\def\Gcoll{\mathbf{\G}}
\def\latp{\lat^\perp}

\def\alphaT{\Tcomp{\alpha}}

\def\sumwT{\alphaT^\circ}
\def\sumbT{\alphaT^\bullet}

\def\sumwTll{\sumw}
\def\sumbTll{\sumb}
\def\SO{\operatorname{SO}}
\def\Re{\operatorname{Re}}
\def\Im{\operatorname{Im}}
\def\ovlmargin{0.1em}
\def\ovl#1{\hspace{\ovlmargin}\overline{\hspace{-\ovlmargin}#1\hspace{-\ovlmargin}}\hspace{\ovlmargin}}

\def\bt{\mathbf{t}}
\def\tdiag{\bt}
\def\epsK{\varepsilon}
\def\FX_#1(#2){#2(\bdrypath{#1})}
\def\pFw{\partial\Fw}
\def\pFb{\partial\Fb}
\def\pFwl{\partial\Fwl}
\def\pFbl{\partial\Fbl}

\def\APMSop{\Acal}
\def\APMS{\Acal}
\def\APMSGbd(#1){\APMS_{#1}(\G)}
\def\apm{\bm{a}}
\def\by{\bm{\xi}}
\def\byt{\bm{\tilde \xi}}
\def\fp{\f}

\def\RV#1{{\normalfont(R$#1$)}\xspace} %
\def\MV#1{{\normalfont(M$#1$)}\xspace} %
\def\MVbd{\MV{\bdryaccent{1}}}

\def\Itarget{I}
\def\Ir{\Itarget}

\def\Icalr{\Ical}
\def\laext{\la^\circ}
\def\latext{\lat^\bullet}
\def\Chat{\widehat C}
\def\labf{\bm{\la}}
\def\latbf{\bm{\lat}}
\def\lalatbf{\bm{\la\!^\perp\!\!\tilde\la}}
\def\LaLatbf{\bm{\La\!\tilde\La}}
\def\lak{\labf_{k,n}^+}
\def\lakx_#1{\labf_{#1}^+}
\def\lakpr{\labf_{k',n}^+}
\def\latk{\latbf_{k,n}^+}
\def\latkx_#1{\latbf_{#1}^+}
\def\LaLak{\LaLatbf_{k,n}^+}
\def\LaLat{(\La,\Lat)}
\def\Lap{\La^\perp}
\def\LapLat{(\La^\perp,\Lat)}
\def\lakdk{\labf_{k,2k}^+}

\def\fdd{\ddot f}

\def\bdryarcs{\partial^{\mathtt{arcs}}}
\def\nbdryarcs#1{|\bdryarcs#1|}
\def\laextsub{\la}
\def\latextsub{\lat}
\def\brlaw<#1,#2>{\brxx<#1,#2>_{\laextsub}}
\def\brlatb[#1,#2]{[#1\,#2]_{\latextsub}}
\def\brlawf<#1,#2>{\brxx<#1,#2>_{\laextsub}}
\def\brlatbf[#1,#2]{[#1\,#2]_{\latextsub}}
\def\brguess<#1,#2>{\brxx<#1,#2>^{\mathtt{RHS}}}

\def\Kop{\operatorname{K}}
\def\wtK{\Kop}
\def\ddwt{\ddot{\wt}}
\def\ddG{\ddot{\G}}
\def\Sig{\Sigma}
\def\ddfperm{\ddot{\fap}}
\def\ddItarget{\ddot{\Itarget}}
\def\ddIr{\ddItarget}
\def\w{\wv}
\def\b{\bv}

\def\zzpath{\gamma}
\def\s{s}
\def\cshift{\sigma_k}
\def\leq{\leqslant}
\def\geq{\geqslant}
\def\ge{\geqslant}
\def\epsvar{\varepsilon}

\def\GO{\vec \G}
\def\GOp{\vec \G'}
\def\fC{f_{C}}
\def\fG{f_\G}
\def\VandM{M}
\def\slp{\nu}
\def\Rtp{\R_{>0}}
\def\csop{\mathtt{cs}}
\def\Lacs{\La^{\csop}}
\def\Latcs{\Lat^{\csop}}
\def\lacs{\la^{\csop}}
\def\latcs{\lat^{\csop}}
\def\Momcs{\Mcal_{\Lacs,\Latcs}}
\def\Phics{\Phi_{\Lacs,\Latcs}}
\def\Shift{\sigma_k}
\def\Xcs{X^{\csop}_k}
\def\Xcsx#1{X^{\csop}_{#1}}
\def\XcsLa{X^{\csop}_{n-k+2}}
\def\XcsLat{X^{\csop}_{k+2}}
\def\eig{\xi}
\def\eigvec{u}
\def\Eig{E}
\def\bzero{\bm{0}}
\def\Id{\bm{1}}
\def\medmatrix#1#2{\text{\scalebox{#1}{$\left(\begin{matrix}#2\end{matrix}\right)$}}}
\def\Hast{\ast}
\def\fkn{f_{k,n}}
\def\Matroid{\Mcal}
\def\Ibar{\bar I}
\def\Icalbar{\bar\Ical}

\def\cShift{\cshift} %
\def\epsKpr{\epsK'}
\def\fpr{f'}
\def\comp#1{#1^{c}}
\def\Fwl{\Fw_{\la}}
\def\Fbl{\Fb_{\lat}}
\def\brn{\bm{[n]}}
\def\brm{\bm{[m]}}
\def\brnm{\bm{[n-1]}}
\def\brd{\bm{[d]}}
\def\brx#1{\bm{[#1]}}
\def\Iij{I_{\bm{(i,j]}}}
\def\Jijyy_#1{I_{\bm{[#1)}}}

\def\Stau{S(\tau)}
\def\Ttaus{\bm{\tau\!T}}
\def\Ttaukn{\Ttaus_{k,n}}
\def\Ttauknij{\Ttaus^{i\ssep j}_{k,n}}
\def\Lmark{\operatorname{L}}
\def\Rmark{\operatorname{R}}
\def\marking{\mu}
\def\Lmu{L_{\marking}}
\def\Rmu{R_{\marking}}
\def\dmu{d_\mu}%
\def\brk{\brx{k}}
\def\BRLAP<#1>{\<#1\>^\perp}
\def\BRLATP<#1>{\<\tilde{#1}\>}
\def\Ttauij{(\tau,T,i,j)}
\def\Pmom{P}
\def\MPop{\Mcal^{\ambsup}}
\def\lalapp{\MPop_{k,n}}
\def\lalappf{\MPop_f}
\def\lalappg{\MPop_g}

\def\brat[#1|{[#1|}
\def\kett|#1]{|#1]}
\def\bralat[#1|{[#1|_{\lat}}
\def\ketlat|#1]{|#1]_{\lat}}
\def\Tcomp#1{\hat{#1}}
\def\Ocomp#1{\check{#1}}
\def\jo{{j_0}}
\def\io{{i_0}}
\def\jop{{j'_0}}
\def\PmomT{\Tcomp\Pmom}
\def\PmomO{\Ocomp\Pmom}
\def\Qmom{Q}
\def\QmomT{\Tcomp\Qmom}
\def\QmomO{\Ocomp\Qmom}
\def\ycomp#1{#1^{\y}}
\def\ytcomp#1{#1^{\yt}}
\def\Pmomy{\ycomp\Pmom}
\def\Pmomyt{\ytcomp{\Pmom}}
\def\Qmomy{\ycomp\Qmom}
\def\Qmomyt{\ytcomp\Qmom}
\def\r{r}
\def\Rmom{R}
\def\RmomT{\Tcomp\Rmom}
\def\RmomO{\Ocomp\Rmom}
\def\Rmomy{\ycomp\Rmom}
\def\Rmomyt{\ytcomp\Rmom}
\def\Psum_#1{\Pmom_{\bm{(\io,#1]}}}
\def\line_#1{\ell_{#1}}
\def\raypoint_#1{\ray(\r_{#1})}

\def\Sone(#1,#2){\|\xd(#1)-\xd(#2)\|_1}

\def\Phmom{\Pmom'}
\def\Qhmom{\Qmom'}
\def\QhmomT{\QmomT'}

\def\lalatL{(\la_L,\lat_L)}
\def\lalatR{(\la_R,\lat_R)}

\def\Pcurve{\bm{p}}
\def\PcurveT{\Tcomp{\bm{p}}}
\def\PcurveO{\Ocomp{\bm{p}}}
\def\Pbd{\Pcurve^{\partial}}
\def\PbdT{\PcurveT^{\partial}}
\def\PbdO{\PcurveO^{\partial}}
\def\PbdxT{\PbdT_{\xd}}
\def\PbdxO{\PbdO_{\xd}}
\def\Pbdx{\Pbd_{\xd}}

\def\sumD{\alphaT_{\Dcurve}}
\def\sumDa{\alphaT_{\Dcurve_1}}
\def\sumDb{\alphaT_{\Dcurve_2}}

\def\sumwAL{\alphaT_{\AcurveL}^\circ}
\def\sumwBR{\alphaT_{\BcurveR}^\circ}
\def\sumwCL{\alphaT_{\CcurveL}^\circ}
\def\sumwCR{\alphaT_{\CcurveR}^\circ}
\def\sumwD{\alphaT_{\Dcurve}^\circ}
\def\sumwDa{\alphaT_{\Dcurve_1}^\circ}
\def\sumwDb{\alphaT_{\Dcurve_2}^\circ}

\def\sumbAR{\alphaT_{\AcurveR}^\bullet}
\def\sumbBL{\alphaT_{\BcurveL}^\bullet}
\def\sumbCL{\alphaT_{\CcurveL}^\bullet}
\def\sumbCR{\alphaT_{\CcurveR}^\bullet}
\def\sumbD{\alphaT_{\Dcurve}^\bullet}
\def\sumbDa{\alphaT_{\Dcurve_1}^\bullet}
\def\sumbDb{\alphaT_{\Dcurve_2}^\bullet}

\def\LG{T}
\def\LGp{T_+}
\def\LGm{T_-}
\def\LGpm{T_{\pm}}
\def\argp(#1){\arg_{[0,\pi)}(\pm#1)}

\def\pinv{\Theta}
\def\tlai{t^\la_i}
\def\tla_#1{t^\la_{#1}}
\def\btla{\tdiag^\la}
\def\mupt{\nu}
\def\Hyp{H}
\def\HArr{\Acal}
\def\HArrC{\Acal_C}
\def\bHArrC{\overline{\Acal}_C}
\def\chiC{\chi_{\HArrC}}
\def\bbar{\beta}
\def\ba{\bm{a}}
\def\dop{\operatorname{d}\!}
\def\dlog{\operatorname{dlog}}

\def\TO{\xd}
\def\ITO{\Ical[\TO]}
\def\bH{\bar H}

\def\pifl{\pi_{k-2,k+2}}
\def\Gtp{\GL^{>0}_n(\R)}
\def\Tcirc{\mathbb{T}}

\def\dref#1#2{\hyperref[#2]{\ref*{#1}\ref*{#2}}}

\def\Rg{R}
\def\RW{\Rg^\circ}
\def\RB{\Rg^\bullet}

\def\RgWV{\Rg^\circ}
\def\RgBV{\Rg^\bullet}

\def\GR{\G\ind[\Rg]}

\def\helW{k^\circ_\G}
\def\helB{k^\bullet_\G}
\def\helWsub_#1{k^\circ_{#1}}
\def\helBsub_#1{k^\bullet_{#1}}

\def\helWmin{k^\circ_{\min}}
\def\helBmin{k^\bullet_{\min}}

\def\helWp{k^\circ_{\G'}}
\def\degG{\deg_\G}
\def\degGD{\deg_{\GD}}

\def\xd{\bm{x}}
\def\xT{\Tcomp{\xd}}
\def\xO{\Ocomp{\xd}}
\def\xs{x}
\def\xsT{\Tcomp{x}}
\def\xsO{\Ocomp{x}}
\def\ys{y}
\def\ysT{\Tcomp{y}}

\def\leftop{\oldmathtt{L}}
\def\rightop{\oldmathtt{R}}

\def\Lop{\leftop}
\def\Rop{\rightop}

\def\Edges{\mathbf{E}}

\def\outop{\mathtt{out}}

\def\bdx{\xs^\partial}
\def\bdxM{M^\partial}
\def\lalappx_#1{\MPop_{#1}}

\def\xTout{\Tcomp{\xdout}}
\def\xOout{\Ocomp{\xdout}}

\def\xM{M}
\def\ap{a_+}
\def\am{a_-}
\def\CMI{\eta}

\def\lap{\la^\perp}
\def\Qlapp{Q_\la^\vee}
\def\ddC{\ddot{C}}
\def\AA{\Acal^{\ambsup}}
\def\AAkn{\AA_{k-2,n}}
\def\AAkntree{\AAkn}
\def\AAddf{\AA_{\ddfperm}}

\def\var{\operatorname{var}}
\def\brV[#1,#2,#3,#4]{[#1\,#2\,#3\,#4]_V}
\def\brVLnopar[#1,#2,#3]{[#1\,#2\,#3]_{V}}
\def\brVLy[#1,#2]{\brVLnopar[#1,#2,\Liney]}

\def\MPnla{\MPop_{n;\la^\perp}}
\def\AAnla{\AA_{n;(\la)}}
\def\bdxT{\xsT^{\partial}}
\def\bdxO{\xsO^{\partial}}

\def\MPkn{\MPop_{k,n}}
\def\MPkntree{\MPkn}
\def\MPf{\MPop_f}
\def\Grsup^#1{\Gr^{(#1)}}
\def\Grsupla{\Grsup^\la}
\def\ladual{\la^\vee}
\def\Qladual{Q_{\ladual}}
\def\Grtnnsupla{\Gr_{\geq0}^{(\la)}}
\def\flagsup{\mathtt{flag}}
\def\MAop{\Acal^{\flagsup}}
\def\MAkn{\MAop_{k-2,n}}
\def\MAddf{\MAop_{\ddfperm}}

\def\Bounddb{\BND_{0,2}(k-2,n)}

\def\Disk{\mathbb{D}}

\def\ORATITLE{Origami reconstruction algorithm\xspace}
\def\oraTITLE{origami reconstruction algorithm\xspace}
\def\ora{origami reconstruction algorithm\xspace}

\def\tembendings{weak t-embeddings\xspace}

\def\intsup{\diamond}

\def\ebar{\bar\e}
\def\ebarast{\bar{e}^\ast}
\def\Ebarast{\bar{\E}^\ast}

\def\outacc{\tilde}

\def\partEast{\partial_{\East}}
\def\partE{\partial_{\Edges}}

\def\ORst{origami reconstruction step\xspace}

\def\Cyc{\zeta}

\def\adatr{an \datr}
\def\datr{algebraic t-realization\xspace}
\def\datrs{algebraic t-realizations\xspace}

\def\Rtpgauge{\Rtp^{|\Faces|-1}}
\def\RtpVint{\Rtp^{|\Vint|}}
\def\RtpE{\Rtp^{|\Edges|}}

\def\pmoneVint{\{\pm1\}^{|\Vint|}}
\def\decacc#1{#1}

\def\talgop{\mathtt{ATR}}

\def\Mdatr{\decacc\Mcal_{\talgop}}

\def\Neigh{\operatorname{N}}
\def\NeighG{\Neigh_{\G}}
\def\XV{X}
\def\XW{\XV^{\circ}}
\def\XB{\XV^{\bullet}}

\def\Rdd{\R^{2,2}}

\def\ind[#1]{[#1]}

\def\sconn{simply connected\xspace}
\def\WKmat{\Kop^{\circ}}
\def\BKmat{\Kop^{\bullet}}
\def\hasMbd{has M-positive boundary\xspace}
\def\Mbd{M-positive boundary\xspace}
\def\MbdTITLE{Mandelstam-positive boundary\xspace}
\def\weakMbd{weakly M-positive boundary\xspace}
\def\datrQ{(\wt,\epsK,\Fw,\Fb,\xd)}
\def\datrQf{(\wtf,\epsKf,\Fwf,\Fbf,\xdGf)}
\def\datrQnox{(\wt,\epsK,\Fw,\Fb)}
\def\datrQL{\bm{T}}
\def\datrQLll{\bm{T}_{\la,\lat}}
\def\datrQll{(\wt,\epsK,\Fw,\Fb,\xll)}
\def\datrQLf{\fouracc{\datrQL}}
\def\quintuple{quintuple\xspace}

\def\fouracc#1{\widetilde{#1}}
\def\Gf{\fouracc{\G}}
\def\Gfint{\Gf_{\intop}}
\def\GDf{\fouracc{\G}^\ast}
\def\Ef{\fouracc{\Edges}}
\def\Vf{\fouracc{\Verts}}
\def\Vfint{\fouracc{\Verts}_{\intop}}
\def\Vfbd{\fouracc{\Verts}_\partial}
\def\WVf{\fouracc{\Verts}^\circ}
\def\BVf{\fouracc{\Verts}^\bullet}
\def\BVfint{\fouracc{\Verts}^\bullet_{\intop}}
\def\WVfint{\fouracc{\Verts}^\circ_{\intop}}
\def\wtf{\fouracc{\wt}}
\def\wtKf{\fouracc{\wtK}}
\def\Cf{\fouracc{C}}
\def\epsKf{\fouracc{\epsK}}
\def\kf{\fouracc{k}}
\def\nf{\fouracc{n}}
\def\Fwf{F^\circ}
\def\Fbf{\tilde F^\bullet}
\def\xdGf{\fouracc{\xd}}
\def\laextf{\la^\circ}
\def\latextf{\lat\vphantom{\lat}^\bullet}
\def\Omf{\fouracc{\Om}}
\def\Ombot{\botacc{\Om}}
\def\Omfvec{\Omf^{\uparrow}}
\def\Omfvecbw{\Omf^{\bullet\to\circ}}
\def\Omvecbw{\Om^{\bullet\to\circ}}
\def\epsOmf{\varepsilon_{\mathtt{alt}}}

\def\apmwf{\fouracc{\apm}^\circ}
\def\apmbf{\fouracc{\apm}^\bullet}

\def\Cut{\xi}
\def\CutP{\xi_{\Path}}
\def\CutPb{\xi_{\Pathb}}
\def\CutPw{\xi_{\Pathw}}
\def\CutPwp{\xi_{\Pathwp}}
\def\constw{c^\circ}
\def\constb{c^\bullet}
\def\Grem#1{\G\rem #1}
\def\Gfrem#1#2{\Gf\rem\{#1,#2\}}
\def\Gfirem#1#2{\Gfint\setminus\{#1,#2\}}
\def\Gfiww{\Gfirem{\w_1}{\w_2}}
\def\Gfbb{\Gfrem{\b_1}{\b_2}}
\def\Giww{\G\rem\{\w_1,\w_2\}}
\def\Gbb{\G\rem\{\b_1,\b_2\}}

\def\Hwspace_#1{\Hcal_{#1}^\circ}
\def\Hbspace_#1{\Hcal_{#1}^\bullet}
\def\Hwbspace_#1{\Hcal_{#1}^{\circ\bullet}}
\def\HHspaceKV{\Hwbspace_{\Vspace}\HtripK}
\def\HHspaceKC{\Hwbspace_{\C}\HtripK}
\def\HHspaceKCf{\Hwbspace_{\C}\HtripKf}
\def\HHspaceKRd{\Hwbspace_{\R^2}\HtripK}
\def\HHspaceKRdf{\Hwbspace_{\R^2}\HtripKf}

\def\Omegas{\bm{\Omega}}
\def\Om{\Omega}
\def\Oms{\Omegas}
\def\ncycop(#1){|\operatorname{cyc}(#1)|}
\def\ncycOm{\ncycop(\Om)}
\def\ncycOmf{\ncycop(\Omf)}

\def\ssep{\|}
\def\OmsG{\Oms_{\G}}
\def\OmsGtauT{\Oms_{\G}\tauT}
\def\OmsGbot{\botacc{\Oms}_{\Gbot}}
\def\OmsGf{\fouracc{\Oms}_{\Gf}}
\def\OmsGfij{\OmsGf^{\ijsepsup}}
\def\OmsGij{\OmsG^{\ijsepsup}}

\def\Pathbd{\Path^\partial}
\def\bdbd{boundary-to-boundary\xspace}

\def\bweps{\delta}
\def\orUD_#1{\sigma_{#1}(\Omfvec)}
\def\epsOm{\varepsilon_{\mathtt{alt}}}
\def\Pathw{\Path^\circ}
\def\Pathwp{\Path^{\circ\prime}}
\def\Pathb{\Path^\bullet}

\def\apmw{\apm^\circ}
\def\apmb{\apm^\bullet}

\def\bisy{\zeta}
\def\bisyt{\tilde\zeta}
\def\CtoM[#1]{(#1)_{+}}
\def\CtoMt[#1]{(#1)_{-}}
\def\Matdnr{\Mat_{2,n}(\R)}
\def\Matrkone{\Mat^{(1)}_{2,n}(\R)}
\def\GLp{\GL_2^+(\R)}
\def\Matddr{\Mat_{2,2}(\R)}
\def\ambsup{\mathtt{flip}}
\def\topacc#1{#1}
\def\botacc#1{\breve{#1}}
\def\Gtop{\topacc{\G}}
\def\kbot{\botacc{k}}
\def\Gbot{\botacc{\G}}
\def\nbot{\botacc{n}}
\def\Vbot{\botacc{\Verts}}
\def\WVbot{\botacc{\Verts}^\circ}
\def\BVbot{\botacc{\Verts}^\bullet}
\def\BVintbot{\botacc{\Verts}^\bullet_{\intop}}
\def\WVintbot{\botacc{\Verts}^\circ_{\intop}}
\def\Fbot{\botacc{\Verts}^\ast}
\def\apmbot{\botacc{\apm}}
\def\apmbotp{\botacc{\apm}^+}
\def\apmbotm{\botacc{\apm}^-}
\def\apmbotpm{\botacc{\apm}^\pm}
\def\ffbot{\botacc{\ff}}
\def\ffbotout{\botacc{\ffletter}^\ast_{\outop}}

\def\Ebot{\botacc{\Edges}}

\def\Gbotp{\Gbot^+}
\def\Gbotm{\Gbot^-}
\def\Gbotpm{\Gbot^\pm}
\def\wttop{\topacc{\wt}}
\def\wtbot{\botacc{\wt}}
\def\wtbotK{\botacc{\wtK}}
\def\wtbotp{\botacc{\wt}\vphantom{\wt}'}
\def\epsKtop{\topacc{\epsK}}
\def\epsKbot{\botacc{\epsK}}
\def\Etop{\topacc{\Edges}}

\def\bdvtop{\topacc{\bdv}}
\def\bdebotin{\botacc{\e}^{\partial\insub}}
\def\bdfbotin{\botacc{\ffletter}^{\partial\ast\insub}}
\def\bdfbotout{\botacc{\ffletter}^{\partial\ast\outsub}}

\def\Cbot{\botacc{C}}
\def\Cbotp{\Cbot^+}
\def\Cbotm{\Cbot^-}
\def\Cbotpm{\Cbot^\pm}

\def\bapletter{f}
\def\fap{\bapletter}
\def\wtemb{weak t-embedding\xspace}
\def\wtembs{\tembendings}
\def\SuppGD{|\GD|}
\def\oac{origami-amplituhedron correspondence\xspace}
\def\OACTITLE{Origami-amplituhedron correspondence\xspace}
\def\oacTITLE{origami-amplituhedron correspondence\xspace}
\def\Mdash{M-}
\def\MdashTITLE{Mandelstam-}
\def\wdash{$\circ$-}
\def\bdash{$\bullet$-}
\def\wclosed{\wdash closed\xspace}
\def\bclosed{\bdash closed\xspace}
\def\Rgs{\bm{R}}

\def\clsub{\mathtt{cl}}

\def\WNEIg{\mathrlap{\overline{\phantom{\Rgs^{\circ}}}}\Rgs^{\circ}_{\clsub}}
\def\BNEIg{\mathrlap{\overline{\phantom{\Rgs^{\bullet}}}}\Rgs^{\bullet}_{\clsub}}
\def\WNEI{\Rgs^{\circ}_{\clsub}}
\def\BNEI{\Rgs^{\bullet}_{\clsub}}

\def\kw{k^\circ}
\def\kb{k^\bullet}
\def\rem{\setminus}
\def\KSprim{Kenyon--Smirnov primitive\xspace}
\def\KSprims{Kenyon--Smirnov primitives\xspace}
\def\Vspace{\mathbb{U}}
\def\HtripK{(\G,\wtK)}
\def\HtripKp{(\G',\wtK')}
\def\HtripKf{(\Gf,\wtKf)}
\def\HtripKbotp{(\Gbotp,\wtbotK)}
\def\HtripKbotm{(\Gbotm,\wtbotK)}
\def\Kawangle{Kawasaki angle\xspace}
\def\Deltaw{\Delta^\circ}
\def\Deltab{\Delta^\bullet}
\def\Ann{\botacc{\mathbb{A}}}
\def\insub{\mathtt{in}}
\def\outsub{\mathtt{out}}
\def\nin{\botacc{n}_{\insub}}
\def\nout{\botacc{n}_{\outsub}}
\def\bdvin{\botacc{\bdvletter}^{\partial\insub}}
\def\bdvout{\botacc{\bdvletter}^{\partial\outsub}}
\def\partin{\partial_{\insub}}
\def\partout{\partial_{\outsub}}
\def\Annin{\partin\Ann}
\def\Annout{\partout\Ann}
\def\Iin{I_{\insub}}
\def\Iout{I_{\outsub}}
\def\Iinp{I_{\insub}^+}
\def\Iinm{I_{\insub}^-}
\def\Iinpm{I_{\insub}^\pm}
\def\Ioutp{I_{\outsub}^+}
\def\Ioutm{I_{\outsub}^-}
\def\Ioutpm{I_{\outsub}^\pm}
\def\Arep{$\Ann$-representable\xspace}
\def\Areped{$\Ann$-represented\xspace}
\def\epstra{\delta}

\def\latextalt{\latextf_{\alt}}
\def\Fltpsymbol{\Fl_{>0}}
\def\Fltpn(#1,#2){\Fltpsymbol(#1,#2;n)}
\def\Fln(#1,#2){\Fl(#1,#2;n)}
\def\Xmat{X}
\def\Latrest{\Lat_{\mathtt{rest}}}
\def\twosep{$2$-separated\xspace}
\def\Twosep{$2$-separated\xspace}

\def\mtpref{momentum-twistor\xspace}
\def\mta{\mtpref amplituhedron\xspace}

\def\varxxx{\var_{123\ast}}
\def\varx{\var_{1\ast}}
\let\oldmu\mu
\def\mu{\oldmu}

\def\BCFWop{\mathtt{BCFW}}
\def\BCFWf{\bm{\fap}^{\BCFWop}}
\def\BCFWfkn{\BCFWf_{k,n}}

\def\BBf{\bm{\fap}}
\def\ddBBf{\ddot{\bm{\fap}}}
\def\BBfkn{\BBf_{k,n}}
\def\ddBBfkn{\ddBBf_{k-2,n}}
\def\ddRel{\ddot \Rel}
\def\mtiling{$\Rel$-tiling\xspace}
\def\Mtiling{$\Rel$-Tiling\xspace}
\def\mtilings{$\Rel$-tilings\xspace}
\def\amtiling{an $\Rel$-tiling\xspace}

\def\Reliso{\Rphitop}
\def\RelAAforg{(\id,\AAforg)}

\def\pperf{p}
\def\ddpperf{\ddot{p}}

\def\MomLL{\MomLLproj}

\def\projsup{\mathtt{proj}}
\def\Mproj{\Mcal^{\projsup}}
\def\Aproj{\Acal^{\projsup}}
\def\MomLLproj{\Mproj_{k,n}(\La,\Lat)}
\def\AZproj{\Aproj_{k-2,n}(Z)}
\def\MomLLprojf{\Mproj_{\fap}(\La,\Lat)}
\def\MomLLprojG{\Mproj_{\G}(\La,\Lat)}
\def\AZprojddf{\Aproj_{\ddfperm}(Z)}
\def\MomLLamb{\Mcal^{\ambsup}_{k,n}(\La,\Lat)}
\def\AZamb{\Acal^{\ambsup}_{k-2,n}(Z)}
\def\clMomLLamb{\overline{\Mcal^{\ambsup}_{k,n}(\La,\Lat)}}

\def\clAZamb{\overline{\Acal^{\ambsup}_{k-2,n}(Z)}}
\def\PsiZ{\Phi_Z}

\def\LaLaimmp{\LaLatbf_{k,n}^{\mathtt{imm>0}}}
\def\LaLaimmnn{\LaLatbf_{k,n}^{\mathtt{imm\geq0}}}
\def\LaLaimmnndk{\LaLatbf_{k,2k}^{\mathtt{imm\geq0}}}

\def\LaLafl{\LaLatbf_{k,n}^{\mathtt{flag}}}
\def\LaLaflp{\LaLatbf_{k,n}^{\mathtt{flag>0}}}

\def\LaLaArep{\LaLatbf_{k,n}^{\Ann}}
\def\fullysep{fully \twosep}

\def\isfullysep{is \fullysep}
\def\rhoij{\rho_{i,j}}

\def\xll{\xd_{\la,\lat}}
\def\xTll{\xT_{\la,\lat}}
\def\xOll{\xO_{\la,\lat}}
\def\Tll{\xTll}
\def\Oll{\xOll}
\def\Tcal{\xT}
\def\Ocal{\xO}
\def\T{\xT}
\def\O{\xO}
\def\y{\bisy}
\def\yt{\bisyt}
\def\paptwo#1{\cite[#1]{origami2}\xspace}
\def\justpaptwo{\cite{origami2}\xspace}

\def\Lamid{\widehat{\La}}
\def\mat[#1]{(#1)}

\let\oldmathtt\mathtt
\def\ttscl{0.85}
\def\mathtt#1{\text{\scalebox{\ttscl}{$\oldmathtt{#1}$}}}

\def\Nref#1{\paptwo{\cref*{N-#1}}}
\def\special{separating\xspace}
\def\Iijspec{$(\bdf_i,\bdf_j)$-\special}

\def\botoutacc#1{\botacc{#1}^{\partial\outsub}}

\def\OmsGbotcoef{\botacc{\Oms}^{\ijsepsup}_{\tau,T}}
\def\bbdv{\bm{\bdvletter}^\partial}
\def\bbdvij{\bbdv_{\bm{(i,j]}}}
\def\bbdvji{\bbdv_{\bm{(j,i]}}}
\def\gp{\gamma}
\def\gl{\ell}
\def\windGL{\wind_{\mathtt{GL}}}
\def\GLPath{\gp}
\def\windRd{\wind}
\def\glrot{\rho}
\def\gm{g}
\def\gmpol{\bm{\gm}}
\def\gmpoll{\gl_{\gmpol}}
\def\winding{winding\xspace}
\def\TURN{\operatorname{turn}}

\def\PllO{\PbdO_{\la,\lat}}
\def\PllT{\PbdT_{\la,\lat}}
\def\Pll{\Pbd_{\la,\lat}}
\def\muof(#1){\mu^{#1}}
\def\muiof_#1(#2){\mu_{#1}^{#2}}
\def\windaround(#1,#2){\wind(#1,#2)}
\def\Liney{\Lcal_{\ys}}
\def\Arg{\operatorname{Arg}}
\def\LVVp{\mat[V_1|V_2]^\perp}
\def\laeps{\la^{\eps}}
\def\laieps_#1{\la_{#1}^{\eps}}
\def\Veps{V^{\eps}}
\def\muot{\mu_{03}}
\def\yes{\eps\ys}
\def\yesM{\eps\xM_{\ys}}
\def\muepsy{\muof(\yes)}
\def\muepsyi_#1{\muiof_{#1}(\yes)}
\def\subev{\mathtt{ev}}
\def\subodd{\mathtt{odd}}
\def\oend{parity-nondegenerate\xspace}
\def\brnodd{\brx{2k}_{\subodd}}
\def\brnev{\brx{2k}_{\subev}}
\def\DP{D}
\def\Groe{\Grtnn^{\mathtt{par}}}
\def\AAkdk{\AA_{k-2,2k}}
\def\ABJMsup{\mathtt{ABJM}}
\def\MBJ{\Mcal^{\ABJMsup}}
\def\MBJk{\MBJ_{k,2k}}

\def\RelBJ{\Rel_{\ABJMsup}}
\def\ReloBJ_#1{\Rel^{\intsup}_{\ABJMsup}(#1)}

\def\BJG{\mathbf{\G}^{\ABJMsup}}
\def\BJGs{\BJG_{k,2k}}
\def\BJfs{\bm{f}^{\ABJMsup}_{k,2k}}
\def\brkk{\brx{2k}}
\def\BJO{\Tcomp{O}}
\def\JIS{\r}
\def\QuadT{\Tcomp{\bm{q}}}
\def\sech{\operatorname{sech}}
\def\tanh{\operatorname{tanh}}
\def\Circ{\Ccal}
\def\BJDO{\Ocomp{D}}
\def\BJa{a}
\def\BJb{b}
\def\BJc{c}
\def\BJd{d}
\def\BJx{x}
\def\BJA{A}
\def\BJB{B}
\def\BJC{C}
\def\BJD{D}
\def\BJX{X}
\def\CircX{\Circ_{|\BJx|}}
\def\CircB{\Circ_{\BJb}}
\def\CircD{\Circ_{\BJd}}
\def\CircC{\Circ_{|\BJc|}}
\def\invCX{\CircX^{-1}}
\def\invCB{\CircB^{-1}}
\def\invCC{\CircC^{-1}}
\def\invCD{\CircD^{-1}}
\def\BJDofr{\rayT(\r)}
\def\BJDofx(#1){\rayT(#1)}
\def\invCXj{\Circ_{\BJx_j}^{-1}}
\def\GIS{G^{\mathtt{Ising}}}
\def\Mijll{(\bdx_i-\bdx_j)^2}
\def\Mxxll#1#2{(\bdx_{#1}-\bdx_{#2})^2}
\def\ssgen{single-step generic\xspace}
\def\msgen{fully generic\xspace}

\def\twondsup{\mathtt{2^\partial}\text{-}\mathtt{nd}}
\def\Grnd{\Gr^{\twondsup}_{\geq0}}
\def\Grndxx#1#2{\Gr^{\mathtt{(#1,#2)}^\partial\text{-}\mathtt{nd}}_{\geq0}}
\def\GrndAB{\Gr^{\text{\scalebox{\ttscl}{$(a,b)^\partial$}}\text{-}\mathtt{nd}}_{\geq0}}

\def\APMnoacr{almost perfect matching\xspace}
\def\APM{APM\xspace}
\def\APMs{APMs\xspace}
\def\Omfg{(\Om,\ff,\f)}
\def\Omfgij{(\Om,\bdf_i,\bdf_j)}
\def\markings{\bm{\marking}}
\def\musOmfg{\markings^{\fgsup}_{\Om}}
\def\musOmee{\markings^{\e_1,\e_2}_{\Om}}
\def\dmufg{d_\mu}
\def\dOmee{d^{\;\e_1,\e_2}_\Om}
\def\ijsepsup{i\ssep j}
\def\fgijsup{\bdf_i\ssep\bdf_j}
\def\hasnofloat{is boundary-connected\xspace}
\def\hasnofloatAPM{\hasnofloat and admits an \APM}

\def\fo{\ffout}

\def\xdout{\xd}

\def\xdoutll{\xdout_{\la,\lat}}
\def\xToutll{\xTout_{\la,\lat}}
\def\xdoutllp{\xdout_{\la,\lat}'}

\def\Stack{\operatorname{Glue}}

\def\Tmark{\operatorname{T}}
\def\Emark{\emptyset}

\def\Path{\zzpath}
\def\Pathdir{\Path}
\def\onebdfour{p}
\def\onebdfoursum{p(\f)}
\def\DeltaGfiww{\Deltaw_{\Gfiww}}
\def\DeltaGfbb{\Deltab_{\Gfbb}}

\def\trT{\Tcomp{\xs}}

\def\MPkdk{\MPop_{k,2k}}

\def\mtL_#1{\Lcal(#1)}
\def\perfsup{\mathtt{perf}}
\def\MPerfkdk{\Mcal_{k,2k}^{\perfsup}}
\def\MPerff{\Rel_{\fap}^{\perfsup}}
\def\APerfkdk{\Acal_{k-2,2k}^{\perfsup}}
\def\APerff{\ddRel_{\ddfperm}^{\perfsup}}
\def\flAPerff{\ddRel_{\ddfperm}^{\perfsup,\flagsup}}
\def\BCFWfkdk{\BCFWf_{k,2k}}

\def\trpi{p}
\def\supA{1}
\def\supB{2}
\def\supC{3}
\def\supS{\s}

\def\RGbf{\bm{\fap}}
\def\RG{\fap}
\def\RGg{g}

\def\RelA{\Rel_{\supA}}
\def\RelB{\Rel_{\supB}}
\def\RelC{\Rel_{\supC}}
\def\RelS{\Rel_{\supS}}

\def\Rclacc{\widetilde}

\def\RXcl{\Rclacc{X}}
\def\RYcl{\Rclacc{Y}}
\def\Wcl{\Rclacc{W}}
\def\RPhicl{\Rclacc{\Phi}}
\def\Relcl{\Rclacc{\Rcal}}
\def\RX{X}
\def\RY{Y}
\def\W{W}
\def\Rel{\Rcal}
\def\RPhi{\Phi}
\def\RXint{\RX^{\intsup}}
\def\RYint{\RY^{\intsup}}
\def\Relint{\Rel^{\intsup}}
\def\Relo_#1{\Relint(#1)}
\def\Wint{\W^{\intsup}}

\def\RXA{\RX_{\supA}}
\def\RXB{\RX_{\supB}}
\def\RXC{\RX_{\supC}}
\def\RXS{\RX_{\supS}}

\def\RXAcl{\RXcl_{\supA}}
\def\RXCcl{\RXcl_{\supC}}
\def\RYAcl{\RYcl_{\supA}}
\def\RYCcl{\RYcl_{\supC}}
\def\WAcl{\Wcl_{\supA}}
\def\WCcl{\Wcl_{\supC}}
\def\RelAcl{\Relcl_{\supA}}
\def\RelCcl{\Relcl_{\supC}}

\def\RYA{\RY_{\supA}}
\def\RYB{\RY_{\supB}}
\def\RYC{\RY_{\supC}}
\def\RYS{\RY_{\supS}}
\def\Rproj{\trpi}
\def\RprojY{\trpi_{\RY}}
\def\RprojX{\trpi_{\RX}}

\def\RprojB{\Rproj_{\RYB}}
\def\RprojC{\Rproj_{\RYC}}
\def\RprojS{\Rproj_{\RYS}}

\def\RprojoB_#1{\Rproj_{\RYoB_{#1}}}
\def\RprojoC_#1{\Rproj_{\RYoC_{#1}}}
\def\RprojoY_#1{\Rproj_{\RYo_{#1}}}

\def\RprojoS_#1{\Rproj_{\RYoS_{#1}}}

\def\Rphitop{\widehat\phi}
\def\Rphibot{\phi}
\def\Rpi{\pi}

\def\ReloA_#1{\Relint_{\supA}(#1)}
\def\ReloB_#1{\Relint_{\supB}(#1)}
\def\ReloC_#1{\Relint_{\supC}(#1)}
\def\ReloS_#1{\Relint_{\supS}(#1)}
\def\RXo_#1{\RXint(#1)}
\def\RXocl_#1{\RXcl(#1)}
\def\RYo_#1{\RYint(#1)}
\def\RXoA_#1{\RXint_{\supA}(#1)}
\def\RXoAcl_#1{\RXcl_{\supA}(#1)}
\def\RXoB_#1{\RXint_{\supB}(#1)}
\def\RXoC_#1{\RXint_{\supC}(#1)}
\def\RXoS_#1{\RXint_{\supS}(#1)}
\def\RYoA_#1{\RYint_{\supA}(#1)}
\def\RYoB_#1{\RYint_{\supB}(#1)}
\def\RYoC_#1{\RYint_{\supC}(#1)}
\def\RYoS_#1{\RYint_{\supS}(#1)}
\def\Wo_#1{\Wint(#1)}
\def\WoA_#1{\Wint_{\supA}(#1)}
\def\WoC_#1{\Wint_{\supC}(#1)}

\def\amtilingA{an $\RelA$-tiling\xspace}
\def\amtilingB{an $\RelB$-tiling\xspace}
\def\amtilingC{an $\RelC$-tiling\xspace}
\def\amtilingBJ{an $\RelBJ$-tiling\xspace}

\def\mtilingsB{$\RelB$-tilings\xspace}
\def\mtilingsC{$\RelC$-tilings\xspace}

\def\WA{\W_{\supA}}
\def\WC{\W_{\supC}}

\def\RGraph{\operatorname{G}}
\def\RPhiA{\RPhi_{\supA}}
\def\RPhiC{\RPhi_{\supC}}
\def\RPhiAcl{\RPhicl_{\supA}}
\def\RPhiCcl{\RPhicl_{\supC}}

\def\supX#1{%
 \ifcase#1\relax
 \or \supA %
 \or \supB %
 \or \supC %
 \fi
}

\def\Rtila#1{\hyperref[Rtiling1]{(\supX{#1}a)}}
\def\Rtilb#1{\hyperref[Rtiling2]{(\supX{#1}b)}}
\def\Rtilc#1{\hyperref[Rtiling3]{(\supX{#1}c)}}

\def\twoind{fully $2$-\independent}
\def\independent{independent\xspace}

\def\fsepsup{\mathtt{2}\text{-}\mathtt{sep}}
\def\Grfsep{\Gr^{\fsepsup}_{\geq0}}
\def\ddfsepsup{\mathtt{2}\text{-}\mathtt{ind}}
\def\Grddfsep{\Gr^{\ddfsepsup}_{\geq0}}
\def\BNDfsep{\BND^{\fsepsup}}
\def\bdconn{boundary-connected\xspace}

\def\AAiso{\Rphibot}
\def\AAshift{\overline{\Rphibot}}
\def\AAforg{\Rpi}
\def\Jij{J_{i,j}}
\def\Jijc{\comp{J_{i,j}}}
\def\rtw{\vec\tau}
\def\ltw{\cev\tau}
\def\omf{\omega_{\fap}}

\def\tripsup{\mathtt{tripod}}
\def\suptripsup{^{\tripsup}}

\def\epsww_#1{\epsK_{#1}\suptripsup}
\def\epsbb_#1{\epsK_{#1}\suptripsup}
\def\epswwf_#1{\fouracc{\epsK}_{#1}^{\,\;\tripsup}}
\def\epsbbf_#1{\fouracc{\epsK}_{#1}^{\,\;\tripsup}}
\def\epsKfp{\epsKf\vphantom{\epsK}^{\;\prime}}
\def\wtfp{\wtf\vphantom{\wt}'}
\def\nfloat{\operatorname{float}}

\def\dwcor{d^{\circ}}
\def\Amat{A}
\def\Atmat{\tilde A}
\def\regangle{\alpha^{\csop}}
\def\epsOmPathw{(-1)^{\<\Om,\Cut\>}}
\def\source{s}
\def\tink{t}
\def\Gro{\Gr^{\diamond}}
\def\perfY{y}
\def\perfZ{z}
\def\perfYT{\Tcomp{\perfY}}
\def\perfYO{\Ocomp{\perfY}}
\def\perfZT{\Tcomp{\perfZ}}
\def\perfZO{\Ocomp{\perfZ}}

\def\ABJW{ABJM-BCFW\xspace}
\def\ddr{\ddot r}
\def\rinf{\r_\infty}
\def\ddsub{{2\times2}}
\def\bzerodd{\bzero_\ddsub}

\def\Matsup{\mathtt{Mat}}
\def\lalatsMAT{\lalatbf_n^{\Matsup}}
\begin{document}
\numberwithin{equation}{section}

\thinmuskip=3mu 

\title{Amplituhedra and origami, I: tree level}
\author{Pavel Galashin}
\address{Department of Mathematics, Cornell University, Ithaca, NY 14850, USA}
\email{{\href{mailto:galashin@cornell.edu}{galashin@cornell.edu}}}
\address{Department of Mathematics, University of California, Los Angeles, CA 90095, USA}
\email{{\href{mailto:galashin@math.ucla.edu}{galashin@math.ucla.edu}}}
\thanks{P.G.\ was supported by the National Science Foundation under Grants No.~DMS-2046915 and No.~DMS-2554159.}
\date{\today}



\begin{abstract}
We establish a precise correspondence between points of the $m=4$ tree momentum amplituhedron and origami crease patterns. 
As an application, we prove that the BCFW (Britto--Cachazo--Feng--Witten) cells triangulate the $m=4$ tree amplituhedron both in momentum space and in momentum-twistor space. As another application, we show that every nondegenerate weighted planar bipartite graph $\Gamma$ admits a t-embedding, i.e., an embedding of the planar dual of $\Gamma$ such that the sum of angles of white (equivalently, black) faces around each vertex is equal to $\pi$.
\end{abstract}

\maketitle

\setcounter{tocdepth}{1}
\hypersetup{bookmarksdepth=subsection}
\tableofcontents

\begin{figure}
\includegraphics[scale=1.3]{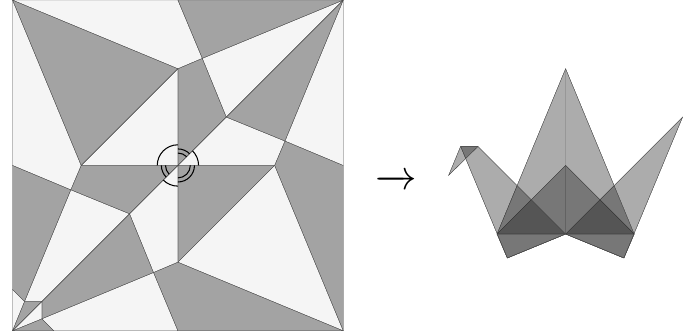}
\caption{\label{fig:orig}An origami crease pattern $\Tcal(\GD)$ (left) and its image $\Ocal(\GD)$ under the origami map (right). 
The \Kawangle condition is satisfied at each interior vertex of $\Tcal(\GD)$, which allows it to be folded consistently. 
Figure reproduced from~\cite[Figure~1]{Hull}.}
\end{figure}

\section*{Introduction}%

The goal of this work is to relate the physics of scattering amplitudes in $\Ncal=4$ supersymmetric Yang--Mills theory to origami crease patterns. In this paper, we focus on the case of \emph{tree} amplitudes; in a companion paper~\justpaptwo, we extend our results to the case of \emph{loop} amplitudes.

We begin with a brief overview of the physics side. 
Given $n$ gluons with incoming momenta $\Pmom_1,\Pmom_2,\dots,\Pmom_n$, each with either positive or negative helicity $h_i\in\{\pm1\}$, one can calculate the associated \emph{scattering amplitude} $A(\Pmom_1,\Pmom_2,\dots,\Pmom_n)$ using Feynman diagrams. A much more efficient way is to use the \emph{BCFW recursion}~\cite{BCFW} to compute $A(\Pmom_1,\Pmom_2,\dots,\Pmom_n)$ inductively. 
 The intermediate steps of the recursion are substantially more complicated than the final answer, and involve a large amount of cancellation and spurious singularities. Moreover, there are multiple ways to run the recursion, resulting in nontrivial identities between the terms. It was proposed in~\cite{AHT}, building on the ideas of Hodges~\cite{Hodges}, that the individual terms of the recursion correspond to ``volumes'' of a collection of pieces that ``triangulate'' a geometric object called the \emph{amplituhedron}, and that the scattering amplitude represents the total ``volume'' of the amplituhedron. They defined the amplituhedron for all integer parameters $m$, with the case $m=4$ being the one relevant to physics. Throughout this paper, we focus on the $m=4$ amplituhedron.

The original definition~\cite{AHT} of the amplituhedron is set up in \emph{momentum-twistor space}, where the BCFW tiling prediction was recently confirmed in~\cite{ELT,ELPTSBW}. The problem has remained wide open in the momentum space (which is where the particle momenta $\Pmom_1,\Pmom_2,\dots,\Pmom_n$ naturally live).
In our first main result, we provide a simultaneous solution to both problems; see \cref{thm:intro:BCFW,lemma:proj_tiling} for precise statements.
\renewcommand{\thetheoremintro}{A}
\begin{theoremintro}\label{thm:intro:A}
The BCFW cells triangulate the $m=4$ tree amplituhedron, both in momentum space and in momentum-twistor space.
\end{theoremintro}
\noindent In momentum space, the above result applies only to \emph{Mandelstam-positive} (tree) momentum amplituhedra, where the planar Mandelstam variables $(\Pmom_i + \Pmom_{i+1} + \cdots + \Pmom_j)^2$ take positive values. It was conjectured in~\cite{DFLP,FL} that such momentum amplituhedra exist, and we confirm this conjecture in \cref{thm:intro:Mand_pos_exists}. Mandelstam variables appear in the denominators of scattering amplitudes, so it is natural to impose their positivity when defining the amplituhedron.

We deduce \cref{thm:intro:A} from a more general BCFW tiling result (\cref{thm:f_triang}) for \emph{ambient} momentum amplituhedra. We develop \emph{T-duality} for ambient amplituhedra in \cref{sec:TREE}, extending the results of~\cite{abcgpt,LPW}, and use it to show that the BCFW tiling conjecture for ambient momentum amplituhedra is equivalent to that for ambient momentum-twistor amplituhedra. This allows us to deduce \cref{thm:intro:A} for both types of amplituhedra simultaneously; see \cref{lemma:ddf_triang,lemma:proj_tiling}. 
In particular, the BCFW tiling results of~\cite{ELT,ELPTSBW} follow from \cref{thm:f_triang} (but not vice versa; see \cref{rmk:ELT_ELPTSBW}). We remark that our proof strategy is very different from that of~\cite{ELT,ELPTSBW}.
For example, our approach allows us to show the equivalence of the \emph{linear projection}~\cite{AHT} and the \emph{sign flip}~\cite{AHTT} definitions of the amplituhedron, confirming a conjecture of~\cite{AHTT}. 
On the other hand, the \emph{cluster adjacency} results of~\cite{ELPTSBW} remain quite mysterious from our point of view. 

An \emph{origami crease pattern}~\cite{Hull,KLRR,CLR1} is a $2$-dimensional cell complex $\Tcal(\GD)$ in the plane such that the faces of $\Tcal(\GD)$ are convex polygons colored black and white, and such that the \emph{\Kawangle condition}~\cite{Kawasaki,Hull} is satisfied: for each interior vertex $\f$ of $\Tcal(\GD)$, the sum of angles of white faces around $\f$ equals the sum of angles of black faces around $\f$, with both sums equal to $\pi$. See \figref{fig:orig}(left) and \cref{dfn:intro:t_imm}. The \Kawangle condition ensures that the crease pattern is \emph{flat-foldable}, meaning there exists a piecewise-linear map of the plane (called the \emph{origami map}) that is an isometry on each face of $\Tcal(\GD)$ preserving (resp., reversing) the orientation of all white (resp., black) faces of $\Tcal(\GD)$. See \figref{fig:orig}(right). Such origami crease patterns were studied under the name \emph{circle patterns} in~\cite{KLRR} and \emph{t-embeddings} in~\cite{CLR1}, generalizing the work of Smirnov~\cite{Smirnov_cluster} and Chelkak--Smirnov~\cite{CheSmi} on the Ising model.%

A t-embedding of a \emph{weighted} planar bipartite graph $(\G,\wt)$ is an embedded origami crease pattern $\Tcal(\GD)$ planar dual to $\G$ 
 such that the edge weights $\wt$ are gauge equivalent to the Euclidean geometric edge weights $\wtT$ associated with $\Tcal$; see~\eqref{eq:intro:wtT_dfn}--\eqref{eq:intro:gauge}. One may view a t-embedding as an explicit geometric realization of $(\G,\wt)$ in the plane. This is particularly useful for studying convergence questions on sequences of such graphs. In particular, rich asymptotic results on the convergence of various discrete observables to conformally invariant limits (such as the Gaussian Free Field) have been obtained for t-embeddings and related objects in~\cite{Chelkak_ICM,Chelkak_s_emb,KLRR,CLR1,CLR2,BNR}. 

However, one fundamental problem in the theory of t-embeddings has not been addressed: their existence is not known in general. 
Only the special case when $\G$ has $n=4$ boundary vertices was considered in~\cite{KLRR} (although see \cref{rmk:KLRR}).
\renewcommand{\thetheoremintro}{B}
\begin{theoremintro}\label{thm:intro:B}
Any sufficiently nondegenerate weighted planar bipartite graph $(\G,\wt)$ admits a t-embedding.
\end{theoremintro}
\noindent The nondegeneracy condition here is that $\helmin(\G)\geq2$ for a certain \emph{surplus} function $\helmin$ introduced in~\eqref{eq:surplus_dfn}. This class of graphs includes \emph{reduced} planar bipartite graphs of~\cite{Pos}; see \cref{rmk:reduced_helmin2}. For graphs $\G$ satisfying a weaker assumption $\helmin(\G)\geq1$ (i.e., such that every edge of $\G$ appears in an almost perfect matching), we introduce \emph{\wtembs} and prove their existence in~\justpaptwo. As we explain in~\justpaptwo, the condition $\helmin(\G)\geq2$ (resp., $\helmin(\G)\geq1$) is necessary in order for $\G$ to admit a t-embedding (resp., a \wtemb). 

Our proof of Theorems~\ref{thm:intro:A} and~\ref{thm:intro:B} relies on a bijection (called the \emph{\oacTITLE}; see \cref{thm:intro:t_imm_vs_triples}) between t-embeddings and points in the momentum amplituhedron. Every edge $\xT(\e)$ of $\Tcal(\GD)$ is a $2$-dimensional vector, and its image $\xO(\e)$ under the origami map is also a $2$-dimensional vector of the same length. Together, they form a $4$-dimensional vector $\Pmom(\e)$ which is \emph{null}, i.e., has zero norm in the Minkowski space $\R^{2,2}$. The null vectors $\Pmom(\bde_1),\Pmom(\bde_2),\dots,\Pmom(\bde_n)$ associated to the boundary edges of $\Tcal(\GD)$ correspond precisely to the particle momenta $\Pmom_1,\Pmom_2,\dots,\Pmom_n$ discussed above. 

We formulate the bijection using the spinor-helicity formalism. Namely, we show that t-embeddings (more precisely, \emph{t-immersions} introduced in \cref{dfn:intro:t_imm}) are in bijection with triples $\la\subset C\subset\latp$, where $\lalat$ is a pair of perpendicular $2$-planes in $\R^n$ satisfying a sign flip condition, and $C\in\Grtnn(k,n)$ is a totally nonnegative $k$-plane in $\R^n$; cf.~\cite{Lus2,Pos}. In the terminology of~\cite{KLRR,CLR1}, the $2$-planes $\la$ and $\lat$ extend to complex-valued \emph{discrete holomorphic functions}~\cite{Kenyon} defined on white and black vertices of $\G$, respectively, which allows one to construct the associated t-embedding. 

One may express~\cite{AHCC,abcgpt} the scattering amplitude as a certain integral over the space of triples $(\la,\lat,C)$ satisfying the above conditions. The null vectors $\Pmom(\e)$ are precisely the parameters associated to the edges of $\G$ in the integrals computing the scattering amplitudes as well as the more general \emph{on-shell functions}; see e.g.~\cite[Equation~(2.39)]{abcgpt}. 
Thus, one can think of the scattering amplitude $A(\Pmom_1,\Pmom_2,\dots,\Pmom_n)$ as a certain integral over the space of origami crease patterns. We leave this direction for future work; see \cref{rmk:omega_form}.

\subsection*{Changelog}
We highlight the major changes compared to the first version of this paper.
{
 \setlength{\leftmargini}{20pt}
\begin{itemize}
\item Replaced the original proof of the \oacTITLE (\cref{thm:intro:t_imm_vs_triples}) 
which relied on the twist map of~\cite{MuSp}
with a much simpler argument using double-dimer model formulae, and extended it to not necessarily reduced planar bipartite graphs (\crefrange{sec:dimer_KS}{sec:proof_main_bij}).
\item Interpreted the momentum amplituhedron map $C\mapsto \PhiLL(C)$ of~\cite{DFLP} in terms of gluing an annular planar bipartite graph representing $\LaLat$ to the boundary of a planar bipartite graph representing $C$
(\cref{fig:intro-annular} and \cref{ssec:annular}).
\item Developed \emph{T-duality} for amplituhedra and used it to prove that the BCFW tiling conjectures for ambient momentum and momentum-twistor amplituhedra are equivalent (\cref{sec:TREE}).
\item Added a proof of equivalence of the sign flip and linear projection definitions of the amplituhedron (\crefi{lemma:proj_tiling}{proj_tiling3}).
\item Moved all previous results on T-duality for planar bipartite graphs to~\justpaptwo.
\item Introduced \emph{T-dual perfect t-embeddings} and proved their existence and uniqueness (\cref{ssec:T_dual_perfect}).
\item Outlined how our results can be extended to study the \emph{ABJM amplituhedron} of~\cite{Huang_Wen} and \emph{s-embeddings} of~\cite{Chelkak_s_emb} (\cref{sec:ABJM}). 
\end{itemize}
}

\section{Main results}\label{sec:main_results}

We give precise statements of our main results. We start by covering some background material; see \cref{sec:backgr} for further details. 
\subsection{Background on total positivity and the dimer model}\label{sec:intro:plabic}
Let $(\G,\wt)$ be a weighted bipartite graph embedded in a disk $\Disk$ with $n$ boundary vertices of degree $1$ denoted $\bdv_1,\bdv_2,\dots,\bdv_n$ in clockwise order. We assume that the edge weights $\wt:\E\to\R_{>0}$ are positive real numbers, where $\E$ is the edge set of $\G$. Let $\GD$ be the planar dual of $\G$, and let $\bdf_1,\bdf_2,\dots,\bdf_n$ be the boundary faces of $\G$ such that $\bdf_i$ is located between $\bdv_i$ and $\bdv_{i+1}$. We always label boundary vertices and faces modulo $n$ so that $\bdv_{i+n}=\bdv_i$ and $\bdf_{i+n}=\bdf_i$.

Let $\Verts$ and $\Faces$ denote the sets of vertices and faces of $\G$, respectively. We write $\Verts = \Vint\sqcup \Vbd$ and $\Faces = \Fint \sqcup \Fbd$, where $\Vbd := \{\bdv_1,\bdv_2,\dots,\bdv_n\}$, $\Fbd := \{\bdf_1,\bdf_2,\dots,\bdf_n\}$, $\Vint :=\Verts\setminus\Vbd$, and $\Fint:=\Faces\setminus\Fbd$.
 We have a natural bijection $\e\mapsto \east$ between $\E$ and the edge set $\East$ of $\GD$.

Throughout, we denote $\brn:=\{1,2,\dots,n\}$ and ${\brn\choose k}:=\{J\subset\brn: |J| = k\}$. 

The weighted graph $(\G,\wt)$ carries a dimer model. Namely, an \emph{\APMnoacr (\APM)} of $\G$ is a collection $\Apm\subset\E$ of edges covering each interior (resp., boundary) vertex of $\G$ exactly once (resp., at most once). We let $\APMS(\G)$ be the set of \APMs of $\G$, and for $\apm\in\APMS(\G)$, let
$\partial\Apm\subset\brn$ be the set of indices $i$ such that either $\bdv_i$ is black and used in $\Apm$ or $\bdv_i$ is white and not used in $\Apm$. 
There is an integer $k=k(\G)$ such that $|\partial\Apm| = k$ for any $\Apm\in\APMS(\G)$, and we say that $\G$ is \emph{of type $(k,n)$}. For $I\in{\brn\choose k}$, 
 we set $\APMSGbd(I):=\{\apm\in\APMS(\G):\ \partial\apm=I\}$ and 
\begin{equation}\label{eq:intro:Delta_G_wt}
 \Delta_I(\G,\wt) := \sum_{\Apm\in\APMSGbd(I)} \wt(\Apm),\quad\text{where}\quad \wt(\Apm) := \prod_{\e\in\Apm} \wt(\e).
\end{equation}
\noindent Throughout this paper, $(\G,\wt)$ denotes a weighted planar bipartite graph of type $(k,n)$ embedded in a disk. For the rest of this section, we assume that $\G$ admits at least one \APM.

Let $\BV$, $\BVint$, $\WV$, $\WVint$ denote the sets of black, interior black, white, and interior white vertices of $\G$, respectively. 
Given a subset $\Rg\subset\Verts$, we let $\RgWV:=\Rg\cap\WV$ and $\RgBV:=\Rg\cap\BV$. 
We say that $\Rg\subset\Verts$ is \emph{\wclosed} if $\RgBV\subset\BVint$ and $\NeighG(\RgBV)\subset\RgWV$, where $\NeighG(\RgBV)\subset\WV$ denotes the neighborhood of $\RgBV$ in $\G$. We let $\WNEIg(\G)$ be the set of \wclosed subsets in $\G$. 
For $\Rg\in\WNEIg(\G)$, we set $\helW(\Rg):=|\RgWV| - |\RgBV|$. We set $\WNEI(\G):=\{\Rg\in\WNEIg(\G)\mid \RgBV\neq\emptyset\}$. 
We similarly define the set $\BNEIg(\G)$ of \emph{\bclosed} subsets and set $\helB(\Rg):=|\RgBV| - |\RgWV|$ for $\Rg\in\BNEIg(\G)$. We denote
\begin{equation}\label{eq:DIM:helWmin_dfn}
 \helWmin(\G) := \min\{\helW(\Rg)\mid \Rg\in\WNEI(\G)\},
 \quad 
 \helBmin(\G) := \min\{\helB(\Rg)\mid \Rg\in\BNEI(\G)\}, \quad\text{and}
\end{equation}
\begin{equation}\label{eq:surplus_dfn}
 \helmin(\G):=\min(\helWmin(\G),\helBmin(\G)).
\end{equation}
\noindent The quantity $\helmin(\G)$ is called the \emph{surplus} of $\G$; see~\cite[Section~1.3]{Lovasz_Plummer}.
Thus, for example, we have $\helmin(\G)\geq0$ if and only if $\G$ admits an \APM and $\helmin(\G)\geq1$ if and only if every edge of $\G$ appears in some \APM; see \cref{ssec:surplus} for further discussion. 
By \cref{rmk:reduced_helmin2} below, each (connected) \emph{reduced} planar bipartite graph of~\cite{Pos} satisfies $\helmin(\G)\geq1$, and each such graph satisfies $\helmin(\G)\geq2$ after contracting all edges incident to degree-$2$ vertices.

The \emph{Grassmannian} $\Gr(k,n)$ is the space of linear $k$-dimensional subspaces of $\R^n$; equivalently, it is the space of row spans of $k\times n$ matrices of rank $k$. We usually identify such matrices with their row spans. For a $k\times n$ matrix $C\in\Gr(k,n)$, we denote by $\Delta_I(C)$ its maximal $k\times k$ minor with column set $I\in{\brn\choose k}$. The \emph{totally nonnegative Grassmannian} $\Grtnn(k,n)$ is the subset of $\Gr(k,n)$ consisting of $k\times n$ matrices $C$ such that all nonzero maximal minors $\Delta_I(C)$ have the same sign. 
When $\G$ admits an \APM, 
 there exists a unique $k$-plane $C\in \Grtnn(k,n)$, denoted $\Meas(\G,\wt)$, such that $\Delta_I(\G,\wt) = \Delta_I(C)$ for all $I\in{\brn\choose k}$. 

We say that $C=\mat[C_1|C_2|\cdots|C_n]\in\Gr(k,n)$ is \emph{\twonondeg} if for all $i\in\brn$, 
$\rank\mat[C_i|C_{i+1}] = 2$ and $\rank\mat[C_{i+2}|\dots|C_{i+n-1}] = k$. Here, we set $C_{i+n}=(-1)^{k-1}C_i$ for all $i\in\brn$. Similarly, a graph $\G$ is \emph{\twonondeg} if for all $i\in\brn$, $\G$ admits \APMs $\apm_+,\apm_-$ such that $i,i+1\in\partial\apm_+$ and $i,i+1\notin \partial\apm_-$. 
 It follows from standard properties of the boundary measurement map that $\G$ is \twonondeg if and only if for some (equivalently, any) $\wt:\Edges\to\Rtp$, $C:=\Meas(\G,\wt)$ is \twonondeg. 
We denote
\begin{equation}\label{eq:intro:Grnd_dfn}
 \Grnd(k,n):=\{C\in\Grtnn(k,n)\mid C\text{ is \twonondeg}\}.%
\end{equation}

\subsection{T-embeddings and t-immersions}\label{sec:intro:t_emb}
We consider a $2$-dimensional cell complex $\SuppGD$ embedded in a disk $\Disk$, with vertex set $\Faces$, edge set $\East$, and face set $\Vint$.%

Given a map $\Tcal:\SuppGD\to\C$ such that the image of each face of $\GD$ is a convex polygon, 
we refer to the images $\Tcal(\Faces),\Tcal(\East),\Tcal(\Vint)$ as the \emph{vertices}, \emph{edges}, and \emph{faces of $\Tcal(\GD)$}, respectively. The faces of $\Tcal$ are naturally colored black and white. We also define the geometric edge weights
\begin{equation}\label{eq:intro:wtT_dfn}
 \wtT(\e) := |\Tcal(\f_1)-\Tcal(\f_2)| \quad\text{for all $\e\in\E$ with $\east = \{\f_1,\f_2\}\in\East$.}
\end{equation}
We say that $\wt,\wt':\E\to\R_{>0}$ are \emph{weakly gauge equivalent} if there exists a function $\gauge: \Verts\to \R_{>0}$ such that 
\begin{equation}\label{eq:intro:gauge}
	\wt'(\e) = \gauge(\wv)\wt(\e)\gauge(\bv) \quad\text{for any edge $\e=\{\wv,\bv\}\in\E$.}
\end{equation}

\begin{definition}\label{dfn:intro:sum_angles}
For $\f\in\Faces$, we let $\sumwT(\f)$ (resp., $\sumbT(\f)$) denote the sum of angles at $\Tcal(\f)$ of all white (resp., black) faces of $\Tcal(\GD)$ incident to $\Tcal(\f)$. 
\end{definition}

\begin{definition}[\cite{KLRR,CLR1}]\label{dfn:intro:t_imm}
Assume that $\G$ is connected. 
A \emph{t-immersion} of $(\G,\wt)$ is 
a map 
$\Tcal:\SuppGD \to \C$
 such that the following conditions are satisfied.
\begin{enumerate}[label=(TE\arabic*)]
\item\label{intro:t_imm_straight_convex} $\Tcal(\east)$ is a straight line segment of nonzero length for each $\east\in\East$.
\item\label{intro:t_imm_orientation} \emph{Immersion condition}: for each face $\v\in\Vint$ of $\GD$ of degree $d$, $\Tcal(\v)$ is a convex $d$-gon\footnote{
We slightly relax this condition when $\G$ has parallel edges; see \cref{dfn:DIM:R1}.
} and the restriction of $\Tcal$ to the face $\v$ is an orientation-preserving homeomorphism.
\item\label{intro:t_imm_gauge} The edge weights $\wt$ and $\wtT$ are weakly gauge equivalent.
\item\label{intro:t_imm_angles_int} \emph{\Kawangle condition}: for each $\f\in\Fint$,
\begin{equation}\label{eq:intro:t_imm_angles_int}
 \sumwT(\f) = \sumbT(\f) = \pi.
\end{equation}
\item\label{intro:t_imm_angles_bdry} \emph{Boundary angle condition}: for each $\bdf_i\in\Fbd$,
\begin{equation}\label{eq:intro:t_imm_angles_bdry}
 0<\sumwT(\bdf_i)<\pi \quad\text{and}\quad 0<\sumbT(\bdf_i) < \pi.
\end{equation}
\end{enumerate}
A \emph{t-embedding} of $(\G,\wt)$ is a t-immersion of $(\G,\wt)$ which is injective as a map $\SuppGD\to\C$.
\end{definition}
See \figref{fig:t-imm}(left) for an example of a t-embedding and \figref{fig:t-imm}(right) for an example of a t-immersion that is not a t-embedding. One can show that any t-immersion whose boundary is a \emph{simple} (i.e., non-self-intersecting) closed polygonal chain is a t-embedding; see \cref{lemma:Jordan_curve}.

\begin{figure}
\def\scl{1.95}
\begin{tabular}{cc}
\includegraphics[scale=\scl]{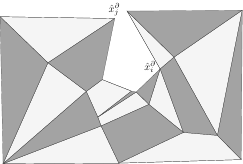}
&
\includegraphics[scale=\scl]{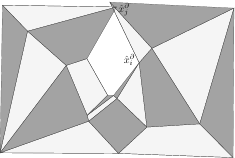}
\end{tabular}
 \caption{\label{fig:t-imm} A t-embedding (left) and a t-immersion (right).}
\end{figure}

A t-immersion may be alternatively viewed as a \emph{circle pattern}~\cite{KLRR}: one can draw a circle centered at each vertex $\Tcal(\f)$ of $\Tcal(\GD)$ in such a way that for each face $\xT(\v)$ of $\Tcal(\GD)$, the circles centered at all vertices of the polygon $\Tcal(\v)$ have a common intersection point. See e.g.~\cite[Figure~2]{KLRR}.

\subsection{\OACTITLE}
Let 
\begin{equation}\label{eq:intro:LALATS}
 \lalats:=\{\lalat\in\Gr(2,n)\times\Gr(2,n)\mid \la \perp \lat\}
\end{equation}
denote the \emph{kinematic space}, i.e., the space of pairs of perpendicular $2$-planes in $\R^n$. The condition $\la\perp\lat$ is referred to as \emph{momentum conservation}.

Fix $2\leq k\leq n-2$. Given a pair $\lalat\in\lalats$ of $2$-planes (viewed as $2\times n$ matrices), we denote their columns by $\la_1,\la_2,\dots,\la_n\in \R^2$ and $\lat_1,\lat_2,\dots,\lat_n\in\R^2$, respectively. We extend these sequences to $(\la_i)_{i\in\Z}$ and $(\lat_i)_{i\in\Z}$ periodically using the \emph{twisted cyclic symmetry}: $\la_{i+n} = (-1)^{k-1}\la_i$ and $\lat_{i+n} = (-1)^{k-1}\lat_i$ for all $i\in\Z$. For $i,j\in\Z$, we introduce the brackets 
\begin{equation*}%
 \brla<i,j>:=\det\mat[\la_i|\la_j] \quad\text{and}\quad \brlat[i,j]:=\det\mat[\lat_i|\lat_j]. 
\end{equation*}

When the columns of a $2\times n$ matrix $\la$ are all nonzero and $\la_i$ is not antiparallel to $\la_{i+1}$ (i.e., $\la_{i+1}\notin\R_{<0}\cdot \la_i$) for all $i\in\brn$, we define %
\begin{equation}\label{eq:intro:wind}
\wind(\la) := \sum_{i=1}^n \Arg_{(-\pi,\pi]}(\la_i,\la_{i+1})
\end{equation}
to be the total turning angle of the column vectors of $\la$ around the origin in the counterclockwise direction, where $\Arg_{(-\pi,\pi]}(\la_i,\la_{i+1})$ denotes the angle between $\la_i$ and $\la_{i+1}$. 
 Since $\la_{n+1}=(-1)^{k-1}\la_1$, $\wind(\la)$ equals $(k-1)\pi$ modulo $2\pi$. 
 Let 
\begin{equation}\label{eq:intro:LALAK}
 \lalak := \left\{\lalat\in\lalats\middle| \text{\begin{tabular}{l}
 $\brla<i,i+1> >0$ and $\brlat[i,i+1]>0$ for all $i\in\brn$,\\
 $\wind(\la) = (k-1)\pi$, and $\wind(\lat) = (k+1)\pi$
 \end{tabular}
 }\right\}.
\end{equation}
The latter two conditions are usually referred to as $\lalat$ having \emph{correct sign flips}~\cite{AHTT}. 

A space of interest in particle physics is the space of \emph{triples}
\begin{equation}\label{eq:intro:TRIPLES}
 \TRIPLES := \{(\la,\lat,C)\in\lalak\times\Grtnn(k,n)\mid \la \subset C \subset \lat^\perp\}.
\end{equation}

As we discuss in \cref{rmk:intro:Minkowski} below, a t-immersion $\xT:\SuppGD\to\C$ may be naturally viewed as a $2$-dimensional piecewise-linear surface $\xd:\SuppGD\to\Rdd$ inside the $4$-dimensional Minkowski spacetime\footnote{
Here, the ``space'' and ``time'' components of $\Rdd$ represent the origami crease pattern and its folding, respectively. 
 Following a suggestion of Stanislav Smirnov, we refer to these as \emph{kami} and \emph{origami} planes (Japanese for ``paper'' and ``folded paper''), respectively.} $\R^{2,2}$ with signature $(+,+,-,-)$. 
The isometry group $\SO(2,2)$ of \emph{Lorentz transformations} therefore acts naturally on the space of t-immersions. This corresponds to the left action of $\SL_2(\R)\times \SL_2(\R)$ on pairs $\lalat$ of $2\times n$ matrices (which does not change $\lalat$ as an element of $\Gr(2,n)\times\Gr(2,n)$). 

Our first main result is a correspondence between t-immersions and elements of $\TRIPLES$. 
\begin{theorem}[\OACTITLE]\label{thm:intro:t_imm_vs_triples}
Assume that $\G$ is connected and satisfies $\helmin(\G)\geq2$. 
Let $C:=\Meas(\G,\wt)\in\Grtnn(k,n)$. Then t-immersions of $(\G,\wt)$ (viewed up to rescaling and Lorentz transformations) are in bijection with pairs $\lalat\in\lalak$ such that $\la \subset C \subset \lat^\perp$. %
\end{theorem}

\noindent See \cref{sec:intro:holomorphic} for the explicit construction of the bijection in \cref{thm:intro:t_imm_vs_triples} using the theory of discrete holomorphic functions on $\G$ (cf.~\cite{CLR1,KLRR}). 
Given $C = \Meas(\G,\wt)\in\Grtnn(k,n)$, 
this construction associates a map $\Tll:\SuppGD\to\C$ to any pair $\lalat\in\lalats$ satisfying $\la\subset C \subset \latp$. 
The bulk of the proof is devoted to showing that if 
we additionally have $\lalat\in\lalak$ then $\Tll$ is indeed a t-immersion and, in particular, preserves the orientations of all faces of $\GD$. 

\subsection{Momentum amplituhedron}\label{sec:intro:mom}
Let $\Grtp(k,n):=\{C\in\Grtnn(k,n)\mid \Delta_I(C)>0\text{ for all $I$}\}$. 
For a subspace $C\in\Gr(k,n)$, we let $C^\perp\in\Gr(n-k,n)$ be the orthogonal complement of $C$. Let $\alt:\Gr(k,n)\to\Gr(k,n)$ be the map on matrices changing the sign of every second column. 
 We have $C\in\Grtnn(k,n)$ if and only if $\alt(C^\perp)\in\Grtnn(n-k,n)$; see \cref{sec:backgr:cyc}. 

Let $2\leq k\leq n-2$. Define
\begin{equation}\label{eq:intro:LaLat}
\LaLak := \alt(\Grtp(n-k+2,n)) \times \Grtp(k+2,n).
\end{equation}
Fix $\LaLat\in\LaLak$. Define the \emph{momentum amplituhedron map}
\begin{equation}\label{eq:intro:PhiLL_dfn}
 \PhiLL: \Grtnn(k,n)\to \lalats,\quad C \mapsto (C\cap \La,C^\perp\cap \Lat).
\end{equation}
The intersections $\la := C\cap \La$ and $\lat := C^\perp\cap \Lat$ are always $2$-dimensional. 
Moreover, if $C\in\Grnd(k,n)$ then $\PhiLL(C)\in\lalak$; see~\cite{DFLP} and \cref{prop:momLL_basic}. 

It follows that for any $C\in\Grnd(k,n)$, there exists a pair $\lalat\in\lalak$ (namely, $\lalat := \PhiLL(C)$) such that $\la\subset C\subset \latp$. Combining this with \cref{thm:intro:t_imm_vs_triples}, we get the following result.
\begin{corollary}
Assume that $\G$ is connected, \twonondeg, and satisfies $\helmin(\G)\geq2$. Then $(\G,\wt)$
 admits a t-immersion for all $\wt:\Edges\to\Rtp$.
\end{corollary}
\noindent This provides a proof of \cref{thm:intro:B} modulo the difference between t-immersions and t-embeddings. See \cref{cor:intro:t_emb_exists} for a full solution for t-embeddings.

\begin{figure}
\def\inputscl{1.7}
\setlength{\tabcolsep}{1pt}
\begin{tabular}{ccc}
\includegraphics[scale=\inputscl]{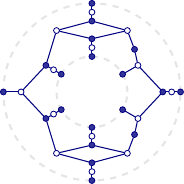}
&
\includegraphics[scale=\inputscl]{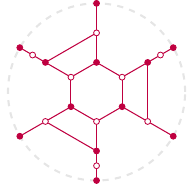}
&
\includegraphics[scale=\inputscl]{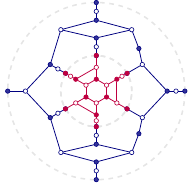}
\\
(a) Annular graph $\Gbot$ & (b) $\G$ & (c) $\Gf = \Stack(\Gbot,\G)$
\end{tabular}
 \caption{\label{fig:intro-annular} 
When $\LaLat$ is \Arep, the momentum amplituhedron map $\PhiLL$ corresponds 
to gluing an annular graph $\Gbot$ to the boundary of $\G$; see \cref{ssec:annular}.
}
\end{figure}

In \cref{ssec:annular}, we introduce a subset of \emph{\Arep} pairs $\LaLat$. For such $\LaLat$,
 $\PhiLL$ corresponds to gluing an \emph{annular bipartite graph} $\Gbot$ representing $\LaLat$ to the boundary of $\G$; see \cref{fig:intro-annular}.

\begin{remark}\label{rmk:KLRR}
Conceptually, the operation in \cref{fig:intro-annular} reduces the problem of constructing t-embeddings to the class of graphs $\Gf$ of type $(\kf,\nf)=(2,4)$ (although making this reduction fully rigorous requires significant care). The case $(\kf,\nf)=(2,4)$ was considered in~\cite[Theorem~2]{KLRR}. However, we found this approach problematic as the maps produced by the proof of~\cite[Theorem~2]{KLRR} 
 need not be injective on the faces of $\Gf$ (and may, for example, send all faces of $\G$ to a single point). See \Nref{ex:non-einj} for further discussion. 
\end{remark}

\begin{definition}[{\cite{DFLP}}]\label{dfn:intro:MomLL}
Given fixed $(\La,\Lat)\in\LaLak$,
 the \emph{(tree) momentum amplituhedron} is defined
 as the image
\begin{equation}\label{eq:intro:MomLL_dfn}
 \MomLL := \PhiLL(\Grtnn(k,n)).
\end{equation}
\end{definition}
\begin{remark}\label{rmk:dim_MomLL}
The dimension of $\MomLL$ is $2n-4$: generically, there are $2(n-k) = \dim\Gr(2,\La)$ degrees of freedom to choose $\la\subset\La$ and $2k = \dim\Gr(2,\Lat)$ degrees of freedom to choose $\lat \subset \Lat$, subject to $4$ constraints coming from momentum conservation $\la\perp\lat$. 
\end{remark}

\subsection{Origami map and Mandelstam variables}\label{sec:intro:origami_mand}
Given a t-embedding $\Tcal:\SuppGD\to\C$, the \emph{origami map}
$\Ocal:\SuppGD\to\C$ is
 the unique (up to global shift and rotation) piecewise-linear map such that
\begin{equation}\label{eq:intro:O_preserves_edges}
 |\Ocal(\f_1)-\Ocal(\f_2)| = |\Tcal(\f_1) - \Tcal(\f_2)| \quad\text{for all $\f_1,\f_2\in\SuppGD$ sharing a face of $\GD$},
\end{equation}
and such that $\Ocal$ 
preserves (resp., reverses) the orientations of all white (resp., black) faces of $\GD$.

It was pointed out in~\cite{CLR1} that the origami map ``clearly does not increase Euclidean distances in the complex plane.'' 
Indeed, if $\Tcal(\SuppGD)$ is embedded then for any two points $\pt_1,\pt_2\in\SuppGD$ such that the line segment $[\Tcal(\pt_1),\Tcal(\pt_2)]$ is contained in $\Tcal(\SuppGD)$, we have
\begin{equation}\label{eq:intro:1_Lipschitz}
 |\Ocal(\pt_1)-\Ocal(\pt_2)| \leq |\Tcal(\pt_1)-\Tcal(\pt_2)|,
\end{equation}
since folding a straight line segment cannot increase the distance between its endpoints. We emphasize that~\eqref{eq:intro:1_Lipschitz} \emph{need not hold} when the polygon $\Tcal(\SuppGD)$ is non-convex. For example,~\eqref{eq:intro:1_Lipschitz} is violated for $\pt_1=\bdf_{i}$ and $\pt_2=\bdf_{j}$ in both \figref{fig:t-imm}(left) and \figref{fig:t-imm}(right). 

\begin{definition}\label{dfn:intro:1_Lipschitz}
We say that a t-immersion $\Tcal:\SuppGD\to\C$ has \emph{\MdashTITLE nonnegative boundary} (\emph{\Mdash nonnegative boundary} for short) 
 if~\eqref{eq:intro:1_Lipschitz} holds for any two boundary vertices $\pt_1,\pt_2\in\Fbd$ of $\GD$.
\end{definition}

The \emph{Mandelstam variables} are an important family of functions on the kinematic space $\lalats$.
 For $i,j\in\Z$ satisfying $i+2\leq j\leq i+n-2$, the associated Mandelstam variable is given by
\begin{equation}\label{eq:intro:Mand_dfn}
 \Mand_{i,j}(\la,\lat) := \sum_{i< p<q\leq j} \brla<p,q>\brlat[p,q].
\end{equation}
\begin{definition}
We say that $\lalat\in\lalats$ is \emph{\Mdash nonnegative} (resp., \emph{\Mdash positive}) if $\Mand_{i,j}(\la,\lat) \geq 0$ (resp., $\Mand_{i,j}(\la,\lat) > 0$) for all $i+2\leq j\leq i+n-2$. 
\end{definition}

Mandelstam variables naturally appear in the denominators of scattering amplitudes. It was conjectured in \cite[Section~5]{HZ_notes} that the scattering amplitude should be given as an integral over some $(2n-4)$-dimensional 
slice of the ``positive region'' 
\begin{equation}\label{eq:intro:M_amb_dfn}
\lalapp:=\{\lalat\in\lalak\mid\lalat\text{ is \Mdash positive}\}. 
\end{equation}
We refer to $\lalapp$ as the \emph{ambient momentum amplituhedron}. 
Recall that the momentum amplituhedron $\MomLL$ 
 is contained inside the closure of $\lalak$ and has the correct dimension $2n-4$. However, it was observed in~\cite{DFLP} that for certain choices of $(\La,\Lat)\in\LaLak$, some Mandelstam variables have negative sign on a ``very small region'' of $\MomLL$. 

It turns out that \Mdash nonnegativity of $\lalat$ is equivalent to $\Tll$ having \Mdash nonnegative boundary, 
 as the following lemma demonstrates. %
\begin{lemma}\label{lemma:intro:Mand_vs_norm}
Let $(\la,\lat,C)\in\TRIPLES$. Then for any $i+2\leq j\leq i+n-2$,
\begin{equation}\label{eq:intro:Mand_vs_norm}
 4\Mand_{i,j}(\la,\lat) = |\bdxT_{i} - \bdxT_{j}|^2 - |\bdxO_{i} - \bdxO_{j}|^2.
\end{equation}
Here and below, we denote $\bdxT_i:=\Tll(\bdf_{i})$ and $\bdxO_i:=\Oll(\bdf_i)$.
\end{lemma}
\begin{remark}\label{rmk:intro:Minkowski}
Consider the Minkowski space $\R^{2,2}\cong \C^2$ with norm $\|\Pmom\|^2:=|\PmomT|^2 - |\PmomO|^2$ for $\Pmom=(\PmomT,\PmomO)\in\Rdd$. Taking a t-immersion $\T$ and its associated origami map $\O$ together, we get a map $\TO:\Faces\to\R^{2,2}$, $\f\mapsto(\T(\f),\O(\f))$, such that the image $\TO(\east)$ of every edge of $\GD$ is \emph{null}, i.e., has zero Minkowski norm; cf.~\eqref{eq:intro:O_preserves_edges}. The space of \emph{null polygons} $\Pbdx:=(\bdx_1,\bdx_2,\dots,\bdx_n)$ satisfying $(\bdx_i-\bdx_{i-1})^2=0$ for all $i\in\brn$
is known under the name \emph{dual space}. It makes the \emph{dual conformal invariance} of scattering amplitudes manifest; see e.g. \cite{DHKS}. We discuss this further in relation to t-immersions in \cref{rmk:dual_conformal}.
\end{remark}

We write $\bdx_i=(\bdxT_i,\bdxO_i)\in\Rdd$ so that~\eqref{eq:intro:Mand_vs_norm} becomes $4\Mand_{i,j}(\la,\lat)=(\bdx_i-\bdx_j)^2$ and~\eqref{eq:intro:O_preserves_edges} becomes
\begin{equation}\label{eq:inv_null}
 \left(\TO(\f_1) - \TO(\f_2)\right)^2 = 0 \quad\text{for all $\f_1,\f_2\in\SuppGD$ sharing a face of $\GD$.}
\end{equation}
We say that a map $\xd=(\xT,\xO):\Faces\to\Rdd$ is a \emph{t-immersion} (resp., \emph{t-embedding}) if so is $\xT$, where we always assume that $\xO$ is the origami map obtained from $\xT$ via~\eqref{eq:intro:O_preserves_edges}.

The following result was conjectured in~\cite{DFLP}.
\begin{theorem}\label{thm:intro:Mand_pos_exists}
 For all $2\leq k\leq n-2$, there exist matrices 
$(\La,\Lat)\in\LaLak$
 such that 
$\PhiLL(C)$ is \Mdash nonnegative for all $C\in\Grtnn(k,n)$. 
\end{theorem}
 Our proof of \cref{thm:intro:Mand_pos_exists} relies on the theory of \emph{Temperley--Lieb immanants}~\cite{RhSk,Lam_dimers}. 
Specifically, we show that \cref{thm:intro:Mand_pos_exists} holds for \emph{flag-positive} pairs $(\La,\Lat)$ obtained by choosing a totally positive matrix $M\in\Gtp$ (with positive minors of all sizes) and taking $\La^\perp$ (resp., $\Lat$) to be the span of the first $k-2$ (resp., $k+2$) rows of $M$. 
More generally, we introduce a subset $\LaLaimmnn\subset\LaLak$ of \emph{immanant-nonnegative} pairs $\LaLat$ and show that the conclusion of \cref{thm:intro:Mand_pos_exists} holds for $\LaLat\in\LaLaimmnn$. 
In particular, we show that each \Arep pair $\LaLat$ belongs to $\LaLaimmnn$. 
For $\LaLat\in\LaLaimmnn$, we use \Mdash positivity of $\lalat = \PhiLL(C)$ to deduce that the boundary polygon of $\Tll(\GD)$ is simple, in which case $\Tll$ is a t-embedding. This completes the proof of \cref{thm:intro:B}.

\begin{corollary}\label{cor:intro:t_emb_exists}
Assume that $\G$ is connected, \twonondeg, and satisfies $\helmin(\G)\geq2$. Then $(\G,\wt)$
admits a t-embedding for all $\wt:\Edges\to\Rtp$. 
\end{corollary}

 \subsection{\ORATITLE and BCFW tilings}\label{sec:intro:BCFW}

The \emph{BCFW recursion} discovered in~\cite{BCFW} generates a collection $\BCFWGkn$ of planar bipartite graphs of type $(k,n)$ for each $2\leq k\leq n-2$. A single step of the recursion is shown in \figref{fig:BCFW}(a). We pick an index $i_0\in\brn$, and add either a \emph{white-black bridge} or a \emph{black-white bridge} at boundary vertices $\bdv_{i_0},\bdv_{i_0+1}$. 
For example, in \figref{fig:BCFW}(a), we added a white-black bridge. 
For each collection $(k_L,n_L,k_R,n_R)$ of integers such that $k_L+k_R = k+1$ and $n_L+n_R = n+2$, we take all pairs $(\G_L,\G_R)\in\BCFWGknL\times\BCFWGknR$ and combine them as shown in \figref{fig:BCFW}(a) to obtain a graph $\G=\G_L\otimes\G_R\in\BCFWGkn$. See \cref{sec:BCFW:backgr} for a precise description of the recursion.

After contracting the edges incident to degree-$2$ vertices (see \cref{rmk:reduced_helmin2}), we observe that each graph $\G\in\BCFWGkn$ satisfies $\helmin(\G)\geq2$. The associated \emph{positroid cell} 
$\Pip_{\G} := \{\Meas(\G,\wt)\mid \wt:\Edges\to\R_{>0}\}$ 
is $(2n-4)$-dimensional. We consider the \emph{tiles} $\MomLLprojG:=\PhiLL(\Pip_{\G})$. 

\begin{definition}\label{dfn:intro:triang}
Let $\Gcoll$ be a finite collection of planar bipartite graphs. We say that the tiles $\{\MomLLprojG\mid\G\in\Gcoll\}$ form a \emph{tiling} of $\MomLL$ if the following conditions are satisfied.
\begin{enumerate}[label=(\alph*)]
\item\label{intro:triang1} \emph{Injectivity:} 
For each $\G\in\Gcoll$, the map $\PhiLL$ restricts to a homeomorphism $\Pip_{\G}\xrasim\MomLLprojG$.
\item\label{intro:triang2} \emph{Disjointness:} The images $\{\MomLLprojG\mid\G\in\Gcoll\}$ are pairwise disjoint.
\item\label{intro:triang3} \emph{Surjectivity:} The union $\bigsqcup_{\G\in\Gcoll} \MomLLprojG$ is dense in $\MomLL$.
\end{enumerate}
\end{definition}

\begin{theorem}\label{thm:intro:BCFW}
 For all immanant-nonnegative $(\La,\Lat)\in\LaLaimmnn$, %
 the graphs in $\BCFWGkn$ form a tiling of the momentum amplituhedron $\MomLL$.
\end{theorem}
\begin{remark}
There are multiple ways of running the BCFW recursion, resulting in many possible collections of planar bipartite graphs for given $(k,n)$; cf. \cref{rmk:BCFW_multiple}. It was shown in~\cite{ELT} that one such collection yields a tiling of the \emph{momentum-twistor amplituhedron} $\AZproj$ introduced in~\cite{AHT}, and the result was later extended in~\cite{ELPTSBW} to all other such collections. Both of these problems have been open for the momentum amplituhedron $\MomLL$, and in \cref{thm:intro:BCFW}, we solve both.
\end{remark}

Instead of proving \cref{thm:intro:BCFW} directly, we introduce 
``multivalued'' versions of tilings in 
 \cref{dfn:BCFW:tiling_amb} and prove a 
BCFW tiling result
 for the ambient amplituhedron $\lalapp$ first (\cref{thm:f_triang}). We then deduce the results for momentum and momentum-twistor amplituhedra in \cref{lemma:proj_tiling}.

\begin{figure}
 \hspace{-0.1in}\includegraphics[scale=1.1]{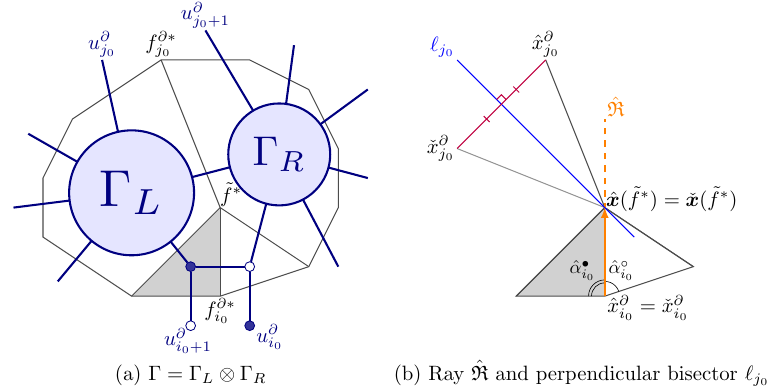}
 \caption{\label{fig:BCFW} Applying a step of the BCFW recursion (left) and of the \ora (right), 
recovering the point $\xTout(\fo)=\xOout(\fo)$ as the intersection of $\line_\jo$ and $\ray$; see \cref{sec:intro:BCFW}.}
\end{figure}

We sketch the proof of \cref{thm:intro:BCFW}. Recall from \cref{thm:intro:t_imm_vs_triples} that a t-immersion $\Tcal=\Tll$ is fully determined by a triple $(\la,\lat,C) \in \TRIPLES$. As we discuss in \cref{sec:BCFW:boundary}, specifying just a pair $\lalat$ of $2\times n$ matrices is equivalent to specifying the following partial information about $\xd=\xll$:
\begin{enumerate}[label=(\roman*)]
\item\label{lalat_bdry1} the kami boundary polygon $\PbdxT = (\bdxT_1,\bdxT_2,\dots,\bdxT_n)$, where $\bdxT_i:=\xT(\bdf_i)$;
\item\label{lalat_bdry2} the origami boundary polygon $\PbdxO = (\bdxO_1,\bdxO_2,\dots,\bdxO_n)$, where $\bdxO_i:=\xO(\bdf_i)$;
\item\label{lalat_bdry3} the boundary angle sums $\sumwT_i:=\sumwT(\bdf_i)$ and $\sumbT_i:=\sumbT(\bdf_i)$ for each $i\in\brn$; cf. \cref{sec:intro:t_emb}.
\end{enumerate} 
We note that when~\itemref{lalat_bdry1} is specified,~\itemref{lalat_bdry2} and~\itemref{lalat_bdry3} determine each other uniquely.

Next, we describe the \emph{\oraTITLE}. 
Let $\G\in\BCFWGkn$. By construction, $\G$ contains an interior face $\fo$ adjacent to two boundary faces $\bdf_{\io}$ and $\bdf_{\jo}$, where $\jo:=\io+n_L-1$. Let $\RmomT\in\C$ be such that the angle between the \emph{folding ray}
 $\ray = \{\bdxT_\io + \r \RmomT\mid \r\geq0\}$ and the edge $\bdxT_{\io-1}-\bdxT_\io$ (resp., $\bdxT_{\io+1}-\bdxT_\io$) is $\sumwT_{\io}$ (resp., $\sumbT_{\io}$). This ray is fully determined by $\lalat$ in view of~\itemref{lalat_bdry3}; see \figref{fig:BCFW}(b). 
Since the face $\bdf_{\io}$ is incident to a unique interior white (resp., black) vertex of~$\G$, the point $\xTout(\fo)$ must lie on the ray $\ray$.%

Next, recall that the origami map $\Ocal$ is defined up to shift and rotation. Let us normalize it so that $\bdxO_{\io} = \bdxT_{\io}$ and $\xOout(\fo) = \xTout(\fo)$. Since $\fo$ is adjacent to $\bdf_{\jo}$, we see from~\eqref{eq:intro:O_preserves_edges} that 
 $|\bdxO_{\jo} - \xOout(\fo)| = |\bdxT_{\jo} - \xTout(\fo)|.$
 In other words, the point $\xOout(\fo) = \xTout(\fo)$ must lie on the perpendicular bisector $\line_\jo$ between the points $\bdxT_{\jo}$ and $\bdxO_{\jo}$, both of which are determined by $\lalat$. See \figref{fig:BCFW}(b).

To summarize, the point $\xOout(\fo) = \xTout(\fo)$ is recovered uniquely from $\lalat$ as the intersection point of the ray $\ray$ and the perpendicular bisector $\line_\jo$. Assuming that the parameters $(k_L,n_L,k_R,n_R)$ are given in advance, finding $\xOout(\fo)=\xTout(\fo)$ is enough to recover the pairs of $2$-planes $(\la_L,\lat_L)$ and $(\la_R,\lat_R)$ encoding the boundary null polygons of $\G_L$ and $\G_R$ via~\itemref{lalat_bdry1}--\itemref{lalat_bdry2}. We then proceed to recover the rest of the t-immersion $\Tcal$ recursively. This verifies part~\itemref{intro:triang1} of \cref{dfn:intro:triang}. 

\begin{remark}
The original proof of~\cite{BCFW} introduces an auxiliary meromorphic function $A(z)$ of a complex variable $z$ such that $A(0)=A(\Pmom_1,\Pmom_2,\dots,\Pmom_n)$ is the scattering amplitude. One then computes the poles $z_I$ of $A(z)/z$ and expresses the residue $A(0)$ of $A(z)/z$ at $0$ as a sum of residues at the poles $z_I$. It turns out that the above description of $\xTout(\fo)$ in terms of the ray $\ray$ and the perpendicular bisector $\line_\jo$ precisely matches the formula for the corresponding pole $z_I$; see \cref{rmk:coincides_BCFW}. Thus, our construction may be viewed as a positive-geometric incarnation of the complex-analytic approach of~\cite{BCFW}.
\end{remark}

To check part~\itemref{intro:triang2} (disjointness), it suffices to show that the parameters $(k_L,n_L,k_R,n_R)$ can also be recovered from $\lalat$. 
 For $\r\in\R$, let $\ray(\r):=\bdxT_\io + \r \RmomT$ as above. For each $\io+2\leq\s\leq\io+n-2$, 
let 
$\ray(\r_\s)$ be the intersection point of the line containing $\ray$ with the perpendicular bisector $\line_\s$ between $\bdxT_\s$ and $\bdxO_\s$.
 Let $\jo$ be the index (with $\io+2\leq\jo\leq\io+n-2$) at which $\r_{\s}$ achieves its minimum positive value. Thus, $\jo$ is determined by $\lalat$. We show in \cref{prop:BCFW:from_lalat_to_r} that---assuming $\lalat$ is \Mdash positive---the parameter $n_L$ satisfies $\jo = \io + n_L - 1$, and is therefore determined by $\lalat$.

Thus, given $\lalat$, we have explained how to recover $n_L$ and the corresponding point $\xTout(\fo)=\xOout(\fo):=\ray(\r_\jo)$. As above, we then recover the rest of the t-immersion recursively.

Finally, to check part~\itemref{intro:triang3} (surjectivity),
 we show that the \ora outputs a valid t-embedding for all generic \Mdash positive points $\lalat\in\lalak$. 

\begin{remark}
The \ora gives a constructive way to extend a \emph{$1$-Lipschitz} (i.e., non-expansive) map sending $\bdxT_i\mapsto\bdxO_i$ for each $i\in\brn$ to a $1$-Lipschitz piecewise-linear isometry $\xT(\SuppGD)\to\xO(\SuppGD)$ of the entire polygon $\xT(\SuppGD)$. The statement that partial $1$-Lipschitz mappings of the Euclidean space to itself can be extended to the entire space is known as \emph{Kirszbraun’s Theorem}~\cite{Valentine}. 
Explicit constructions of such extensions that are similar in flavor to the \ora can be found in~\cite{Akop_Tar,Tasmuratov}. See also~\cite[Exercise~39.11 and Section~40.4]{Pak_book}.\footnote{We thank Igor Pak for bringing these references to our attention.}
\end{remark}

\subsection{\KSprims}\label{sec:intro:holomorphic}
We briefly explain the construction of the \oac in \cref{thm:intro:t_imm_vs_triples} following~\cite{KLRR,CLR1}. We let $\wtK:\E\to\R$ be the \emph{Kasteleyn edge weights} on $\G$: we have $\wtK(\e) = \pm\wt(\e)$ for all $\e\in\E$ so that 
 for each interior face $\f\in\Fint$ of $\G$ incident to $2d$ edges, the product of Kasteleyn signs of edges around $\f$ is $(-1)^{d-1}$. There are additional conditions for boundary faces; see \cref{dfn:DIM:Kast}.

Let $\Vspace$ be an $\R$-vector space. (We will work with $\Vspace=\R^d$ for $d\geq1$ or $\Vspace=\C$.) A function $\Fw: \WV\to\Vspace$ (resp., $\Fb:\BV\to\Vspace$) is called \emph{\wdash holomorphic} (resp., \emph{\bdash holomorphic}) if for each $\bv\in \BVint$ (resp., $\wv\in\WVint$), 
\begin{equation}\label{eq:intro:white_holom}
\sum_{\wv'\sim\bv} \wtK(\wv',\bv) \Fw(\wv') = 0,
\quad\text{resp.,}\quad
\sum_{\bv'\sim\wv} \wtK(\wv,\bv') \Fb(\bv') = 0.
\end{equation} 
Here, the summations are taken over all vertices of $\G$ adjacent to $\bv$ (resp., $\wv$), and $\wtK(\wv',\bv')$ denotes the sum of Kasteleyn weights of the edges connecting $\wv'$ to $\bv'$. 
We refer to \wdash\ and \bdash holomorphic functions collectively as \emph{($\Vspace$-valued) discrete holomorphic functions}. 
We let $\Hwspace_{\Vspace}\HtripK$ (resp., $\Hbspace_{\Vspace}\HtripK$) denote the space of $\Vspace$-valued \wdash holomorphic (resp., \bdash holomorphic) functions. 
Thus, $\Hwspace_{\R^k}\HtripK$ is closely related to the space of \emph{vector-relation configurations} on $\G$ studied in~\cite{AGPR}.
We denote $\HHspaceKV:=\Hwspace_{\Vspace}\HtripK\times\Hbspace_{\Vspace}\HtripK$. 

When $\G$ admits an \APM (i.e., $\helmin(\G)\geq0$), 
 a discrete holomorphic function $\Fw\in\Hwspace_{\Vspace}\HtripK$ (resp., $\Fb\in\Hbspace_{\Vspace}\HtripK$) can be uniquely recovered from its \emph{boundary restriction} $\pFw\in\Vspace^n$ (resp., $\pFb\in\Vspace^n$) defined in~\eqref{eq:partial_F_dfn}. 
 For $C=\Meas(\G,\wt)$, we have identifications
\begin{equation}\label{eq:intro:holom_C_and_Cp}
 \alt(C) = \{\pFw \mid \Fw\in\Hwspace_{\R}\HtripK\} \quad\text{and}\quad
\alt(C^\perp) = \{\pFb\mid \Fb\in\Hbspace_{\R}\HtripK\}
\end{equation}
as elements of $\Gr(k,n)$ and $\Gr(n-k,n)$, respectively.

Given a pair $(\Fw,\Fb)\in\HHspaceKC$ of $\C$-valued discrete holomorphic functions, the \emph{\KSprim} $\xd:\Faces\to\Rdd$ is defined up to an overall additive constant by the conditions
\begin{equation}\label{eq:intro:primitive_TO}
 \xT(\f_2)-\xT(\f_1) = \Fw(\wv) \wtK(\e) \Fb(\bv) 
 \quad\text{and}\quad
 \xO(\f_2)-\xO(\f_1) = \ovl{\Fw(\wv)} \wtK(\e) \Fb(\bv) 
 \quad\text{for all $\e\in\Edges$,}
\end{equation}
 where $\bv,\wv$ are the endpoints of $\e$ and $\f_1,\f_2$ are the faces incident to $\e$, with $\f_2$ located to the left of the oriented edge $\wv\to\bv$. 
By~\eqref{eq:intro:white_holom}, the increments in~\eqref{eq:intro:primitive_TO} add up to zero around each (black or white) face of $\GD$. Thus, $\xd$ is globally well defined on $\Faces$ up to an overall shift. 

The bijection of \cref{thm:intro:t_imm_vs_triples} is obtained as follows. Let 
 $C:=\Meas(\G,\wt)\in\Grtnn(k,n)$, and let $\lalat\in\lalak$ be such that $\la\subset C \subset \latp$. Let $(\y_i)_{i=1}^n,(\yt_i)_{i=1}^n\in\C^n$ be such that
\begin{equation}\label{eq:intro:y_to_lalat}
 \la_i = \CtoM[\y_i] := \begin{pmatrix}
\Re(\y_i)\\
\Im(\y_i)
\end{pmatrix} \quad\text{and}\quad
\lat_i = \CtoMt[\yt_i] := \begin{pmatrix}
\Re(\yt_i)\\
-\Im(\yt_i)
\end{pmatrix} \quad\text{for all $i\in\brn$.}
\end{equation}
\noindent Since $\la\subset C$ (resp., $C\subset \latp$), there exist discrete holomorphic functions $(\Fwl,\Fbl)\in\HHspaceKC$ such that 
\begin{equation}\label{eq:intro:Fw_Fb_y}
 \alt(\pFwl) = (\y_1,\y_2,\dots,\y_n),\quad\text{resp.,}\quad \alt(\pFbl) = (\yt_1,\yt_2,\dots,\yt_n).
\end{equation} 
The t-immersion $\Tll$ of $(\G,\wt)$ corresponding to $\lalat$ is obtained as the Kenyon--Smirnov primitive of $(\Fwl,\Fbl)$. Conversely, one can recover $(\Fwl,\Fbl)$---and thus $\lalat$---from $\Tll$ as we explain in \cref{sec:finishing_main_bij}.

\subsection{Outline}
We review some further background in \cref{sec:backgr}. The proof of the \oac (\cref{thm:intro:t_imm_vs_triples}) occupies \crefrange{sec:mom}{sec:proof_main_bij}. In \cref{sec:imm}, we complete the proof of existence of t-embeddings (\cref{thm:intro:B}) by studying Temperley--Lieb immanants.
 In \cref{sec:BCFW}, we prove the BCFW tiling result for ambient momentum amplituhedra. 
In \cref{sec:TREE}, we discuss T-duality for amplituhedra and complete the proof of the BCFW tiling conjecture (\cref{thm:intro:A}) for momentum and momentum-twistor amplituhedra. 
In \cref{sec:perf}, we discuss \emph{perfect t-embeddings} of~\cite{CLR2} and show existence and uniqueness of T-dual perfect t-embeddings. 
In \cref{sec:ABJM}, we discuss analogs of our results for ABJM amplituhedra of~\cite{Huang_Wen,Huang_Wen_Xie} and s-embeddings of~\cite{Chelkak_s_emb}. 
In \cref{sec:Varc}, we relate the number of t-embeddings with prescribed boundary to the number of bounded regions of a positroid hyperplane arrangement. 
Interactive applets illustrating the \oac, the \ora, and its ABJM version are available at~\cite{applets}.

\subsection*{Acknowledgments}
I am grateful to Thomas Lam for the numerous conversations throughout the years that have influenced many of the ideas in this work. I thank Terrence George for discussions related to discrete holomorphic functions. I also thank Marianna Russkikh and Misha Basok for their explanations regarding some of the concepts in~\cite{KLRR,CLR1,CLR2}. I thank 
Nima Arkani-Hamed, 
Dmitry Chelkak, 
Sasha Goncharov,
Rick Kenyon,
Allen Knutson,
Igor Pak,
Matteo Parisi,
Jara Trnka,
and Lauren Williams
for their comments on the first version of the paper. Finally, I am grateful to Daniel Galashin for the memorable one-sided discussions that facilitated the preparation of this manuscript.

\vspace{0.05in}
{\small AI was used for proofreading this document and for creating the interactive applets~\cite{applets}.\par}

\section{Preliminaries}\label{sec:backgr}
\subsection{Planar bipartite graphs}\label{sec:backgr:plabic}

We introduce some notation and basic definitions related to planar bipartite graphs. 
Let $\G$ be a planar bipartite graph embedded in the disk $\Disk$ with boundary vertices $\bdv_1,\dots,\bdv_n$ of degree $1$. We denote the \emph{boundary edges} of $\G$ by $\bde_i=\{\bdv_i,\bdvx_i\}$ for $i\in\brn$, and assume that each $\bdvx_i$ is an interior vertex. We refer to the $\bdvx_i$-s as \emph{next-to-boundary vertices}. 
We let $\Ebd:=\{\bde_1,\bde_2,\dots,\bde_n\}$ and $\Eint:=\E\setminus\Ebd$. 
For a subset $\Rg\subset\Verts$, we let $\GR$ be the induced subgraph with vertex set $\Rg$. For a subset $X\subset\Verts$, we let $\Grem X:=\G\ind[\Verts\setminus X]$ be the induced subgraph on the complement of $X$. 

We denote some graphs using accents (e.g., $\Gbot$ or $\Gf$ as in \cref{fig:intro-annular}). In this case, we implicitly assume that the same accent is used to denote the corresponding set of vertices ($\Vbot$ or $\Vf$), edges ($\Ebot$ or $\Ef$), etc. 

Since every \APM of $\G$ uses the same number of black and white vertices, we have 
\begin{equation}\label{eq:DIM:k_dfn}
 k=|\WV| - |\BVint|,\quad n-k = |\BV| - |\WVint|, \quad\text{and}\quad n = |\Verts| - |\Vint|.
\end{equation}

A connected component of $\G$ is called \emph{floating} if it does not contain any boundary vertices. We say that $\G$ is \emph{boundary-connected} if it has no floating connected components. 
When $\G$ \hasnofloat, Euler's formula yields
\begin{equation}\label{eq:Euler_no_float}
 |\Vint| - |\Edges| + |\Faces| = 1,
 \quad\text{and when $\G$ is connected,}\quad 
|\Verts| - |\Edges| + |\Fint| = 1.
\end{equation}
\begin{definition}\label{dfn:DIM:GD}
Suppose that $\G$ \hasnofloat. 
In this case, every face $\f\in\Faces$ of $\G$ is homeomorphic to an open disk, and we consider the \emph{dual graph} 
$\GD=(\Faces,\East)$ 
with $n$ boundary vertices $\bdf_1,\bdf_2,\dots,\bdf_n$ (some of which may coincide when $\G$ is not connected).
\end{definition}
The graph $\GD$ may have loop or parallel edges. For a dual edge $\east\in\East$ with endpoints $\ff,\f\in\Faces$, we denote $\ebarast:=\{\ff,\f\}$, and we denote by $\Ebarast:=\{\ebarast\mid\east\in\East\}$ the set of parallelism classes of edges of $\GD$. 
The set of edges of $\G$ incident to a face $\ff\in\Faces$ of $\G$ is denoted $\partE\ff$. We treat $\partE\ff$ as a multiset: if $\e\in\Edges$ is incident to $\ff$ on both sides 
then $\e$ appears twice in $\partE\ff$. The set of edges of $\GD$ incident to a face $\v\in\Vint$ of $\GD$ is denoted $\partEast\v$. 

\begin{definition}\label{dfn:positroid_Grneck}
The set $\Matroid_\G:=\{\partial\Apm\mid\Apm\in\APMS(\G)\}$ is called the \emph{positroid} of $\G$. For $i\in\brn$, let $\prec_i$ be the cyclically shifted total ordering on $\brn$ given by $i\prec_i i+1\prec_i\cdots\prec_i i-1$. We let $\Ibar_i$ be the lexicographically-minimal element of $\Matroid_\G$ with respect to $\prec_i$. The sequence $\Icalbar_\G:=(\Ibar_1,\Ibar_2,\dots,\Ibar_n)$ is called the \emph{Grassmann necklace} of $\G$.
\end{definition}

\subsection{Totally nonnegative Grassmannian}\label{sec:backgr:Grtnn}
We continue to review the theory of the totally nonnegative Grassmannian and refer to~\cite{Pos,LamCDM,KLS} for further details.
\begin{definition}\label{dfn:bound_kn}
A \emph{bounded affine permutation of type $(k,n)$} is a bijection $f:\Z\to\Z$ such that $f(i+n) = f(i) + n$ and $i\leq f(i)\leq i+n$ for all $i\in\Z$, and such that $\frac1n\sum_{i=1}^n(f(i)-i) = k$.
\end{definition}
\noindent The (finite) set of bounded affine permutations of type $(k,n)$ is denoted $\Boundkn$. An element $f\in\Boundkn$ is completely determined by the \emph{window} $[f(1),f(2),\dots,f(n)]$. 

The \emph{Grassmann necklace} $\Icalr_f=(\Ir_i)_{i\in\Z}$ associated to $f\in\Boundkn$ is defined by
\begin{equation}\label{eq:gr_neck_dfn}
 \Ir_i:=\{f(j)\mid j<i\text{ and }f(j)\geq i\}, \quad\text{for all $i\in\Z$.}
\end{equation}
We have $|\Ir_i| = k$ for all $i\in\Z$. 

Let $C$ be a full rank $k\times n$ matrix. Define $\fC:\Z\to\Z$ by
\begin{equation}\label{eq:dfn_fC}
 \fC(i):=\min\{j\geq i\mid C_i\in\Span(C_{i+1},\dots,C_j)\} \quad\text{for all $i\in\Z$.}
\end{equation}
It is a nontrivial fact~\cite{KLS} that the resulting map $\fC$ belongs to $\Boundkn$. 
 We get a decomposition $\Gr(k,n) = \bigsqcup_{f\in\Boundkn} \Pio_f$ of the (real) Grassmannian into \emph{open positroid varieties} given by $\Pio_f:=\{C\in\Gr(k,n)\mid \fC = f\}$.
 The \emph{positroid cells} $\Ptp_f$ are defined as
\begin{equation}\label{eq:Ptp_intersection}
 \Ptp_f:=\Pio_f\cap \Grtnn(k,n), \quad\text{and thus}\quad
\Grtnn(k,n) = \bigsqcup_{f\in\Boundkn} \Ptp_f.
\end{equation}
The above decomposition contains a unique top-dimensional piece (called the \emph{top cell}) labeled by $\fkn\in\Boundkn$, with $\fkn(i) = i+k$ for all $i\in\Z$. We have $\Ptp_{\fkn} = \Grtp(k,n)$.

Suppose that $\G$ admits an \APM. Then there exists a unique $\fG\in\Boundkn$ such that 
\begin{equation}\label{eq:Ptp_vs_Meas}
\Ptp_{\fG} = \{\Meas(\G,\wt)\mid \wt:\E\to\Rtp\}. 
\end{equation}
The Grassmann necklace $\Icalbar_\G$ in \cref{dfn:positroid_Grneck} is obtained from $\Icalr_{\fG}$ by taking all elements modulo $n$, and we denote $\Icalr_\G:=\Icalr_{\fG}$.

We say that $\wt,\wt':\E\to\R_{>0}$ are \emph{gauge equivalent} if they are related by~\eqref{eq:intro:gauge} for a function $\gauge: \Verts\to \R_{>0}$ satisfying $\gauge(\bdv_1) = \gauge(\bdv_2) = \dots = \gauge(\bdv_n) = 1$.

\begin{definition}\label{dfn:reduced}
$\G$ is called \emph{reduced} if 
the map $\Meas:\Rtpgauge\to\Ptp_{\fG}$ is a homeomorphism, where 
 $\Rtpgauge:=\Rtp^{\E}/\Rtp^{\Vint}$ is the space of edge weightings $\wt:\E\to\Rtp$ modulo gauge transformations at interior vertices; cf.~\eqref{eq:Euler_no_float}. Alternatively, $\G$ is reduced if and only if $\G$ has the minimal number of faces among all graphs with bounded affine permutation $\fG$.\footnote{We additionally require that any leaf interior vertex of a reduced graph must be adjacent to a boundary vertex.}
\end{definition}
\noindent When $\G$ is reduced, $\fG$ may be computed using the combinatorics of \emph{zig-zag paths} in $\G$. We emphasize that when $\G$ is not reduced, $\fG$ must be computed via~\eqref{eq:Ptp_vs_Meas}.

\begin{definition}\label{dfn:top_cell_graph}
We say that $\G$ is a \emph{top cell graph} if $\fG=\fkn$, equivalently, if $\Meas(\G,\wt)\in\Grtp(k,n)$ for some (equivalently, any) $\wt\in\Rtpgauge$.
\end{definition}

\begin{definition}\label{dfn:perfect_orient}
A \emph{perfect orientation} $\GO$ of $\G$ is an orientation of all edges of $\G$ such that each black (resp., white) interior vertex of $\G$ has exactly one outgoing (resp., incoming) edge. 
\end{definition}
\noindent Perfect orientations of $\G$ are in bijection with \APMs of $\G$: if $\GO$ is a perfect orientation then $\Apm(\GO)\in\APMS(\G)$ consists of all edges oriented from black to white in $\GO$. The boundary $\Ir(\GO):=\partial\Apm(\GO)$ of $\Apm(\GO)$ is the set of $i\in\brn$ such that $\bdv_i$ is a source of $\GO$.
Given $\apm\in\APMS(\G)$, we denote by $\GO(\apm)$ the corresponding perfect orientation.

\begin{definition}\label{dfn:nondeg}
For $a,b\geq0$, we say that $f\in\Boundkn$ is \emph{$(a,b)^\partial$-nondegenerate} and write $\fap\in\BND_{a,b}(k,n)$ if 
\begin{equation*}%
 i+a \leq f(i) \leq i+n-b \quad\text{for all $i\in\Z$.}
\end{equation*}
A $k\times n$ matrix $C$ (resp., a planar bipartite graph $\G$) is \emph{$(a,b)^\partial$-nondegenerate} if $\fC$ (resp., $\fG$) is \emph{$(a,b)^\partial$-nondegenerate}. Let $\GrndAB(k,n)\subset\Grtnn(k,n)$ be the set of $(a,b)^\partial$-nondegenerate $C\in\Grtnn(k,n)$.
\end{definition}
\noindent In particular, a matrix $C\in\Grtnn(k,n)$ (resp., a graph $\G$) is \twonondeg as defined in \cref{sec:intro:plabic} if and only if it is $(2,2)^\partial$-nondegenerate. 

\subsection{Kasteleyn theory and \KSprims}\label{ssec:Kast_KSprims}
We discuss the classical \emph{Kasteleyn sign condition}~\cite{Kasteleyn} for planar bipartite graphs.
Assume that $\G$ \hasnofloat. For $\ff\in\Faces$, we denote by $\bdryarcs\ff:=\{i\in\brn\mid\bdf_i = \ff\}$ the set of boundary arcs incident to $\ff$. 

Given a directed path $\Pathdir$ in $\G$ with edges $(\e_1,\e_2,\dots,\e_d)$ (or, more generally, any collection of oriented edges of $\G$) and a function $h:\E\to\Cast:=\C\setminus\{0\}$, we set
\begin{equation}\label{eq:h_Path_dfn}
 h(\Pathdir) := \prod_{i=1}^d 
 \begin{cases}
 h(\e_i), &\text{if $\e_i$ is directed from white to black in $\Pathdir$,}\\
 h(\e_i)^{-1}, &\text{if $\e_i$ is directed from black to white in $\Pathdir$.}\\
 \end{cases}
\end{equation}
 For $\ff\in\Faces$, let $\bdrypath\ff$ be the collection of boundary edges of $\ff$ directed clockwise around $\ff$, and let $h(\bdrypath\ff)$ be the \emph{face weight} of $\ff$. 
Recall that $\partE\ff$ denotes the \emph{multiset} of edges incident to $\ff$. 
Thus, an edge $\e$ incident to $\ff$ on both sides appears twice in $\bdrypath\ff$ with opposite orientations and contributes~$1$ to $h(\bdrypath\ff)$. 
For $\ff\in\Faces$, we denote by $\dwcor(\ff)$ the number of white vertices of $\G$ incident to $\ff$, again counted with multiplicity. In other words, $\dwcor(\ff)$ is the number of \emph{white corners} of $\G$ incident to $\ff$. 

\begin{definition}[{\cite{AGPR,SpeyerVariations}}]\label{dfn:DIM:Kast}
We say that $\epsK:\E\to\{\pm1\}$ is a choice of \emph{Kasteleyn signs} for $\G$ if for each face $\ff\in\Faces$,
\begin{equation}\label{eq:Kast_sign}
 \epsK(\bdrypath\ff) = (-1)^{\dwcor(\ff)+\nbdryarcs\ff+\epstra(\ff)},\quad\text{where}\quad
\epstra(\ff) = 
\begin{cases}
 1, &\text{if $\ff\neq\bdf_n$;}\\
 k+n, &\text{if $\ff=\bdf_n$.}\\
\end{cases}
\end{equation}
We define \emph{Kasteleyn edge weights} $\wtK:\Edges\to\R$ by $\wtK(\e):=\epsK(\e)\wt(\e)$ for all $\e\in\Edges$. 
\end{definition}
\noindent For example, if $\ff\in\Fint$ is an interior face then $\bdryarcs\ff=\emptyset$ and $\epstra(\ff)=1$, so~\eqref{eq:Kast_sign} reduces to the classical Kasteleyn sign condition~\cite{Kasteleyn,TF_exact}. When $\G$ is connected, we have $|\bdryarcs\bdf_i|=1$ for all $i\in\brn$.

\begin{remark}\label{rmk:DIM:Kast_sign_bdfn}
By~\eqref{eq:DIM:k_dfn} and~\eqref{eq:Euler_no_float}, the choice of $\epstra(\bdf_n)$ in~\eqref{eq:Kast_sign} is consistent with the equation $\prod_{\ff\in\Faces} \epsK(\bdrypath\ff)=1$ which holds for any $\epsK:\Edges\to\{\pm1\}$.
\end{remark}

\noindent As explained in~\cite[Proposition~4.8]{AGPR}, a choice of Kasteleyn signs exists for any $\G$.

\begin{remark}\label{rmk:Kast_sign_float}
When $\G$ is not necessarily \bdconn, the Kasteleyn sign condition~\eqref{eq:Kast_sign} becomes
$\epsK(\bdrypath\ff)=(-1)^{\dwcor(\ff)+\nbdryarcs\ff+\epstra(\ff)+\nfloat(\ff)}$, where $\nfloat(\ff)$ denotes the number of ``holes'' in the face $\ff$ (i.e., the number of bounded connected components of the complement of $\ff$ in the plane). 
 By convention, isolated white vertices contained in $\ff$ contribute to $\nfloat(\ff)$ but not to $\dwcor(\ff)$.
\end{remark}

\begin{lemma}\label{lemma:OCP:Kast_even}
Let $\epsK$ be a choice of Kasteleyn signs for $\G$. Let $\G'$ be a subgraph of $\G$ with the same set of boundary vertices. Then $\epsK$ restricts to a choice of Kasteleyn signs for $\G'$ if and only if each face of $\G'$ that does not contain $\bdf_n$ encloses an even number of vertices in $\Vint\setminus \Vint'$.
\end{lemma}
\begin{proof}
For interior faces of $\G'$, this is~\cite[Lemma~1]{Kenyon_lectures}. We extend the argument to arbitrary faces of $\G'$. First, we show that if $\G'$ is obtained from $\G$ by deleting any collection of edges (but keeping the set of vertices intact) then $\epsK$ restricts to a choice of Kasteleyn signs for $\G'$. By induction, it suffices to consider the case where $\G'$ is obtained by deleting a single edge $\e\in\Edges$. Let $\ff_1,\ff_2\in\Faces$ be the two faces of $\G$ incident to $\e$ and let $\f$ be the face of $\G'$ containing $\ff_1$ and $\ff_2$. If $\ff_1=\ff_2$ then $\nfloat(\f) = \nfloat(\ff_1) + 1$, $\dwcor(\f)=\dwcor(\ff_1)-1$, $\nbdryarcs{\f}=\nbdryarcs{\ff_1}$, $\epstra(\f)=\epstra(\ff_1)$, and $\epsK(\bdrypathp\f)=\epsK(\bdrypath\ff_1)$. Otherwise, $\nfloat(\f) = \nfloat(\ff_1) + \nfloat(\ff_2)$, $\dwcor(\f)=\dwcor(\ff_1)+\dwcor(\ff_2)-1$, 
$\nbdryarcs{\f}=\nbdryarcs{\ff_1}+\nbdryarcs{\ff_2}$, $\epstra(\f)=\epstra(\ff_1) + \epstra(\ff_2)-1$, and 
$\epsK(\bdrypathp\f)=\epsK(\bdrypath\ff_1) \cdot \epsK(\bdrypath\ff_2)$. In either case, it follows that the restriction of $\epsK$ to $\Edges'$ satisfies the Kasteleyn sign condition in \cref{rmk:Kast_sign_float}. 
 Finally, deleting an isolated interior vertex from a face $\ff$ changes $\nfloat(\ff)$ by $1$, changes the parity of $k+n$ in~\eqref{eq:Kast_sign}, and does not affect any other quantities in \cref{rmk:Kast_sign_float}. Thus, $\epsK$ restricts to a choice of Kasteleyn signs for $\G'$ if and only if each face of $\G'$ (excluding the face containing $\bdf_n$) contains an even number of vertices in $\Vint\setminus \Vint'$.
\end{proof}

Given a discrete holomorphic function $\Fw\in\Hwspace_{\C}\HtripK$ or $\Fb\in\Hbspace_{\C}\HtripK$ (cf.~\eqref{eq:intro:white_holom}), we let
\begin{equation}\label{eq:partial_F_dfn}
 \pFw_i:=
 \begin{cases}
 -\Fw(\bdv_i), &\text{if $\bdv_i$ is white;}\\
 -\wtK(\bde_i) \Fw(\bdvx_i), &\text{if $\bdv_i$ is black;}
 \end{cases}
 \quad
 \pFb_i:=
 \begin{cases}
 \Fb(\bdv_i), &\text{if $\bdv_i$ is black;}\\
 -\wtK(\bde_i) \Fb(\bdvx_i), &\text{if $\bdv_i$ is white.}
 \end{cases}
\end{equation}
By~\eqref{eq:partial_F_dfn}, the \KSprim $\xd:\Faces\to\Rdd$ introduced in~\eqref{eq:intro:primitive_TO} satisfies
\begin{equation}\label{eq:TE:t_imm_bdry_vs_pFw_pFb}
 \bdxT_i - \bdxT_{i-1} = \pFw_i\pFb_i 
 \quad\text{and}\quad
 \bdxO_i - \bdxO_{i-1} = \ovl{\pFw_i}\pFb_i 
\quad\text{for all $i\in\brn$},
\end{equation}
where the index $i-1$ is taken modulo $n$; cf. \cref{sec:backgr:cyc} below.

\begin{definition}\label{dfn:MCE:WKmat}
Assume that $\G$ \emph{has white boundary}, i.e., $\BVint = \BV$. 
 In this case, we denote the $\WV\times\BVint$ Kasteleyn matrix with entries $\wtK(\wv,\bv)$ by $\WKmat$. 
 We assume that the rows of $\WKmat$ are ordered so that the boundary vertices $\bdv_1,\bdv_2,\dots,\bdv_n$ appear first (and in this order). For $I\in{\brn\choose k}$, we let $\Delta_{\WV\setminus I}(\WKmat)$ be the minor of $\WKmat$ with row set $\{\bdv_i\mid i\notin I\}\sqcup\WVint$ and column set $\BVint$. 
\end{definition}

For a proof of the following result, see \cite[Corollary~6.8 and Proposition~6.9]{AGPR}.

\begin{proposition}[\cite{AGPR}]%
\label{lemma:MCE:Delta_vs_Kast}
We have $\Delta_I(\G,\wt) = \eps\Delta_{\WV\setminus I}(\WKmat)$ for a fixed $\eps\in\{\pm1\}$ and all $I\in{\brn\choose k}$.
\end{proposition}
\begin{corollary}\ \label{lemma:MCE:apm_vs_Kast}
Assume that $\G$ has white boundary and admits an \APM. Then 
\begin{enumerate}[label=(\arabic*)]
\item\label{apm_vs_Kast1}
 $\WKmat$ is of full rank: $\rank\WKmat = |\BVint|$,
\item\label{apm_vs_Kast2}
 $\dim\Hwspace_{\R}\HtripK = k$, 
\item\label{apm_vs_Kast3} the linear operator $\partial:\Hwspace_{\R}\HtripK\to\R^n$ 
defined in~\eqref{eq:partial_F_dfn}
 is injective, and for $C:=\Meas(\G,\wt)$, %
\begin{equation}\label{eq:MCE:alt(C)_vs_pFw}
 \alt(C) = \{\pFw \mid \Fw\in\Hwspace_{\R}\HtripK\} \quad\text{as elements of $\Gr(k,n)$}.
\end{equation}
\end{enumerate}
\end{corollary}
\begin{proof}
Part~\itemref{apm_vs_Kast1} follows directly from \cref{lemma:MCE:Delta_vs_Kast}. For part~\itemref{apm_vs_Kast2}, observe that $\Fw\in\R^{\WV}$ is \wdash holomorphic if and only if it belongs to the left kernel of $\WKmat$, which by part~\itemref{apm_vs_Kast1} and~\eqref{eq:DIM:k_dfn} has dimension~$k$. Part~\itemref{apm_vs_Kast3} follows from~\cite[Proposition~6.9]{AGPR}.
\end{proof}

\begin{remark}\label{rmk:MCE:BKmat}
Similar results hold for the $\BV\times\WVint$ matrix $\BKmat$ when $\G$ has black boundary: there exists $\eps\in\{\pm1\}$ such that for all $I\in{\brn\choose k}$, we have $\Delta_I(\G,\wt) = (-1)^{\sum_{i\in I} i}\eps \Delta_{\{\bdv_i\mid i\in I\}\sqcup\BVint}(\BKmat)$, and if $\G$ admits an \APM then 
\begin{equation}\label{eq:MCE:alt(Cp)_vs_pFb}
\dim\Hbspace_{\R}\HtripK = n-k
\quad\text{and}\quad
 \alt(C^\perp) = \{\pFb \mid \Fb\in\Hbspace_{\R}\HtripK\} \quad\text{as elements of $\Gr(n-k,n)$}.
\end{equation}
\end{remark}

\begin{proposition}\label{lemma:DIM:remove_bivertex}
Suppose that $\w_1,\w_2\in\WVint$ share a face of $\G$, and consider the restriction $\wt'$ of $\wt$ to the edge set of $\G':=\G\rem\{\w_1,\w_2\}$. Assume that both $\G$ and $\G'$ admit \APMs. Then
\begin{equation}\label{eq:DIM:remove_bivertex_w}
 C'\subset C,\quad\text{where}\quad
 C':=\Meas(\G',\wt')\in\Grtnn(k-2,n) \quad\text{and}\quad 
C:=\Meas(\G,\wt)\in\Grtnn(k,n).
\end{equation}
Similarly, if $\b_1,\b_2\in\BVint$ share a face of $\G$ and $\G'':=\G\rem\{\b_1,\b_2\}$ admits an \APM then 
\begin{equation}\label{eq:DIM:remove_bivertex_b}
 C\subset C'',\quad\text{where}\quad
C'':=\Meas(\G'',\wt'')\in\Grtnn(k+2,n)
\quad\text{for}\quad
\wt'':=\wt|_{\Edges''}.
\end{equation}
\end{proposition}
\begin{proof}
We prove~\eqref{eq:DIM:remove_bivertex_w}. Fix a choice of Kasteleyn signs $\epsK$ for $\G$. By \cref{lemma:OCP:Kast_even}, the restriction $\epsK'$ of $\epsK$ to the edge set of $\G'$ is a choice of Kasteleyn signs for $\G'$. It follows that $\Hwspace_{\R}\HtripKp \cong \{\Fw\in\Hwspace_{\R}\HtripK\mid \Fw(\w_1)=\Fw(\w_2)=0\}$. Since both $\G$ and $\G'$ admit \APMs, \cref{lemma:MCE:apm_vs_Kast} applies to each of them. This shows~\eqref{eq:DIM:remove_bivertex_w}. 
 The proof of~\eqref{eq:DIM:remove_bivertex_b} follows by an analogous argument or via replacing $C$ with $\alt(C^\perp)$ and changing the colors of all vertices of $\G$ as discussed in \cref{sec:backgr:cyc} below.
\end{proof}
\noindent See~\cite[Section~3]{BHL} for a related construction.

\begin{definition}[Boundary restriction and extension]\label{dfn:OCP:restr_ext}
Assume that $\G$ admits an \APM and let $\Vspace=\R^d$. 
The \emph{boundary restriction} of a discrete holomorphic function $\Fw\in\Hwspace_{\Vspace}\HtripK$ (resp., $\Fb\in\Hbspace_{\Vspace}\HtripK$) is the $d\times n$ matrix $\alt(\pFw)\subset C$ (resp., $\alt(\pFb)\subset C^\perp$). 
In this case, we say that $\Fw$ (resp., $\Fb$) is the \emph{discrete holomorphic extension} of $\alt(\pFw)$ (resp., $\alt(\pFb)$). By \cref{lemma:MCE:apm_vs_Kast}, a discrete holomorphic extension of any $d\times n$ matrix $\Amat\subset C$ (resp., $\Atmat\subset C^\perp$) exists and is unique.
\end{definition}

\begin{figure}
 \includegraphics[scale=1]{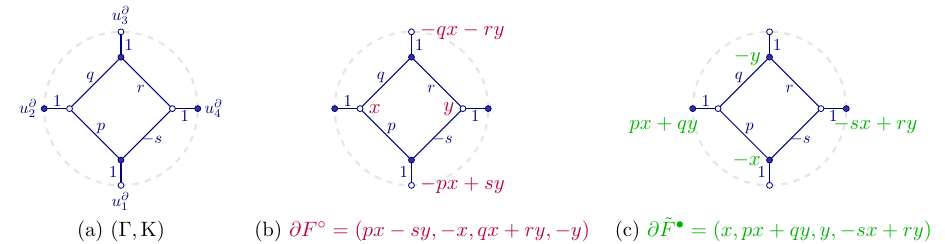}
 \caption{\label{fig:holom}Examples of discrete holomorphic functions.}
\end{figure}

\begin{example}\label{ex:KSprim24}
Let $k=2,n=4$, and let $(\G,\wt)$ be the weighted graph whose Kasteleyn edge weights are shown in \figref{fig:holom}(a), with $p,q,r,s>0$. We calculate using~\eqref{eq:intro:Delta_G_wt} that $C:=\Meas(\G,\wt) = \begin{pmatrix}
p & 1 & q & 0\\
-s & 0 & r & 1
\end{pmatrix}$, and thus $\alt(C^\perp) = \begin{pmatrix}
1 & p & 0 & -s\\
0 & q & 1 & r
\end{pmatrix}$. We clearly have $C,\alt(C^\perp)\in\Grtp(2,4)$. Examples of discrete holomorphic functions $\Fw,\Fb$ on $\G$ are shown together with $\pFw,\pFb$ in \figref{fig:holom}(b,c). 
We see that $\pFw = \begin{pmatrix}
x & y
\end{pmatrix}\cdot \alt(C)$ and $\pFb = \begin{pmatrix}
x & y
\end{pmatrix}\cdot \alt(C^\perp)$. This agrees with~\eqref{eq:intro:holom_C_and_Cp}.
\end{example}

\begin{remark}\label{rmk:DIM:Kast_gauge_eq}
The action of the gauge group $\RtpVint$ on $\wt\in\RtpE$ extends to $\HHspaceKC$: 
for $\gauge\in\RtpVint$, the value $\Fw(\wv)$ (resp., $\Fb(\bv)$) gets divided by $\gauge(\wv)$ (resp., by $\gauge(\bv)$). The \emph{sign gauge group} $\pmoneVint$ acts on $(\epsK,\wtK,\Fw,\Fb)$ by changing the signs at interior vertices. It is well known~\cite[Section~3.2]{Kenyon_lectures} that any two choices $\epsK,\epsK'$ of Kasteleyn signs on $\G$ are related by $\pmoneVint$-action. By~\eqref{eq:intro:primitive_TO}, the \KSprim $\xd$ is invariant under the action of the gauge groups $\RtpVint$ and $\pmoneVint$.
\end{remark}

\subsection{Cyclic symmetry and duality} \label{sec:backgr:cyc}
We briefly summarize the effect of the \emph{cyclic shift map $\Shift$} and the \emph{duality map $\altp$} on $\Grtnn(k,n)$ and related objects.

Let $\Shift:\R^n\to\R^n$ be the linear operator sending $(x_1,\dots,x_{n-1},x_n)\mapsto ((-1)^{k-1}x_n,x_1,\dots,x_{n-1})$. We let it act on columns of matrices in $\Gr(k,n)$ as
\begin{equation}\label{eq:Shift_dfn}
 \mat[C_1|\cdots|C_{n-1}|C_n]\cdot \Shift = \mat[(-1)^{k-1}C_n|C_1|\dots|C_{n-1}].
\end{equation}
It is easy to see that $\Shift$ preserves $\Grtnn(k,n)$. 

For $f\in\Boundkn$, $\Shift$ sends $\Ptp_f$ to $\Ptp_{\fpr}$, where $\fpr = \cShift(f) \in\Boundkn$ is defined by $\fpr(i) = f(i-1)+1$ for all $i\in\Z$. Let $\cShift(\G)$ be the graph obtained from $\G$ by cyclically relabeling the boundary vertices so that the new boundary vertices are $(\bdv_n,\bdv_1,\dots,\bdv_{n-1})$. The edge weights $\wt$ on $\G$ are preserved under $\cShift$. If $\epsK$ is a choice of Kasteleyn signs for $\G$ then we define $\epsKpr:=\cShift(\epsK)$ by
\begin{equation}\label{eq:cshift_eps_dfn}
 \epsKpr(\e) = 
 \begin{cases} 
 (-1)^{k+n-1}\epsK(\e), &\text{if $\e=\bde_n$};\\
 \epsK(\e), &\text{otherwise.}
 \end{cases}
\end{equation}
By \cref{rmk:DIM:Kast_sign_bdfn}, $\epsKpr$ is a choice of Kasteleyn signs for $\cShift(\G)$. 

We extend the vectors $\partial \Fw$ and $\pFb$ and the associated sequences $(\y_i)_{i=1}^n,(\yt_i)_{i=1}^n\in\C^n$ given by~\eqref{eq:intro:Fw_Fb_y} to sequences labeled by $i\in\Z$ via %
\begin{equation*}%
\pFw_{i+n} = (-1)^{k+n-1}\pFw_i,\quad \pFb_{i+n} = (-1)^{k+n-1}\pFb_i,\quad \y_{n+i} = (-1)^{k-1}\y_i,\quad \yt_{n+i} = (-1)^{k-1}\yt_i
\end{equation*}
for all $i\in\Z$. This is consistent with~\eqref{eq:intro:holom_C_and_Cp}, \eqref{eq:intro:Fw_Fb_y}, and~\eqref{eq:Shift_dfn}.

Next, we discuss the \emph{duality map} $\altp$. Recall from \cref{sec:intro:mom} that 
for $C = \mat[C_1|C_2|C_3|\dots|C_n]$, we have $\alt(C)=\mat[C_1|-C_2|C_3|\dots|(-1)^{n-1}C_n]$. 
 Let $\altp:\Gr(k,n)\to\Gr(n-k,n)$ be given by $\altp(C):=\alt(C^\perp) = \alt(C)^\perp$. This map sends $\Grtnn(k,n)$ homeomorphically to $\Grtnn(n-k,n)$. 
The map $\altp$ restricts to a homeomorphism
\begin{equation}\label{eq:altp_vs_fhat}
 \altp:\Ptp_f\xrasim\Ptp_{\fhat}, \quad\text{where}\ \ %
 \fhat\in\Boundxx(n-k,n) \ \ \text{is given by}\ \ %
\fhat(i) = f^{-1}(i) + n
\ \ \text{for all $i\in\Z$}.
\end{equation}
If $C = \Meas(\G,\wt)$ then $\altp(C) = \Meas(\Gvee,\wt)$, where $\Gvee$ is obtained from $\G$ by changing the colors of all vertices (i.e., swapping the roles of black and white). 
We have
\begin{equation}\label{eq:altp_Delta}
 \Delta_I(C) = \Delta_{\comp{I}} (\altp(C)) \quad\text{for all $I\in{\brn\choose k}$, where \quad $\comp{I}:=\brn\setminus I$.}
\end{equation}
 We record the following obvious property of the map $\alt$.
\begin{lemma}\label{lemma:wind_alt}
Let $\la$ be a $2\times n$ matrix satisfying $\brla<i,i+1>>0$ for all $i\in\brn$. Let $\la':=\diag(1,-1)\cdot \alt(\la)$ be obtained from $\alt(\la)$ by changing the sign of the second row. Then 
\begin{equation}\label{eq:wind_alt}
\text{$\brlapr<i,i+1>=\brla<i,i+1>$\quad for all $i\in\brn$},
\quad\text{and}\quad
 \wind(\la') = n\pi - \wind(\la),
\end{equation}
where we set $\la'_{i+n}:=(-1)^{k+n-1}\la'_i$ for all $i\in\Z$.
\end{lemma}

\subsection{Moves on planar bipartite graphs}\label{sec:backgr:moves}
We discuss well-known local moves on planar bipartite graphs. 
Each of these moves extends to a transformation of (Kasteleyn and ordinary) edge weights that preserves the boundary measurements, and induces appropriate transformations of discrete holomorphic functions and t-immersions. The moves are shown in \cref{fig:moves}; see also~\cite[Figures~3 and~4]{KLRR}.
The \emph{square/spider move} \MV2 will not be used in what follows.

\begin{definition}[Degree-$2$ vertex insertion and removal] \label{dfn:backgr:moves:contr_uncontr}
Let $\v$ be an interior vertex of degree $2$ in $\G$ such that both of its neighbors are interior non-leaf vertices. The move \MV1 consists of removing $\v$ from $\G$ and identifying its two neighbors. The reverse move is also denoted \MV1. 
This move creates or removes a degenerate bigon in the interior of a t-immersion of $\GD$, turning it into a \emph{near t-immersion} in the terminology of~\cite{KLRR}. 
\end{definition}

\begin{definition}[Boundary edge insertion/contraction]
\label{dfn:moves:boundary_edge}
Let $i\in\brn$. The move \MVbd consists of declaring $\bdv_i$ to be an interior vertex, introducing a new boundary vertex $\bdvp_i$ of color opposite to that of $\bdv_i$, and connecting it to $\bdv_i$ by a new edge $\bdep_i$. 
 We set $\wtK(\bdep_i):=1$ if $\bdv_i$ is white and $\wtK(\bdep_i):=-1$ if $\bdv_i$ is black. This move creates or removes a degenerate bigon at the boundary of a t-immersion of $\GD$.
\end{definition}

\begin{figure}
\def\scl{1.22}
\begin{tabular}{c|c|c}
 \includegraphics[scale=\scl]{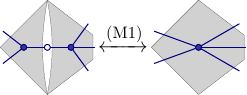}
&
 \includegraphics[scale=\scl]{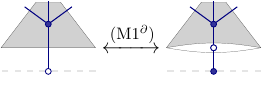}
&
 \includegraphics[scale=\scl]{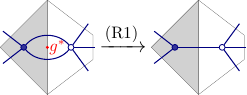}
\end{tabular}
 \caption{\label{fig:moves} Local moves on planar bipartite graphs.}
\end{figure}

\begin{definition}[Parallel edge reduction] \label{dfn:DIM:R1}
Suppose that $\G$ contains a bigonal face incident to a pair $\e_1,\e_2\in\Edges$ of parallel edges connecting vertices $\vv,\v\in\Vint$. The move \RV1 consists of replacing $\e_1,\e_2$ with a single edge from $\vv$ to $\v$. 
 This move replaces a degree-$2$ interior vertex of $\GD$ with a single edge. 
We relax condition~\itemref{intro:t_imm_orientation} so that whenever $\ff\in\Faces$ is incident to $\v\in\Vint$, the angle of the convex polygon $\xT(\v)$ at $\xT(\ff)$ is equal to $\pi$ if $\degGD(\ff)=2$ and belongs to $(0,\pi)$ if $\degGD(\ff)\geq3$. 
\end{definition}

\section{Momentum amplituhedron}\label{sec:mom}

As a first step towards our proof of the \oac (\cref{thm:intro:t_imm_vs_triples}), we develop some fundamental properties of the momentum amplituhedron and use them to show in \cref{prop:magic_homeo} that the \emph{magic projector} $\Qla$ defined in~\eqref{eq:intro:Qla} below preserves total positivity.

\subsection{Properties of the momentum amplituhedron map}
In our analysis, we will mostly consider $\la$ and $\lat$ independently. To that end, denote
 \begin{align}\label{eq:lak_latk}
 \lak &:= \{\la\in\Gr(2,n)\mid \text{$\brla<i,i+1> > 0$ for all $i\in\brn$ and $\wind(\la) = (k-1)\pi$}\};
\\
\label{eq:lak_latk2}
 \latk &:= \{\lat\in\Gr(2,n)\mid \text{$\brlat[i,i+1] > 0$ for all $i\in\brn$ and $\wind(\lat) = (k+1)\pi$} \}.
\end{align}
Thus, $\lalak = (\lak\times\latk)\cap \lalats$. Recall the map $\PhiLL$ defined in~\eqref{eq:intro:PhiLL_dfn}.

\begin{proposition}[{\cite[Section~2.2]{DFLP}}]\label{prop:momLL_basic}
Let $\LaLat\in\LaLak$ and $C\in\Grtnn(k,n)$.
\begin{enumerate}[label=(\arabic*)]
\item\label{mom1} The intersections $\la:=\La\cap C$ and $\lat:=\Lat\cap C^\perp$ are $2$-dimensional.
\item\label{mom2} If $C\in\Grnda(k,n)$ then $\la\in\lak$.
\item\label{mom3} If $C\in\Grndb(k,n)$ then $\lat\in\latk$.
\item\label{mom4} If $C\in\Grnd(k,n)$ then $\lalat\in\lalak$ (and thus $(\la,\lat,C)\in\TRIPLES$).
\end{enumerate}
\end{proposition}

\subsection{Extending \texorpdfstring{$(\la,\lat)$}{(𝜆,𝜆̃)} to \texorpdfstring{$(\La,\Lat)$}{(𝚲,𝚲̃)}}
Given a fixed $C\in\Grnd(k,n)$, one may view \cref{prop:momLL_basic} as a convenient way to find a pair $\lalat\in\lalak$ 
 such that $\la\subset C\subset\latp$: one just needs to choose any pair $\LaLat\in\LaLak$ and set $\la:=C\cap \La$, $\lat:=C^\perp\cap \Lat$. In this subsection, given a fixed pair $\lalat\in\lalak$, we would like to find $\LaLat\in\LaLak$ such that $\la\subset\La$ and $\lat\subset\Lat$.\footnote{Given a fixed pair $\lalat\in\lalak$, it is not always possible to also find $C\in\Grnd(k,n)$ such that $\la\subset C \subset\latp$. On the other hand, it is possible to find such $C$ when $\lalat\in\lalapp$ is \Mdash positive by \cref{thm:f_triang} below.}

\begin{proposition}\label{prop:from_la_to_La} \
\begin{itemize}
\item For each $\la\in\lak$, there exists $\La\in\alt(\Grtp(n-k+2,n))$ such that $\la\subset\La$.
\item For each $\lat\in\latk$, there exists $\Lat\in\Grtp(k+2,n)$ such that $\lat\subset\Lat$.
\end{itemize}
\end{proposition}
\begin{proof}
Let $\la'\in\lakpr$. We will show that there exists $C'\in\Grtp(k',n)$ such that $\la'\subset C'$. This will imply both of the statements above: we get the first statement using~\eqref{eq:wind_alt} for $\la':=\diag(1,-1)\cdot \alt(\la)$, $k':=n-k+2$, and $\La:=\alt(C')$ and the second statement for $\la':=\lat$, $k'=k+2$, and $\Lat:=C'$.

Our goal is to find real numbers $y_1<y_2<\cdots<y_n$ and $t_1,t_2,\dots,t_n>0$ such that $\la'\subset C'$ for $C':=\VandM\cdot \tdiag$, where $\VandM:=(y_j^{i-1})_{(i,j)\in\brx{k'}\times\brn}$ is a $k'\times n$ Vandermonde matrix and $\tdiag = \diag(t_1,t_2,\dots,t_n)$. It is well known that any such matrix $C'$ belongs to $\Grtp(k',n)$. 

Let $P(x) = a_0 + a_1 x + \cdots + a_{k'-1}x^{k'-1}$ be a degree-$(k'-1)$ polynomial with real coefficients such that $a_{k'-1} = (-1)^{k'-1}$. Assume that $P(x)$ has $k'-1$ distinct real roots $x_1 < x_2 < \cdots < x_{k'-1}$. 
 Set $x_0:=-\infty$ and $x_{k'}:=+\infty$. Consider the function $h(x) := P'(x) / P(x)$. It is monotone decreasing on each interval $(x_{i-1},x_{i})$, $i\in\brx{k'}$, since the second derivative of $\log|P(x)| = \sum_{i=1}^{k'-1} \log|x-x_i|$ is negative. See \figref{fig:Cheb}(right) for an example.

\begin{figure}
 \includegraphics[scale=1]{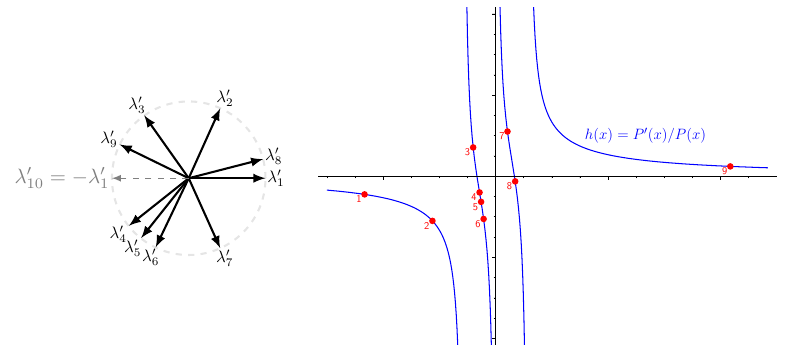}
 \caption{\label{fig:Cheb}The column vectors of $\la'$ (left) and the corresponding solutions to $h(x)=\slp_i$ (right); see the proof of \cref{prop:from_la_to_La}.}
\end{figure}

Let us adjust $\la'$ using $\SL_2(\R)$-action so that $\la'_1 = (1,\epsilon)^T$ for some small $\epsilon>0$, and so that the first row of $\la'$ has no zero entries. We denote $\la'_0:=(1,0)^T$ and $\la'_{n+1}:=(-1)^{k'-1}\la'_0$. For $j=1,\dots,n,n+1$, we let 
$\phi_j:= \sum_{i=1}^j \Arg_{(-\pi,\pi]}(\la'_{i-1},\la'_{i})$. 
 Thus, $0<\phi_1<\pi/2$ and $\phi_{n+1} = (k'-1)\pi$. Let us consider a partition $\brn=B_0\sqcup B_1\sqcup\cdots\sqcup B_{k'-1}$ of $\brn$ into $k'$ intervals, so that all $i\in B_j$ satisfy 
$\pi/2+(j-1)\pi< \phi_i< \pi/2+j\pi$. For example, if $\la'$ is given in \figref{fig:Cheb}(left) then we get the partition $\brn = \{1,2\}\sqcup\{3,4,5,6\}\sqcup\{7,8\}\sqcup\{9\}$.

Let $\slp_i:=-\la'_{2,i}/\la'_{1,i}$ for $i\in\brn$. It is clear that if $i<i'$ belong to the same block then $\slp_i > \slp_{i'}$. 
For $i\in B_j$, let $x=y_i$ be the unique solution to $h(x) = \slp_i$ on the interval $(x_{j},x_{j+1})$; see \figref{fig:Cheb}(right). 
It follows that $y_1<y_2<\cdots<y_n$, and that there exist $t_1,t_2,\dots,t_n\in\Rtp$ such that $\la'_{1,i} = t_iP(y_i)$ and $\la'_{2,i} = -t_i P'(y_i)$. Letting $A$ be the $2\times k'$ matrix comprised of the coefficients of $P$ and $-P'$, we see that $\la' = A\cdot \VandM \cdot \tdiag$, as desired. 
\end{proof}

\begin{corollary}
 The set $\TRIPLES$ is connected.
\end{corollary}
\begin{proof}
By \crefi{prop:momLL_basic}{mom4}, we have a continuous map $\LaLak\times\Grnd(k,n)\to\TRIPLES$ sending $(\La,\Lat,C)\mapsto (\PhiLL(C),C)$. By \cref{prop:from_la_to_La}, this map is surjective: given $(\la,\lat,C)\in\TRIPLES$, letting $\LaLat$ be obtained from $\lalat$ via \cref{prop:from_la_to_La}, we get $\la\subset\La\cap C$ and $\lat\subset\Lat\cap C^\perp$. By \crefi{prop:momLL_basic}{mom1}, $\dim(\La\cap C)=\dim(\Lat\cap C^\perp) = 2$, so $\la=\La\cap C$ and $\lat=\Lat\cap C^\perp$, i.e., $\lalat=\PhiLL(C)$. 
 Since the set $\LaLak\times\Grnd(k,n)$ is clearly connected, we conclude that $\TRIPLES$ is also connected.
\end{proof}
\begin{remark}\label{rmk:connected}
A similar argument shows that the individual sets $\{(\la,C)\in\lak\times\Ptp_f\mid\la\subset C\}$ and $\{(\lat,C)\in\latk\times\Ptp_g\mid C\subset\latp\}$
are connected for any $f\in\Bounda$ and $g\in\Boundb$.
\end{remark}

\subsection{The cyclically symmetric momentum amplituhedron}\label{sec:cs_mom_ampl}
We develop the notion of a cyclically symmetric momentum amplituhedron following~\cite[Section~5]{GKL1}. 

Let $\Shift:\R^n\to\R^n$ be the cyclic shift operator given by~\eqref{eq:Shift_dfn}. 
We will choose two special matrices $(\Lacs,\Latcs)\in\LaLak$ such that $\Momcs$ is invariant under a certain natural action of $\Shift$.

Recall from \cref{sec:backgr:cyc} that $\Shift$ preserves $\Grtnn(k,n)$ and $\Grtp(k,n)$. 
 Thus, the map $(\La,\Lat)\mapsto (\La\cdot \Shift,\Lat\cdot \Shift)$ preserves $\LaLak$. 

By~\cite{Karp}, for all $0\leq k\leq n$, there exists a unique point $\Xcs\in\Grtnn(k,n)$ such that $\Xcs = \Xcs\cdot \Shift$, and in fact, we have $\Xcs\in\Grtp(k,n)$. We will describe $\Xcs$ explicitly below.
\begin{definition}
The \emph{cyclically symmetric momentum amplituhedron} $\Momcs$ is the momentum amplituhedron associated with $\Lacs:=\alt(\XcsLa)$ and $\Latcs:=\XcsLat$.
\end{definition}

The most convenient way to work with $\Momcs$ is to use the complex Grassmannian $\Gr(k,n;\C)$. An element $X\in\Gr(k,n)$ of the real Grassmannian gives rise to an element $X\oplus\I X\in\Gr(k,n;\C)$ represented by the same $k\times n$ matrix, where $\I:=\sqrt{-1}$. Conversely, given a subspace $Y\in\Gr(k,n;\C)$, we say that $Y$ is \emph{real} and write $Y\in\Gr(k,n)$ if it is invariant under complex conjugation. In that case, the subspace (over $\R$) of vectors in $Y$ preserved by complex conjugation is an element of the real Grassmannian $\Gr(k,n)$. Alternatively, $Y\in\Gr(k,n;\C)$ is real if and only if the ratio of any two nonzero Pl\"ucker coordinates of $Y$ is real.

Let $\eig_1,\eig_2,\dots,\eig_n\in\C$ be the eigenvalues of $\Shift$. They are the $n$-th roots of $(-1)^{k-1}$, and we order them so that $\Re(\eig_1)\geq\Re(\eig_2)\geq\cdots\geq\Re(\eig_n)$. Let $\eigvec_1,\eigvec_2,\dots,\eigvec_n\in\C^n$ be the (orthonormal) eigenvectors corresponding to $\eig_1,\eig_2,\dots,\eig_n$, and let $\Eig$ be the (unitary) $n\times n$ matrix whose rows are $\eigvec_1,\eigvec_2,\dots,\eigvec_n$. Let $\Id_r$ denote the $r\times r$ identity matrix and $\bzero_{a\times b}$ denote the $a\times b$ zero matrix. The point $\Xcs\in\Grtnn(k,n)$ is given (as a real element of $\Gr(k,n;\C)$) by
\begin{equation*}%
 \Xcs := \mat[\Id_k|\bzero_{k\times(n-k)}] \cdot \Eig = \Span(\eigvec_1,\eigvec_2,\dots,\eigvec_k).
\end{equation*}
Here, $\mat[\Id_k|\bzero_{k\times(n-k)}]$ denotes a block $k\times n$ matrix. Consequently,
\begin{equation*}%
 \Lacs = \mat[\bzero_{(n-k+2)\times(k-2)} | \Id_{n-k+2}]\cdot \Eig \quad\text{and}\quad
 \Latcs = \mat[\Id_{k+2}|\bzero_{(k+2)\times(n-k-2)}]\cdot \Eig.
\end{equation*}

In these coordinates, the map $\PhiLL$ has a particularly simple form. It was shown in~\cite[Proposition~3.4]{GKL1} that any point $C\in\Grtnn(k,n)$ can be represented as $\mat[\Id_k|P_{k\times(n-k)}]\cdot \Eig$ for some $k\times(n-k)$ complex matrix $P_{k\times(n-k)}$. We have $C^\perp = \mat[-P^\Hast_{(n-k)\times k}|\Id_{n-k}]\cdot \Eig$, where $P^\Hast_{(n-k)\times k}$ denotes the conjugate transpose of $P_{k\times(n-k)}$.

\def\mmscl{1}

Let us split $P_{k\times(n-k)} = \begin{pmatrix}
A_{(k-2)\times 2} & B_{(k-2)\times(n-k-2)} \\
C_{2\times 2} & D_{2\times(n-k-2)} 
\end{pmatrix}$ into four blocks, indicating the block sizes in subscripts. 
 Computing $\lalat = \Phics(C)$ with $\la = C\cap \Lacs$ and $\lat = C^\perp\cap \Latcs$, we find
\begin{equation*}%
 \la = 
\medmatrix{\mmscl}{
 \bzero_{2\times(k-2)} & \Id_{2} & C_{2\times 2} & D_{2\times(n-k-2)} 
}\cdot \Eig,
\quad\text{and}\quad
 \lat = 
\medmatrix{\mmscl}{
-A^\Hast_{2\times(k-2)} & -C^\Hast_{2\times 2} & \Id_{2} & \bzero_{2\times(n-k-2)} 
}\cdot \Eig.
\end{equation*}
In particular, $(\lacs,\latcs):=\Phics(\Xcs)$ is obtained by setting $P_{k\times(n-k)}:=\bzero_{k\times(n-k)}$. Thus,
\begin{equation*}%
 \lacs = \Span(\eigvec_{k-1},\eigvec_k) \quad\text{and}\quad \latcs = \Span(\eigvec_{k+1},\eigvec_{k+2}).
\end{equation*}

\noindent The vectors $\eigvec_{k-1},\eigvec_k$ are complex conjugates of each other, and similarly for $\eigvec_{k+1},\eigvec_{k+2}$. Thus, $\lacs$ and $\latcs$ are indeed real elements of $\Gr(2,n)$. After acting by $\GL_2(\C)$ on $\lacs$ and $\latcs$, for all $i\in\brn$, we get
\begin{equation}\label{eq:lacs_latcs_regangle}
 \lacs_i = (\cos(i\regangle_{k-1}),\sin(i\regangle_{k-1})) \quad\text{and}\quad
 \latcs_i = (\cos(i\regangle_{k+1}),\sin(i\regangle_{k+1})),
 \quad\text{where $\regangle_{k\pm1} := \frac{k\pm1}{n}\pi$}.
\end{equation}
 Observe that $\lacs\cdot \Shift = \lacs$ and $\latcs\cdot \Shift = \latcs$. Thus, $(\lacs,\latcs)$ is a cyclically symmetric point inside the cyclically symmetric momentum amplituhedron $\Momcs$. A t-immersion of the form $\Tllcs$ is bounded by a regular $n$-gon and satisfies
$\sumwT_i = \regangle_{n-k-1}$ and $\sumbT_i = \regangle_{k-1}$ for all $i\in\brn$.

Directly adapting the argument in~\cite[Section~5]{GKL1} to the case of the momentum amplituhedron yields the following result that will not be used in the rest of the paper.
\begin{proposition}
The cyclically symmetric momentum amplituhedron $\Momcs$ is homeomorphic to a closed $(2n-4)$-dimensional ball. 
\end{proposition}

\subsection{The magic projector \texorpdfstring{$\Qla$}{Q\_lambda}}\label{ssec:MOM:magic_Qla_mom_ampl}
Fix a $2$-plane $\la\in\Gr(2,n)$ satisfying $\brla<i,i+1> \neq0$ for all $i\in\brn$. 
We discuss properties of a certain linear operator $\Qla$ introduced in~\cite{AHCC}, where it was used to rewrite the amplitude as an integral over the momentum-twistor space rather than the momentum space. Following the wording of~\cite{abcgpt}, we refer to $\Qla$ as \emph{the magic projector}. By~\cite[Equation~(8.23)]{abcgpt}, for a matrix $C=\mat[C_1|C_2|\cdots|C_n]\in\Gr(k,n)$, the matrix $C\cdot \Qla$ has columns
\begin{equation}\label{eq:intro:Qla}
 (C\cdot \Qla)_i = \frac{1}{\brla<i-1,i>\brla<i,i+1>} \left(
C_{i-1}\brla<i,i+1> + C_i\brla<i+1,i-1> + C_{i+1}\brla<i-1,i>
\right) 
\quad\text{for $i\in\brn$.}
\end{equation}
 Here, we again use the twisted cyclic symmetry and set $C_{i+n} := (-1)^{k-1}C_i$ for all $i\in\Z$.

We denote $\Gr(k-2,\lap):=\{\ddC\in\Gr(k-2,n)\mid \ddC\subset\lap\}$ and $\Grsupla(k,n):=\{C\in\Gr(k,n)\mid \la\subset C\}$ and set 
$\Grtnn(k-2,\lap):=\Grtnn(k-2,n)\cap\Gr(k-2,\lap)$ and $\Grtnnsupla(k,n):=\Grsupla(k,n)\cap \Grtnn(k,n)$. 
 It is well known that $\Qla^T = \Qla$ and $\Ker\Qla = \la$. 
Thus, for $C\in\Grsupla(k,n)$, we have $C\cdot \Qla\in\Gr(k-2,\lap)$. Abusing notation, we denote the resulting map $\Qla:\Grsupla(k,n)\to\Gr(k-2,\lap)$ by the same symbol.

\begin{definition}\label{dfn:TREE:Qlapp}
Let $\Qlapp:\Gr(k-2,\lap)\to\Grsupla(k,n)$ be the operator sending 
$\ddC\in\Gr(k-2,\lap)$ to the set-theoretic preimage $\Qlapp(\ddC)$ of $\ddC$ under the linear map $\Qla:\R^n\to\lap$. 
\end{definition}

The following result is straightforward.
\begin{lemma}\label{lemma:TREE:magic_homeo}
$\Qla$ and $\Qlapp$ are mutually inverse isomorphisms between $\Grsupla(k,n)$ and $\Gr(k-2,\lap)$. 
\end{lemma}

Recall from~\eqref{eq:altp_Delta} that we also have an isomorphism $\altp:\Gr(k,n)\xrasim \Gr(n-k,n)$ that restricts to a homeomorphism $\altp:\Grtnn(k,n)\xrasim\Grtnn(n-k,n)$. 

\begin{lemma}\label{lemma:Qlapp_duality}
Let $\la\in\Matdnr$ be such that $\brla<i,i+1> >0$ for all $i\in\brn$, and let $\ladual:=\alt(\la)$. Then 
\begin{equation}\label{eq:Qlapp_duality}
 \Qlapp = \altp\circ \Qladual \circ \altp \quad\text{as maps from $\Gr(k-2,\lap)$ to $\Grsupla(k,n)$}.
\end{equation}
\end{lemma}
\begin{proof}
Let $J:=\diag(1,-1,1,\dots,(-1)^{n-1})$. Then $\ladual = \la\cdot J$, so $\Qladual = J \Qla J$. Let $\ddC\in\Gr(k-2,\lap)$. 
Since $\altp(\ddC) = \ddC^\perp J = (\ddC J)^\perp$, 
\begin{equation}\label{eq:Qlapp_duality_proof}
 C:=(\altp\circ\Qladual\circ\altp)(\ddC) = (\ddC^\perp J\cdot J\Qla J\cdot J)^\perp = (\ddC^\perp\Qla)^\perp.
\end{equation}
We have $\ddC^\perp\in\Grsupla(n-k+2,n)$ and $\ddC^\perp\Qla\in\Gr(n-k,\lap)$, so $C\in\Grsupla(k,n)$. 
 Since $\Ker\Qla = \la$, it suffices to check that $C\Qla = \ddC$, or equivalently, $C\Qla\perp \ddC^\perp$. Indeed, since $\Qla=\Qla^T$, we have $C\Qla\cdot (\ddC^\perp)^T = C\cdot (\ddC^\perp \Qla)^T$, which is zero by~\eqref{eq:Qlapp_duality_proof}. 
\end{proof}

For the following definition, see~\cite[Equation~(8.25)]{abcgpt} or~\cite[Section~5.1]{LPW}.
\begin{definition}[T-duality for bounded affine permutations]\label{dfn:T_dual_perm}
For any $\fap\in\Bounda$, we define $\fdd\in\BND_{0,2}(k-2,n)$ by
\begin{equation}\label{eq:fdd_dfn}
 \fdd(i) := f(i-1)-1 \quad\text{for all $i\in\Z$.}
\end{equation}
\end{definition}

\begin{lemma}[{\cite[Lemma~5.8]{LPW}}]
\label{lemma:la_vs_Bounda_Boundb}
Let $\la\in\Matdnr$ be such that $\brla<i,i+1> \neq0$ for all $i\in\brn$. Then $\Grsupla(k,n)\cap \Ptp_\fap=\emptyset$ unless $\fap\in \Bounda$ and $\Gr(k-2,\lap)\cap\Ptp_{\ddfperm}=\emptyset$ unless $\ddfperm\in\Bounddb$. 
\end{lemma}
\begin{proof}
If $C\in\Grsupla(k,n)\cap\Ptp_{\fap}$ then since $\brla<i,i+1> \neq0$, we have $\rank\mat[C_i|C_{i+1}]=2$, which by~\eqref{eq:dfn_fC} implies that $\fap(i)\geq i+2$. The case of $C\in\Gr(k-2,\lap)\cap\Ptp_{\ddfperm}$ is similar. 
\end{proof}

\begin{lemma}[{\cite[Equation~(8.25)]{abcgpt}}]\label{lemma:Qla_Pio}
Let $\la\in\Matdnr$ be such that $\brla<i,i+1> \neq0$ for all $i\in\brn$. Let $C\in\Pio_f$ for some $f\in\Bounda$ be such that $\la\subset C$. Then $C\cdot \Qla\in\Pio_{\fdd}$.
\end{lemma}
\begin{proof}
Write $C = \begin{pmatrix}
\la \\
\Chat
\end{pmatrix}$ for some $(k-2)\times n$ matrix $\Chat$. Let $\Icalr_f = (\Ir_i)_{i\in\Z}$ and $\Icalr_{\ddfperm} = (\ddIr_i)_{i\in\Z}$. By~\eqref{eq:gr_neck_dfn}, they are related as
$\ddIr_{i+1} = \{j-1\mid j\in \Ir_i\setminus\{i,i+1\}\}$ for all $i\in\Z$. 
Our goal is to show that $\Icalr_{\ddfperm}$ coincides with the Grassmann necklace $\Icalr_{C\cdot \Qla}=(\Ir'_i)_{i\in\Z}$ of $C\cdot \Qla$. 
By the twisted cyclic symmetry, it suffices to check that $\Ir'_2 = \ddIr_2$. Recall that $\la$ satisfies $\brla<1,2>\neq0$. This implies that $1,2\in\Ir_1$, and so $f(1)\geq 3$. After applying some row operations to $C$, we may assume that the columns $\Chat_1,\Chat_2$ are zero. We further put $\Chat$ into reduced row echelon form so that the columns $(\Chat_j)_{j\in\Ir_1\setminus\{1,2\}}$ form an identity submatrix. The submatrix of $\Chat\cdot \Qla$ with column set $\ddIr_2$ is upper triangular with diagonal entries $\frac1{\brla<j,j+1>}$ for $j\in \ddIr_2$. 
This implies that $\Ir'_2 = \ddIr_2$ and $\Delta_{\Ir_1}(C)=\brla<1,2>\cdot \Delta_{\ddIr_2}(\Chat\cdot\Qla)\cdot \prod_{j\in\ddIr_2}\brla<j,j+1> $. More generally, %
\begin{equation}\label{eq:gr_neck_minors_C_vs_Chat_Qla}
 \Delta_{\Ir_i}(C) = \Delta_{\ddIr_{i+1}}(\Chat\cdot\Qla) \cdot \prod_{j\in\{i\}\sqcup\ddIr_{i+1}} \brla<j,j+1> \quad\text{for $i\in\Z$.}
\end{equation}
This yields a generalization of~\cite[Equation~(8.24)]{abcgpt}.
\end{proof}

Next, we show
 that $\Qla$ preserves total positivity when $\la\in\lak$.
\begin{proposition}\label{prop:Qla_Ptp}\ 
 If $\la\in\lak$ and $C\in\Ptp_f\cap\Grsupla(k,n)$ then $C\cdot \Qla\in\Ptp_{\fdd}\cap\Gr(k-2,\lap)$.
\end{proposition}
\begin{proof}
Let $\laCf:=\{(\la,C)\in\lak\times\Ptp_f\mid\la\subset C\}$. We have a continuous map $\laCf\to\Gr(k-2,n)$ sending $(\la,C)\mapsto C\cdot \Qla$. By \cref{lemma:Qla_Pio}, this map lands inside $\Pio_{\fdd}$. By \cref{rmk:connected}, the set $\laCf$ is connected, and therefore its image is contained inside a single connected component of $\Pio_{\fdd}$. It is well known~\cite{Rietsch_alg} that $\Ptp_{\fdd}$ is a connected component of $\Pio_{\fdd}$. Thus, it is enough to show that $C\cdot \Qla\in\Ptp_{\fdd}$ for at least one pair $(\la,C)\in\laCf$.

We first prove the result for the case of the top cell where $f = \fkn$ and $\Ptp_f = \Grtp(k,n)$; see \cref{sec:backgr:Grtnn}. In this case, we may take cyclically symmetric $\la:=\lacs\in\lak$ and $C:=\Xcs\in\Grtp(k,n)$ defined in \cref{sec:cs_mom_ampl}. We have $\Qla = \frac{-2}{\tan(\regangle_{k-1})}\Id_n + \frac1{\sin(\regangle_{k-1})} (\Shift+\Shift^{-1})$ for $\regangle_{k-1} = \frac{(k-1)\pi}{n}$ as in~\eqref{eq:lacs_latcs_regangle}. 
Thus, $\Qla$ is a linear combination of $\Id_n,\Shift,\Shift^{-1}$. Since $\la$ and $C$ are spanned by eigenvectors of $\Shift$ and since $\Qla$ annihilates $\la$, we see that $C\cdot \Qla$ is the span of the eigenvectors $\eigvec_1,\eigvec_2,\dots,\eigvec_{k-2}$ of $\Shift$ contained in $C$ but not in $\la$. Therefore, $C\cdot \Qla = \Xcsx{k-2} \in \Grtp(k-2,n)$. This completes the proof for the top cell $f = \fkn$. 

Now, let $f\in\Bounda$ be arbitrary. Using \cref{prop:from_la_to_La}, we find $\La\in\alt(\Grtp(n-k+2,n))$ such that $\la = C \cap \La$. Approximating $C$ by elements $C'\in\Grtp(k,n)$, we let $\la':=C'\cap \La$.
By \cref{prop:momLL_basic}, we have $\la'\in\lak$ and the map $\PhiLL$ is continuous at $C$, so $\la'$ approximates $\la$. 
Thus, $C\cdot \Qla$ can be approximated by elements $C'\cdot Q_{\la'}$ that were shown above to belong to $\Grtp(k-2,n)$. It follows that $C\cdot \Qla\in\Grtnn(k-2,n)$. By \cref{lemma:Qla_Pio}, $C\cdot \Qla\in\Pio_{\fdd}$, so by~\eqref{eq:Ptp_intersection}, we get $C\cdot \Qla\in\Ptp_{\fdd}$.
\end{proof}

\begin{corollary}\label{prop:magic_homeo}
Let $\la\in\lak$. 
Then the operators $\Qla$ and $\Qlapp$ restrict to mutually inverse homeomorphisms between $\Grtnnsupla(k,n)$ and $\Grtnn(k-2,\lap)$. Furthermore, for any $\fap\in\Bounda$, they restrict to mutually inverse homeomorphisms between $\Grsupla(k,n)\cap \Ptp_\fap$ and $\Gr(k-2,\lap)\cap\Ptp_{\ddfperm}$.
\end{corollary}
\begin{proof}
By \cref{lemma:TREE:magic_homeo}, 
the operators $\Qla$ and $\Qlapp$ are mutually inverse and continuous. 
By \cref{prop:Qla_Ptp}, $\Qla$ restricts to a map 
$\Grsupla(k,n)\cap \Ptp_\fap\to\Gr(k-2,\lap)\cap\Ptp_{\ddfperm}$. 
The statement that $\Qlapp$ restricts to a map $\Gr(k-2,\lap)\cap\Ptp_{\ddfperm}\to\Grsupla(k,n)\cap \Ptp_\fap$ follows from~\eqref{eq:altp_vs_fhat} and~\eqref{eq:Qlapp_duality}.
\end{proof}

\section{Dimer formulae for \KSprims}\label{sec:dimer_KS}
We give explicit identities for discrete holomorphic functions and their \KSprims introduced in \cref{sec:intro:holomorphic} in terms of single-dimer and double-dimer configurations on $\G$. 

\subsection{Surplus}\label{ssec:surplus}

We study the surplus function $\helmin(\G)$ introduced in~\eqref{eq:surplus_dfn}. 
See~\cite[Section~4.1]{Kenyon_Sheffield} for closely related analysis.\footnote{We thank Richard Kenyon for bringing the results of~\cite[Section~4.1]{Kenyon_Sheffield} to our attention.} 
The following result is a variant of Hall's theorem; see~\cite[Theorem~1.3.1]{Lovasz_Plummer}.
\begin{proposition}\label{prop:DIM:helWBmin_geq_0=>apm_exists}
$\G$ admits an \APM if and only if $\helmin(\G)\geq0$.
\end{proposition}

\begin{lemma}\label{lemma:DIM:deleting_hel_verts}
Fix integers $\kw,\kb\geq0$. The following are equivalent:
\begin{enumerate}[label=(\arabic*)]
\item\label{DIM:del1} $\helWmin(\G)\geq\kw$ and $\helBmin(\G)\geq\kb$;
\item\label{DIM:del2} for any $\XW\subset\WV$ and $\XB\subset\BV$ with $|\XW|\leq \kw$ and $|\XB|\leq\kb$, $\Grem{(\XW\sqcup\XB)}$ 
 admits an \APM.\footnote{Some of the boundary vertices of $\Grem{(\XW\sqcup\XB)}$ may have degree $0$. The definition of an \APM extends verbatim to graphs whose boundary vertices have arbitrary degree.
}
\end{enumerate}
\end{lemma}
\begin{proof}
 \itemref{DIM:del1}$\Longrightarrow$\itemref{DIM:del2}: Let $\G':=\Grem{(\XW\sqcup\XB)}$. By \cref{prop:DIM:helWBmin_geq_0=>apm_exists}, it suffices to show that $\helmin(\G')\geq0$. Let $\Rg\in\WNEI(\G')$. Thus, $\RgBV\neq\emptyset$. Let 
$\Rg_+:=\Rg\cup\NeighG(\RgBV)$. 
It follows that $\Rg_+\in\WNEI(\G)$, so $\helW(\Rg_+)\geq\helWmin(\G)$. We have $|\RgWV|\geq |\RgWV_+| - |\XW|\geq |\RgWV_+| - \helWmin(\G)\geq |\RgWV_+| - \helW(\Rg_+) = |\RgBV_+| = |\RgBV|$. Thus, $\helWp(\Rg)\geq0$ for all $\Rg\in\WNEI(\G')$, so $\helWmin(\G')\geq0$. We similarly obtain $\helBmin(\G')\geq0$.

 \itemref{DIM:del2}$\Longrightarrow$\itemref{DIM:del1}: Let $\Rg\in\WNEI(\G)$. We need to show that $\helW(\Rg)\geq\kw$. Applying \cref{prop:DIM:helWBmin_geq_0=>apm_exists} and part~\itemref{DIM:del2} for $\XW=\XB=\emptyset$, we get $\helW(\Rg)\geq0$. Thus, $|\RgWV|\geq|\RgBV|\geq1$.
Let $\XB:=\emptyset$ and let $\XW\subset\RgWV$ be any subset with $|\XW|\leq\kw$. Then $\Rg\setminus\XW \in\WNEI(\G\rem\XW)$ and $\helWsub_{\G\rem\XW}(\Rg\setminus\XW)=\helW(\Rg)-|\XW|$. 
Applying \cref{prop:DIM:helWBmin_geq_0=>apm_exists} and part~\itemref{DIM:del2} again, we get
$\helWmin(\G\rem\XW)\geq0$, so $\helW(\Rg)\geq|\XW|$. If $|\RgWV|<\kw$ then we can set $\XW:=\RgWV$. This gives $\helW(\Rg)\geq|\RgWV|$, contradicting $\RgBV\neq\emptyset$. Otherwise, if $|\RgWV|\geq\kw$ then we can take $\XW$ to be any $\kw$-element subset of $\RgWV$ which gives $\helW(\Rg)\geq\kw$, as desired. 
\end{proof}

\begin{corollary}\label{lemma:hel_geq_1}
Suppose that $\helmin(\G)\geq1$. Then every edge of $\G$ appears in an \APM of $\G$.
\end{corollary}

We discuss the effect of moves \MVbd and \MV1 on $\helWmin(\G)$ and $\helBmin(\G)$. First, observe that for any graph $\G$ satisfying $\helWmin(\G)\geq2$, we have $\degG(\b)\geq3$ for all $\b\in\BVint$. In \cref{lemma:Tll_is_t_imm}, we will need to apply uncontraction moves \MV1 to make all interior black vertices of $\G$ trivalent. The following lemma shows that such an operation preserves the inequality $\helWmin(\G)\geq2$. 

\begin{lemma}\label{lemma:DIM:moves_vs_helWmin_trivalent}
Suppose that $\G_2$ is obtained from $\G_1$ by applying move~\MV1, replacing a black vertex $\b\in\BVint$ 
of degree at least $4$
 with a pair $(\b_{\Lop},\b_{\Rop})$ of black vertices 
 adjacent to a degree-$2$ white vertex $\w$. Suppose that 
$|\Neigh_{\G_2}(\b_{\Lop})|,|\Neigh_{\G_2}(\b_{\Rop})|\geq3$. 
If $\helWmin(\G_1)\geq2$ and $\helBmin(\G_1)\geq0$ then $\helWmin(\G_2)\geq2$ and $\helBmin(\G_2)\geq0$. 
\end{lemma}
\begin{proof}
Let $\Rg_2 \in \WNEI(\G_2)$. We show $\helWsub_{\G_2}(\Rg_2) \geq 2$. 
If $\b_{\Lop}, \b_{\Rop} \in \RB_2$ then $\w \in \RW_2$. Letting $\Rg_1:=(\Rg_2\setminus\{\b_{\Lop},\w,\b_{\Rop}\})\sqcup\{\b\}$, we have $\Rg_1 \in \WNEI(\G_1)$ and $\helWsub_{\G_2}(\Rg_2) = \helWsub_{\G_1}(\Rg_1) \geq 2$. 
If $\b_{\Lop}, \b_{\Rop} \notin \RB_2$
 then $\Rg_1:=\Rg_2\setminus\{\w\} \in \WNEI(\G_1)$ and $\helWsub_{\G_2}(\Rg_2)\geq\helWsub_{\G_1}(\Rg_1) \geq 2$. 
Suppose exactly one of $\b_{\Lop},\b_{\Rop}$ is in $\Rg_2$, say, $\b_{\Lop}\in\RB_2$. Then $\w \in \RW_2$. Let $\Rg_1 := \Rg_2 \setminus \{\b_{\Lop}, \w\}$. We have $\b \notin \Rg_1$ and $\Rg_1 \in \WNEIg(\G_1)$. If $\RB_1 = \emptyset$ then $\helWsub_{\G_2}(\Rg_2) = |\RW_2| - 1\geq2$ since $\RB_2=\{\b_{\Lop}\}$ and thus $|\RW_2|\geq|\Neigh_{\G_2}(\b_{\Lop})| \geq 3$ by assumption. 
 If $\RB_1 \neq \emptyset$ then $\Rg_1 \in \WNEI(\G_1)$. In this case, $\helWsub_{\G_2}(\Rg_2) = \helWsub_{\G_1}(\Rg_1) \geq 2$. Thus $\helWmin(\G_2) \geq 2$.

Now let $\Rg_2 \in \BNEI(\G_2)$. We show $\helBsub_{\G_2}(\Rg_2) \geq 0$. 
If $\w \in \Rg_2$ then $\b_{\Lop}, \b_{\Rop} \in \Rg_2$. Setting $\Rg_1:=(\Rg_2\setminus\{\b_{\Lop},\w,\b_{\Rop}\})\sqcup\{\b\}$ yields $\Rg_1 \in \BNEIg(\G_1)$ with $\helBsub_{\G_2}(\Rg_2) = \helBsub_{\G_1}(\Rg_1) \geq 0$. Here, we have used that if $\helBmin(\G_1)\geq0$ then $\helBsub_{\G_1}(\Rg_1) \geq 0$ for all $\Rg_1\in\BNEIg(\G_1)$. 
If $\w \notin \Rg_2$, let $\Rg_1 \in \BNEI(\G_1)$ be obtained by replacing $\Rg_2 \cap \{\b_{\Lop}, \b_{\Rop}\}$ with $\{\b\}$ (if nonempty). We have $|\RW_2| = |\RW_1|$ and $|\RB_2| \geq |\RB_1|$. Thus, $\helBsub_{\G_2}(\Rg_2) \geq \helBsub_{\G_1}(\Rg_1) \geq 0$, and so $\helBmin(\G_2) \geq 0$.
\end{proof}

\begin{lemma}\label{lemma:DIM:moves_vs_helWmin}
Suppose that $\G_1,\G_2$ are related by moves \MVbd or \MV1. Then for all $\kw,\kb\in\{0,1\}$, 
\begin{equation*}%
 \helWmin(\G_1)\geq\kw,\ \helBmin(\G_1)\geq\kb \quad \Longleftrightarrow\quad 
 \helWmin(\G_2)\geq\kw,\ \helBmin(\G_2)\geq\kb.
\end{equation*}
\end{lemma}
\begin{proof}
A local check, similar to the proof of \cref{lemma:DIM:moves_vs_helWmin_trivalent}, shows that moves \MVbd and \MV1 preserve both $\min(\helWmin(\G),1)$ and $\min(\helBmin(\G),1)$.
\end{proof}

\begin{lemma}\label{lemma:reduced_helmin}
Let $\G$ be a reduced (\cref{dfn:reduced}) planar bipartite graph of type $(k,n)$. 
\begin{enumerate}[label=(\arabic*)]
\item\label{reduced_helmin1} If $\degG(\v)\geq2$ for all $\v\in\Vint$ then $\helmin(\G)\geq1$.
\item\label{reduced_helmin2} If $\degG(\v)\geq3$ for all $\v\in\Vint$ then $\helmin(\G)\geq2$.
\end{enumerate}
\end{lemma}

\begin{remark}\label{rmk:reduced_helmin2}
By \cref{dfn:reduced}, a connected reduced graph $\G$ of type $(k,n)$ with $1\leq k\leq n-1$ contains no internal leaves, and thus satisfies the assumption of \crefi{lemma:reduced_helmin}{reduced_helmin1}. Furthermore, if $n\geq3$ and $\G$ is connected then no two boundary vertices of $\G$ are connected by a path consisting entirely of interior degree-$2$ vertices. Thus, one can apply moves \MV1 and \MVbd until each interior vertex of $\G$ has degree at least $3$. The resulting reduced graph is called \emph{contracted}, and \crefi{lemma:reduced_helmin}{reduced_helmin2} applies to it.
\end{remark}

\begin{proof}[Proof of \cref{lemma:reduced_helmin}]
\itemref{reduced_helmin1}: It is well known that every reduced planar bipartite graph admits an \APM, so by \cref{prop:DIM:helWBmin_geq_0=>apm_exists}, $\helmin(\G)\geq0$. 
As we explain in \Nref{lemma:BCFW:Rg_simply_conn_lower_bound}, if $\helmin(\G)=0$ then one can find a simple cycle $\Cyc$ in $\GD$ such that the graph $\G\ind[\Cyc]$ obtained by intersecting $\G$ with the disk enclosed by $\Cyc$ admits an \APM and is of type $(k',n')$ with $k'\in\{0,n'\}$. Thus, any \APM of $\G$ restricts to an \APM of $\G\ind[\Cyc]$. Removing all vertices used by any \APM of $\G\ind[\Cyc]$ from $\G$ 
 results in a graph $\G'$ with the same boundary measurements as $\G$. 
If $\G\ind[\Cyc]$ is not reduced then one can replace it by a reduced graph with the same boundary measurements but strictly fewer faces by \cref{dfn:reduced}. In this case, $\G$ is also not reduced, a contradiction. On the other hand, if $\G\ind[\Cyc]$ is reduced of type $(0,n')$ or $(n',n')$ then it must contain an interior leaf, contradicting the assumption that $\degG(\v)\geq2$ for all $\v\in\Vint$.

\itemref{reduced_helmin2}: By part~\itemref{reduced_helmin1}, $\helmin(\G)\geq1$. 
Suppose that $\helmin(\G)=1$ so that, say, $\helWmin(\G)=1$ and $\helBmin(\G)\geq1$. Then there exists $\Rg\in\WNEI(\G)$ such that $\helW(\Rg)=1$. By \Nref{lemma:DIM:coll_sconn}, we may assume that $\Rg$ is \emph{\sconn} meaning one can find a simple cycle $\Cyc$ in $\GD$ separating the vertices in $\Rg$ from those in $\Verts\setminus\Rg$. 
 Since $\helW(\Rg)=1$, $\G\ind[\Cyc]$ is of type $(1,n')$ for some $n'\geq1$. Similarly to the above, we see that $\G\ind[\Cyc]$ must itself be reduced. A reduced graph of type $(1,n')$ cannot contain an interior black vertex of degree $\geq3$. Since $\Rg\in\WNEI(\G)$, $\G\ind[\Cyc]$ contains an interior black vertex $\b$, and by assumption, $\degG(\b)\geq3$, a contradiction.
\end{proof}

\subsection{Type $(2,4)$: single-dimer configurations}\label{ssec:TE:Kast}
Throughout this subsection, let $(\Gf,\wtf)$ be a weighted planar bipartite graph of type $(\kf,\nf)=(2,4)$;
 see e.g. \figref{fig:intro-annular}(c).
 Assume that $\Gf$ admits an \APM and has black boundary.
 We set $\Gfint:=\Gf\ind[\Vfint]$ and fix a choice $\epsKf$ of Kasteleyn signs on $\Gf$. 

By \cref{lemma:MCE:apm_vs_Kast}, $(\Gf,\wtf)$ admits unique (up to $\GL_2(\R)\times\GL_2(\R)$-action) discrete holomorphic functions 
$(\laextf,\latextf)\in\HHspaceKRdf$. 
 For $\w_1,\w_2\in\WVf$ (resp., $\b_1,\b_2\in\BVf$), we denote
\begin{equation*}%
 \brlaw<\w_1,\w_2>:=\det\mat[\laextf(\w_1)|\laextf(\w_2)]\quad\text{resp.,}\quad \brlatb[\b_1,\b_2]:=\det\mat[\latextf(\b_1)|\latextf(\b_2)].
\end{equation*}
Similarly to~\eqref{eq:intro:y_to_lalat}, we let $(\Fwf,\Fbf)\in\HHspaceKCf$ be given by
\begin{equation}\label{eq:TE:y_to_lalat}
 \laextf(\w)=\CtoM[\Fwf(\w)] \quad\text{and}\quad \latextf(\b)=\CtoMt[\Fbf(\b)] \quad\text{for all $\w\in\WVf$ and $\b\in\BVf$}.
\end{equation} 
Let $\xdGf$ be the \KSprim of $(\Fwf,\Fbf)$. Thus, $\datrQLf:=\datrQf$ is an \emph{\datr} of $\Gf$ in the sense of \cref{dfn:TE:datr} below.

\begin{figure}
\includegraphics[scale=1.2]{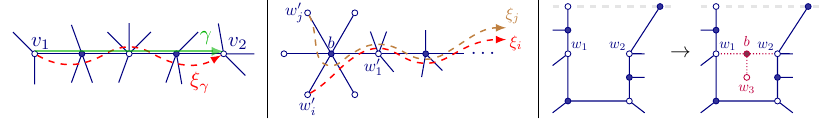}
 \caption{\label{fig:zig-cut} A path $\Path$ and the corresponding zig-zag cut $\CutP$ (left); the zig-zag cuts $\Cut_i,\Cut_j$ in the proof of \cref{lemma:TE:single_dimer_formula} (middle); 
tripod insertion (right).
}
\end{figure}

A \emph{cut} $\Cut$ in $\Gf$ is a walk in $\GDf$. 
 Abusing notation, we identify $\Cut$ with the (multi)set of edges of $\Gf$ that it crosses. 
 For $\v_1,\v_2\in\Vf$, let $\f_1,\f_2$ be the faces of $\Gfrem{\v_1}{\v_2}$ containing $\v_1$ and $\v_2$, respectively. To any simple path $\Path$ in $\Gf$ from $\v_1$ to $\v_2$ we associate a \emph{zig-zag cut} $\CutP$ in $\Gfrem{\v_1}{\v_2}$ that intersects every edge of $\Path$ not incident to $\v_1$ or $\v_2$, keeping all white (resp., black) vertices of $\Path$ other than $\v_1,\v_2$ to the right (resp., to the left); see \figref{fig:zig-cut}(left). More precisely, denote the edges and vertices of $\Path$ by $\Path = (\v_1=\vv_0,\e_1,\vv_1,\e_2,\dots,\e_m,\vv_m = \v_2)$. Then for each intermediate vertex $\vv_i$, $i\in\brx{m-1}$, if $\vv_i$ is white (resp., black) then $\CutP$ crosses all edges incident to $\vv_i$ that are located weakly between $\e_i$ and $\e_{i+1}$ in clockwise (resp., counterclockwise) order around $\vv_i$.

Observe that for any two distinct vertices $\w_1,\w_2\in\WVf$, the graph $\Gfiww$ is planar bipartite of type $(0,0)$, so $\APMS(\Gfiww)$ is the set of perfect matchings of $\Gfiww$. Similarly, for distinct $\b_1,\b_2\in\BVf$, $\Gfbb$ is of type $(4-d,4-d)$ for $d:=|\{\b_1,\b_2\}\cap\Vfbd|$ and $\APMS(\Gfbb)$ is the set of perfect matchings of $\Gfbb$. 
\begin{proposition}\label{lemma:TE:single_dimer_formula}
There exists a constant $\constw\neq0$ such that for distinct $\w_1,\w_2\in\WVf$ and for any simple path $\Pathw$ from $\w_1$ to $\w_2$ in $\Gfint$, 
\begin{equation}\label{eq:IP:brlaw}
 \brlawf<\w_1,\w_2> = \constw \cdot \epsKf(\Pathw)\cdot 
 \sum_{\apmwf\in\APMS(\Gfiww)} (-1)^{|\apmwf\cap \CutPw|} \wtf(\apmwf).
\end{equation}
Similarly,
there exists $\constb\neq0$ such that for distinct $\b_1,\b_2\in\BVf$ and any simple path $\Pathb$ from $\b_1$ to $\b_2$ in $\Gf$,
\begin{equation}\label{eq:IP:brlatb}
 \brlatbf[\b_1,\b_2] = \constb \cdot \epsKf(\Pathb)\cdot 
 \sum_{\apmbf\in\APMS(\Gfbb)} (-1)^{|\apmbf\cap \CutPb|} \wtf(\apmbf).
\end{equation}
\end{proposition}

Recall that we have not assumed that $\Gf$ is connected or \bdconn. The case of arbitrary $\Gf$ can be reduced to the case of connected $\Gf$ using the following operation. 
\begin{definition}[Tripod insertion]\label{dfn:eksbb}
Suppose that $\w_1,\w_2\in\WVf$ share a face of $\Gf$. 
Let $\Gf'$ be obtained from $\Gf$ by adding a trivalent black vertex $\b$ adjacent to $\w_1$, $\w_2$, and a white leaf $\w_3$ by edges $\e_1,\e_2,\e_3$ in clockwise order. 
See \figref{fig:zig-cut}(right).
 Let $\epsKfp$ be an extension of $\epsKf$ to a choice of Kasteleyn signs on $\Gf'$. 
(Such an extension $\epsKfp$ always exists; cf. \cref{lemma:OCP:Kast_even}.)
We set $\epswwf_{\w_1,\w_2}:=\epsKfp(\e_1)\epsKfp(\e_2)$. Similarly, if $\b_1,\b_2\in\BVf$ share a face of $\Gf$, we let $\Gf'$ be obtained by adding a trivalent white vertex $\w$ adjacent to $\b_1$, $\b_2$, and a black leaf $\b_3$ by edges $\e_1,\e_2,\e_3$ in clockwise order. We set $\epsbbf_{\b_1,\b_2}:=\epsKfp(\e_1)\epsKfp(\e_2)$. 
\end{definition}

\begin{remark}\label{rmk:tripods_commute}
Let $\vv_1,\vv_2$ (resp., $\v_1,\v_2$) be vertices of $\Gf$ of the same color and sharing a face. Suppose in addition that the pairs $(\vv_1,\vv_2)$ and $(\v_1,\v_2)$ are \emph{non-crossing}, i.e., adding edges connecting $\vv_1$ to $\vv_2$ and $\v_1$ to $\v_2$ results in a planar (non-bipartite) graph. 
Then tripod insertions at $(\vv_1,\vv_2)$ and $(\v_1,\v_2)$ commute; that is, applying them in any order yields the same choice of Kasteleyn signs on the resulting graph.
This is the case since by \cref{lemma:OCP:Kast_even}, deleting a tripod vertex and its leaf neighbor preserves the Kasteleyn signs on the remaining graph. 
For any graph $\Gf$, one can always insert enough tripods to make $\Gf$ connected. 
Tripod insertions do not affect the boundary measurements of $(\Gf,\wtf)$ as every \APM must use the leaf edge of the tripod. 
Similarly, every discrete holomorphic function on $\Gf$ extends uniquely to the vertices of the inserted tripods.
Thus, for arbitrary $\Gf$, we obtain an analog of~\eqref{eq:IP:brlaw}--\eqref{eq:IP:brlatb}, with the sign of each term adjusted by a product of signs $\epswwf_{\w'_1,\w'_2}$ and $\epsbbf_{\b'_1,\b'_2}$ for the tripods that appear along the paths $\Pathw$ and $\Pathb$. 
\end{remark}

\begin{proof}[Proof of \cref{lemma:TE:single_dimer_formula}]
We prove~\eqref{eq:IP:brlaw}; the proof of~\eqref{eq:IP:brlatb} is similar. 
In view of \cref{rmk:tripods_commute}, we may assume that $\Gf$ is connected. 
Denote the right-hand side of~\eqref{eq:IP:brlaw} (omitting the constant $\constw$) by $\brguess<\w_1,\w_2>_{\Pathw}=\brguess<\w_1,\w_2>_{\Pathw,\CutPw}$. More generally, we let $\brguess<\w_1,\w_2>_{\Pathw,\Cut}$ be obtained from $\brguess<\w_1,\w_2>_{\Pathw,\CutPw}$ by replacing the zig-zag cut $\CutPw$ with an arbitrary cut $\Cut$ from the face of $\Gfiww$ containing $\w_1$ to the face containing $\w_2$.

Our first goal is to show that $\brguess<\w_1,\w_2>_{\Pathw}$ does not depend on the choice of $\Pathw$. Suppose that a cut $\Cut'$ is obtained by passing another cut $\Cut$ through a vertex $\v$ of $\Gfiww$.
 Since every $\apmwf\in\APMS(\Gfiww)$ uses $\v$ exactly once, we see that $(-1)^{|\apmwf\cap\Cut|} = -(-1)^{|\apmwf\cap\Cut'|}$. Thus, $\brguess<\w_1,\w_2>_{\Pathw,\Cut} = -\brguess<\w_1,\w_2>_{\Pathw,\Cut'}$. Similarly, if $\Pathwp$ is another path from $\w_1$ to $\w_2$ such that $\Gf$ contains a unique (necessarily interior) face $\f$ located between $\Pathw$ and $\Pathwp$ then $\epsKf(\Pathw)$ differs from 
$\epsKf(\Pathwp)$ by 
$\epsKf(\bdrypath\f)=(-1)^{\dwcor(\f)+1}$; 
 see~\eqref{eq:Kast_sign}. A straightforward check shows that there are exactly $\dwcor(\f)-1$ vertices of $\Gfint$ located between $\CutPw$ and $\CutPwp$. Thus, $\brguess<\w_1,\w_2>_{\Pathw} = \brguess<\w_1,\w_2>_{\Pathwp}$, as desired. From now on, we denote $\brguess<\w_1,\w_2>:=\brguess<\w_1,\w_2>_{\Pathw}$ for any choice of a simple path $\Pathw$ from $\w_1$ to $\w_2$.

Next, we claim that $\brguess<\w_1,\w_2>$ is a \wdash holomorphic function of $\w_1$ and of $\w_2$ satisfying $\brguess<\w_1,\w_2>=-\brguess<\w_2,\w_1>$. First, since every path $\Pathw$ from $\w_1$ to $\w_2$ contains an odd number of vertices, and since reversing the path $\Pathw$ corresponds to passing $\CutPw$ through each vertex on $\Pathw$ other than $\w_1$ and $\w_2$, we see that indeed $\brguess<\w_1,\w_2>=-\brguess<\w_2,\w_1>$. 
We now check \wdash holomorphicity in each argument. 
Let $\w\in\{\w_1,\w_2\}$. Let $\b\in\BVfint$, and let $\e_1,\e_2,\dots,\e_m$ be the edges connecting $\b$ to white vertices $\w'_1,\w'_2,\dots,\w'_m$, respectively, listed in clockwise order around $\b$. 
After applying \MV1 if necessary, we may assume that $\b$ is not incident to any parallel edges so that the vertices $\w'_1,\w'_2,\dots,\w'_m$ are pairwise distinct. 
Pick a shortest path $\Pathw$ from $\b$ to $\w$ and suppose that it first passes through $\w'_1$. Let $\Pathw_1$ be the subpath of $\Pathw$ connecting $\w'_1$ to $\w$. For $2\leq i\leq m$, let $\Pathw_i$ be the concatenation of $\e_i$ and $\Pathw$. For $i\in\brm$, denote $\Cut_i:=\Cut_{\Pathw_i}$, $\APMSop_{\b,i}:= \APMS(\Gfirem{\w'_i}{\w})$, and set
$\APMSop_\b:=\bigsqcup_{i=1}^m \APMSop_{\b,i}$. For $\apmwf\in\APMSop_{\b,i}$, let $j\in\brm\setminus\{i\}$ be such that $\apmwf$ contains the edge $\e_j$. Let $\phi(\apmwf)\in\APMSop_{\b,j}$ be obtained by replacing $\e_j$ with $\e_i$. We see that $\phi:\APMSop_\b\to\APMSop_\b$ is an involution. We claim that it is sign-reversing, i.e., if $\apmwf\in\APMSop_{\b,i}$ and $\phi(\apmwf)\in\APMSop_{\b,j}$ then 
$(-1)^{|\apmwf\cap \Cut_i|} = -(-1)^{|\phi(\apmwf)\cap \Cut_j|}$. 
Indeed, $\apmwf$ and $\phi(\apmwf)$ differ by one edge ($\e_j\in\apmwf$ and $\e_i\in\phi(\apmwf)$), and assuming $i<j$, the intersections $\apmwf\cap \Cut_i$ and $\phi(\apmwf)\cap \Cut_j$ differ only in the edge $\e_i\in\phi(\apmwf)\cap \Cut_j$; see \figref{fig:zig-cut}(middle). Thus, $\phi$ is sign-reversing. We need to check that the function $\brguess<\w,\cdot>$ satisfies~\eqref{eq:intro:white_holom} at $\b$, i.e., $\sum_{\s=1}^m \wtKf(\e_\s) \brguess<\w,\w'_\s> = 0$. 
The respective contributions of $\apmwf$ and $\phi(\apmwf)$ to the terms $\wtKf(\e_\s) \brguess<\w,\w'_\s>$ for $\s=i$ and $\s=j$ are 
$\wtKf(\e_i)\cdot \epsKf(\e_i)\cdot \epsKf(\Pathw) \cdot (-1)^{|\apmwf\cap \Cut_i|} \wtf(\apmwf)$ and 
$\wtKf(\e_j)\cdot \epsKf(\e_j)\cdot \epsKf(\Pathw) \cdot (-1)^{|\phi(\apmwf)\cap \Cut_j|} \wtf(\phi(\apmwf))$. Noting that $\wtKf(\e_i) \epsKf(\e_i)=\wtf(\e_i)$ and $\wtKf(\e_j) \epsKf(\e_j)=\wtf(\e_j)$, we see that these contributions cancel out. Thus, the function $\brguess<\w,\cdot>$ is \wdash holomorphic. 

Recall from \cref{lemma:MCE:apm_vs_Kast} that $\Hwspace_{\R}\HtripKf$ is a two-dimensional vector space. Since $\brguess<\w_1,\w_2>$ is \wdash holomorphic in each argument, it belongs to $\Hwspace_{\R}\HtripKf^{\otimes2}$. Since it is antisymmetric, it belongs to $\bigwedge^2\Hwspace_{\R}\HtripKf$. Clearly, the same is true for $\brlawf<\w_1,\w_2>$. 
It remains to show that both $\brguess<\w_1,\w_2>$ and $\brlawf<\w_1,\w_2>$ are nonzero elements of the $1$-dimensional space $\bigwedge^2\Hwspace_{\R}\HtripKf$, since in this case, they must differ by multiplication by a nonzero constant~$\constw$.

For $\Cf:=\Meas(\Gf,\wtf)$, we have $\Delta_{i,j}(\Cf)\neq0$ for some $1\leq i<j\leq 4$.
Let $\bdwx_i,\bdwx_j$ be the corresponding next-to-boundary vertices of $\Gf$. By~\eqref{eq:MCE:alt(C)_vs_pFw}, $\brlawf<\bdwx_i,\bdwx_j>\neq0$. 
 Let
 $\Pathw$ be any path from $\bdwx_i$ to $\bdwx_j$ and let
 $\Cut$ be the cut in $\Gfrem{\bdwx_i}{\bdwx_{j}}$ that only crosses boundary edges of $\Gfrem{\bdwx_i}{\bdwx_{j}}$. We showed above that $\brguess<\bdwx_i,\bdwx_j>=\pm\brguess<\bdwx_i,\bdwx_j>_{\Pathw,\Cut}$. The graph $\Gfrem{\bdwx_i}{\bdwx_{j}}$ is of type $(0,4)$, so none of the edges crossed by $\Cut$ appear in any $\apmwf\in\APMS(\Gfrem{\bdwx_i}{\bdwx_j})$. Thus, the signs $(-1)^{|\apmwf\cap \Cut|}$ of all terms contributing to $\brguess<\bdwx_i,\bdwx_j>_{\Pathw,\Cut}$ are the same. Since $\Delta_{i,j}(\Cf)>0$, we have $\APMS(\Gfrem{\bdwx_i}{\bdwx_j})\neq\emptyset$, so $\brguess<\bdwx_i,\bdwx_j>=\pm\brguess<\bdwx_i,\bdwx_j>_{\Pathw,\Cut}\neq0$.
\end{proof}

After adjusting $\laextf$ and $\latextf$ by an element of $\GL_2(\R)$, from now on we assume that 
\begin{equation}\label{eq:TE:constw=constb=1}
 \constw=\constb=1.
\end{equation}

The next result specializes \cref{lemma:TE:single_dimer_formula} to the case where the vertices $\w_1,\w_2$ or $\b_1,\b_2$ share a face of $\Gf$. We denote 
$\DeltaGfiww:=\sum_{\apmwf\in\APMS(\Gfiww)} \wtf(\apmwf)$ and 
$\DeltaGfbb:=\sum_{\apmbf\in\APMS(\Gfbb)} \wtf(\apmbf)$. 

\begin{corollary}\label{lemma:TE:single_dimer_same_face}
If distinct vertices $\w_1,\w_2\in\WVf$ (resp., $\b_1,\b_2\in\BVf$) share a face of $\Gf$ then 
\begin{equation}\label{eq:TE:same_face}
\brlawf<\w_1,\w_2> = -\epswwf_{\w_1,\w_2} \DeltaGfiww,
\quad\text{resp.,}\quad
 \brlatbf[\b_1,\b_2] = \epsbbf_{\b_1,\b_2} \DeltaGfbb.
\end{equation}
\end{corollary}

\begin{proof}%
We prove the first identity in~\eqref{eq:TE:same_face}. Let $(\Gf',\epsKfp)$ be obtained from $(\Gf,\epsKf)$ by inserting a black tripod $\b$ adjacent to $\w_1,\w_2,\w_3$ as in \cref{dfn:eksbb}. 
By \cref{rmk:tripods_commute}, 
 $\laextf$ and $\latextf$ extend to discrete holomorphic functions on $(\Gf',\wtfp,\epsKfp)$. The result follows from~\eqref{eq:IP:brlaw} by choosing the path $\Pathw=(\w_1,\b,\w_2)$ since $|\apmwf\cap \CutPw|=1$ for all $\apmwf\in\APMS(\Gfiww)$. The proof of the second identity is similar, except here we have $|\apmbf\cap \CutPb|=0$ for all $\apmbf\in\APMS(\Gfbb)$.
\end{proof}

We will be particularly interested in the case where $\w_1,\w_2$ (resp., $\b_1,\b_2$) share not only a common face but also a common neighbor in $\Gf$. 
 Let $\b\in\BVfint$ (resp., $\w\in\WVfint$) be connected by edges $\e_1,\e_2,\dots,\e_d$ to $\w_1,\w_2,\dots,\w_d$ (resp., $\b_1,\b_2,\dots,\b_d$) in clockwise order.
Then $\epswwf_{\w_\s,\w_{\s+1}} = \epsKf(\e_\s)\epsKf(\e_{\s+1})$ 
(resp., $\epsbbf_{\b_\s,\b_{\s+1}} = \epsKf(\e_\s)\epsKf(\e_{\s+1})$). 
 Applying~\eqref{eq:TE:same_face} and taking the extra minus sign in~\eqref{eq:TE:y_to_lalat} into account, we obtain
\begin{equation}\label{eq:*}
 \epsKf(\e_\s)\epsKf(\e_{\s+1})\det\mat[\Fwf(\w_{\s})|\Fwf(\w_{\s+1})]\leq0,
 \quad\text{resp.,}\quad
 \epsKf(\e_\s)\epsKf(\e_{\s+1})\det\mat[\Fbf(\b_{\s})|\Fbf(\b_{\s+1})]\leq0
\end{equation}
for all $\s\in\brd$, 
where $\det\mat[\y|\y']:=\Re(\y)\Im(\y') - \Im(\y)\Re(\y')$ for $\y,\y'\in\C$.

\subsection{Type $(2,4)$: double-dimer configurations}\label{ssec:TE:Kast2}
We continue to assume that $\Gf$ is a planar bipartite graph of type $(2,4)$ 
that admits an \APM and has black boundary, 
 and continue to study the \quintuple $\datrQLf=\datrQf$.

\begin{definition}\label{dfn:double_dimer_Omf}
A \emph{double-dimer configuration} on $\Gf$ is a multiset $\Omf$ of edges of $\Gf$ that covers every interior (resp., boundary) vertex of $\Gf$ exactly twice (resp., at most twice). The set of double-dimer configurations on $\Gf$ using every boundary vertex $\bdvf_1,\bdvf_2,\bdvf_3,\bdvf_4$ exactly once is denoted $\OmsGf$.
\end{definition}
\noindent Thus, $\Omf\in\OmsGf$ consists of (i) doubled edges, (ii) cycles of even length, and (iii) two \emph{\bdbd paths} $\Pathbd_1(\Omf),\Pathbd_2(\Omf)$ starting and ending at the boundary of $\Gf$. We denote $\wtf(\Omf):=2^{\ncycOmf}\prod_{\e\in\Omf} \wtf(\e)$, where the product is taken with multiplicity and $\ncycOmf$ is the number of cycles in $\Omf$ (excluding doubled edges). 

An \emph{orientation} $\Omfvec$ of $\Omf\in\OmsGf$ is a choice of an orientation for each of the two \bdbd paths $\Pathbd_1(\Omf),\Pathbd_2(\Omf)$ of $\Omf$. We say that $\Omfvec$ is \emph{alternating} if the induced orientations on the four boundary edges $\bdef_1,\bdef_2,\bdef_3,\bdef_4$ of $\Gf$ alternate between ``in'' and ``out''; otherwise, $\Omfvec$ is called \emph{non-alternating}.

Given two edges $\e_1,\e_2$ of $\Gf$, let 
\begin{equation*}%
 \OmsGf(\e_1,\e_2):=\{\Omf\in\OmsGf\mid \e_1\in\Pathbd_1(\Omf)\text{ and }\e_2\in\Pathbd_2(\Omf)
\text{ or }
\e_1\in\Pathbd_2(\Omf)\text{ and }\e_2\in\Pathbd_1(\Omf)\}.
\end{equation*}
For each $\Omf\in\OmsGf(\e_1,\e_2)$, there is a unique orientation $\Omfvecbw_{\e_1,\e_2}$ of $\Omf$ such that the edges $\e_1,\e_2$ are both oriented from black to white. We set $\epsOmf(\Omfvecbw_{\e_1,\e_2}):=+1$ if $\Omfvecbw_{\e_1,\e_2}$ is non-alternating and $\epsOmf(\Omfvecbw_{\e_1,\e_2}):=-1$ if $\Omfvecbw_{\e_1,\e_2}$ is alternating; see \cref{fig:eps-alt}. 

\begin{figure}
\def\inputscl{1.3}
\setlength{\tabcolsep}{2pt}
\begin{tabular}{cccc}
 \includegraphics[scale=\inputscl]{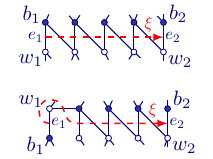}
& 
\includegraphics[scale=\inputscl]{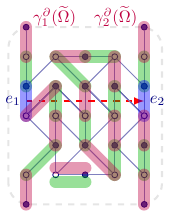}
&
 \includegraphics[scale=\inputscl]{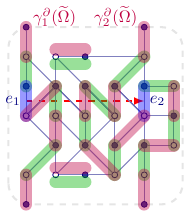}
&
 \includegraphics[scale=\inputscl]{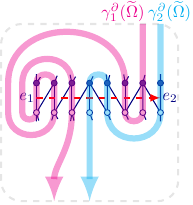}
\\
(a) \textcolor{red}{$\CutPw=\CutPb$} 
& 
(b) $\epsOmf(\Omfvecbw_{\e_1,\e_2})=+1$
& 
(c) $\epsOmf(\Omfvecbw_{\e_1,\e_2})=-1$
& 
(d) $\epsOmf(\Omfvecbw_{\e_1,\e_2})=+1$
\end{tabular}
 \caption{\label{fig:eps-alt} Computing $\epsOmf(\Omfvecbw_{\e_1,\e_2})$; see \cref{lemma:TE:Omvec,lemma:TE:OmsGf_pos}.}
\end{figure}

\begin{lemma}\label{lemma:TE:Omvec}
Let $\e_1,\e_2$ be two edges of $\Gf$ with $\ebar_1=\{\b_1,\w_1\}$ and $\ebar_2=\{\b_2,\w_2\}$. Then 
\begin{equation}\label{eq:TE:Omvec}
 \wtKf(\e_1)\wtKf(\e_2)\cdot \brlawf<\w_1,\w_2> \cdot \brlatbf[\b_1,\b_2] = 
\sum_{\Omf\in\OmsGf(\e_1,\e_2)}\epsOmf(\Omfvecbw_{\e_1,\e_2})\wtf(\Omf). 
\end{equation}
\end{lemma}
\begin{proof}
By \cref{rmk:tripods_commute}, we may assume that $\Gf$ is connected. 
Pick a path $\Pathw$ from $\w_1$ to $\w_2$ and let $\Pathb$ be the path from $\b_1$ to $\b_2$ that differs from $\Pathw$ by either adding or removing each of $\e_1$ and $\e_2$; see \figref{fig:eps-alt}(a). 
Applying~\eqref{eq:IP:brlaw}--\eqref{eq:TE:constw=constb=1} 
 and taking into account that $\wtKf(\e_\s)=\epsKf(\e_\s)\wtf(\e_\s)$ for $\s=1,2$, we find 
\begin{equation*}%
 \wtKf(\e_1)\wtKf(\e_2)\cdot \brlawf<\w_1,\w_2> \cdot \brlatbf[\b_1,\b_2] = 
\wtf(\e_1)\wtf(\e_2) 
\sum_{\substack{\apmwf\in\APMS(\Gfiww),\\\apmbf\in\APMS(\Gfbb)}}
(-1)^{|\apmwf\cap \CutPw| +|\apmbf\cap\CutPb|} \wtf(\apmwf)\wtf(\apmbf).
\end{equation*}
Let $\Omf = \apmwf\cup\apmbf\cup\{\e_1,\e_2\}$ (multiset union). We see that $\Omf\in\OmsGf$. Since $\Gf$ has black boundary, each of $\Pathbd_1(\Omf),\Pathbd_2(\Omf)$ contains an even number of edges. However, all boundary edges belong to $\apmbf$ and none belong to $\apmwf$. Thus, for parity reasons, each \bdbd path $\Pathbd_1(\Omf),\Pathbd_2(\Omf)$ must pass through either $\e_1$ or $\e_2$; see \figref{fig:eps-alt}(b,c) for examples. Thus, $\Omf\in\OmsGf(\e_1,\e_2)$. 
 Conversely, any $\Omf\in\OmsGf(\e_1,\e_2)$ can be split into a union $\apmwf\cup\apmbf\cup\{\e_1,\e_2\}$ for $\apmwf\in\APMS(\Gfiww)$ and $\apmbf\in\APMS(\Gfbb)$ in $2^{\ncycOmf}$-many ways. It remains to show that $(-1)^{|\apmwf\cap \CutPw| +|\apmbf\cap\CutPb|} = \epsOmf(\Omfvecbw_{\e_1,\e_2})$. 

Since $\CutPw$ is a zig-zag cut in $\Gfiww$ and $\CutPb$ is a zig-zag cut in $\Gfbb$, we can deform them so that they coincide, both starting in the middle of $\e_1$ and ending in the middle of $\e_2$, and so that both $\w_1$ and $\w_2$ are located to the right of $\CutPw=\CutPb$; see \figref{fig:eps-alt}(a). Since the cuts coincide, every edge in $\Omf\setminus\{\e_1,\e_2\}$ that intersects one of them contributes a minus sign to either $(-1)^{|\apmwf\cap \CutPw|}$ or $(-1)^{|\apmbf\cap\CutPb|}$. Thus, we have a path $\Pathbd_\s(\Omf)$ passing through the edge $\e_\s$ for $\s=1,2$ and a directed cut $\CutPw=\CutPb$ connecting the midpoints of $\e_1$ and $\e_2$ and keeping both $\w_1$ and $\w_2$ to the right, and $(-1)^{|\apmwf\cap \CutPw| +|\apmbf\cap\CutPb|}$ counts the parity of the number of intersections between $\CutPw=\CutPb$ and $\Pathbd_1(\Omf)\cup\Pathbd_2(\Omf)$. By a topological argument (see \cref{fig:eps-alt}), it follows that the sign $(-1)^{|\apmwf\cap \CutPw| +|\apmbf\cap\CutPb|}$ is $+1$ precisely when $\Omfvecbw_{\e_1,\e_2}$ is non-alternating.
\end{proof}

Given two faces $\ff,\f$ of $\Gf$, let $\OmsGf(\ff\ssep\f)$ be the set of $\Omf\in\OmsGf$ such that both \bdbd paths $\Pathbd_1(\Omf),\Pathbd_2(\Omf)$ separate $\ff$ from $\f$. 

\begin{proposition}\label{lemma:TE:OmsGf_pos}
For any two faces $\ff,\f$ of $\Gf$, 
\begin{equation}\label{eq:TE:OmsGf_pos}
 \frac14(\xdGf(\ff) - \xdGf(\f))^2 = \sum_{\Omf\in\OmsGf(\ff\ssep\f)} \wtf(\Omf).
\end{equation}
\end{proposition}
\begin{proof}
 Consider a simple path $\Cut$ in $\GDf$ from $\ff$ to $\f$. After possibly adjusting the topology of the embedding of $\Gf$ in the disk, we may assume that $\Cut$ is horizontal and oriented from left to right as in \cref{fig:eps-alt}. Let
$E(\Cut):=\{\e_1,\e_2,\dots,\e_m\}$
 be the edges crossed by $\Cut$, listed from left to right, and let $\ebar_i=\{\b_i,\w_i\}$ for $i\in\brm$. For each $i$, we set $\bweps_i:=+1$ (resp., $\bweps_i:=-1$) if $\w_i$ is located above (resp., below) the cut $\Cut$. Let $\ff=\f_0,\f_1,\dots,\f_m=\f$ be the faces appearing along the cut, so that $\e_i$ separates $\f_{i-1}$ from $\f_i$. 
 By~\eqref{eq:intro:primitive_TO} and~\eqref{eq:bispinor_cdot} below,
\begin{equation*}%
 \frac14(\xdGf(\ff) - \xdGf(\f))^2 = \frac14\biggl(\sum_{i=1}^m (\xdGf(\f_i)-\xdGf(\f_{i-1}))\biggr)^2 
= \sum_{1\leq i<j\leq m} \bweps_i\bweps_j\wtKf(\e_i)\wtKf(\e_j)\brlawf<\w_i,\w_j>\brlatbf[\b_i,\b_j].
\end{equation*}
By \cref{lemma:TE:Omvec}, 
\begin{equation}\label{eq:TE:Mand_RHS_proof}
 \frac14(\xdGf(\ff) - \xdGf(\f))^2 
= \sum_{1\leq i<j\leq m} \sum_{\Omf\in\OmsGf(\e_i,\e_j)}\bweps_i\bweps_j\epsOmf(\Omfvecbw_{\e_i,\e_j})\wtf(\Omf).
\end{equation}

Let $\Omf\in\OmsGf$. Our goal is to show that the contribution of $\Omf$ to the right-hand side of~\eqref{eq:TE:Mand_RHS_proof} equals $\wtf(\Omf)$ if $\Omf\in\OmsGf(\ff\ssep\f)$ and $0$ otherwise. Pick a non-alternating orientation $\Omfvec$ of $\Omf$ and let
 $E^{(\s)}:=\{\e_i\in E(\Cut)\mid\e_i\in\Pathbd_\s(\Omf)\}$ for $\s=1,2$. Thus, the orientation $\Omfvec$ induces an orientation on $\e_i$ for each $\e_i\in E^{(\s)}$. For $\s=1,2$, the orientations of edges in $E^{(\s)}$ alternate between up and down along the cut $\Cut$ since $\Pathbd_\s(\Omf)$ separates the disk into two connected components. 
We write $\orUD_i=+1$ (resp., $\orUD_i=-1$) if $\e_i\in E^{(1)}\sqcup E^{(2)}$ is oriented up (resp., down) in $\Omfvec$.

A given double-dimer configuration $\Omf\in\OmsGf$ contributes to the $(i,j)$-th term on the right-hand side of~\eqref{eq:TE:Mand_RHS_proof} whenever $\e_i\in E^{(1)}$ and $\e_j\in E^{(2)}$ or vice versa. Assume without loss of generality that $\e_i\in E^{(1)}$ and $\e_j\in E^{(2)}$. A simple check shows that $\bweps_i\bweps_j\epsOmf(\Omfvecbw_{\e_i,\e_j})=\orUD_i\orUD_j$. 
 Thus, the total contribution of $\Omf$ to the right-hand side of~\eqref{eq:TE:Mand_RHS_proof} is 
$\left(\sum_{\e_i\in E^{(1)}} \orUD_i\right)\left(\sum_{\e_j\in E^{(2)}} \orUD_j\right)\wtf(\Omf)$. This contribution is 
zero when at least one path $\Pathbd_\s(\Omf)$ does not separate $\ff$ from $\f$, since in this case we have $\sum_{\e_i\in E^{(\s)}} \orUD_i = 0$. Otherwise, $\Omf\in\OmsGf(\ff\ssep\f)$, and since we chose $\Omfvec$ to be a non-alternating orientation of $\Omf$, we have $\sum_{\e_i\in E^{(1)}} \orUD_i = \sum_{\e_j\in E^{(2)}} \orUD_j\in\{\pm1\}$, 
 so the total contribution is $\wtf(\Omf)$.
\end{proof}

\subsection{Annular graphs and momentum amplituhedron map}\label{ssec:annular}
We explain how to represent $\LaLat$ as a bipartite graph embedded in an annulus, and use it to reduce the case of a general graph $\G$ of type $(k,n)$ to the case of a graph $\Gf$ of type $(2,4)$.

\begin{definition}
An \emph{annular bipartite graph} of type $(\kbot,\nin,\nout)$ is a planar bipartite graph $\Gbot$ embedded in an annulus $\Ann$, with $\nin$ degree-$1$ boundary vertices $\bdvin_1,\bdvin_2,\dots,\bdvin_{\nin}$ located on the inner boundary circle $\Annin$ of $\Ann$ and $\nout$ degree-$1$ boundary vertices $\bdvout_1,\bdvout_2,\dots,\bdvout_{\nout}$ located on the outer boundary circle $\Annout$. Similarly to~\eqref{eq:DIM:k_dfn}, the parameters $\kbot$ and $\nbot:=\nin+\nout$ satisfy 
\begin{equation*}%
 \kbot:=|\WVbot| - |\BVintbot| \quad\text{and}\quad \nbot-\kbot = |\BVbot| - |\WVintbot|.
\end{equation*}
\end{definition}
As before, an \emph{\APM} $\apmbot$ of $\Gbot$ is required to cover each interior vertex of $\Gbot$ exactly once. Given a positive edge weight function $\wtbot:\Ebot\to\Rtp$ and subsets $\Iin\subset\brx{\nin}$ and $\Iout\subset\brx{\nout}$ with $|\Iin| + |\Iout|=\kbot$, we consider partition functions $\Delta_{\Iin,\Iout}(\Gbot,\wtbot)$ obtained as sums of weights of \APMs $\apmbot$ of $\Gbot$ with $\partin\apmbot=\Iin$ and $\partout\apmbot=\Iout$, where $(\partin\apmbot,\partout\apmbot)$ are defined analogously to the boundary $\partial\apm$ of an ordinary \APM using the inner and outer boundary vertices, respectively.

Let $(\Gbot,\wtbot)$ be a weighted connected annular bipartite graph of type 
$(\kbot,\nin,\nout) = (n-k+2,n,4)$
 with black boundary. Assume that $\Gbot$ admits \APMs $\apmbotp,\apmbotm$ such that $\partout\apmbotp=\Ioutp:=\emptyset$ and $\partout\apmbotm=\Ioutm:=\brx4$. Define $\Cbotpm\in\Grtnn(n-k\pm2,n)$ so that %
\begin{equation}\label{eq:OAC:Cbotpm_meas}
 \Delta_I(\Cbotpm) = \Delta_{I,\Ioutpm}(\Gbot,\wtbot) \quad\text{for all $I\in{\brn\choose n-k\pm2}$}.
\end{equation}
\begin{remark}\label{rmk:OAC:turn_Gbot_into_sphere}
If we embed $\Gbot$ into a sphere, declare the outer boundary vertices $\bdvout_1,\bdvout_2,\bdvout_3,\bdvout_4$ interior, and treat the inner circle $\Annin$ as the boundary of a disk, we obtain a graph whose boundary measurements are given by $\Cbotm$. Since $\Cbotp$ is obtained via the same process after deleting the four (now interior) black vertices $\bdvout_1,\bdvout_2,\bdvout_3,\bdvout_4$ that share a face, \cref{lemma:DIM:remove_bivertex} implies that
\begin{equation}\label{eq:OCA:Cbotm_subset_Cbotp}
 \Cbotm\subset\Cbotp.
\end{equation}
\end{remark}
\begin{definition}\label{dfn:OAC:Arep}
We say that $\LaLat\in\LaLak$ is \emph{\Arep} if there exists a weighted connected annular graph $(\Gbot,\wtbot)$ as above with boundary measurements $\Cbotpm$ given by~\eqref{eq:OAC:Cbotpm_meas} such that 
\begin{equation}\label{eq:OAC:Arep}
\Lap=\altp(\Cbotp)\in\Grtp(k-2,n) \quad\text{and}\quad \Lat=\altp(\Cbotm)\in\Grtp(k+2,n).
\end{equation}
In this case, we say that $\LaLat$ is \emph{\Areped} by $(\Gbot,\wtbot)$. We set 
\begin{equation*}%
 \LaLaArep:=\{\LaLat\in\LaLak\mid \LaLat\text{ is \Arep}\}.
\end{equation*}
\end{definition}
\noindent In particular, by~\eqref{eq:OCA:Cbotm_subset_Cbotp}, we have $\Lap\subset\Lat$ when $\LaLat$ is \Arep (cf. \cref{que:immp_vs_flag_vs_Arep} below).

\begin{definition}[Gluing planar bipartite graphs]\label{def:OAC:stacking_graphs}
Assume that $\G$ admits an \APM, is of type $(k,n)$ with $2\leq k\leq n-2$, and has black boundary. 
Let $\LaLat$ be \Areped by $(\Gbot,\wtbot)$. %
Let $\Gf = \Stack(\Gbot,\G)$ be obtained by identifying the boundary vertex $\bdvtop_i$ of $\G$ with the inner boundary vertex $\bdvin_i$ of $\Gbot$ for each $i\in\brn$; see \cref{fig:intro-annular}. We have $\Ef=\Etop\sqcup\Ebot$ and we let $\wtf:\Ef\to\Rtp$ be equal to $\wttop$ (resp., $\wtbot$) on $\Etop$ (resp., $\Ebot$). 
\end{definition}
\noindent It follows that $(\Gf,\wtf)$ is a weighted planar bipartite graph of type $(2,4)$. 

\begin{definition}
We say that $\epsKbot$ is a choice of \emph{Kasteleyn signs} for $\Gbot$ if it satisfies~\eqref{eq:Kast_sign} with $\epstra(\ffbot)=1$ for all faces $\ffbot$ of $\Gbot$ unless $\ffbot=\bdfbotin_n$ is adjacent to the arc of $\Annin$ between $\bdvin_n$ and $\bdvin_1$ or $\ffbot=\bdfbotout_4$ is adjacent to the arc of $\Annout$ between $\bdvout_4$ and $\bdvout_1$. If $\bdfbotin_n\neq\bdfbotout_4$ then we set $\epstra(\bdfbotin_n)=k$ and $\epstra(\bdfbotout_4)=0$, and if $\bdfbotin_n=\bdfbotout_4$ then we set $\epstra(\bdfbotin_n)=\epstra(\bdfbotout_4)=k+1$.\footnote{This choice ensures that when $\Annout$ turns into an interior face $\ffbotout$ as in \cref{rmk:OAC:turn_Gbot_into_sphere}, $\epsKbot(\bdrypathbot\ffbotout)$ satisfies~\eqref{eq:Kast_sign} with $\epstra(\ffbotout)=1$. It is also consistent with $\prod_{\ffbot\in\Fbot} \epsKbot(\bdrypathbot\ffbot)=1$ since $\kbot+\nbot\equiv k\pmod 2$.}
\end{definition}
Fix Kasteleyn signs $\epsKtop$ on $\Gtop$ and $\epsKbot$ on $\Gbot$. 
Let $\epsKf:\Ef\to\{\pm1\}$ be given by
\begin{equation}\label{eq:OAC:epsKf_stacked}
 \epsKf(\e):=
 \begin{cases}
 \epsKtop(\e), &\text{if $\e\in\Etop$;}\\
 (-1)^i\epsKbot(\e), &\text{if $\e=\bdebotin_i\in\Ebot$ for some $i\in\brn$;}\\
 \epsKbot(\e), &\text{otherwise,}
 \end{cases}
\end{equation}
where for $i\in\brn$, $\bdebotin_i$ is the sole edge of $\Gbot$ incident to $\bdvin_i$.

\begin{lemma}
$\epsKf$ is a valid choice of Kasteleyn signs for $\Gf$.
\end{lemma}
\begin{proof}
For simplicity, we assume that $\Gtop$ is \bdconn; the argument below extends to arbitrary $\Gtop$ using the Kasteleyn sign condition in \cref{rmk:Kast_sign_float}. 

Let $\ff$ be a face of $\Gf$. If $\ff$ is an interior face of $\Gtop$ or a face of $\Gbot$ not incident to $\Annin$ 
then~\eqref{eq:Kast_sign} is clearly satisfied for it. 
Otherwise, 
since $\Gbot$ is connected, $\ff$ is a union of a single (necessarily boundary) face $\f$ 
 of $\Gtop$ together with a disjoint union of $\nbdryarcs\f$-many boundary faces of $\Gbot$. 
Since $\G$ and $\Gbot$ have black boundary, $\dwcor(\ff) = \dwcor(\f) + \sum_{i\in\bdryarcs\f}\dwcor(\bdfbotin_i)$. 

For $i\in\brn$, let $\onebdfour_i$ be $1$ if $\bdfbotin_i = \bdfbotout_4$ and $0$ otherwise. Thus, $\epstra(\bdfbotin_i) \equiv 1+\onebdfour_i\pmod2$ for $i\in\brx{n-1}$. Let $\onebdfoursum:=\sum_{i\in\bdryarcs\f} \onebdfour_i$. Since $\Gbot$ is connected, $\onebdfoursum$ is either $1$ or $0$ depending on whether $\ff$ contains $\bdfbotout_4$ or not. 
 If $n\notin\bdryarcs\f$ then $(-1)^{\epstra(\f)+\sum_{i\in\bdryarcs\f}\epstra(\bdfbotin_i)} = (-1)^{\nbdryarcs\f+1+\onebdfoursum}$. 
If $n\in\bdryarcs\f$ then $\epstra(\f)=k+n$ and $\epstra(\bdfbotin_n)= k+\onebdfour_n$. In this case, $(-1)^{\epstra(\f)+\sum_{i\in\bdryarcs\f}\epstra(\bdfbotin_i)} = (-1)^{\nbdryarcs\f+n+1+\onebdfoursum}$. 

 On the other hand, for each $i\in\bdryarcs\f\setminus\{n\}$, the extra $(-1)^i$ sign for $\bdebotin_i$ in~\eqref{eq:OAC:epsKf_stacked} appears in $\epsKf(\bdrypathf\ff)$ together with the sign $(-1)^{i+1}$ for $\bdebotin_{i+1}$. If $n\notin\bdryarcs\f$ then the combined contribution of these signs to $\epsKf(\bdrypathf\ff)$ is $(-1)^{\nbdryarcs\f}$, and taking the $(-1)^{\nbdryarcs\f+1+\onebdfoursum}$ sign computed above into account, we find that indeed $\epsKf(\bdrypathf\ff)$ satisfies~\eqref{eq:Kast_sign} with $\epstra(\ff):=1+\onebdfoursum$. If $n\in\bdryarcs\f$ then the extra signs for $\bdebotin_1$ and $\bdebotin_n$ in~\eqref{eq:OAC:epsKf_stacked} contribute $(-1)^{n+1}$, so the total contribution of signs in~\eqref{eq:OAC:epsKf_stacked} is $(-1)^{|\bdryarcs\f|+n}$. Multiplying it by the sign $(-1)^{\nbdryarcs\f+n+1+\onebdfoursum}$ computed above, we see again that~\eqref{eq:Kast_sign} is satisfied with $\epstra(\ff):=1+\onebdfoursum$. It remains to note that $1+\onebdfoursum\equiv1\pmod2$ when $\bdfbotout_4$ is not contained in $\ff$ and $1+\onebdfoursum\equiv0\pmod2$ otherwise. This agrees with~\eqref{eq:Kast_sign} since $\kf+\nf = 6$ is even. 
\end{proof}

\begin{lemma}\label{lemma:OAC:Gf_satisfies_ass}
Let $\LaLat\in\LaLaArep$ be \Areped by $(\Gbot,\wtbot)$.
Then $\Gf = \Stack(\Gbot,\G)$ 
 admits an \APM. 
\end{lemma}
\begin{proof}

 For any $I\in{\brn\choose k}$ such that $\Delta_I(\Gtop,\wttop)\neq0$, let $\Iin:=\brn\setminus I$. We claim that $\Gbot$ admits an \APM $\apmbot$ with $(\partin\apmbot,\partout\apmbot)=(\Iin,\Iout)$ for some $\Iout\in{\brx4\choose 2}$. Fix any distinct $a,b\in I$ and $c,d\in\brn\setminus I$; such indices exist since $2\leq k\leq n-2$. Let $\Iinp:=\Iin\sqcup\{a,b\}$ and $\Iinm:=\Iin\setminus\{c,d\}$, and recall that we set $\Ioutp:=\emptyset$ and $\Ioutm:=\brx4$. By~\eqref{eq:OAC:Arep}, $\Gbot$ admits \APMs $\apmbotp,\apmbotm$ such that 
$(\partin\apmbotpm,\partout\apmbotpm)=(\Iinpm,\Ioutpm)$. Since $\Gbot$ has black boundary, $\Ombot:=\apmbotm\cup\apmbotp$ contains four \bdbd paths connecting $\bdvout_1,\bdvout_2,\bdvout_3,\bdvout_4$ to $\bdvin_a,\bdvin_b,\bdvin_c,\bdvin_d$ in some order. Swapping the roles of $\apmbotp$ and $\apmbotm$ along the paths that involve $\bdvin_a$ and $\bdvin_b$, we obtain an \APM $\apmbot$ of $\Gbot$ with boundary $(\Iin,\Iout)$, as desired. Choosing any $\apm\in\APMSGbd(I)$, we find that $\apmbot\sqcup\apm\in\APMS(\Gf)$. 
\end{proof}

 The next result allows one to ``split'' boundary measurements of $\Gf$ into those of $\G$ and $\Gbot$. 
\begin{corollary}\label{cor:Ann_split_Delta}
Assume that $\LaLat$ is \Areped by $(\Gbot,\wtbot)$ and let $(\Gf,\wtf)$ be obtained by gluing $(\G,\wt)$ and $(\Gbot,\wtbot)$ as above. 
For all distinct $\w_1,\w_2\in\WVint$ (resp., $\b_1,\b_2\in\BVint$), 
\begin{align}
\label{eq:Ann_split_Deltaw}
 \DeltaGfiww &= \sum_{L\in{\brn\choose k-2}} \Delta_L(\Giww,\wt)\cdot \Delta_L(\Lap),\quad\text{resp.,}\\
\label{eq:Ann_split_Deltab}
 \DeltaGfbb &= \sum_{R\in{\brn\choose k+2}} \Delta_R(\Gbb,\wt)\cdot \Delta_R(\Lat).
\end{align}
In particular, $\DeltaGfiww>0$ (resp., $\DeltaGfbb>0$) if and only if $\Giww$ (resp., $\Gbb$) admits an \APM.
\end{corollary}
\begin{proof}
Every perfect matching $\apmwf$ of $\Gfiww$ is a union of an \APM $\apm$ of $\Giww$ and an \APM $\apmbot$ of $\Gbot$, with $\wtf(\apmwf) = \wt(\apm)\cdot \wtbot(\apmbot)$. Since $\Giww$ is of type $(k-2,n)$, we have $L:=\partial\apm \in {\brn\choose k-2}$ and $\comp{L} = \partin\apmbot$. By~\eqref{eq:OAC:Cbotpm_meas}, \eqref{eq:altp_Delta}, and~\eqref{eq:OAC:Arep}, $\Delta_{\comp{L},\Ioutp}(\Gbot,\wtbot) = \Delta_{\comp{L}}(\Cbotp) = \Delta_L(\Lap)$. This shows~\eqref{eq:Ann_split_Deltaw}. Since $\Lap\in\Grtp(k-2,n)$ by~\eqref{eq:OAC:Arep}, the right-hand side of~\eqref{eq:Ann_split_Deltaw} is nonzero if and only if $\Giww$ admits an \APM. The argument for~\eqref{eq:Ann_split_Deltab} is similar.
\end{proof}
\noindent See~\eqref{eq:OmsGf_vs_Ttaucoef} for an analogous ``splitting'' formula for double-dimer measurements introduced in~\eqref{eq:TE:OmsGf_pos}.

\subsection{Flag-positive pairs \texorpdfstring{$(\La,\Lat)$}{(𝚲,𝚲̃)}} \label{ssec:flag_positive}

Let $\Gtp\subset\GL_n(\R)$ be the subset consisting of totally positive $n\times n$ matrices (with positive minors of all sizes). For $M\in\GL_n(\R)$ and $0\leq d\leq n$, let $\pi_{d}:\GL_n(\R)\to\Gr(d,n)$ be the map sending $M$ to the span $M_{\brd}$ of the first $d$ rows of $M$. For $2\leq k\leq n-2$, let $\pifl(M):=(M_{\brx{k-2}},M_{\brx{k+2}})$. 
Let $\Fltpn(k-2,k+2)$ be the totally positive part of the $2$-step flag variety $\Fln(k-2,k+2)$ as defined in~\cite{Lus2}. Explicitly, for each $0\leq d\leq n$ and $2\leq k\leq n-2$,
\begin{equation}\label{eq:OAC:pifl}
 \Grtp(d,n)=\pi_d(\Gtp) \quad\text{and}\quad \Fltpn(k-2,k+2)=\pifl(\Gtp).
\end{equation}
Here, the first identity is due to K.~Rietsch; see~\cite[Remark~3.8]{LamCDM} for a proof and~\cite[Theorem~1.1(i) and Section~1.4]{BlKa} for further discussion.

\begin{definition}
We write $\LaLafl:=\{\LaLat\mid(\Lap,\Lat)\in\Fln(k-2,k+2)\}$ and 
 $\LaLaflp:=\{\LaLat\mid(\Lap,\Lat)\in\Fltpn(k-2,k+2)\}$. We say that $\LaLat$ is \emph{flag-positive} if $\LaLat\in\LaLaflp$.
\end{definition}

\begin{lemma}\label{lemma:OAC:flag_pos=>Arep}
If $\LaLat\in\LaLaflp$ is flag-positive then it is \Arep.
\end{lemma}
\begin{proof}
Let $M\in\Gtp$ be such that $\pifl(M)=(\Lap,\Lat)$, so that
 $\Lap = M_{\brx{k-2}}$ and $\Lat = M_{\brx{k+2}}=\begin{pmatrix}
\Lap \\ \Latrest
\end{pmatrix}$ as $(k\pm2)\times n$ matrices. 
Consider a $(k+2)\times(n+4)$ matrix $\Xmat:=\begin{pmatrix}
\Lap & \bzero_{(k-2)\times4}\\
\Latrest & I'_4\\
\end{pmatrix}$ with $I'_4 = \left(\begin{smallmatrix}
0 & 0 & 0 & -1\\
0 & 0 & 1 & 0\\
0 & -1 & 0 & 0\\
1 & 0 & 0 & 0
\end{smallmatrix}\right)$. As explained in~\cite[Lemma~3.9]{Pos}, $\Xmat\in\Grtp(k+2,n+4)$. Let $(\Gbot',\wtbot)$ be a (reduced) weighted planar bipartite graph in a disk such that $\Xmat = \Meas(\Gbot',\wtbot)$. Assume that $\Gbot'$ has white boundary. Denote the boundary vertices of $\Gbot'$ by $(\bdvin_1,\bdvin_2,\dots,\bdvin_n,\bdvout_4,\bdvout_3,\bdvout_2,\bdvout_1)$. 
Let $\Gbot$ be obtained from $\Gbot'$ by reflecting it and changing the colors of all vertices. Choose a cut $\Cut$ in $\Ann$ connecting $\Annin$ to $\Annout$. Cutting $\Ann$ along $\Cut$, we obtain a disk $\Disk$, and we embed $\Gbot$ in $\Disk$ so that $\bdvin_1,\bdvin_2,\dots,\bdvin_n$ appear clockwise around $\Annin$ and $\bdvout_1,\bdvout_2,\bdvout_3,\bdvout_4$ appear clockwise around $\Annout$. 
As explained in \cref{sec:backgr:cyc}, applying $\altp$ corresponds to swapping the colors of the vertices of $\Gbot$, so the elements $\Cbotpm\in\Grtnn(n-k\pm2,n)$ defined in~\eqref{eq:OAC:Cbotpm_meas} satisfy~\eqref{eq:OAC:Arep}. Thus, $\LaLat$ is \Arep.
\end{proof}

Recall from \cref{ssec:TE:Kast} that $\Gf$ admits unique (up to $\GL_2(\R)\times\GL_2(\R)$-action) discrete holomorphic functions $(\laextf,\latextf)\in\HHspaceKRdf$. Their restrictions $(\laextf|_{\WV},\latextf|_{\BV})$ to the vertices of $\G$ belong to $\HHspaceKRd$. Let 
$\lalat = (\alt\partial(\laextf|_{\WV}),\alt\partial(\latextf|_{\BV}))\in\Gr(2,n)\times\Gr(2,n)$ be their boundary restrictions.

\begin{lemma}\label{lemma:OAC:stacking_vs_PhiLL}
The pair $\lalat$ defined above satisfies $\lalat = \PhiLL(C)$, where $C:=\Meas(\Gtop,\wttop)$.
\end{lemma}
\begin{proof}
 By~\eqref{eq:MCE:alt(C)_vs_pFw}--\eqref{eq:MCE:alt(Cp)_vs_pFb}, $\la\subset C\subset\latp$. 
Let $\Gbotm$ (resp., $\Gbotp$) be obtained from $\Gbot$ by making the vertices $\bdvout_1,\bdvout_2,\bdvout_3,\bdvout_4$ interior (resp., deleting the vertices $\bdvout_1,\bdvout_2,\bdvout_3,\bdvout_4$); cf. \cref{rmk:OAC:turn_Gbot_into_sphere}. By~\eqref{eq:OAC:Cbotpm_meas}, 
 $\Cbotpm = \Meas(\Gbotpm,\wtbot)$. Let $\latextalt$ be obtained from $\latextf|_{\BVbot}$ by substituting $\latextf(\bdvin_i)\mapsto(-1)^i\latextf(\bdvin_i)$ for $i\in\brn$. This is consistent with the extra $(-1)^i$ sign in~\eqref{eq:OAC:epsKf_stacked}. We have $\laextf|_{\WVbot}\in\Hwspace_{\R^2}\HtripKbotp$ and $\latextalt\in\Hbspace_{\R^2}\HtripKbotm$. Furthermore, $\partial(\laextf|_{\WVbot}) = \la$ and $\partial(\latextalt) = \lat$. By~\eqref{eq:MCE:alt(C)_vs_pFw}--\eqref{eq:MCE:alt(Cp)_vs_pFb}, 
 we have $\la\subset\alt\Cbotp$ and $\lat\subset\altp\Cbotm$. By~\eqref{eq:OAC:Arep}, $\la\subset C\cap \La$ and $\lat\subset C^\perp\cap\Lat$. By \cref{prop:momLL_basic}, both intersections are $2$-dimensional, so $\lalat = \PhiLL(C)$. 
\end{proof}

We summarize our single-dimer model results.

\begin{corollary}\label{cor:lak_implies_black_imm}
Assume that $\G$ admits an \APM, and let $\wt\in\Rtpgauge$ and $C:=\Meas(\G,\wt)$. Consider $2$-planes $\lalat\in\lalats$ satisfying $\la\subset C\subset\latp$. Let $(\laext,\latext)\in\HHspaceKRd$ be the discrete holomorphic extensions of $\lalat$ to the vertices of $\G$.
\begin{enumerate}[label=(\arabic*)]
\item\label{lak_implies1} If $\la\in\lak$ then for all distinct $\w_1,\w_2\in\WV$ sharing a face of $\G$, we have $\epsww_{\w_1,\w_2}\brlaw<\w_1,\w_2> < 0$ if $\Giww$ admits an \APM and $\brlaw<\w_1,\w_2>=0$ otherwise.
\item\label{lak_implies2} If $\lat\in\latk$ then for all distinct $\b_1,\b_2\in\BV$ sharing a face of $\G$, we have $\epsbb_{\b_1,\b_2}\brlatb[\b_1,\b_2]>0$ if $\Gbb$ admits an \APM and $\brlatb[\b_1,\b_2]=0$ otherwise.
\item\label{lak_implies3} If $\la\in\lak$ (resp., $\lat\in\latk$) and $\b\in\BVint$ (resp., $\w\in\WVint$) is connected to $\w_1,\dots,\w_d$ (resp., $\b_1,\dots,\b_d$) 
by edges $\e_1,\dots,\e_d$ in clockwise order then %
\begin{equation}\label{eq:TE:same_face_and_same_vertex}
 \epsK(\e_{\s})\epsK(\e_{\s+1})\brlaw<\w_{\s},\w_{\s+1}>\leq0,\quad\text{resp.,}\quad
 \epsK(\e_{\s})\epsK(\e_{\s+1})\brlatb[\b_{\s},\b_{\s+1}]\geq0
 \quad\text{for all $\s\in\brd$}.
\end{equation} 
The inequality is strict if and only if the vertices $\w_\s,\w_{\s+1}$ (resp., $\b_\s,\b_{\s+1}$) are distinct and the graph $\G\rem\{\w_\s,\w_{\s+1}\}$ (resp., $\G\rem\{\b_\s,\b_{\s+1}\}$) admits an \APM. 
\item\label{lak_implies4} Suppose that $\helmin(\G)\geq1$ and the edge weights $\wt\in\Rtpgauge$ are generic. If $\la\in\lak$ (resp., $\lat\in\latk$) then $\laext(\w)\neq0$ for all $\w\in\WV$ (resp., $\latext(\b)\neq0$ for all $\b\in\BV$). 
\end{enumerate}
\end{corollary}
\begin{proof}
We start by showing~\itemref{lak_implies1}--\itemref{lak_implies3}.
Let $\la\in\lak$. By \cref{prop:from_la_to_La}, we may extend $\la$ to $\La\in\alt(\Grtp(n-k+2,n))$, so by~\eqref{eq:altp_Delta}, $\Lap\in\Grtp(k-2,n)$. By~\eqref{eq:OAC:pifl}, there exists $M\in\Gtp$ such that $\pi_{k-2}(M)=\Lap$. Letting $\Lat':=\pi_{k+2}(M)$, we get $(\La,\Lat')\in\LaLaflp$. By \cref{lemma:OAC:flag_pos=>Arep}, $(\La,\Lat')$ is \Areped 
 by some weighted annular graph $(\Gbot,\wtbot)$. Let $(\Gf,\wtf)$ be obtained by gluing $(\Gbot,\wtbot)$ and $(\G,\wt)$ via \cref{def:OAC:stacking_graphs}. 
By \cref{lemma:OAC:stacking_vs_PhiLL}, the functions
$(\laextf,\latextf)\in\HHspaceKRdf$ restrict to discrete
holomorphic extensions of $(\la,\lat')=\Phi_{\La,\Lat'}(C)$, where $\lat':=\Lat'\cap C^\perp$.
Part~\itemref{lak_implies1} and the first inequality in~\eqref{eq:TE:same_face_and_same_vertex} now follow from \cref{lemma:TE:single_dimer_same_face,cor:Ann_split_Delta}. 
The proof of~\itemref{lak_implies2}--\itemref{lak_implies3} in the case $\lat\in\latk$ is similar.

Next, we show~\itemref{lak_implies4}. 
After applying moves \MVbd if necessary (cf. \cref{lemma:DIM:moves_vs_helWmin}), we may assume that $\G$ has white boundary. 
By \cref{lemma:MCE:apm_vs_Kast}, we have $\dim\Hwspace_{\R}\HtripK = k$, so let $\Fw\in\Hwspace_{\R^k}\HtripK$ be a 
white-holomorphic function whose $k$ coordinates form a basis of $\Hwspace_{\R}\HtripK$.
Let $\la\in\lak$. Let $A\in\Mat_{2,k}(\R)$ be such that $\laext(\w)=A\Fw(\w)$ for all $\w\in\WV$. 
Assume for contradiction that $\laext(\w)=0$ for some $\w\in\WV$.
Consider a sequence $(\w=\w_0,\w_1,\dots,\w_m)$ of white vertices of $\G$ such that $\w_m=\bdv_s$ is a boundary vertex and for each $i\in\brm$, the vertices $\w_{i-1}$ and $\w_i$ share a common face.
Let $0\leq i\leq m$ be the largest index such that $\laext(\w_i)=0$.
Since $\la\in\lak$, we have $\brlaw<\bdv_s,\bdv_{s+1}>\neq0$, so $\laext(\w_m)=\laext(\bdv_s)\neq0$, and thus $i<m$. 
If $\G\rem\{\w_i,\w_{i+1}\}$ admits an \APM then by part~\itemref{lak_implies1}, $\brlaw<\w_i,\w_{i+1}> \neq0$, contradicting $\laext(\w_i)=0$. 
Thus, $\G\rem\{\w_i,\w_{i+1}\}$ does not admit an \APM. 
By \cref{lemma:DIM:deleting_hel_verts}, each of $\G\rem\{\w_i\}$ and $\G\rem\{\w_{i+1}\}$ admits an \APM. 
Thus, when $\wt\in\Rtpgauge$ is generic, $\w_i$ (resp., $\w_{i+1}$) is not a \emph{coloop} of the matroid of the $\WV\times\BVint$ Kasteleyn matrix $\WKmat$. Dually, $\w_i$ (resp., $\w_{i+1}$) is not a \emph{loop} of the matroid of the kernel $\Fw$ of $\WKmat$, i.e., $\Fw(\w_i)\neq0$ (resp., $\Fw(\w_{i+1})\neq0$).
By \crefi{lemma:MCE:apm_vs_Kast}{apm_vs_Kast3}, the boundary vectors $\Fw(\bdv_1),\dots,\Fw(\bdv_n)$ span $\R^k$. 
If the vectors $\Fw(\w_i),\Fw(\w_{i+1})$ are linearly independent, we can complete them to a basis of $\R^k$ using some $k-2$ boundary vectors $\Fw(\bdv_{j_1}),\dots,\Fw(\bdv_{j_{k-2}})$. Dually, the complementary $|\WV|-k$ rows of $\WKmat$ must form a basis of $\R^{\BVint}$, so $\G\rem\{\w_i,\w_{i+1}\}$ admits an \APM, a contradiction. 
Thus, the nonzero vectors $\Fw(\w_i),\Fw(\w_{i+1})$ are proportional. 
Since $\laext(\w_i)=A\Fw(\w_i)=0$, we must have $\laext(\w_{i+1})=A\Fw(\w_{i+1})=0$, contradicting the maximality of $i$.
 The proof of $\latext(\b)\neq0$ when $\lat\in\latk$ is similar.
\end{proof}

\begin{corollary}
Assume that $\G$ admits an \APM, and let $\wt\in\Rtpgauge$, $C:=\Meas(\G,\wt)$, $\LaLat\in\LaLak$, and $\lalat:=\PhiLL(C)$. 
If $\w_1,\w_2\in\WVint$ are connected by a path $\Pathw$ in $\G$ then 
\begin{equation}\label{eq:IP:brlaw_G}
 \brlaw<\w_1,\w_2> = \epsK(\Pathw)\cdot 
 \sum_{\apmw\in\APMS(\Giww)} (-1)^{|\apmw\cap \CutPw|} \wt(\apmw)\cdot \Delta_{\partial\apmw}(\Lap).
\end{equation}
Similarly, if $\b_1,\b_2\in\BVint$ are connected by a path $\Pathb$ in $\G$ then 
\begin{equation}\label{eq:IP:brlatb_G}
 \brlatb[\b_1,\b_2] = \epsK(\Pathb)\cdot 
 \sum_{\apmb\in\APMS(\Gbb)} (-1)^{|\apmb\cap \CutPb|} \wt(\apmb)\cdot \Delta_{\partial\apmb}(\Lat).
\end{equation}
\end{corollary}
\begin{proof}
Similarly to the proof of parts~\itemref{lak_implies1}--\itemref{lak_implies3} of \cref{cor:lak_implies_black_imm} above, we can extend $\La$ to a pair $(\La,\Lat')\in\LaLaflp$ that is \Areped by $(\Gbot,\wtbot)$. In this case,~\eqref{eq:IP:brlaw_G} follows from~\eqref{eq:IP:brlaw} and~\eqref{eq:Ann_split_Deltaw}. Similarly, extending $\Lat$ to a pair $(\La',\Lat)\in\LaLaflp$, we deduce~\eqref{eq:IP:brlatb_G} from~\eqref{eq:IP:brlatb} and~\eqref{eq:Ann_split_Deltab}. Here, we are continuing to assume that $\lalat$ has been adjusted by $\GL_2(\R)\times\GL_2(\R)$-action so that~\eqref{eq:TE:constw=constb=1} holds. 
\end{proof}

\begin{remark}\label{rmk:can_add_edges}
As before, when $\w_1,\w_2\in\WVint$ or $\b_1,\b_2\in\BVint$ are not connected by a path in $\G$, one can apply~\eqref{eq:IP:brlaw_G}--\eqref{eq:IP:brlatb_G} to the graph $\G'$ obtained from $\G$ by inserting tripods as in \cref{rmk:tripods_commute}. 
Furthermore, one can extend~\eqref{eq:IP:brlaw_G}--\eqref{eq:IP:brlatb_G} to boundary vertices of $\G$ by first making them interior using move~\MVbd.
\end{remark}

\section{Proof of \oac}\label{sec:proof_main_bij}

The goal of this section is to prove \cref{thm:intro:t_imm_vs_triples}. 

\subsection{From t-immersions to triples of subspaces}\label{sec:proof_main:from_t_imm_to_triples}
Assume that $\G$ is connected and admits an \APM. 
We start by introducing an algebraic version of the notion of a t-immersion. 

\begin{definition}\label{dfn:TE:datr}
An \emph{\datr} of $\G$ is a \quintuple $\datrQL:=\datrQ$, where $\wt\in\Rtpgauge$, $\epsK$ is a choice of Kasteleyn signs for $\G$, $(\Fw,\Fb)\in\HHspaceKC$ is a pair of $\C$-valued discrete holomorphic functions (where $\wtK(\e)=\epsK(\e)\wt(\e)$ as before), 
 and $\xd:\Faces\to\Rdd$ is the \KSprim of $(\Fw,\Fb)$. The set of \datrs of $\G$ is denoted $\Mdatr(\G)$. 
We view elements of $\Mdatr(\G)$ up to the action of gauge groups $\RtpVint$ and $\{\pm1\}^{\Vint}$ on $\datrQnox$; cf. \cref{rmk:DIM:Kast_gauge_eq}. 
The subset of \datrs with fixed $\wt$ (resp., $\wt$ and $\epsK$) is denoted by $\Mdatr(\G,\wt)$ (resp., $\Mdatr(\G,\wt,\epsK)$). 
\end{definition}

\begin{proposition}[{\cite[Section~3.2]{KLRR} and \cite[Sections~2 and~3]{CLR1}}]\label{prop:t_imm=>holom}
Any t-immersion $\Tcal$ of $(\G,\wt)$ extends to \adatr $\datrQL=\datrQ\in\Mdatr(\G,\wt)$. 
\end{proposition}
\begin{proof}%
Fix a choice $\epsK$ of Kasteleyn signs on $\G$ and let $\wtK$ be the corresponding Kasteleyn edge weights. 
Let $\bv_0\in\BV$ be a fixed vertex, and choose $\Fb(\bv_0)$ to be any nonzero value. 
 For each $\wv\in\WV$ connected by edges $\e_1,\e_2,\dots,\e_d$ to vertices $\bv_1,\bv_2,\dots,\bv_d$, 
 the tuple $(\Fb(\bv_1)\wtK(\e_1):\Fb(\bv_2)\wtK(\e_2):\cdots:\Fb(\bv_d)\wtK(\e_d))$ is determined by $\Tcal$ up to multiplication by an unknown constant (namely, $\Fw(\wv)$). Note that since $\xT$ is a t-immersion, it is injective on the edges of $\GD$, and thus we have $\Fw(\wv)\neq0$ and $\Fb(\bv_\s)\neq0$ for all $\s\in\brd$. 
Since $\G$ is connected, knowing the value $\Fb(\bv_0)$ allows us to recover the values of $\Fb$ at any other black vertex $\bv$ (by choosing a path from $\bv_0$ to $\bv$). The \Kawangle condition at each interior face together with the Kasteleyn sign condition~\eqref{eq:Kast_sign} guarantee that this procedure is consistent, i.e., does not depend on the choice of the path from $\bv_0$ to $\bv$. Taking $\wv_0$ to be any white vertex adjacent to $\bv_0$, we recover the value $\Fw(\wv_0)$ from~\eqref{eq:intro:primitive_TO}, and then proceed as above to recover the values $\Fw(\wv)$ at all other white vertices $\wv\in\WV$. Since the faces $\xT(\wv),\xT(\bv)$ of $\xT(\GD)$ are closed polygons for all $\wv\in\WVint$ and $\bv\in\BVint$, we get $(\Fw,\Fb)\in\HHspaceKC$. Thus, $\datrQL\in\Mdatr(\G,\wt)$. 
See \cite{KLRR,CLR1} for full details.
\end{proof}

\begin{remark}\label{rmk:Fw_Fb_determined_up_to_const}
The pair $(\Fw,\Fb)$ in \cref{prop:t_imm=>holom} is uniquely determined by $(\wt,\epsK,\Tcal)$ up to multiplication $(\Fw,\Fb)\mapsto (z\Fw,z^{-1}\Fb)$ by a global constant $z\in\Cast$. Furthermore, observe that the corresponding origami map $\Ocal$ given by~\eqref{eq:intro:primitive_TO} is invariant under gauge groups $\RtpVint$ and $\{\pm1\}^{\Vint}$. In particular, $\Ocal$ is determined by $\Tcal$ up to global shift and rotation.
\end{remark}

From now on, we fix a t-immersion $\xT$ of $(\G,\wt)$ and let $\datrQL=\datrQ\in\Mdatr(\G,\wt,\epsK)$ be the \datr provided by \cref{prop:t_imm=>holom}. 
 We define the tuples $(\y_i)_{i=1}^n := \alt(\pFw)$ and $(\yt_i)_{i=1}^n := \alt(\pFb)$ as in~\eqref{eq:intro:Fw_Fb_y}, and let $\lalat$ be the pair of $2\times n$ matrices given by~\eqref{eq:intro:y_to_lalat}. We let $C:=\Meas(\G,\wt)$. The goal of this subsection is to show the following result.

\begin{proposition}\label{prop:t_imm=>TRIPLES}
We have $\lalat\in\lalak$ and $(\la,\lat,C)\in\TRIPLES$.
\end{proposition}
We will obtain the proof of \cref{prop:t_imm=>TRIPLES} via a series of lemmas. First, observe that by~\eqref{eq:TE:t_imm_bdry_vs_pFw_pFb}, 
\begin{equation}\label{eq:yy=Tcal_Ocal}
 \y_i\yt_i = \bdxT_i - \bdxT_{i-1} \quad\text{and}\quad
 \ovl{\y_i}\yt_i = \bdxO_i - \bdxO_{i-1} \quad\text{for all $i\in\brn$.}
\end{equation}
Thus, $\sum_{i=1}^n\y_i\yt_i = \sum_{i=1}^n (\bdxT_i - \bdxT_{i-1}) = 0$ and $\sum_{i=1}^n\ovl{\y_i}\yt_i = \sum_{i=1}^n (\bdxO_i - \bdxO_{i-1}) = 0$,
so
 $\lalat\in\lalats$. %

Next, we analyze the brackets $\brla<i,i+1>$ and $\brlat[i,i+1]$.
 In what follows, the argument $\arg(z)$ of $z\in\C$ by definition belongs to $(-\pi,\pi]$. 
Recall the notation $\sumwT_i:=\sumwT(\bdf_i)$ and $\sumbT_i:=\sumbT(\bdf_i)$ from \cref{sec:intro:BCFW}.
\begin{lemma}\label{lemma:angle_sum_arg_rat}
For each $i\in\brn$,
\begin{equation}\label{eq:angle_sum_arg_rat}
 \sumbT_i = \arg(\y_{i+1} / \y_i)
 \in (0,\pi) \quad\text{and}\quad \sumwT_i -\pi = \arg(\yt_{i+1}/\yt_i) 
 \in (-\pi,0).
\end{equation}
\end{lemma}
\noindent In other words, by~\eqref{eq:intro:y_to_lalat}, $\Arg(\la_i,\la_{i+1})=\sumbT_i\in(0,\pi)$ and $\Arg(\lat_i,\lat_{i+1})=(\pi-\sumwT_i)\in(0,\pi)$.
\begin{proof}
Applying \MVbd at $\bdv_i$ and $\bdv_{i+1}$ if necessary, we may assume that both $\bdv_i$ and $\bdv_{i+1}$ are white. This results in potentially adding some boundary bigons to $\Tcal$, and has no effect on any of the terms appearing in~\eqref{eq:angle_sum_arg_rat}. 

Let $\fp_1,\fp_2,\dots,\fp_{2d}$ be the faces of $\G$ adjacent to $\bdf_i$ in counterclockwise order. Let $\e_1,\e_2,\dots,\e_{2d}$ be the edges separating them from $\bdf_i$, and let $K_j:=\wtK(\e_j)$ for $j\in\brx{2d}$ be their Kasteleyn weights. Let $\bdv_i=\v_0,\v_1,\dots,\v_{2d}=\bdv_{i+1}$ be the vertices incident to $\bdf_i$ in counterclockwise order. By~\eqref{eq:intro:Fw_Fb_y} and~\eqref{eq:partial_F_dfn}, setting $\epsvar := (-1)^{i}$, we have $\Fw(\v_0) = \epsvar\y_i$, $\Fb(\v_1) = \epsvar\yt_i/K_1$, $\Fb(\v_{2d-1}) = -\epsvar\yt_{i+1}/K_{2d}$, and $\Fw(\v_{2d}) = -\epsvar\y_{i+1}$.
 We have
\begin{equation*}%
 \sumbT_i = \arg\left(\prod_{j=1}^d\frac{\T(\fp_{2j}) - \bdxT_i}{\T(\fp_{2j-1}) - \bdxT_i}\right)
 \quad\text{and}\quad
 \sumwT_i = \arg\left(\prod_{j=1}^{d-1}\frac{\T(\fp_{2j+1}) - \bdxT_i}{\T(\fp_{2j}) - \bdxT_i}\right),
\end{equation*}
where both arguments belong to $(0,\pi)$ by~\itemref{intro:t_imm_angles_bdry} in \cref{dfn:intro:t_imm}. By~\eqref{eq:intro:primitive_TO},
\begin{equation*}%
\frac{\T(\fp_{2j}) - \bdxT_i}{\T(\fp_{2j-1}) - \bdxT_i} = -\frac{K_{2j}\Fw(\v_{2j})}{K_{2j-1}\Fw(\v_{2j-2})}
\quad\text{and}\quad
\frac{\T(\fp_{2j+1}) - \bdxT_i}{\T(\fp_{2j}) - \bdxT_i} = -\frac{K_{2j+1}\Fb(\v_{2j+1})}{K_{2j}\Fb(\v_{2j-1})}.
\end{equation*}
Comparing this with \cref{dfn:DIM:Kast}, we find 
\begin{equation}\label{eq:T_Kast_bdry}
 \prod_{j=1}^d\frac{\T(\fp_{2j}) - \bdxT_i}{\T(\fp_{2j-1}) - \bdxT_i} = \frac{\y_{i+1}}{\y_i} \FX_{\bdf_i}(\wt) ,
\quad
 \prod_{j=1}^{d-1}\frac{\T(\fp_{2j+1}) - \bdxT_i}{\T(\fp_{2j}) - \bdxT_i} = -\frac{\yt_{i+1}}{\yt_i} \FX_{\bdf_i}(\wt)^{-1}.
\end{equation}
We note that one does not need to treat the case $i=n$ differently since the Kasteleyn signs are consistent with the twisted cyclic symmetry~\eqref{eq:cshift_eps_dfn}.
\end{proof}
\begin{corollary}\label{cor:t_imm_=>_<ii+1>>0}
We have $\brla<i,i+1> >0$ and $\brlat[i,i+1]>0$ for all $i\in\brn$.
\end{corollary}

\begin{lemma}
We have 
\begin{equation}\label{eq:sumw_sumb_pi_k}
 \sum_{i=1}^n \sumwT_i = \pi(n-k-1) \quad\text{and}\quad
 \sum_{i=1}^n \sumbT_i = \pi(k-1).
\end{equation}
\end{lemma}
\begin{proof}

The sum of angles in each white (resp., black) face $\xT(\v)$ of $\xT(\GD)$ of degree $d=\degG(\v)$ is $\pi(d-2)$.
 On the other hand, the sum of white (resp., black) angles around each interior vertex $\xT(\ff)$ of $\xT(\GD)$ is $\pi$ by~\eqref{eq:intro:t_imm_angles_int}. 
Counting the sum of angles of all white (resp., black) faces of $\Tcal(\GD)$ in two different ways, we get
\begin{align}%
\label{eq:sum_up_angles1}
 \sum_{i=1}^n \sumwT_i + \pi |\Fint| &= \pi \sum_{\wv\in\WVint}(\deg(\wv)-2) = \pi(|\E| - |\WVbd| - 2|\WVint|);\\
\label{eq:sum_up_angles2}
 \sum_{i=1}^n \sumbT_i + \pi |\Fint| &= \pi \sum_{\bv\in\BVint}(\deg(\bv)-2) = \pi(|\E| - |\BVbd| - 2|\BVint|).
\end{align}
The result follows by applying~\eqref{eq:DIM:k_dfn}--\eqref{eq:Euler_no_float}. 
\end{proof}
\begin{proof}[Proof of \cref{prop:t_imm=>TRIPLES}]
We have already shown that $\lalat\in\lalats$. In view of \cref{cor:t_imm_=>_<ii+1>>0}, in order to prove that $\lalat\in\lalak$, it remains to check that 
$\wind(\la) = (k-1)\pi$
 and
$\wind(\lat) = (k+1)\pi$. 
 These statements 
readily follow from~\eqref{eq:angle_sum_arg_rat} and~\eqref{eq:sumw_sumb_pi_k}. Since the functions $\Fw$ and $\Fb$ are discrete holomorphic, by~\eqref{eq:MCE:alt(C)_vs_pFw}--\eqref{eq:MCE:alt(Cp)_vs_pFb}, $\la\subset C\subset \latp$, and thus $(\la,\lat,C)\in\TRIPLES$.
\end{proof}

\subsection{From triples of subspaces to t-immersions}\label{sec:finishing_main_bij}
\begin{notation}\label{notn:Tll_xll}
 Assume that $\G$ \hasnofloatAPM. For $C=\Meas(\G,\wt)$ and any $\lalat\in\lalats$ such that $\la\subset C\subset\latp$, let $\datrQLll:=\datrQ\in\Mdatr(\G)$ be the associated \datr of $(\G,\wt)$, where $\epsK$ is any choice of Kasteleyn signs for $\G$, $(\Fw,\Fb)=(\Fwl,\Fbl)\in\HHspaceKC$ are discrete holomorphic extensions of $\lalat$ (cf. \cref{dfn:OCP:restr_ext}), 
 and $\xd=\xll:\Faces\to\Rdd$ is the \KSprim of $(\Fw,\Fb)$.
\end{notation}
We first show that it suffices to only consider \twonondeg graphs $\G$ and matrices $C$. 
\begin{lemma}\label{lemma:OAC:nondeg}
 Assume that $\G$ admits an \APM. 
 If $\G$ is not \twonondeg then it admits no t-immersions. 
On the other hand, if $C\in\Grtnn(k,n)\setminus\Grnd(k,n)$ is not \twonondeg then the set of $\lalat\in\lalak$ satisfying $\la\subset C\subset\latp$ is empty.
\end{lemma}
\begin{proof}
Suppose that $\G$ is not \twonondeg. Let $\wt\in\Rtpgauge$ and $C:=\Meas(\G,\wt)$. If for some $i\in\brn$ we have, say, $\rank\mat[C_i|C_{i+1}] < 2$ then $\det\mat[\pFw_i|\pFw_{i+1}]=0$ for any $\Fw\in\Hwspace_{\C}\HtripK$ by~\eqref{eq:MCE:alt(C)_vs_pFw}. Similarly to~\eqref{eq:angle_sum_arg_rat}, we get $\sumbT_i\equiv0\pmod{\pi}$ for any \datr of $\G$, so by \cref{prop:t_imm=>holom}, $\G$ admits no t-immersions. 
The statement for $C\in\Grtnn(k,n)$ follows from \cref{lemma:la_vs_Bounda_Boundb}. 
\end{proof}

From now on, we assume that $\G$ and $C:=\Meas(\G,\wt)$ are \twonondeg, and that $\G$ is connected and satisfies $\helmin(\G)\geq2$. 
Let $\lalat$ and $\datrQLll=\datrQll$ be as in \cref{notn:Tll_xll}. 
 Let $(\y_i)_{i=1}^n,(\yt_i)_{i=1}^n\in\C^n$ be given by~\eqref{eq:intro:y_to_lalat}.

\begin{proposition}\label{lemma:Tll_is_t_imm}
If $\lalat\in\lalak$ then $\Tll$ is a t-immersion of $(\G,\wt)$.
\end{proposition}
\begin{proof}
First, applying moves \RV1 (\cref{dfn:DIM:R1}), we may assume that $\G$ has no parallel edges. 
Since $\helmin(\G)\geq2$, we have $\degG(\v)\geq3$ for each $\v\in\Vint$. 
Applying uncontraction moves \MV1, we obtain a graph $\G'$ in which every black vertex is trivalent. This corresponds to triangulating black faces of $\GD$ by inserting bigonal white faces. 
Since $\helmin(\G)\geq2$, by \cref{lemma:DIM:moves_vs_helWmin_trivalent,lemma:DIM:moves_vs_helWmin}, $\helWmin(\G')\geq2$ and $\helBmin(\G')\geq1$. Note that \cref{lemma:DIM:moves_vs_helWmin_trivalent} applies since $\G$ has no parallel edges and thus each trivalent black vertex $\b'$ of $\G'$ satisfies $|\Neigh_{\G'}(\b')|=\deg_{\G'}(\b')=3$.

By \cref{lemma:DIM:deleting_hel_verts}, $\G'\rem\{\w_1,\w_{2}\}$ admits an \APM for any two distinct white vertices $\w_1,\w_2$ of $\G'$. 
By \crefi{cor:lak_implies_black_imm}{lak_implies3}, $\epsK(\e_{\s})\epsK(\e_{\s+1})\brlaw<\w_{\s},\w_{\s+1}> <0$ in the notation of~\eqref{eq:TE:same_face_and_same_vertex}. Thus, $\Tll$ is orientation-preserving and injective on each black (triangular) face of $\GDp$. 
 Since the triangulation $\GDp$ of the black faces of $\GD$ was chosen in an arbitrary fashion, it follows that the image of each black face $\b\in\BVint$ of $\GD$ under $\Tll$ is a nondegenerate convex $d$-gon (where $d=\degG(\b)\geq3$) whose orientation is preserved by $\Tll$.

We similarly obtain the result for white faces of $\GD$. Thus, $\Tll$ satisfies conditions~\itemref{intro:t_imm_straight_convex}--\itemref{intro:t_imm_orientation} of \cref{dfn:intro:t_imm}. Condition~\itemref{intro:t_imm_gauge} is satisfied by~\eqref{eq:intro:primitive_TO}. 

Since $\Tll$ satisfies~\itemref{intro:t_imm_straight_convex}--\itemref{intro:t_imm_gauge}, the angle sums $\sumwTll(\f)$ and $\sumbTll(\f)$ are well defined for any vertex $\f\in\Faces$. 
 Similarly to~\eqref{eq:T_Kast_bdry}, for $\f\in\Fint$ adjacent to faces $\f_1,\f_2,\dots,\f_{2d}\in\Faces$ in counterclockwise order, we get $\prod_{j=1}^d\frac{\T(\f_{2j}) - \xT(\f)}{\T(\f_{2j-1}) - \xT(\f)} = - \wt(\bdrypath\f)$. 
This implies that $\sumwTll(\f)$ and $\sumbTll(\f)$ are equal to $\pi$ modulo $2\pi$. Since $\sumwTll(\f)>0$ and $\sumbTll(\f)>0$ for all $\f\in\Faces$, we have
\begin{equation}\label{eq:angle_ineqs_inside}
 \sumwTll(\f)\geq\pi \quad\text{and}\quad \sumbTll(\f)\geq\pi \quad\text{for all $\f\in\Fint$.}
\end{equation}
 For $i\in\brn$, by~\eqref{eq:T_Kast_bdry}, we find that $\sumbT_i \equiv \arg(\y_{i+1} / \y_i)$ and $\sumwT_i \equiv \arg(\yt_{i+1}/\yt_i) + \pi$ modulo $2\pi$. Since $\lalat\in\lalak$, we have 
$\arg(\y_{i+1} / \y_i)\in(0,\pi)$ and $\arg(\yt_{i+1}/\yt_i)\in (-\pi,0)$. 
Thus, 
\begin{equation}\label{eq:angle_ineqs_bdry}
\sumbT_i \geq \arg(\y_{i+1} / \y_i) \quad\text{and}\quad
\sumwT_i \geq \arg(\yt_{i+1}/\yt_i) + \pi \quad\text{for all $i\in\brn$.} 
\end{equation}
Summing up the angles of all white and black faces of $\Tll(\GD)$ as we did in~\eqref{eq:sum_up_angles1}--\eqref{eq:sum_up_angles2}, we see that each of the inequalities in~\eqref{eq:angle_ineqs_inside}--\eqref{eq:angle_ineqs_bdry} must in fact be an equality. This implies that $\Tll$ satisfies the angle conditions~\itemref{intro:t_imm_angles_int}--\itemref{intro:t_imm_angles_bdry}.
\end{proof}

\begin{proof}[Proof of \cref{thm:intro:t_imm_vs_triples}.]
Fix $\lalat\in\lalak$ such that $\la\subset C \subset \latp$. We view $\lalat$ as a pair of $2\times n$ matrices. 
By \cref{lemma:Tll_is_t_imm}, $\Tll$ is a t-immersion of $\G$. Even though the discrete holomorphic functions $(\Fwl,\Fbl)$ depend on a choice of Kasteleyn signs $\epsK$ for $\G$ (all of which are gauge equivalent; cf. \cref{rmk:DIM:Kast_gauge_eq}), the map $\Tll$ is invariant under gauge transformations and thus does not depend on the choice of $\epsK$. 

Let us now view $\lalat$ as a pair of $2$-planes, i.e., elements of $\Gr(2,n)$. In order to construct $\Tll$, we must choose $2\times n$ matrix representatives satisfying $\brla<i,i+1> >0$ and $\brlat[i,i+1]>0$ for all $i\in\brn$. Such representatives are determined up to the action of $\GL^+_2(\R)\times\GL^+_2(\R)$ on $\lalat$, where
$\GL^+_2(\R):=\{g\in\GL_2(\R)\mid\det(g)>0\}\cong \SL_2(\R)\times\Rtp$. By definition, acting by $\SL_2(\R)\times \SL_2(\R)$ on $\lalat$ corresponds to applying Lorentz transformations to $\xll$. Rescaling $\lalat\mapsto (t\cdot \la,\tilde t\cdot \lat)$ for some $t,\tilde t\in\Rtp$ results in rescaling $\xll$ by $t\cdot \tilde t$. Thus, $\xll$ is indeed defined up to global rescaling and Lorentz transformations.

Conversely, let $\Tcal$ be a t-immersion of $(\G,\wt)$. As we showed in \cref{prop:t_imm=>TRIPLES}, it gives rise to a triple $(\la,\lat,C)\in\TRIPLES$. Specifically, from $\Tcal$ we recover the pair $(\Fw,\Fb)$ via \cref{prop:t_imm=>holom} which is then converted into $\lalat$ by 
applying~\eqref{eq:intro:y_to_lalat}--\eqref{eq:intro:Fw_Fb_y}.
 The only ambiguity in this process arises from the 
gauge group action $\RtpVint\times\{\pm1\}^{\Vint}$ and the action $(\Fw,\Fb)\mapsto(z\Fw,z^{-1}\Fb)$ of $z\in\Cast$ from \cref{rmk:Fw_Fb_determined_up_to_const}. The pair $\lalat$ is invariant under $\RtpVint\times\{\pm1\}^{\Vint}$, and the action $(\Fw,\Fb)\mapsto(z\Fw,z^{-1}\Fb)$ corresponds to multiplying $(\la,\lat)\mapsto (g_z\cdot \la,(g_z^{-1})^T\cdot \lat)$ for $g_z\in\GL_2(\R)$ with $\det(g_z)=|z|^2$, and thus leaves the pair $(\la,\lat)\in\Gr(2,n)\times\Gr(2,n)$ invariant. 
\end{proof}

\section{Mandelstam variables, immanants, and t-embeddings}\label{sec:imm}
The goal of this section is to prove \cref{thm:intro:Mand_pos_exists} and use it to establish existence of t-embeddings, completing the proofs of \cref{thm:intro:B} and \cref{cor:intro:t_emb_exists}.

\subsection{Temperley--Lieb immanants}\label{sec:TL_imm}
Building on the results of~\cite{RhSk}, Lam~\cite{Lam_dimers} introduced a family of functions $\Delta_{\tau,T}$ on (the affine cone over) $\Gr(k,n)$. These functions are nonnegative on $\Grtnn(k,n)$, and in fact constitute the \emph{canonical basis} of the degree-$2$ part of the coordinate ring of $\Gr(k,n)$. Our goal is to express the Mandelstam variables $\Mijll$ in terms of the functions $\Delta_{\tau,T}$.

Recall that $\tau:\brn\to\brn$ is an \emph{involution} if $\tau(l) = r$ implies $\tau(r) = l$ for all $l,r\in\brn$. Alternatively, we may think of $\tau$ as a perfect matching of the elements of $\Stau:=\{l\in\brn\mid \tau(l)\neq l\}$. Thus, we can represent $\tau$ as a collection of \emph{arcs} $\{\{l,\tau(l)\}\mid l\in\Stau\}$.
 We say that an involution $\tau:\brn\to\brn$ is \emph{non-crossing} if there do not exist indices $1\leq a<b<c<d\leq n$ such that $\tau(a) = c$ and $\tau(b) = d$. 
A \emph{$(k,n)$-partial non-crossing matching} is a pair $\tauT$ such that $T\subset\brn$ and $\tau$ is a non-crossing involution satisfying $T\cap\Stau = \emptyset$ and $2|T| + |\Stau| = 2k$. We denote by $\Ttaukn$ the set of $(k,n)$-partial non-crossing matchings.
Given $A,B\in{\brn\choose k}$ and $\tauT\in\Ttaukn$, we say that $\tauT$ is \emph{compatible} with $(A,B)$ if $T = A\cap B$, $\Stau = (A\setminus B)\sqcup(B\setminus A)$, and $\tau(A\setminus B) = B\setminus A$.

The \emph{Temperley--Lieb immanants} $\{\Delta_{\tau,T}\mid \tauT\in\Ttaukn\}$ are uniquely defined by the relations
\begin{equation}\label{eq:TL_imm_dfn}
 \Delta_A(C)\Delta_B(C) = \sum_{\tauT} \Delta_{\tau,T}(C) \quad\text{for all $A,B\in{\brn\choose k}$ and $C\in\Gr(k,n)$,}
\end{equation}
where the summation is over all $\tauT\in\Ttaukn$ compatible with $(A,B)$. 

We denote by $\OmsG$ the set of double-dimer configurations $\Om$ on $\G$ (with no restrictions on which boundary vertices are used by $\Om$); cf. \cref{dfn:double_dimer_Omf}. 
For $\Om\in\OmsG$, we let $\tauTOm$ be the induced $(k,n)$-partial non-crossing matching, where $\TOm$ is the set of black boundary vertices used twice and white boundary vertices not used in $\Om$, and the matching $\tauOm$ records the \bdbd path connectivity of $\Om$. The set $\StauOm$ consists of boundary vertices used once by $\Om$.
 For each fixed $\tauT$, we let $\OmsGtauT:=\{\Om\in\OmsG\mid \tauTOm = \tauT\}$. It was shown in~\cite{Lam_dimers} that for $C:=\Meas(\G,\wt)$, 
\begin{equation}\label{eq:Delta_tauT_vs_OmsGtauT}
 \Delta_{\tau,T}(C) = \sum_{\Om\in\OmsGtauT} \wt(\Om).
\end{equation} 
For each $\tauT$ compatible with $(A,B)$, every $\Om\in\OmsGtauT$ can be split (in $2^{\ncycOm}$-many ways) into a union $\Apm_1\cup\Apm_2$ of two \APMs of $\G$, with $A = \partial\Apm_1$ and $B=\partial\Apm_2$. Thus,~\eqref{eq:Delta_tauT_vs_OmsGtauT} implies~\eqref{eq:TL_imm_dfn}.

\subsection{Immanants of \texorpdfstring{$(\La,\Lat)$}{(𝚲,𝚲̃)}}
We introduce a family of functions of $\LaLat\in\LaLak$ that we also call \emph{immanants}; cf. \cref{dfn:immanants_LaLat}. As we show in \cref{lemma:Arep_subset_immnn}, when $\LaLat\in\LaLaArep$ is \Areped by a weighted annular graph $(\Gbot,\wtbot)$, immanants of $\LaLat$ recover Temperley--Lieb immanants of $(\Gbot,\wtbot)$. 

\begin{figure}
\def\inputscl{1}
\setlength{\tabcolsep}{10pt}
\begin{tabular}{cc}
 \includegraphics[scale=\inputscl]{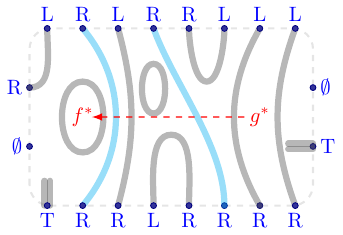}
&
 \includegraphics[scale=\inputscl]{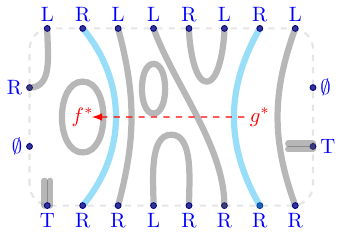}
\\
(a) $\dmu = 1$ & (b) $\dmu=2$
\end{tabular}
 \caption{\label{fig:marking} Two $\Omfg$-markings (see \cref{dfn:markings_Om}). The elements in $\TOm$ (resp., $\brn\setminus(\TOm\sqcup\StauOm)$) are marked by $\textcolor{blue}{\Tmark}$ (resp., $\textcolor{blue}{\Emark}$). The paths $\Pathbd_1,\Pathbd_2$ are shown in blue. 
}
\end{figure}

\begin{definition}[Markings]\label{dfn:markings_Om}
Let $\ff,\f\in\Faces$. We denote by $\OmsG(\ff\ssep\f)$ the set of double-dimer configurations $\Om\in\OmsG$ that contain at least two \bdbd paths separating $\ff$ from $\f$. Given $\Om\in\OmsG(\ff\ssep\f)$, an \emph{$\Omfg$-marking} is a map $\marking:\StauOm\to\{\Lmark,\Rmark\}$ satisfying the following conditions. 
\begin{enumerate}[label=(\arabic*)]
\item\label{marking1} There are exactly two \bdbd paths $\Pathbd_1,\Pathbd_2$ in $\Om$ that have both endpoints marked by $\Rmark$. Both paths $\Pathbd_1,\Pathbd_2$ must separate $\ff$ from $\f$.
\item For all other \bdbd paths in $\Om$, one endpoint is marked by $\Lmark$ and the other one by $\Rmark$.
\end{enumerate}
See \cref{fig:marking} for examples. 
We denote the set of $\Omfg$-markings by $\musOmfg$.
Given $\marking\in\musOmfg$, 
 we denote $\Lmu:=\TOm\sqcup\marking^{-1}(\Lmark)$ and $\Rmu:=\TOm\sqcup\marking^{-1}(\Rmark)$. 
Let $\dmufg$ be the number of \bdbd paths in $\Om$ separating $\ff$ from $\f$ and located strictly between $\Pathbd_1$ and $\Pathbd_2$. 
\end{definition}

\begin{proposition}\label{lemma:immanant_fg}
Assume that $\G$ admits an \APM. Let $C:=\Meas(\G,\wt)$, $\LaLat\in\LaLak$, and $\lalat:=\PhiLL(C)$. Then for all $\ff,\f\in\Faces$, we have (cf. \cref{notn:Tll_xll})
\begin{align}
\label{eq:immanant_fg}
 \frac14(\xll(\ff)-\xll(\f))^2 &= \sum_{\Om\in\OmsG(\ff\ssep\f)} \TtaucoefOmfgof\LaLat\cdot \wt(\Om),
 \quad\text{where}\quad\\
\label{eq:immanant_fg_dfn}
 \TtaucoefOmfgof\LaLat:&=\sum_{\mu\in\musOmfg} (-1)^{\dmufg} \Delta_{\Lmu}(\Lap)\Delta_{\Rmu}(\Lat).
\end{align}
\end{proposition}
\begin{proof}
Our first goal is to give an analog of~\eqref{eq:TE:Omvec} for $\G$ (as opposed to $\Gf$). Let $\e_1,\e_2\in\Edges$
be edges of $\G$ with endpoints $\ebar_1=\{\b_1,\w_1\}$ and $\ebar_2=\{\b_2,\w_2\}$. 
 Let $\OmsG(\e_1,\e_2)$ be the set of $\Om\in\OmsG$ such that $\e_1,\e_2$ belong to two different \bdbd paths $\Pathbd_1,\Pathbd_2$ of $\Om$. For $\Om\in\OmsG(\e_1,\e_2)$, we let $\Omvecbw_{\e_1,\e_2}$ be the unique orientation of $\Pathbd_1,\Pathbd_2$ such that the edges $\e_1,\e_2$ are both oriented from black to white. 
Similarly to \cref{ssec:TE:Kast2}, we set $\epsOm(\Omvecbw_{\e_1,\e_2}):=-1$ if $\Omvecbw_{\e_1,\e_2}$ is \emph{alternating} (i.e., if the endpoints of $\Pathbd_1,\Pathbd_2$ alternate between ``in'' and ``out'' around the boundary of the disk) and $\epsOm(\Omvecbw_{\e_1,\e_2}):=+1$ otherwise. 
 Let $\musOmee$ be the set of maps $\marking:\StauOm\to\{\Lmark,\Rmark\}$ such that both endpoints of each of $\Pathbd_1,\Pathbd_2$ are marked with $\Rmark$, and for all other \bdbd paths in $\Om$, one endpoint is marked by $\Lmark$ and the other one by $\Rmark$. 
Let $\dOmee$ be the number of \bdbd paths in $\Om$ not passing through $\e_1,\e_2$ and separating $\e_1$ from $\e_2$. 
We claim that
\begin{align}\label{eq:TE:Omvec_G}
 \wtK(\e_1)\wtK(\e_2)\cdot \brlaw<\w_1,\w_2> \cdot \brlatb[\b_1,\b_2] &= 
 \sum_{\Om\in\OmsG(\e_1,\e_2)}\wt(\Om)\cdot \TtaucoefOmeeof\LaLat, 
\quad\text{where}\\
\label{eq:TE:Omvec_G_TtaucoefOmee}
 \TtaucoefOmeeof\LaLat :&=(-1)^{\dOmee}\epsOm(\Omvecbw_{\e_1,\e_2}) \sum_{\mu\in\musOmee}\Delta_{\Lmu}(\Lap)\Delta_{\Rmu}(\Lat). 
\end{align}
Indeed, by \cref{rmk:can_add_edges}, we may assume that $\w_1,\w_2$ are connected by a path $\Pathw$ in $\G$.
Following the proof of~\eqref{eq:TE:Omvec}, we choose a path $\Pathb$ from $\b_1$ to $\b_2$ that differs from $\Pathw$ only in the edges $\e_1,\e_2$ and let $\Cut:=\CutPw=\CutPb$ be the associated zig-zag cut. 
 We apply~\eqref{eq:IP:brlaw_G}--\eqref{eq:IP:brlatb_G} and let $\Om:=\apmw\cup\apmb\cup\{\e_1,\e_2\}$ (multiset union), where $\apmw\in\APMS(\Giww)$ and $\apmb\in\APMS(\Gbb)$. We claim that $\Om\in\OmsG(\e_1,\e_2)$. Indeed, let us orient all edges in $\apmw$ (resp., $\apmb$) from black to white (resp., from white to black). The vertices $\w_1,\b_1,\w_2,\b_2$ are sources in the resulting directed graph, and every other interior vertex $\v\in\Vint\setminus\{\w_1,\b_1,\w_2,\b_2\}$ has one incoming and one outgoing edge, so each of the four paths in $\apmw\cup\apmb$ starting at $\w_1,\b_1,\w_2,\b_2$ must terminate at the boundary. In particular, every cycle in $\Om$ intersects the cut $\CutPw=\CutPb$ an even number of times. 
It follows that the left-hand side of~\eqref{eq:TE:Omvec_G} equals
$\sum_{\Om\in\OmsG(\e_1,\e_2)} \epsOmPathw \wt(\Om)\cdot \sum_{\mu\in\musOmee}\Delta_{\Lmu}(\Lap)\Delta_{\Rmu}(\Lat)$, where $\epsOmPathw$ counts the parity of the number of intersections between $\Cut$ and the union of \bdbd paths of $\Om$. Similarly to the proof of~\eqref{eq:TE:Omvec}, we see that $\epsOmPathw=\epsOm(\Omvecbw_{\e_1,\e_2})\cdot (-1)^{\dOmee}$. This shows~\eqref{eq:TE:Omvec_G}--\eqref{eq:TE:Omvec_G_TtaucoefOmee}. 
We deduce~\eqref{eq:immanant_fg}--\eqref{eq:immanant_fg_dfn}
from~\eqref{eq:TE:Omvec_G}--\eqref{eq:TE:Omvec_G_TtaucoefOmee} using an argument entirely analogous to the one in the proof of~\eqref{eq:TE:OmsGf_pos}.
\end{proof}

We will be particularly interested in the case where $\ff=\bdf_i$ and $\f=\bdf_j$ are boundary faces of~$\G$. For $i+2\leq j\leq i+n-2$, we denote 
$\OmsGij:=\OmsG(\bdf_i\ssep\bdf_j)$ and 
$\Ttaucoef(\La,\Lat):=\TtaucoefOmfgijof\LaLat$ and refer to $\Omfgij$-markings as \emph{$\Ttauij$-markings}, 
where $\tauT:=\tauTOm$. 
Specializing \cref{lemma:immanant_fg} to this case, we obtain the following result. 
\begin{corollary}\label{thm:Mand_Ttau}
For $\LaLat\in\LaLak$, $C\in\Grtnn(k,n)$, and $\lalat = \PhiLL(C)$, we have
\begin{equation}\label{eq:Mand_Ttau}
 \frac14\Mijll = \sum_{\tauT\in\Ttaukn} \Ttaucoef(\La,\Lat)\cdot \Delta_{\tau,T}(C) 
\quad\text{for all $i+2\leq j\leq i+n-2$.}
\end{equation}
\end{corollary}

\begin{figure}
 \includegraphics[scale=1]{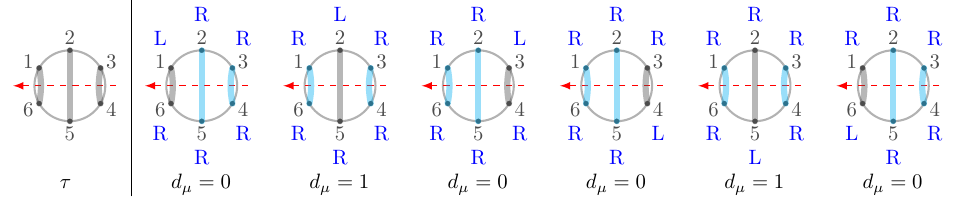}
 \caption{\label{fig:tau} A matching $\tau$ and the six $\Ttauij$-markings; see \cref{ex:tau_markings}.}
\end{figure}

\begin{example}\label{ex:tau_markings}
Let $k=3$, $n=6$, $i=0$, $j=3$, $T = \emptyset$, and 
$\tau=\{\{1,6\},\{2,5\},\{3,4\}\}$. 
 The six possible $\Ttauij$-markings are shown in \cref{fig:tau}. Exactly two of them have $\dmu=1$, and the rest have $\dmu=0$. Thus,~\eqref{eq:immanant_fg_dfn} specializes to
\begin{equation}\label{eq:intro:imm_example}
 \Ttaucoefxx{0}{3}(\La,\Lat) = 
 \BRLAP<1> \BRLATP<1>
- \BRLAP<2> \BRLATP<2>
+ \BRLAP<3> \BRLATP<3>
+ \BRLAP<4> \BRLATP<4>
- \BRLAP<5> \BRLATP<5>
+ \BRLAP<6> \BRLATP<6>,
\end{equation}
where $\BRLAP<i>$ (resp., $\BRLATP<i>$) denotes the $i$-th entry of the totally positive $1\times 6$ matrix $\La^\perp$ (resp., $\altp(\Lat)$). As discussed in~\cite[Equation~(3.24)]{DFLP} and~\cite[Equation~(6.1.52)]{Parisi_thesis}, the condition $\Ttaucoefxx{0}{3}(\La,\Lat)>0$ is sufficient for the Mandelstam variable $\Mxxll03$ to be nonnegative on $\MomLL$.
\end{example}

\begin{definition}[Immanants of $\LaLat$]\label{dfn:immanants_LaLat}
Let $i,j\in\Z$ be such that $i+2\leq j\leq i+n-2$. 
An arc of $\tau$ is called \emph{\Iijspec} if 
 the faces $\bdf_i,\bdf_j$ are on the opposite sides of it. 
Let $\Ttauknij$ be the set of $\tauT\in\Ttaukn$ such that $\tau$ has at least two \Iijspec arcs. We refer to the functions $\{\Ttaucoef\mid i+2\leq j\leq i+n-2\text{ and }\tauT\in\Ttauknij\}$ as \emph{immanants of $\LaLat$}.
\end{definition}

\begin{definition}\label{dfn:Ttauknij}
We say that $(\La,\Lat)\in\LaLak$ is \emph{immanant-positive} (resp., \emph{immanant-nonnegative}) if $\Ttaucoef(\La,\Lat)>0$ (resp., $\Ttaucoef(\La,\Lat)\geq0$) for all $i+2\leq j \leq i+n-2$ and all $\tauT\in\Ttauknij$. We denote by $\LaLaimmp$ (resp., $\LaLaimmnn$) the set of immanant-positive (resp., immanant-nonnegative) pairs $\LaLat\in\LaLak$.
\end{definition}

\begin{remark}\label{rmk:Ttauknij_>0_vs_0}
It follows from~\eqref{eq:immanant_fg_dfn} that if $\tauT\in\Ttauknij$ has exactly two \Iijspec arcs then $\Ttaucoef(\La,\Lat)>0$ for all $\LaLat\in\LaLak$ since each of the $2^{|\tau|-2}$ $\Ttauij$-markings satisfies $\dmu=0$. Furthermore, if $\tauT\notin\Ttauknij$ then $\Ttaucoef(\La,\Lat)=0$ since in this case, the set of $\Ttauij$-markings is empty.
\end{remark}

\begin{remark}
Let $\ff,\f\in\Faces$, $\Om\in\OmsG(\ff\ssep\f)$, and $\tauT:=\tauTOm$. Pick any indices $i,j\in\brn$ such that $\ff$ and $\bdf_{i}$ (resp., $\f$ and $\bdf_{j}$) belong to the same connected component of
the complement of the union of \bdbd paths of $\Om$ inside the disk $\Disk$.
 It follows from~\eqref{eq:immanant_fg_dfn} that 
\begin{equation}\label{eq:Ttaucoef_fg_vs_ij}
 \TtaucoefOmfgof\LaLat = \Ttaucoef\LaLat.
\end{equation}
Thus, each coefficient $\TtaucoefOmfgof\LaLat$ on the right-hand side of~\eqref{eq:immanant_fg} equals some immanant $\Ttaucoef\LaLat$.
\end{remark}

\subsection{Immanants of annular graphs}\label{ssec:annular_immanants}
Assume that $(\Gf,\wtf)$ is a planar bipartite graph of type $(2,4)$ obtained by gluing graphs $(\G,\wt)$ and $(\Gbot,\wtbot)$ as in \cref{ssec:annular}. 
For $i+2\leq j\leq i+n-2$, we denote $\OmsGfij:=\OmsGf(\bdf_i\ssep\bdf_j)$, where $\bdf_i,\bdf_j$ denote boundary faces of $\G$ viewed as faces of $\Gf$.

Let $\OmsGbot$ be the set of double-dimer configurations $\Ombot$ on $\Gbot$ 
 such that each (black) outer boundary vertex $\bdvout_1,\bdvout_2,\bdvout_3,\bdvout_4$ is used exactly once by $\Ombot$.
We similarly record the boundary connectivity of $\Ombot\in\OmsGbot$ by a pair $\tauTOmbot$, where $\TOmbot:=\{i\in\brn\mid \bdvin_i\text{ is used twice by $\Ombot$}\}$ and $\tauOmbot$ is a non-crossing matching
(when drawn inside $\Ann$)
 on the set $\StauOmbot=\StauOmbotin\sqcup\StauOmbotout$ with
 $\StauOmbotout:=\{\botoutacc1,\botoutacc2,\botoutacc3,\botoutacc4\}$.
\begin{definition}
We say that $\Om\in\OmsG$ and $\Ombot\in\OmsGbot$ are \emph{gluable} if their union $\Omf$ belongs to $\OmsGf$.
\end{definition}
\noindent If $\Omf\in\OmsGfij$ then its restriction to $\G$ belongs to $\OmsGij$. 
 Conversely, for gluable $\Om\in\OmsGij$ and $\Ombot\in\OmsGbot$, one can recover the topological connectivity $\tauTOmf$ of their union $\Omf$ and determine whether $\Omf$ belongs to $\OmsGfij$ by ``gluing'' $\tauTOm$ and $\tauTOmbot$ and checking whether both arcs of the resulting matching $\tauOmf$ on $\{\botoutacc1,\botoutacc2,\botoutacc3,\botoutacc4\}$ separate $\bdf_i$ from $\bdf_j$. 
\begin{definition}
For $\tauT\in\Ttauknij$, we let $\OmsGbotcoef$ be the set of $\Ombot\in\OmsGbot$ such that for some (equivalently, any, by the above discussion) $\Om\in\OmsGtauT$, the union $\Omf=\Ombot\cup\Om$ belongs to $\OmsGfij$. 
\end{definition}
\noindent We denote by $\ncycop(\tau,\Ombot):=\ncycop(\Om,\Ombot)$ the number of cycles of $\Omf=\Ombot\cup\Om$ that cross the boundary circle of $\G$, i.e., that are not among the cycles of $\Om$ or of $\Ombot$.

Since each $\Omf\in\OmsGfij$ splits uniquely as a union $\Om\cup\Ombot$, we obtain the following result.
\begin{corollary}
Suppose that $(\Gf,\wtf)$ is obtained by gluing $(\Gbot,\wtbot)$ and $(\G,\wt)$ as above. 
Let $C:=\Meas(\G,\wt)$. 
Then for each $i+2\leq j\leq i+n-2$,
\begin{equation}\label{eq:OmsGf_vs_Ttaucoef}
 \sum_{\Omf\in\OmsGfij} \wtf(\Omf) = 
 \sum_{\tauT\in\Ttauknij} \Ttaucoef(\Gbot,\wtbot)\cdot \Delta_{\tau,T}(C),
 \ \ \ \text{where}\ \ \ 
 \Ttaucoef(\Gbot,\wtbot):=\sum_{\Ombot\in\OmsGbotcoef} 2^{\ncycop(\tau,\Ombot)}\wtbot(\Ombot).
\end{equation}
\end{corollary}

\begin{figure}
\def\inputscl{1.3}
\begin{tabular}{ccc}
 \includegraphics[scale=\inputscl]{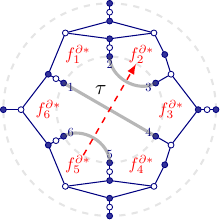}
& 
 \includegraphics[scale=\inputscl]{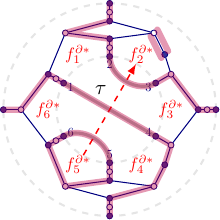}
& 
 \includegraphics[scale=\inputscl]{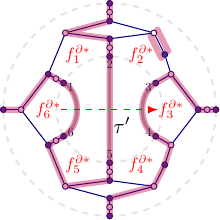}
\\
 (a) $\tau=\{\{1,4\},\{2,3\},\{5,6\}\}$
& (b) \textcolor{magenta}{$c^{2\ssep5}_{\tau,\emptyset}(\Gbot,\wtbot) >0$}
& (c) \textcolor{magenta}{$c^{3\ssep6}_{\tau',\emptyset}(\Gbot,\wtbot)=0$}
\end{tabular}
 \caption{\label{fig:Arep-not-immp} The graph $\Gbot$ gives rise to an \Arep pair $\LaLat$ which is 
immanant-nonnegative but not immanant-positive; see \cref{ex:Arep-not-immp}.
}
\end{figure}

\begin{example}\label{ex:Arep-not-immp}
Let $k=3$, $n=6$, and let $\LaLat$ be \Areped by $(\Gbot,\wtbot)$ for $\Gbot$ shown in \figref{fig:intro-annular}(a). One can check that the boundary measurements $\Cbotpm$ defined in~\eqref{eq:OAC:Cbotpm_meas} belong to $\Grtp(1,6)$ and $\Grtp(5,6)$, so by~\eqref{eq:OAC:Arep}, $\LaLat\in\LaLaArep$. Let $T=\emptyset$ and consider matchings $\tau=\{\{1,4\},\{2,3\},\{5,6\}\}$ and $\tau':=\{\{1,6\},\{2,5\},\{3,4\}\}$ shown in \cref{fig:Arep-not-immp}. We see that $c^{2\ssep5}_{\tau,T}(\Gbot,\wtbot)>0$ since $\bdf_2$ and $\bdf_5$ are separated by two \bdbd paths in some double-dimer configuration $\Omf$ whose restriction to $\G$ is $\tau$. An example of such $\Omf$ is shown in \figref{fig:Arep-not-immp}(b). On the other hand, an example of a double-dimer configuration $\Omf'$ whose restriction to $\G$ is $\tau'$ is shown in \figref{fig:Arep-not-immp}(c). One can check that at most one \bdbd path of any such $\Omf'$ separates $\bdf_3$ from $\bdf_6$, so $c^{3\ssep6}_{\tau',T}(\Gbot,\wtbot)=0$, and thus $\LaLat$ is not immanant-positive. This shows that $\LaLaArep\not\subset\LaLaimmp$.
\end{example}

\begin{proposition}\label{lemma:Arep_subset_immnn}
We have $\LaLaArep\subset\LaLaimmnn$, and when $\LaLat\in\LaLaArep$ is \Areped by $(\Gbot,\wtbot)$,
\begin{equation}\label{eq:Ttaucoef(LaLat)=Ttaucoef(Gbot,wtbot)}
 \Ttaucoef(\La,\Lat) = \Ttaucoef(\Gbot,\wtbot) \quad\text{for all $i+2\leq j \leq i+n-2$ and all $\tauT\in\Ttauknij$}.
\end{equation}
\end{proposition}
\begin{proof}
We show~\eqref{eq:Ttaucoef(LaLat)=Ttaucoef(Gbot,wtbot)}. Fix $\LaLat\in\LaLak$. For $C\in\Grtnn(k,n)$ and $\lalat:=\PhiLL(C)$, by~\eqref{eq:Mand_Ttau}, $\Mijll$ is a polynomial function of degree $2$ in the Pl\"ucker coordinates of $C$. By~\cite[Theorem~3.10]{Lam_dimers}, the functions $\{\Delta_{\tau,T}\mid \tauT\in\Ttaukn\}$ form a basis of the degree-$2$ part of the coordinate ring of $\Gr(k,n)$. 
In particular, they are linearly independent. 
Assume that $\LaLat\in\LaLaArep$ is \Areped by $(\Gbot,\wtbot)$. By~\eqref{eq:TE:OmsGf_pos} and~\eqref{eq:OmsGf_vs_Ttaucoef},
\begin{equation}\label{eq:Mand_vs_OmsGf}
 \frac14\Mijll %
 = \sum_{\Omf\in\OmsGfij} \wtf(\Omf) 
=\sum_{\tauT\in\Ttauknij} \Ttaucoef(\Gbot,\wtbot)\cdot \Delta_{\tau,T}(C).
\end{equation} 
Since the functions $\Delta_{\tau,T}$ are linearly independent, comparing~\eqref{eq:Mand_vs_OmsGf} and~\eqref{eq:Mand_Ttau},
 we get~\eqref{eq:Ttaucoef(LaLat)=Ttaucoef(Gbot,wtbot)}. Since the coefficients $\Ttaucoef(\Gbot,\wtbot)$ are manifestly nonnegative, we get $\LaLat\in\LaLaimmnn$.
\end{proof}
\noindent We remark that~\eqref{eq:Ttaucoef(LaLat)=Ttaucoef(Gbot,wtbot)} can also be shown by a direct combinatorial argument similar to the one in the proof of \cref{lemma:TE:OmsGf_pos}.

\begin{lemma}\label{lemma:Ttaucoef_nonzero_on_Fl}
For all $i+2\leq j\leq i+n-2$ and $\tauT\in\Ttauknij$, the immanant $\Ttaucoef$ is not identically zero on $\LaLafl$.
\end{lemma}
\begin{proof}
Let us fix two \Iijspec arcs $\{l_1,r_1\},\{l_2,r_2\}$ of $\tau$.
\def\zr{_0}
 Let $M\zr\in\GL_n(\R)$ be an $n\times n$ matrix defined as follows. Let $e_1,e_2,\dots,e_n\in\R^n$ be the standard basis of $\R^n$. 
The first $k-2$ rows of $M\zr$ consist of vectors $\{e_\s\mid\s\in T\}$ together with $\{\e_l+\e_r\mid \{l,r\}\text{ is an arc of $\tau$ with }\{l,r\}\neq\{l_1,r_1\},\{l_2,r_2\}\}$. The next four rows of $M\zr$ consist of vectors $e_{l_1},e_{r_1},e_{l_2},e_{r_2}$. The remaining rows are chosen arbitrarily. Next, we change the signs of entries of $M\zr$ so that $(\Lap\zr,\Lat\zr):=\pifl(M\zr)$ would belong to $\Grtnn(k-2,n)\times\Grtnn(k+2,n)$; such a choice of signs is possible because $\tau$ is non-crossing. 

Let $\marking$ be a $\Ttauij$-marking. In order to have $\Delta_{\Lmu}(\La\zr^\perp)\Delta_{\Rmu}(\Lat\zr)\neq0$, we must have $l_1,r_1,l_2,r_2\in\Rmu$. 
For any $\Ttauij$-marking $\marking$ satisfying $\marking(l_1)=\marking(l_2)=\marking(r_1)=\marking(r_2)=\Rmark$, we have $\Delta_{\Lmu}(\La\zr^\perp)=\Delta_{\Rmu} (\Lat\zr)=1$.
Since $\{l_1,r_1\}$ and $\{l_2,r_2\}$ are arcs of $\tau$, all nonzero terms on the right-hand side of~\eqref{eq:immanant_fg_dfn} have the same sign (because $\dmu$ is fixed). 
 Thus, $\Ttaucoef(\La\zr,\Lat\zr)\neq0$, and so $\Ttaucoef$ is not identically zero on the $2$-step flag variety $\LaLafl$. 
\end{proof}

\begin{corollary}\label{thm:imm_pos_Fltp}
We have $\LaLaflp\subset\LaLaimmp$.
\end{corollary}
\begin{proof}
Recall from \cref{lemma:OAC:flag_pos=>Arep,lemma:Arep_subset_immnn} that $\LaLaflp\subset\LaLaArep\subset\LaLaimmnn$. Fix $i+2\leq j\leq i+n-2$ and $\tauT\in\Ttauknij$. The function $\Ttaucoef$ is nonnegative on $\LaLaflp$. By \cref{lemma:Ttaucoef_nonzero_on_Fl}, it is not identically zero on $\LaLafl$. Since $\LaLaflp$ is Zariski dense inside $\LaLafl$, there exists $\LaLat\in\LaLaflp$ such that $\Ttaucoef\LaLat>0$. 
Let $(\Gbot,\wtbot)$ be the annular bipartite graph constructed in the proof of \cref{lemma:OAC:flag_pos=>Arep} so that $\LaLat$ is \Areped by $(\Gbot,\wtbot)$.
By~\eqref{eq:Ttaucoef(LaLat)=Ttaucoef(Gbot,wtbot)}, 
$\Ttaucoef(\Gbot,\wtbot) = \Ttaucoef\LaLat > 0$.
 Thus, $\OmsGbotcoef\neq\emptyset$ and so 
$\Ttaucoef(\Gbot,\wtbotp)>0$ for all $\wtbotp\in\Rtp^{|\Ebot|}$. 
It follows from the proof of \cref{lemma:OAC:flag_pos=>Arep} that any other pair $(\La',\Lat')\in\LaLaflp$ may be \Areped by the same annular graph $(\Gbot,\wtbotp)$ for a suitable choice $\wtbotp\in\Rtp^{|\Ebot|}$ of edge weights. Thus, $\Ttaucoef(\La',\Lat')>0$ for all $(\La',\Lat')\in\LaLaflp$.
\end{proof}

The subset $\LaLaimmp\subset\LaLak$ is open (by definition) but \emph{a priori}, it could be empty. On the other hand, $\LaLaflp$ is manifestly nonempty (since $\Gtp\neq\emptyset$) but it is contained in a proper subvariety $\LaLafl$ of $\LaLatbf_{k,n}:=\Gr(n-k+2,n)\times\Gr(k+2,n)$. Thanks to \cref{thm:imm_pos_Fltp}, we obtain the following. 

\begin{corollary}\label{lemma:LaLaimmp_nonempty_Z_dense}
The open subset $\LaLaimmp\subset\LaLak$ is nonempty (and therefore Zariski dense).
\end{corollary}

Since $\Delta_{\tau,T}(C)\geq0$ for all $C\in\Grtnn(k,n)$ (cf. \cite{Lam_dimers}), we see from~\eqref{eq:immanant_fg} that $\PhiLL(C)$ is \Mdash nonnegative when 
$\LaLat\in\LaLaimmnn$ 
 and $C\in\Grtnn(k,n)$, finishing the proof of \cref{thm:intro:Mand_pos_exists}.

\begin{remark}
A specific pair $\LapLat$ of Vandermonde matrices was conjectured in~\cite[Equation~(2.33)]{DFLP} 
to satisfy the conclusion of \cref{thm:intro:Mand_pos_exists}. 
 \Cref{thm:imm_pos_Fltp} confirms this prediction. 
\end{remark}

We have shown that 
\begin{equation}\label{eq:inclusions_LaLa}
 \LaLaflp\subset\LaLaimmp \quad\text{and}\quad \LaLaflp\subset\LaLaArep\subset\LaLaimmnn.
\end{equation}
Observe also that $\LaLaflp\subset\LaLafl$ and by~\eqref{eq:OCA:Cbotm_subset_Cbotp}--\eqref{eq:OAC:Arep}, $\LaLaArep\subset\LaLafl$. Since $\LaLafl$ is a proper subvariety of $\LaLatbf_{k,n}$, we get $\LaLaflp\subsetneq\LaLaimmp$ and $\LaLaArep\subsetneq\LaLaimmnn$ by \cref{lemma:LaLaimmp_nonempty_Z_dense}. By \cref{ex:Arep-not-immp}, $\LaLaArep\not\subset\LaLaimmp$, so $\LaLaflp\subsetneq\LaLaArep$. Thus, all inclusions in~\eqref{eq:inclusions_LaLa} are proper. 
This motivates the following question.
\begin{question}\label{que:immp_vs_flag_vs_Arep}
Do we have $(\LaLaimmp\cap\LaLafl)\subset\LaLaArep$ or
$(\LaLaimmnn\cap\LaLafl)=\LaLaArep$?
\end{question}

\subsection{\Twosep faces}

\begin{definition}%
\label{dfn:PROP:2sep}
We say that $\ff,\f\in\Faces$ are \emph{\twosep} in $\G$ if $\OmsG(\ff\ssep\f)\neq\emptyset$ (cf. \cref{dfn:markings_Om}).
\end{definition}

\begin{remark}\label{rmk:exactly_twosep}
Equivalently, $\ff,\f\in\Faces$ are \twosep if and only if there exists $\Om\in\OmsG$ containing \emph{exactly} two \bdbd paths separating $\ff$ from $\f$. Indeed, if $\Om=\apm_1\cup\apm_2$ contains more than two such paths, we can swap the edges of $\apm_1$ for the edges of $\apm_2$ along each of the remaining \bdbd paths.
\end{remark}

\begin{corollary}\label{lemma:PROP:Mpos}
Assume that $\G$ admits an \APM. Let $(\La,\Lat)\in\LaLaimmnn$, $C:=\Meas(\G,\wt)$, and $\lalat = \PhiLL(C)$. 
Then for any $\ff,\f\in\Faces$, we have $(\xll(\ff)-\xll(\f))^2>0$ if $\ff,\f$ are \twosep in $\G$ and $(\xll(\ff)-\xll(\f))^2=0$ otherwise. 
\end{corollary}
\begin{proof}
Since $\LaLat\in\LaLaimmnn$, by~\eqref{eq:Ttaucoef_fg_vs_ij}, all coefficients on the right-hand side of~\eqref{eq:immanant_fg} are nonnegative. If $\ff,\f$ are not \twosep in $\G$ then the right-hand side of~\eqref{eq:immanant_fg} is zero. Otherwise, by \cref{rmk:exactly_twosep}, there exists $\Om\in\OmsG$ separating $\ff$ from $\f$ by exactly two \bdbd paths, in which case the coefficient $\TtaucoefOmfgof\LaLat$ of $\wt(\Om)$ is strictly positive by \cref{rmk:Ttauknij_>0_vs_0}. 
\end{proof}

\begin{definition}\label{dfn:fullysep}
We say that $\G$ \emph{\isfullysep} if for all $i+2\leq j\leq i+n-2$, the faces $\bdf_i$ and $\bdf_j$ of $\G$ are \twosep.%
\end{definition}

\begin{definition}
Let $\fap\in\Boundkn$. For $i+2\leq j\leq i+n-2$, let
\begin{equation*}%
 \rhoij(\fap):= |\{i<s\leq j\mid j<\fap(s)\leq i+n\}|.
\end{equation*}
\end{definition}

In the following lemma, we denote $\bbdvij:=\{\bdv_s\mid s\in\Iij\}$ and $\bbdvji:=\{\bdv_s\mid s\in\Iji\}$, where $\Iij,\Iji\subset\brn$ are cyclic intervals obtained by taking the elements of $(i,j]$ and $(j,i+n]$ modulo $n$. 
\begin{lemma}\label{lemma:TREE:fullysep_vs_rhoij}
Assume that $\G$ admits an \APM and let $\fap=\fG\in\Boundkn$. 
Let $\Icalbar_\G=(\Ibar_1,\Ibar_2,\dots,\Ibar_n)$ (cf. \cref{dfn:positroid_Grneck}). 
Fix $i+2\leq j\leq i+n-2$. Pick an \APM $\apm_{i+1}$ with $\partial\apm_{i+1} = \Ibar_{i+1}$
 and let $\GO_{i+1}$ be the associated perfect orientation of $\G$ (cf. \cref{dfn:perfect_orient}). 
The following are equivalent.
\begin{enumerate}[label=(\arabic*)] 
\item\label{rho1:twosep} $\bdf_i$ and $\bdf_j$ are \twosep in $\G$;
\item\label{rho2:GO} $\GO_{i+1}$ contains a pair of vertex-disjoint directed paths from $\bbdvij$ to $\bbdvji$;
\item\label{rho3:rhoij} $\rhoij(\fap)\geq2$.
\end{enumerate}
Furthermore, if $\helmin(\G)\geq2$ then the above conditions are equivalent to 
\begin{enumerate}[label=(\arabic*)]
 \setcounter{enumi}{3}
\item\label{rho4:share_face} $\bdf_i,\bdf_j$ do not share a common face of $\GD$.
\end{enumerate}
\end{lemma}
\begin{proof}
Let $r:2^{\brn}\to\Z$ be the rank function of the positroid $\MatroidG$ given by $r(I):=\max_{J\in\MatroidG}|I\cap J|$ for all $I\subset\brn$. 
It is well known that 
$r(\Iij) = |\Ibar_{i+1}\cap \Iij|=|\{s\in(i,j]\mid \fap(s)> j\}|$. 
 Thus,
\begin{equation}\label{eq:rhoij_k_Iij_Iji}
 k+\rhoij(\fap) = r(\Iij)+r(\Iji). 
\end{equation}
Assume that~\itemref{rho1:twosep} holds. 
Let $\Om\in\OmsGij$. 
We can represent $\Om=\apm_1\cup\apm_2$ 
as a union of two \APMs such that $J_1:=\partial\apm_1$ and $J_2:=\partial\apm_2$ satisfy 
$|J_1\cap\Iij|\geq |J_2\cap\Iij|+2$. 
We have $r(\Iij)\geq|J_1\cap\Iij|$ and 
$r(\Iji)\geq |J_2\cap\Iji|=k - |J_2\cap\Iij|$ by the definition of $r$.
Substituting this into~\eqref{eq:rhoij_k_Iij_Iji}, we get $\rhoij(\fap)\geq2$. 
Thus,~\itemref{rho1:twosep}$\Longrightarrow$\itemref{rho3:rhoij}. 
Next, assume that~\itemref{rho3:rhoij} holds and let $\apm_{i+1}$ and $\apm_{j+1}$ be \APMs with boundaries $\Ibar_{i+1}$ and $\Ibar_{j+1}$. 
By~\eqref{eq:rhoij_k_Iij_Iji}, we get $|\Ibar_{i+1}\cap\Iij|\geq|\Ibar_{j+1}\cap\Iij|+2$. 
Recall from \cite{Talaska} that double-dimer configurations of the form $\apm_{i+1}\cup\apm$ for some $\apm\in\APMS(\G)$ are in bijection with \emph{flows} in $\GO_{i+1}$.
Here, since $\GO_{i+1}$ is perfectly oriented, each flow is a vertex-disjoint union of directed \bdbd paths and directed cycles.
The perfect orientation $\GO(\apm)$ (cf. \cref{dfn:perfect_orient}) is obtained from $\GO_{i+1}$ by reversing all edges of the corresponding flow. 
In particular, since $|\Ibar_{i+1}\cap\Iij|\geq|\Ibar_{j+1}\cap\Iij|+2$, $\GO_{j+1}:=\GO(\apm_{j+1})$ is obtained by reversing at least two vertex-disjoint directed paths in $\GO_{i+1}$ from $\bbdvij$ to $\bbdvji$. 
This shows~\itemref{rho3:rhoij}$\Longrightarrow$\itemref{rho2:GO}. The implication~\itemref{rho2:GO}$\Longrightarrow$\itemref{rho1:twosep} follows similarly: if $\GO$ is obtained from $\GO_{i+1}$ by reversing two vertex-disjoint paths from $\bbdvij$ to $\bbdvji$ then $\Om:=\apm_{i+1}\cup\Apm(\GO)$ contains two \bdbd paths separating $\bdf_i$ from $\bdf_j$. 

It is clear that if~\itemref{rho4:share_face} is violated (for any $\G$) then $\bdf_i$ and $\bdf_j$ cannot be \twosep. Assume now that $\helmin(\G)\geq2$ and that~\itemref{rho4:share_face} holds for $\G$. 
Suppose for contradiction that~\itemref{rho2:GO} is violated. 
Let $\GOp_{i+1}$ be obtained from $\GO_{i+1}$ by adding a source vertex $\source$ adjacent to all sources of $\GO_{i+1}$ in $\bbdvij$ and a sink vertex $\tink$ adjacent to all sinks of $\GO_{i+1}$ in $\bbdvji$. 
Thus, $\GOp_{i+1}$ is still planar. By Menger's theorem for directed graphs (vertex version), there exists $\v\in\Verts$ such that $\GOp_{i+1}\rem\{\v\}$ contains no directed path from $\source$ to $\tink$. 
After possibly replacing $\v$ with the vertex connected to $\v$ by an edge in $\apm_{i+1}$ (or with the sole black neighbor when $\v$ is a white boundary source), we may assume that $\v$ is black. 

Let $\Rg_\source\subset\Verts\setminus\{\v\}$ be the set of vertices of $\GOp_{i+1}\rem\{\v\}$ \emph{reachable} from $\bbdvij$ (i.e., admitting a directed walk in $\GOp_{i+1}\rem\{\v\}$ from some vertex in $\bbdvij$), and let $\Cut$ be a cut in $\GOp_{i+1}\rem\{\v\}$ connecting $\bdf_i$ to $\bdf_j$, separating the vertices in $\Rg_\source$ from the vertices in its complement $\Rg_\tink:=(\Verts\setminus\{\v\})\setminus\Rg_\source$. 

Let $\e=(\w\to\b)\notin\apm_{i+1}$ be an edge of $\GOp_{i+1}\rem\{\v\}$ crossed by $\Cut$. 
Suppose that there exists a directed path $\Pathdir_{\e\to\bbdvij}$ in $\GOp_{i+1}\rem\{\v\}$ starting with $\e$ and terminating at some vertex in $\bbdvij$. We claim that in this case,
\begin{equation}\label{eq:rho_claim}
	\text{$\apm_{i+1}$ contains an edge connecting $\v$ to $\w$}.
\end{equation}
Let $\Rg_{\w}:=\{\vv\in\Verts\setminus\{\v\}:\w\text{ is reachable from $\vv$ in $\GOp_{i+1}\rem\{\v\}$}\}$. 
Since $\w\in\Rg_\tink$, we have $\Rg_{\w}\subset\Rg_\tink$. 
Furthermore, we have $\Rg_{\w}\subset\Vint$; otherwise, concatenating a directed path $\Pathdir_{\bbdvji\to\e}$ in $\Rg_{\w}$ with the path $\Pathdir_{\e\to\bbdvij}$ would produce a directed path $\Pathdir_{\bbdvji\to\bbdvij}$ from $\bbdvji$ to $\bbdvij$. Swapping the roles of $\apm_{i+1}$ and $\Pathdir_{\bbdvji\to\bbdvij}\setminus\apm_{i+1}$ along this path, we see that $\Ibar_{i+1}$ is not lexicographically-minimal in $\MatroidG$ with respect to $\prec_{i+1}$ (cf. \cref{dfn:positroid_Grneck}), a contradiction.

Next, for $\e'=(\b'\to\w')\in\apm_{i+1}$ with $\v\neq\b',\w'$, we have $\b'\in\RB_{\w}$ if and only if $\w'\in\RW_{\w}$, and moreover, if $\b'\in\RB_{\w}$ then $\NeighG(\b')\setminus\{\v\}\subset\RW_{\w}$. Thus, $\Rg_{\w}\in\WNEIg(\G\rem\{\v\})$, and either $\Rg_{\w}$ or $\Rg_{\w}\sqcup\{\v\}$ admits a perfect matching. If $\RB_{\w}\neq\emptyset$, i.e., $\Rg_{\w}\in\WNEI(\G\rem\{\v\})$ then we get a contradiction with $\helmin(\G)\geq2$. Thus, $\RB_{\w}=\emptyset$, and therefore $\RW_{\w}=\{\w\}$ and $(\v\to\w)\in\apm_{i+1}$. 
This completes the proof of~\eqref{eq:rho_claim}. 

Suppose now that $\e=(\b\to\w)\in\apm_{i+1}$ and that $\e$ crosses $\Cut$. If $\w\in\bbdvij$ is a boundary sink then the above defined region $\Rg_\w$ belongs to $\WNEI(\G\rem\{\v\})$, a contradiction. Otherwise, we have $\b\in\Rg_\source$ if and only if $\w\in\Rg_\source$, so $\e$ does not cross $\Cut$, a contradiction. 

We have shown that any edge $\e$ crossing $\Cut$ must be of the form $\e=(\w\to\b)\notin\apm_{i+1}$ for some $\w\in\RW_\tink$ and $\b\in\RB_\source$. 
By \cref{lemma:DIM:deleting_hel_verts}, $\G\rem(\{\v\}\sqcup\ebar)$ admits an \APM $\apm'$. Thus, $\apm:=\apm'\sqcup\{\e\}$ is an \APM of $\G\rem\{\v\}$.

The union $\apm\cup\apm_{i+1}$ contains a path $\Pathdir_\v$ in $\GOp_{i+1}$ connecting $\v$ to some boundary vertex 
and a path (or a cycle) $\Pathdir_\e$ passing through $\e$. 
Suppose first that $\Pathdir_\v\neq\Pathdir_\e$. Then $\Pathdir_\e$ is a \bdbd path vertex-disjoint from $\Pathdir_\v$, directed from $\bbdvji$ to $\bbdvij$. 
(It cannot be a directed cycle because all edges of $\GOp_{i+1}\rem\{\v\}$ cross $\Cut$ with the same orientation.)
Thus, we obtain a contradiction with the minimality of $\Ibar_{i+1}$ as above. 

Suppose now that $\Pathdir_\v=\Pathdir_\e$. 
The part of $\Pathdir_\e$ from $\e$ to the boundary cannot cross the cut $\Cut$, so it must terminate at $\bbdvij$. By~\eqref{eq:rho_claim}, $\w$ must be the vertex matched to $\v$ in $\apm_{i+1}$.
For each $\w'\in(\NeighG(\v)\setminus\{\w\})\cap\Rg_\tink$, we similarly consider the set $\Rg_{\w'}$ of vertices $\vv$ such that $\w'$ is reachable from $\vv$ and arrive at a contradiction as above. Thus, $\NeighG(\v)\setminus\{\w\}\subset\Rg_\source$, and so $\bdf_i$ and $\bdf_j$ share the face $\w$ of $\GD$, contradicting~\itemref{rho4:share_face}. 
\end{proof}

\begin{corollary}\label{lemma:imm_pos_Mand_pos}
Assume that ${\helmin(\G)\geq2}$. 
 Let $(\La,\Lat)\in\LaLaimmnn$, $C:=\Meas(\G,\wt)$, and $\lalat = \PhiLL(C)$. Then for each $i+2\leq j \leq i+n-2$, either $\Mijll > 0$ and $\bdf_i,\bdf_j$ do not share a face of $\GD$, or $\Mijll = 0$ and $\bdf_i,\bdf_j$ share a face of $\GD$.
\end{corollary}
\begin{proof}
By \cref{lemma:PROP:Mpos}, it suffices to show that $\bdf_i,\bdf_j$ do not share a face of $\GD$ if and only if they are \twosep. This follows from the equivalence~\itemref{rho1:twosep}$\Longleftrightarrow$\itemref{rho4:share_face} in \cref{lemma:TREE:fullysep_vs_rhoij}.
\end{proof}

\begin{corollary}\label{lemma:fullysep=>twonondeg}
Suppose that $\fap\in\Boundkn$ \emph{\isfullysep}, i.e., satisfies $\rhoij(\fap)\geq2$ for all $i+2\leq j\leq i+n-2$. Then $\fap$ is \emph{\twonondeg} (i.e., $(2,2)^\partial$-nondegenerate; cf. \cref{dfn:nondeg}).
\end{corollary}
\begin{proof}
Pick a planar bipartite graph $\G$ with $\fG = \fap$. By \cref{lemma:TREE:fullysep_vs_rhoij}, for each $i\in\brn$, $\bdf_{i-1}$ and $\bdf_{i+1}$ are \twosep in $\G$. Let $\Om\in\OmsG$ be a double-dimer configuration containing \bdbd paths starting at $\bdv_i$ and $\bdv_{i+1}$. It follows that we can split $\Om=\apm_-\cup\apm_+$ into a union of two \APMs satisfying $i,i+1\in\partial\apm_+$ and $i,i+1\notin\partial\apm_-$. Thus, $\G$ (and therefore $\fap$) is \twonondeg.
\end{proof}

\begin{definition}\label{dfn:Grfsep}
For $2\leq k\leq n-2$, we denote $\BNDfsep(k,n):=\{\fap\in\Boundkn\mid \fap\text{ \isfullysep}\}$ and 
\begin{equation}\label{eq:Grfsep_dfn}
 \Grfsep(k,n):=\bigsqcup_{\fap\in\BNDfsep(k,n)} \Ptp_{\fap} = \{\Meas(\G,\wt)\mid \G\text{ \isfullysep and }\wt\in\Rtpgauge\}.
\end{equation}
\end{definition}

\begin{lemma}\label{lemma:lalat_MP=>C_fsep}
Let $(\la,\lat,C)\in\TRIPLES$ be such that $\lalat\in\MPkntree$. Then $C\in\Grfsep(k,n)$. 
\end{lemma}
\begin{proof}
By \cref{prop:momLL_basic,prop:from_la_to_La}, there exists $\LaLat\in\LaLak$ such that $\lalat=\PhiLL(C)$. Pick a weighted graph $(\G,\wt)$ such that $C=\Meas(\G,\wt)$. By \cref{lemma:TREE:fullysep_vs_rhoij}, $\G$ \isfullysep if and only if $C\in\Grfsep(k,n)$. If $\bdf_i,\bdf_j$ are not \twosep in $\G$ for some $i+2\leq j\leq i+n-2$ then by~\eqref{eq:immanant_fg}, $(\bdx_i - \bdx_j)^2=0$, contradicting $\lalat\in\MPkntree$.
\end{proof}

We calculate the image of the subset $\Grfsep(k,n)$ under T-duality.
\begin{lemma}\label{lemma:TREE:rank_of_iijj}
Let $\fap\in\Bounda$ and $\ddC\in\Ptp_{\ddfperm}$. For $i+2\leq j\leq i+n-2$, let ${\Jij:=\{i,i+1,j,j+1\}\subset\brn}$, where the indices $i,i+1,j,j+1$ are taken modulo $n$. Then
the submatrix $\ddC_{\Jijc}:=\mat[\ddC_\s]_{\s\in \Jijc}$ has rank
\begin{equation}\label{eq:TREE:rank_of_iijj}
 \rank\ddC_{\Jijc} = \min(k-2,k+\rhoij(\fap) - 4).
\end{equation}
\end{lemma}
\begin{proof}
One can compute the rank function $\ddr:=\rank_{\ddot\Mcal}$ of the positroid $\ddot\Mcal$ of $\ddfperm$ using the results of~\cite{Suho_rank}. Since $\Jijc=\Jijyy_{i+2,j}\sqcup\Jijyy_{j+2,i}$ is a union of two cyclic intervals, \cite[Theorem~25]{Suho_rank} yields $\ddr(\Jijc) = \min\left(k-2,\ddr(\Jijyy_{i+2,j})+\ddr(\Jijyy_{j+2,i})\right)$, and a straightforward computation shows that $\ddr(\Jijyy_{i+2,j})+\ddr(\Jijyy_{j+2,i})=k+\rhoij(\fap) - 4$.
\end{proof}

\begin{definition}\label{dfn:Grfind}
 In view of
\cref{lemma:TREE:fullysep_vs_rhoij,lemma:TREE:rank_of_iijj}, we say that $\ddC\in\Grtnn(k-2,n)$ is \emph{\twoind} if $\rank\ddC_{\Jijc} = k-2$ for all $i+2\leq j\leq i+n-2$. Similarly to \cref{dfn:Grfsep},
 we set
\begin{equation*}%
 \Grddfsep(k-2,n):=\bigsqcup_{\fap\in\BNDfsep(k,n)}\Ptp_{\ddfperm} = 
\{\ddC\in\Grtnn(k-2,n)\mid \ddC\text{ is \twoind}\}.
\end{equation*}
\end{definition}

\noindent Thus, by \cref{prop:magic_homeo}, for $\la\in\lak$, $C\in\Grtnnsupla(k,n)$, and $\ddC:=C\cdot \Qla\in\Grtnn(k-2,\lap)$, $C$ is \fullysep if and only if $\ddC$ is \twoind. 

\subsection{\MbdTITLE and t-embeddings}\label{ssec:Mbd_and_temb}
Motivated by \cref{dfn:intro:1_Lipschitz,lemma:imm_pos_Mand_pos}, we introduce the following notion. 

\begin{definition}
We say that a map $\xd:\Faces\to\Rdd$ 
 has \emph{\Mbd} if for all $i+2\leq j\leq i+n-2$, we have $(\bdx_i - \bdx_j)^2>0$. 
We say that $\xd$ has \emph{\weakMbd} if $(\bdx_i - \bdx_j)^2>0$ is satisfied for all $i+2\leq j\leq i+n-2$ such that $\bdf_i$ and $\bdf_j$ do not share a face of $\GD$.
\end{definition}

We start with a general topological lemma.
\begin{lemma}\label{lemma:non_self_intersecting_Sigma}
Let $\Sig$ be a simple undirected graph and $\xd=(\xT,\xO):\Sig\to\Rdd$ be such that the image of any edge of $\Sig$ under $\xT$ is a straight line segment of positive length. Suppose that for any two vertices $\f_1,\f_2$ of $\Sig$, we have $(\xd(\f_1)-\xd(\f_2))^2\geq0$, with equality if and only if $\{\f_1,\f_2\}$ is an edge of $\Sig$. Suppose further that the image under $\xT$ of every triangle of $\Sig$ is a non-degenerate triangle in the plane, and that $\Sig$ does not contain a subgraph isomorphic to $K_4$ (the complete graph on $4$ vertices). Then $\xT$ is a straight-line embedding of $\Sig$.%
\end{lemma}
\begin{proof}
We classify the possible self-intersections of $\xT(\Sig)$ into three types:
\begin{enumerate}[label=(\roman*),wide,labelindent=5pt]
\item\label{bdrysi1} two noncollinear line segments $\Tcal(\east_1)$ and $\Tcal(\east_2)$ intersecting in their relative interiors;
\item\label{bdrysi2} a vertex $\Tcal(\f)$ belonging to the relative interior of an edge $\Tcal(\east)$;
\item\label{bdrysi3} two vertices $\Tcal(\f_1) = \Tcal(\f_2)$ coinciding, for $\f_1\neq \f_2$. 
\end{enumerate}
We first show that no self-intersections of type~\itemref{bdrysi1} are possible. Suppose otherwise that such a self-intersection occurs, and denote $\ebarast_1 = \{\f_1,\f_3\}$ and $\ebarast_2 = \{\f_2,\f_4\}$. 
Consider a convex quadrilateral with boundary vertices $(\Tcal(\f_1),\Tcal(\f_2),\Tcal(\f_3),\Tcal(\f_4))$ listed in cyclic order. The map $\Ocal$ preserves the lengths of the diagonals but weakly decreases the lengths of the four sides. Moreover, since $\Sig$ does not contain $K_4$, $\Ocal$ strictly decreases the length of at least one of the sides of the quadrilateral. Without loss of generality, we may assume that $\Ocal(\f_1)=\Tcal(\f_1)$ and $\Ocal(\f_3) = \Tcal(\f_3)$. Let $p$ be the intersection point of $\Tcal(\east_1)$ and $\Tcal(\east_2)$. Then $|\Ocal(\f_2)-p|\leq|\Tcal(\f_2)-p|$ and $|\O(\f_4)-p|\leq|\Tcal(\f_4)-p|$, and at least one of these inequalities must be strict, so $|\O(\f_2)-p| + |\O(\f_4)-p| < |\Tcal(\f_2)-p| + |\Tcal(\f_4)-p|$. On the other hand, $|\Tcal(\f_2)-p| + |\Tcal(\f_4)-p| = |\Tcal(\f_2)-\Tcal(\f_4)| = |\O(\f_2)-\O(\f_4)|\leq |\O(\f_2)-p| + |\O(\f_4)-p|$, a contradiction.

Similarly, in a self-intersection of type~\itemref{bdrysi2}, the point $\Tcal(\f)$ belongs to the line segment $\Tcal(\east)$, and $\Ocal$ weakly decreases the distance from $\Tcal(\f)$ to each of the endpoints of the line segment while preserving the length of the line segment itself, a contradiction. (Note that $\f$ cannot be adjacent in $\Sig$ to both endpoints of $\east$ by the assumption that $\Tcal$ is injective on the triangles of $\Sig$.) Finally, if a boundary self-intersection of type~\itemref{bdrysi3} occurs then we have $|\Tcal(\f_1)-\Tcal(\f_2)|=0$, so we must have $|\O(\f_1)-\O(\f_2)|=0$ and therefore $\{\f_1,\f_2\}$ must be an edge of $\Sig$. But then the image of this edge under $\Tcal$ was assumed to be a line segment of nonzero length, a contradiction.
\end{proof}

\begin{remark}\label{rmk:K4}
It follows from the above proof that when $\Sig$ is not necessarily $K_4$-free but satisfies all other conditions in the lemma then
 $\xT(\Sig)$ does not contain self-intersections of types~\itemref{bdrysi2} and~\itemref{bdrysi3}, and for any self-intersection of type~\itemref{bdrysi1} between edges $\Tcal(\east_1)$ and $\Tcal(\east_2)$, the vertices in $\ebarast_1\sqcup\ebarast_2$ form a clique in $\Sig$.
\end{remark}

Applying \cref{lemma:non_self_intersecting_Sigma} to the case where $\Sig$ is either an $n$-cycle or an $n$-cycle with one extra isolated vertex, 
 we obtain the following. 
\begin{corollary}\label{lemma:Mbd=>simple}
Suppose that $\Pbdx=(\bdx_1,\bdx_2,\dots,\bdx_n)$ is an \emph{\Mdash positive} null polygon in $\Rdd$, meaning $(\bdx_i-\bdx_{i-1})^2=0$, $\bdx_i\neq\bdx_{i-1}$, and $(\bdx_i-\bdx_j)^2>0$ for all $i+2\leq j\leq i+n-2$.
 Then the polygon $\PbdxT$ is simple. 
Furthermore, if $\ys\in\Rdd$ satisfies $(\ys-\bdx_i)^2>0$ for all $i\in\brn$ then the point $\ysT$ 
is located either strictly inside or strictly outside $\PbdxT$.
\end{corollary}

\begin{corollary}\label{lemma:fullysep=>simple_and_Mbd}
Let $C\in\Grfsep(k,n)$. Let $\LaLat\in\LaLaimmnn$ and $\lalat:=\PhiLL(C)$. Then $\lalat\in\MPkntree$ is \Mdash positive and the polygon $\PllT$ is simple.
\end{corollary}
\begin{proof}
By \cref{lemma:fullysep=>twonondeg}, $C$ is \twonondeg, so $\lalat\in\lalak$ by \cref{prop:momLL_basic}. 
Pick a weighted graph $(\G,\wt)$ such that $C=\Meas(\G,\wt)$. 
By \cref{lemma:TREE:fullysep_vs_rhoij} (\itemref{rho3:rhoij}$\Longrightarrow$\itemref{rho1:twosep}), $\G$ is \fullysep. 
By \cref{lemma:PROP:Mpos}, $\lalat$ is \Mdash positive, so $\lalat\in\MPkntree$. 
By \cref{lemma:Mbd=>simple}, $\PllT$ is simple.
\end{proof}

\begin{lemma}\label{lemma:helmin>2=>no_triangles}
Suppose that $\helmin(\G)\geq2$ and that $1\leq a<b<c\leq n$ are such that $\bdf_i$ and $\bdf_j$ share a face of $\GD$ for each $i,j\in\{a,b,c\}$. Then the three vertices $\bdf_a,\bdf_b,\bdf_c$ share a common face of $\GD$. 
\end{lemma}
\begin{proof}
If the faces $\bdf_a,\bdf_b,\bdf_c$ are not pairwise distinct then the result follows trivially, so assume that they are pairwise distinct. 
Suppose that $\G$ contains pairwise distinct interior vertices $\v_{ab},\v_{ac},\v_{bc}$ such that for $\{i,j\}\subset\{a,b,c\}$, $\bdf_i$ and $\bdf_j$ are incident to $\v_{ij}$. Let $\Rg$ be the union of $\{\v_{ab},\v_{ac},\v_{bc}\}$ with the set of faces of $\GD$ contained inside the triangle with vertices $\bdf_a,\bdf_b,\bdf_c$ and sides given by curves inside $\v_{ij}$ connecting $\bdf_i$ to $\bdf_j$ for each pair 
$\{i,j\}\subset\{a,b,c\}$. 
 Thus, $\G\ind[\Rg]$ is a planar bipartite graph of type $(k',n')$ for some $0\leq k'\leq n'=3$ with boundary vertices $\v_{ab},\v_{bc},\v_{ac}$. By~\eqref{eq:DIM:k_dfn}, $|\RW|-|\RBint| = k'$ and $|\RB|-|\RWint| = n'-k'$, where $\RWint:=\RW\setminus\{\v_{ab},\v_{bc},\v_{ac}\}$ and $\RBint:=\RB\setminus\{\v_{ab},\v_{bc},\v_{ac}\}$. 
Observe that $\min(k',n'-k')\leq1$. Suppose that, say, $k'\leq 1$. Let $\Rg_2:=\RW\sqcup\RBint\in\WNEIg(\G)$ so that $\helW(\Rg_2) = k'$. If $\RBint=\emptyset$ then $|\RW|=k'\leq 1$. Since 
$\bdf_a,\bdf_b,\bdf_c$ are pairwise distinct, we cannot have $|\RW|=0$, so letting $\w\in\RW$ be the sole white vertex of $\GR$, we see that $\bdf_a,\bdf_b,\bdf_c$ share the face $\w$ of $\GD$. If $\RBint\neq\emptyset$ then $\Rg_2\in\WNEI(\G)$ so $\helWmin(\G)\leq1$, a contradiction. The case $n'-k'\leq1$ is handled similarly by setting $\Rg_2:=\RB\sqcup\RWint\in\BNEIg(\G)$. 
\end{proof}

\begin{lemma}\label{lemma:LaLaimmn+helmin=simple}
 Assume that $\G$ is connected and $\helmin(\G)\geq2$.
 Let $\xd:\Faces\to\Rdd$ be a t-immersion of $\G$ with \weakMbd. Then the boundary polygon $\PbdxT$ is simple.
\end{lemma}
\begin{proof}
Since $\G$ is connected, its boundary faces are pairwise distinct. 
Consider a simple graph $\Sig$ on vertex set $\{\bdf_1,\bdf_2,\dots,\bdf_n\}$, with $\bdf_i,\bdf_j$ connected by an edge of $\Sig$ if and only if they share a face of $\GD$. By \cref{lemma:helmin>2=>no_triangles}, every triangle in $\Sig$ is contained in a face $\v\in\Vint$ of $\GD$. Since $\xd$ is a t-immersion of $\G$, $\xT(\v)$ is a nondegenerate convex polygon, so the image of every such triangle under $\xT$ is nondegenerate. 
Thus, $\Sig$ and $\xd$ satisfy all assumptions of \cref{lemma:non_self_intersecting_Sigma} except that $\Sig$ need not be $K_4$-free. By \cref{rmk:K4}, the only possible self-intersections in $\xT(\Sig)$ may occur between edges $\{\bdf_a,\bdf_c\}$ and $\{\bdf_b,\bdf_d\}$ for some $a,b,c,d\in\brn$ such that $\bdf_a,\bdf_b,\bdf_c,\bdf_d$ form a clique in $\Sig$. Since the graph $K_4$ is not outerplanar, by \cref{lemma:helmin>2=>no_triangles}, it follows that the vertices $\bdf_a,\bdf_b,\bdf_c,\bdf_d$ all share a common face $\v\in\Vint$ of $\GD$. Since $\xT(\v)$ is a nondegenerate convex polygon, no self-intersections of type~\itemref{bdrysi1} are possible between its boundary edges. 
Thus, $\xT$ is indeed injective on the boundary polygon of $\GD$.
\end{proof}

\begin{lemma}\label{lemma:Jordan_curve}
 Assume that $\G$ is connected and $\helmin(\G)\geq2$.
A t-immersion $\xd:\Faces\to\Rdd$ of $\G$ is a t-embedding if and only if the boundary polygon $\PbdxT$ is simple.
\end{lemma}
\begin{proof}
By definition, if $\xT$ is a t-embedding then $\PbdxT$ must be simple. Conversely, by 
 the angle conditions~\eqref{eq:intro:t_imm_angles_int}--\eqref{eq:intro:t_imm_angles_bdry} and the immersion condition~\itemref{intro:t_imm_orientation}, $\xT:\SuppGD\to\C$ is a local homeomorphism (i.e., each point in $\SuppGD$ has an open neighborhood $U$ such that $\xT|_U$ is a homeomorphism onto its image). By the argument principle, the cardinality of the preimage of a point $z\in\C$ under $\xT$ equals the winding number of the boundary polygon $\PbdxT$ around $z$. It follows that every point inside $\PbdxT$ has a unique preimage, and thus $\xT$ is injective, i.e., a t-embedding. See the (last paragraph of the) proof of~\cite[Theorem~4.1]{CLR2} for a similar argument.
\end{proof}

The next result follows from \cref{lemma:imm_pos_Mand_pos} and \crefrange{lemma:LaLaimmn+helmin=simple}{lemma:Jordan_curve}.
\begin{corollary}\label{lemma:LaLaimmn+helmin=t_emb}
Assume that $\G$ is connected, \twonondeg, and satisfies $\helmin(\G)\geq2$. 
Let $C:=\Meas(\G,\wt)$, $\LaLat\in\LaLaimmnn$, and $\lalat=\PhiLL(C)$. 
 Then the t-immersion $\xll$ is a t-embedding.
\end{corollary}

\begin{proof}[Proof of \cref{cor:intro:t_emb_exists}]
Set $C:=\Meas(\G,\wt)\in\Grnd(k,n)$. Fix $\LaLat\in\LaLaimmnn$ and let $\lalat:=\PhiLL(C)$. (Note that $\LaLaimmnn\neq\emptyset$ by \cref{lemma:LaLaimmp_nonempty_Z_dense}.) 
Applying \cref{thm:intro:t_imm_vs_triples}, we obtain a t-immersion $\xll:\SuppGD\to\Rdd$ of $(\G,\wt)$. By \cref{lemma:LaLaimmn+helmin=t_emb}, 
 $\xll$ is a t-embedding.
\end{proof}

\begin{lemma}\label{lemma:weakMbd=>boundary>0_or_=0}
Assume that $\helmin(\G)\geq2$. 
Suppose that $\xd:\SuppGD\to\Rdd$ is a t-embedding of $\G$ with \weakMbd. Then for any two points $p,q$ located on the closed polygonal chain $\Pbdx$ in $\Rdd$, we have $(p-q)^2>0$ if $\Tcomp{p},\Tcomp{q}$ do not share a face of $\xT(\GD)$ 
 and $(p-q)^2=0$ otherwise.
\end{lemma}
\begin{proof}
Let $i,j\in\brn$ be distinct and write $p(t):=(1-t)\bdx_{i-1}+t\bdx_i$ and $q(s):=(1-s)\bdx_{j-1}+s\bdx_{j}$. Since $(\bdx_{j-1}-\bdx_{j})^2=0$, we see that the function $h_t(s):=(p(t)-q(s))^2$ is affine linear in $s$ for each fixed $t\in[0,1]$. Since $\xd$ has \weakMbd, for all $(t',s')\in\{0,1\}^2$, we have $h_{t'}(s')>0$ if $p(t')$ and $q(s')$ do not share a face of $\xd(\GD)$ and $h_{t'}(s')=0$ otherwise. 
 For $t'\in\{0,1\}$, since $h_{t'}(s)$ takes nonnegative values on the endpoints $s'\in\{0,1\}$, we have $h_{t'}(s)\geq0$ for all $s\in[0,1]$. 

Next, we prove the lemma for $t'\in\{0,1\}$ and $0<s<1$. 
If $p(t')$ and $q(s)$ share a face of $\xd(\GD)$ for $0<s<1$ then this face must include the entire boundary edge $[q(0),q(1)]$, so we have $h_{t'}(s)=0$. Conversely, suppose $h_{t'}(s)=0$ for some $0<s<1$. If $p(t')$ does not share a face with $q(s')$ for either $s'=0$ or $s'=1$ then we have $h_{t'}(s')>0$ and therefore $h_{t'}(s)>0$ for all $0<s<1$, a contradiction. Thus, assume that $p(t')$ shares a face with $q(0)$ and with $q(1)$. By \cref{lemma:helmin>2=>no_triangles}, the points $p(t'),q(0),q(1)$ all share a common face, so $p(t')$ and $q(s)$ share a face of $\xd(\GD)$ for all $0<s<1$.

We have shown the result of the lemma when $p=p(t')$ with $t'\in\{0,1\}$ is a boundary vertex of $\xd(\GD)$ and $q=q(s)$ is an arbitrary point on the 
closed polygonal chain $\Pbdx$. 
 By symmetry, the result also holds when $q$ is a boundary vertex and $p$ is an arbitrary point on the boundary. Consider now the case where $p=p(t)$ and $q=q(s)$ with $0<t,s<1$. If $p$ and $q$ share a face of $\xd(\GD)$ then $(p-q)^2=0$ by~\eqref{eq:inv_null}. Otherwise, there must exist $(t',s')\in\{0,1\}^2$ such that $p(t')$ and $q(s')$ do not share a face of $\xd(\GD)$. Thus, $h_{t'}(s')>0$ which implies that $h_{t'}(s)>0$ for $0<s<1$. Since for each fixed $s$, $h_{t}(s)$ is an affine linear function of $t$ which takes a strictly positive value at one endpoint $t=t'$ and a nonnegative value at another endpoint, we see that $h_t(s)>0$ for all $0<t,s<1$. 
\end{proof}

\begin{proposition}\label{lemma:Mbd=>face_Mnn}
Assume that $\helmin(\G)\geq2$. 
Suppose that $\xd:\Faces\to\Rdd$ is a t-embedding of $\G$ with \weakMbd. Then for all $\ff,\f\in\Faces$, we have $(\xd(\ff)-\xd(\f))^2>0$ if $\ff,\f$ do not share a face of $\GD$ and $(\xd(\ff)-\xd(\f))^2=0$ otherwise.
\end{proposition}
\begin{proof}
If $\ff,\f$ share a face of $\GD$ then $(\xd(\ff)-\xd(\f))^2=0$ by~\eqref{eq:inv_null}. Suppose that they do not share a face of $\GD$. 
If the line segment $[\T(\ff),\T(\f)]$ is fully contained in $\T(\SuppGD)$ then we have $(\xd(\ff)-\xd(\f))^2\geq 0$ by~\eqref{eq:intro:1_Lipschitz}. Since $\ff,\f$ do not share a face of $\GD$, this line segment either crosses an edge of $\xT(\GD)$ or it passes through an interior vertex of $\xT(\GD)$. In either case, the line segment gets folded at least once under the origami map so we get $(\xd(\ff)-\xd(\f))^2 > 0$. 

Suppose now that $[\T(\ff),\T(\f)]$ is not fully contained in $\T(\SuppGD)$. Consider a sequence 
$(\T(p_0)=\T(\ff),\T(p_1),\dots,\T(p_d)=\T(\f))$ of points on $[\T(\ff),\T(\f)]$, with $p_1,\dots,p_{d-1}$ located on the boundary of $\GD$,
such that the interior of each line segment $[\T(p_{i-1}),\T(p_i)]$ is either fully inside or fully outside $\xT(\SuppGD)$. For two points $p,q\in\SuppGD$, let $\Sone(p,q):=|\T(p)-\T(q)| - |\O(p)-\O(q)|$. Thus, the sign of $\Sone(p,q)$ coincides with the sign of $(\xd(p)-\xd(q))^2=|\T(p)-\T(q)|^2 - |\O(p)-\O(q)|^2$. Furthermore, by the triangle inequality, $\Sone(\ff,\f)\geq\sum_{i=1}^d \Sone(p_{i-1},p_i)$. 
By~\eqref{eq:intro:1_Lipschitz} and \cref{lemma:weakMbd=>boundary>0_or_=0}, $\Sone(p_{i-1},p_i)\geq0$ for each $i\in\brd$. 
 Furthermore, there exists $i\in\brd$ such that both $p_{i-1}$ and $p_i$ lie on the boundary of $\GD$ and such that the interior of the line segment $[\T(p_{i-1}),\T(p_i)]$ is fully outside $\xT(\SuppGD)$. 
The points $p_{i-1}$ and $p_i$ cannot share a face of $\GD$ because the faces of $\T(\GD)$ are convex since $\xT$ is a t-immersion. Thus, $\Sone(p_{i-1},p_i)>0$, which implies that $\Sone(\ff,\f)>0$ and thus $(\xd(\ff)-\xd(\f))^2>0$.
\end{proof}

\section{BCFW tilings of ambient momentum amplituhedra}\label{sec:BCFW}
Our next goal is to prove \cref{thm:intro:BCFW} (and thus \cref{thm:intro:A}). 
In this section, we prove a BCFW tiling result (\cref{thm:f_triang}) for the \emph{ambient} momentum amplituhedron $\lalapp$ defined in~\eqref{eq:intro:M_amb_dfn}.

\subsection{Background on the BCFW recursion}\label{sec:BCFW:backgr}
Our exposition follows~\cite[Section~17.2]{abcgpt}. 
 The BCFW recursion produces a collection $\BCFWGkn$ of reduced graphs of type $(k,n)$ for each $2\leq k\leq n-2$. For the base case, we let $\BCFWG_{1,3}$ (resp., $\BCFWG_{2,3}$) contain a single graph consisting of an interior degree-$3$ white (resp., black) vertex adjacent to three boundary vertices of opposite color. For the rest of this section, we assume that $2\leq k \leq n-2$.
\begin{definition}[BCFW recursion]\label{dfn:BCFW}
The collection $\BCFWG_{k,n}$ is obtained as follows.
\begin{enumerate}[label=(\arabic*)]
\item\label{step_BCFW1} Fix an index $\io\in\brn$
 and a choice of either a black-white or a white-black bridge at $(\bdv_{\io},\bdv_{\io+1})$.\footnote{Here, a \emph{bridge} is an edge connecting the next-to-boundary vertices $\bdvx_{\io}$ and $\bdvx_{\io+1}$. For a \emph{black-white bridge}, $\bdvx_{\io}$ is black and $\bdvx_{\io+1}$ is white, and for a \emph{white-black bridge}, $\bdvx_{\io}$ is white and $\bdvx_{\io+1}$ is black. 
}
\item\label{step_BCFW2} $\BCFWGkn$ consists of amalgamations
 $\G=\G_L\otimes\G_R$ (cf. \figref{fig:BCFW}(a)) of pairs of graphs $(\G_L,\G_R)\in\BCFWG_{k_L,n_L}\times\BCFWG_{k_R,n_R}$ for each $(k_L,n_L,k_R,n_R)$ satisfying conditions~\itemref{BCFW1}--\itemref{BCFW3} below, where the collections $\BCFWG_{k_L,n_L},\BCFWG_{k_R,n_R}$ are assumed to have already been constructed.%
\begin{enumerate}[label=(BCFW\arabic*)]
\item\label{BCFW1} $n_L,n_R\geq3$, $n_L+n_R = n+2$, $k_L,k_R\geq1$, and $k_L+k_R = k+1$.
\item\label{BCFW2} If $n_L\geq4$ then $2\leq k_L\leq n_L-2$. If $n_R\geq4$ then $2\leq k_R\leq n_R-2$.
\item\label{BCFW3} If $n_L = 3$ (resp., $n_R=3$) then $k_L = 2$ (resp., $k_R=1$) in the case of a black-white bridge and $k_L=1$ (resp., $k_R=2$) in the case of a white-black bridge. 
\end{enumerate}
\end{enumerate}
Let $\BCFWfkn:=\{\fG\mid \G\in\BCFWGkn\}$ be the collection of associated bounded affine permutations. 
\end{definition}
\noindent For $\G\in\BCFWGkn$, we denote by $\fo$ the newly created interior face of $\G$ adjacent to $\bdf_{\io}$ and $\bdf_{\jo}$; see \cref{fig:BCFW}. 

\begin{remark}\label{rmk:jo}
The choice of $(n_L,n_R)$ is equivalent to a choice of $\jo\in\Z$ satisfying $\io+2\leq \jo\leq \io+n-2$, with $n_L=\jo-\io+1$ and $n_R=\io-\jo+n+1$. 
\end{remark}

\begin{remark}\label{rmk:BCFW_multiple}
The above definition involves a choice of $\io$ and a black-white/white-black bridge at each step of the recursion. We do not assume that these choices are consistent in any way. Thus, \cref{dfn:BCFW} gives rise to many possible collections of reduced graphs for each $k,n$. One such choice is shown for $4\leq n\leq 6$ in~\cite[Figure~7]{KWZ}.
\end{remark}

\begin{remark}\label{rmk:BCFW_reduced}
It is well known (see e.g.~\cite[Section~3.2]{abcgpt} or~\cite[Section~6]{KWZ}) that the graphs in $\BCFWGkn$ are reduced of type $(k,n)$. In view of \cref{lemma:reduced_helmin} and \cref{rmk:reduced_helmin2}, after possibly applying some contraction moves \MV1 and \MVbd, we assume from now on that $\helmin(\G)\geq2$ for each $\G\in\BCFWGkn$. 
\end{remark}

The following result can be shown by induction.
\begin{lemma}\label{lemma:BCFW_fullysep}
Let $\G\in\BCFWGkn$. Then $\bdf_i,\bdf_j$ do not share a face of $\GD$ for all $i+2\leq j\leq i+n-2$. 
In other words, each $\G\in\BCFWGkn$ is \fullysep 
(cf. \cref{dfn:fullysep,lemma:TREE:fullysep_vs_rhoij}).
\end{lemma}
\noindent Applying this lemma to the graphs $\G_L,\G_R$ in \cref{dfn:BCFW}, we get the following result.
\begin{corollary}\label{cor:BCFW:share_face}
For $\G\in\BCFWGkn$, $\fo$ and $\bdf_\s$ share a face of $\GD$ if and only if $\s\in\{\io-1,\io,\io+1,\jo\}$.
\end{corollary}

\subsection{BCFW tilings}\label{sec:BCFW:triangulation_MP}
Recall from~\eqref{eq:intro:M_amb_dfn} that the ambient momentum amplituhedron $\lalapp$ is the space of \Mdash positive pairs $\lalat\in\lalak$. For $\fap\in\Boundkn$, we introduce a \emph{tile}
\begin{equation}\label{eq:lalappf_dfn}
 \lalappf:=\{\lalat\in\lalapp\mid \la\subset C\subset \latp\text{ for some $C\in\Ptp_f$}\}. 
\end{equation}

To work with ambient amplituhedra, we consider the following version of the notion of a tiling (\cref{dfn:intro:triang}) where the \emph{momentum amplituhedron map} $\PhiLL$ is replaced with the \emph{momentum amplituhedron correspondence} $\Rel$ (which one may view as a multivalued map) introduced in~\eqref{eq:TREE:Rel_dfn}. 

\begin{definition}[\Mtiling]\label{dfn:BCFW:tiling_amb}
Let $\RX,\RY$ be topological spaces equipped with a relation $\Rel\subset \RX\times \RY$. 
Let $\RGbf$ be a finite set and let $\{\RXo_\RG\mid\RG\in\RGbf\}$ be a collection of subsets of $\RX$. 
For $\RG\in\RGbf$, define a \emph{tile}
\begin{equation}\label{eq:BCFW:tile_amb_dfn}
 \RYo_\RG:=\{y\in \RY\mid (x,y)\in\Rel\text{ for some $x\in \RXo_\RG$}\}.
\end{equation}
In other words, the tile $\RYo_\RG=\Rproj_\RY(\Relo_\RG)$ is the image of $\Relo_\RG:=\Rel\cap(\RXo_\RG\times \RY)$ under the projection map $\Rproj_\RY:\RX\times \RY\to \RY$. 
We say that the tiles $\{\RYo_\RG\mid \RG\in\RGbf\}$ form an \emph{\mtiling} of $\RY$ if the following conditions are satisfied.
\begin{enumerate}[label=(\alph*)]
\item\label{Rtiling1} \emph{Injectivity:} For each $\RG\in\RGbf$, $\Rproj_\RY$ restricts to a homeomorphism $\Relo_\RG\xrasim\RYo_\RG$. 
\item\label{Rtiling2} \emph{Disjointness:} The tiles $\{\RYo_\RG\mid \RG\in\RGbf\}$ are pairwise disjoint.
\item\label{Rtiling3} \emph{Surjectivity:} The union $\bigsqcup_{\RG\in\RGbf}\RYo_\RG$ is dense in $\RY$.
\end{enumerate}
\end{definition}

We will apply this definition to 
$\RX:=\Grfsep(k,n)$ (\cref{dfn:Grfsep}), $\RY:=\MPkn$, and
\begin{align}
\label{eq:TREE:Rel_dfn}
 \Rel&:=\{(C,\la,\lat)\in\Grfsep(k,n)\times\MPkn\mid \la\subset C\subset\latp\}.
\end{align}
For each $\fap\in\BCFWfkn$, we set $\RXo_\fap :=\Ptp_\fap$. Thus,~\eqref{eq:BCFW:tile_amb_dfn} gives rise to tiles $\RYo_\fap=\MPf\subset\MPkn$; cf.~\eqref{eq:lalappf_dfn}. 
The goal of the next several subsections is to prove the following result.
\begin{theorem}[BCFW tilings of ambient momentum amplituhedra]\label{thm:f_triang}
 The tiles $\{\lalappf\mid f\in\BCFWfkn\}$ form an \Mtiling of $\lalapp$.
\end{theorem}
\noindent \Cref{thm:intro:BCFW} will be deduced from \cref{thm:f_triang}; see \cref{lemma:proj_tiling}.
\begin{remark}\label{rmk:BCFW:T_emb_unique}
 Suppose that $f = \fG$ with $\G\in\BCFWGkn$. 
By \cref{lemma:BCFW_fullysep}, $\fap$ \isfullysep. 
The set 
\begin{equation}\label{eq:lalaC_ppf}
 \Relo_f:=\{(C,\la,\lat)\in\Ptp_f\times\lalapp\mid \la\subset C\subset\latp\}
\end{equation}
 is thus identified via \cref{thm:intro:t_imm_vs_triples} with the set of t-immersions $\Tll$ of $\G$ (with unspecified edge weights) with \Mbd. 
By \cref{lemma:Mbd=>simple,lemma:Jordan_curve}, 
 each such t-immersion is a t-embedding. 
When $\G$ and $\lalat$ are fixed, $C$ and $\Tll$ determine each other uniquely. Thus, the content of \crefi{dfn:BCFW:tiling_amb}{Rtiling1} is that a t-embedding of $\G\in\BCFWGkn$ is uniquely recovered from its \Mdash positive boundary data $\lalat\in\MPkntree$.%
\end{remark}
\begin{remark}\label{rmk:dim_MomLL2}
Recall from \cref{rmk:dim_MomLL} that the dimension of the momentum amplituhedron $\MomLL$ is $2n-4$. 
It is well known that $\dim\Ptp_f = 2n-4$ for each $f\in\BCFWfkn$. On the other hand, $\dim\lalapp = \dim\lalats = 4n-12$. 
 For a given $C\in\Ptp_f$, the dimension of the set $\{\lalat\in\lalats\mid\la\subset C\subset\latp\}$ is $2 (k-2)+2 (n-k-2) = 2n-8$. 
Combining this with $\dim\Ptp_f=2n-4$, we get $\dim\Relo_f = 4n-12$. 
This is consistent with the claim of \cref{thm:f_triang} that the tiles $\lalappf\cong\Relo_f$ for $\fap\in\BCFWfkn$ are full-dimensional open subsets of $\lalapp$.
\end{remark}

\begin{remark}[Ambient momentum amplituhedron canonical form]\label{rmk:omega_form}
 Each positroid cell $\Ptp_\fap$, $\fap\in\BCFWfkn$, admits a $(2n-4)$-differential form $\omf$ called the \emph{canonical top form}, defined up to sign. 
The sum $\Omega_{\MomLL}=\sum_{\fap\in\BCFWfkn}\Phi_{\La,\Lat\ast} \omf$ of pushforwards of these forms along $\PhiLL$ (for a particular choice of signs) conjecturally yields the \emph{canonical top form}~\cite{ABL} of $\MomLL$. 
One can similarly define a $(2n-4)$-form $\Omega_{\MPkntree} := \sum_{\fap\in\BCFWfkn} \Rproj_{\RY\ast} \Rproj_{\RX}^\ast \omf$ on $\MPkntree$ by first pulling back the form $\omf$ on $\RXo_\fap=\Ptp_\fap$ along the projection $\Rproj_{\RX}:\Relo_{\fap}\to\RXo_{\fap}$ and then pushing it forward along the diffeomorphism (by \cref{thm:f_triang}) $\Rproj_{\RY}:\Relo_{\fap}\xrasim\RYo_{\fap}$. The form $\Omega_{\MomLL}$ is recovered as the pullback of $\Omega_{\MPkntree}$ under the natural inclusion $\MomLL\hookrightarrow \MPkntree$. Conjecturally, $\Omega_{\MPkntree}$ does not depend on the choice $\BCFWfkn$ of \amtiling. However, since $\Omega_{\MPkntree}$ is not a \emph{top form} for $n>4$ (because $\dim\MPkntree=4n-12$), it is not a canonical form in the sense of~\cite{ABL}. It would be interesting to extend the theory of positive geometries and their canonical forms (defined recursively via taking residues) to this natural setting and to express $\Omega_{\MPkntree}$ directly in terms of t-embeddings. 
\end{remark}

\subsection{Particle momenta and Mandelstam variables}\label{sec:particle_momenta}
We recall some background on the spinor-helicity formalism. We refer to e.g.~\cite[Section~2]{ElHu} for further details.

Consider the Minkowski space $\R^{2,2}\cong\C^2$. For $\Pmom\in\Rdd$, we denote its coordinates by $(\Tcomp\Pmom,\Ocomp\Pmom)\in\C^2$. We equip $\Rdd$ with Minkowski norm $\|\Pmom\|^2 = |\PmomT|^2 - |\PmomO|^2$.
 Define the associated symmetric bilinear form of signature $(+,+,-,-)$ by 
\begin{equation}\label{eq:Pmom_cdot_Qmom}
 \Pmom\cdot \Qmom:=\Re(\PmomT\ovl{\QmomT} - \PmomO\ovl{\QmomO}) \quad\text{for $\Pmom,\Qmom\in\Rdd$,}
\end{equation}
so that $\Pmom^2 := \Pmom\cdot \Pmom = \|\Pmom\|^2$ for all $\Pmom\in\Rdd$.

We say that $\Pmom\in\Rdd$ is a \emph{nonzero null vector} if $\Pmom^2 = 0$ and $\Pmom\neq0$. Given such $\Pmom\in\Rdd$, a \emph{decoration} of $\Pmom$ is a pair $(\y,\yt)$ of complex numbers such that 
\begin{equation}\label{eq:Pmom_vs_yyt}
 (\PmomT,\PmomO) = (\y\yt,\ovl{\y}\yt). 
\end{equation}
 The decoration $(\y,\yt)$ is defined up to the \emph{little group action} $(\y,\yt)\mapsto(t\cdot \y,t^{-1}\yt)$ for $t\in\Rast$. When a decoration $(\y,\yt)$ of $\Pmom$ is fixed, we 
denote $(\y,\yt)$ by $(\Pmomy,\Pmomyt)$ and set 
\begin{equation}\label{eq:Pmom_Qmom_det}
 \<\Pmom\,\Qmom\> := \det\mat[\Pmomy|\Qmomy] \quad\text{and}\quad 
 [\Pmom\,\Qmom] := -\det\mat[\Pmomyt|\Qmomyt].
\end{equation}
Here, recall that $\det\mat[\y|\y']:=\Re(\y)\Im(\y') - \Im(\y)\Re(\y')$. 
It follows that the scalar product~\eqref{eq:Pmom_cdot_Qmom} may be expressed as
\begin{equation}\label{eq:bispinor_cdot}
 \Pmom\cdot\Qmom = 2\<\Pmom\,\Qmom\>[\Pmom\,\Qmom].
\end{equation}

Recall from \cref{rmk:intro:Minkowski} that a \emph{null polygon} in $\Rdd$ is a collection $\Pcurve=(\bdx_1,\bdx_2,\dots,\bdx_n=\bdx_0)$ of points in $\Rdd$ satisfying 
$\Pmom_i^2 = 0$, where $\Pmom_i:=\bdx_i - \bdx_{i-1}\neq0$ is a nonzero null vector for each $i\in\brn$. 
Recall the notion of an \emph{\Mdash positive} null polygon $\Pcurve$ from \cref{lemma:Mbd=>simple}.

\begin{definition}\label{dfn:Pll}
A \emph{decoration} of $\Pcurve$ is a pair $(\la,\lat)\in\Matdnr^2$ of $2\times n$ matrices such that for each $i\in\brn$, $(\bisy_i,\bisyt_i)$ given by~\eqref{eq:intro:y_to_lalat} is a decoration of $\Pmom_i$. 
We say that $\Pcurve=(\bdx_1,\dots,\bdx_n)$ is in \emph{normal form} if $\bdx_1=0$. In this case, $\Pcurve$ is uniquely determined by its decoration $\lalat$ via~\eqref{eq:Pmom_vs_yyt}, and when $\lalat$ is fixed, we denote the corresponding null polygon by $\Pll$. 
\end{definition}
\noindent Note that the momentum conservation condition $\la\perp\lat$ is equivalent to $\Pmom_1+\Pmom_2+\dots+\Pmom_n = 0$.

\begin{proof}[Proof of \cref{lemma:intro:Mand_vs_norm}]
Let $\Pll=(\bdx_1=0,\bdx_2,\dots,\bdx_n)$ and $\Pmom_i:= \bdx_i - \bdx_{i-1}$. By~\eqref{eq:bispinor_cdot},
\begin{equation}\label{eq:Pmom_scalar}
 \Pmom_i^2 = 0 \quad\text{and}\quad \Pmom_i\cdot \Pmom_j = \Pmom_j\cdot \Pmom_i = 2\brla<i,j>\brlat[i,j] \quad\text{for all $i,j\in\brn$.}
\end{equation}
By~\eqref{eq:intro:Mand_dfn} and~\eqref{eq:Pmom_scalar}, we get
\begin{equation}\label{eq:Mand_vs_norm_proof}
(\bdx_i-\bdx_j)^2
= \Big(\sum_{\s=i+1}^j \Pmom_\s\Big)^2 
= 4\Mand_{i,j}(\la,\lat)
\quad\text{for all $i\leq j\leq i+n$.}
\qedhere
\end{equation}
\end{proof}

\begin{remark}[Inversion]\label{rmk:dual_conformal}
Fix a t-immersion $\xd:\Faces\to\Rdd$ of $\G$. 
Assume that $\TO(\f)^2\neq0$ for all $\f\in\Faces$. Then we can define the \emph{inversion} of $\TO$ to be the map $\ITO:\Faces\to\Rdd$ given by $\ITO(\f):=\TO(\f) / \TO(\f)^2$. It is not hard to check that~\eqref{eq:inv_null} also holds for the map $\ITO$. Moreover, if $\TO(\f)^2>0$ for all $\f\in\Faces$ then $(\ITO(\f_1)-\ITO(\f_2))^2$ has the same sign as $(\TO(\f_1)-\TO(\f_2))^2$ for any $\f_1,\f_2\in\Faces$. This ensures that \Mdash positivity is preserved under inversion. We expect that the condition $\TO(\f)^2>0$ for all $\f\in\Faces$ is sufficient in order for $\ITO$ to also be a t-immersion (modulo applying a global reflection), and that this property explains dual conformal invariance of scattering amplitudes after expressing scattering amplitudes in terms of t-immersions. We leave this for future work.
\end{remark}

\begin{remark}
It would be interesting to investigate the relation between the Lorentz-maximal surfaces of~\cite{CheRa,CLR1,CLR2} and the Lorentz-minimal surfaces of~\cite{AlMa}.
\end{remark}

\subsection{Boundary data}\label{sec:BCFW:boundary}
We discuss how the various structures associated with $\lalat$ interact with each other. 
We start by summarizing which boundary data is determined by a pair $\lalat\in\lalak$ of $2\times n$ matrices satisfying the conditions in~\eqref{eq:intro:LALAK}.
{
\begin{itemize}
\item Complex numbers $(\y_i)_{i=1}^n,(\yt_i)_{i=1}^n\in(\Cast)^n$ given by~\eqref{eq:intro:y_to_lalat}.
\item Mandelstam variables $\Mijll$ for $i+2\leq j\leq i+n-2$.
\item Boundary angle sums $\sumbT_i := \arg(\y_{i+1}/\y_i)$ and $\sumwT_i := \arg(\yt_{i+1}/\yt_i) + \pi$; cf.~\eqref{eq:angle_sum_arg_rat}.
\item A decorated null polygon $\Pll=(\bdx_1=0,\bdx_2,\dots,\bdx_n)$ in $\Rdd$ in normal form.
\end{itemize}
}

\begin{remark}\label{rmk:bdry:alpha}
Sometimes we will allow multiplication $(\y_i,\yt_i)\mapsto (\alpha\y_i,\ovl{\alpha}\yt_i)$ by a fixed unit complex number $\alpha\in\C$, $|\alpha| = 1$, which preserves the boundary polygon $\PllT$ but rotates the origami boundary polygon $\PllO$. This results in an $\SL_2(\R)\times\SL_2(\R)$-transformation of $\lalat$ and a Lorentz transformation of $\Pll$. The Mandelstam variables and the boundary angle sums are preserved. 
\end{remark}

\begin{definition}\label{dfn:BCFW:little_group}
The \emph{little group} is the subgroup $\LG\subset\GL_n(\R)$ of diagonal matrices with nonzero real entries. We denote by $\LGp,\LGm\subset\LG$ the subsets consisting of matrices with all diagonal entries positive (resp., negative). We denote $\LGpm:=\LGp\sqcup\LGm$. We refer to the subgroup $\LGp$ (resp., $\LGpm$) as the \emph{positive} (resp., \emph{sign-constant}) little group.
\end{definition}

\begin{remark}\label{rmk:BCFW:little_group}
The little group action on $\lalats$ is given by $(\la,\lat)\mapsto(\la\cdot \tdiag,\lat\cdot \tdiag^{-1})$ for $\tdiag\in\LG$. For $\lalat\in\lalak$, we have $(\la\cdot \tdiag,\lat\cdot \tdiag^{-1})\in\lalak$ if and only if $\tdiag\in\LGpm$. For $\tdiag\in\LG$, the transformation $(\la,\lat)\mapsto(\la\cdot \tdiag,\lat\cdot \tdiag^{-1})$ preserves the Mandelstam variables and the boundary polygon $\Pll$. For $\tdiag\in\LGpm$, it additionally preserves the boundary angle sums. 
It does not affect the null polygon $\Pll$ itself but affects its decoration.
In particular, $\Pll$ determines its decoration $\lalat\in\lalak$ up to the action of $\LGpm$.
\end{remark}

\begin{lemma}\label{lemma:BCFW:null}
Let $\Pcurve = (\bdx_1=0,\bdx_2,\dots,\bdx_n)$ be a null polygon 
satisfying $(\bdx_{i+1}-\bdx_{i-1})^2>0$ 
 for $i\in\brn$. Assume further that
\begin{equation}\label{eq:sum_arg_pmom_-2pi}
 \TURN(\PcurveT) := \sum_{i=1}^n \arg_{(-\pi,\pi]}(\PmomT_{i+1} / \PmomT_{i}) = -2\pi.
\end{equation}
Then there exists a unique $2\leq k\leq n-2$ such that $\Pcurve=\Pll$ for some $\lalat\in\lalak$.
\end{lemma}
\begin{proof}
By \cref{rmk:BCFW:little_group}, the null polygon $\Pcurve$ determines the decoration $\lalat$ up to $\LG$-action.
 Since $(\bdx_{i+1}-\bdx_{i-1})^2=2\Pmom_i\cdot \Pmom_{i+1}$, by~\eqref{eq:bispinor_cdot}, we have $\brla<i,i+1>\brlat[i,i+1]>0$ for all $i\in\brnm$ and also $\brla<n,1>\brlat[n,1]>0$. Thus, after acting by some element of $\LG$, we may assume that $\brla<i,i+1> >0$ and $\brlat[i,i+1]>0$ for all $i\in\brnm$. The resulting pair $\lalat$ is determined up to the action of $\LGpm$.

Let $\eps\in\{\pm1\}$ be such that $\eps\brla<n,1> >0$ and $\eps\brlat[n,1]>0$. Let $k$ be such that
$\Arg(\la_1,\la_2) + \cdots + \Arg(\la_{n-1},\la_n) + \Arg(\la_n,\eps\la_1) = (k-1)\pi$, 
where $\Arg:=\Arg_{(-\pi,\pi]}$.
 Thus, $k$ and $\eps=(-1)^{k-1}$ are invariant under the action of $\LGpm$. Letting $(\y_i)_{i=1}^n,(\yt_i)_{i=1}^n\in(\Cast)^n$ be obtained by~\eqref{eq:intro:y_to_lalat}, we have $\arg(\PmomT_{i+1}/\PmomT_i) = \arg(\y_{i+1}/\y_{i}) + \arg(\yt_{i+1}/\yt_{i})$, with $\arg:=\arg_{(-\pi,\pi]}$, $\arg(\y_{i+1}/\y_{i})\in(0,\pi)$, and $\arg(\yt_{i+1}/\yt_{i})\in(-\pi,0)$. By~\eqref{eq:sum_arg_pmom_-2pi} we get $\wind(\lat) = (k+1)\pi$, so $\lalat\in\lalak$.
\end{proof}

\begin{remark}\label{rmk:arg_Pmom_ratio}
For nonzero null vectors $\Pmom,\Qmom\in\Rdd$, $\Pmom\cdot \Qmom>0$ is equivalent to
 $\cos(\arg(\PmomT/\QmomT)) > \cos(\arg(\PmomO/\QmomO))$. In particular, for $i\in\brn$, the argument $\arg_{(-\pi,\pi]}(\PmomT_{i+1} / \PmomT_{i})$ in~\eqref{eq:sum_arg_pmom_-2pi} is never equal to~$\pi$. Let $\lalat$ and $(\y_i)_{i=1}^n,(\yt_i)_{i=1}^n$ be as in the proof of \cref{lemma:BCFW:null}.
Following~\eqref{eq:angle_sum_arg_rat}, we denote
\begin{equation}\label{eq:BCFW:sumwP_sumbP_dfn}
 \sumbP_i:=\arg(\y_{i+1}/\y_i)\in(0,\pi) \quad\text{and}\quad \sumwP_i:=\arg(\yt_{i+1}/\yt_i) + \pi \in(0,\pi).
\end{equation}
(These angle sums are invariant under $\LGpm$-action.) We get
\begin{equation}\label{eq:wind_Pcurve_vs_wind_lalat}
 \arg(\PmomT_{i+1} / \PmomT_{i}) + \pi= \sumbP_i + \sumwP_i
 \quad\text{and}\quad
 \TURN(\PcurveT) = \wind(\la) - \wind(\lat).
\end{equation}
\end{remark}

\begin{remark}\label{rmk:BCFW:null_triangle}
We extend \cref{lemma:BCFW:null} and~\eqref{eq:intro:LALAK} to the case $k\in\{1,n-1\}$ as follows. 
Let $\Matrkone$ be the set of rank-one matrices $M\in\Matdnr$ such that the columns of $M$ are all positive real multiples of some nonzero vector $u\in\R^2\setminus\{0\}$, viewed modulo $\SL_2(\R)$-action. 
 For $n\geq3$, we set
\begin{equation*}%
\lakx_{1,n}:=\Matrkone
\quad\text{and}\quad
	\latkx_{n-1,n}:=\{\lat\in\Matdnr\mid \alt(\lat)\in\Matrkone\}.
\end{equation*}
 We define $\lakx_{n-1,n}$ and $\latkx_{1,n}$ by~\eqref{eq:lak_latk}--\eqref{eq:lak_latk2}. 
For $n\geq3$ and $k\in\{1,n-1\}$, we continue to denote by $\lalakx_{k,n}$ the set of pairs $\lalat\in\lakx_{k,n}\times\latkx_{k,n}$ of $2\times n$ matrices (viewed modulo $\SL_2(\R)\times\SL_2(\R)$) satisfying $\lat\cdot \la^T = \bzero_{2\times2}$. It follows that for $n=3$ and for any null triangle $\Pcurve = (\bdx_1=0,\bdx_2,\bdx_3)$ satisfying~\eqref{eq:sum_arg_pmom_-2pi}, there exists a unique $k\in\{1,2\}$ such that $\Pcurve=\Pll$ for some $\lalat\in\lalak$.
\end{remark}

\subsection{BCFW boundary data}\label{ssec:BCFW_boundary}

Let us fix a collection $\BCFWG_{k,n}$ as in \cref{dfn:BCFW}. We will always assume that we have chosen a white-black bridge in \crefi{dfn:BCFW}{step_BCFW1} as shown in \cref{fig:BCFW}; the case of a black-white bridge is similar. For $\r\in\R$, we denote $\ray(\r):=\bdxT_\io + \r \y_\io\yt_{\io+1}$. We consider the \emph{folding ray}
$\ray:=\{\ray(\r)\mid\r\geq0\}$.
 We shift the origami boundary polygon so that $\bdxO_{\io} = \bdxT_{\io}$. We furthermore choose the unit complex number $\alpha$ as in \cref{rmk:bdry:alpha} so that $\y_\io$ is real and so that the folding ray $\ray$ is ``fixed by the origami map.'' To be more precise, we let $\Rmom\in\Rdd$ be the null vector with decoration $(\y_\io,\yt_{\io+1})$ and set
\begin{equation}\label{eq:Rmom_dfn}
 \Rmomy = \y_\io, \quad \Rmomyt = \yt_{\io+1},\quad \RmomT = \y_\io\yt_{\io+1},\quad \RmomO = \ovl{\y_\io}\yt_{\io+1}, \quad\text{and}\quad
\rayTO(\r):=\bdx_\io + \r \Rmom.
\end{equation}
 Then, since $\y_\io$ is real, $\RmomT = \RmomO$. For each $\io+2\leq \s\leq \io+n-2$, we denote
\begin{equation}\label{eq:Psum_r_dfn}
 \Psum_\s:=\bdx_\s - \bdx_{\io} = \Pmom_{\io+1}+\cdots+\Pmom_\s, \quad%
 \r_\s:=\frac{\Psum_\s^2}{2\Psum_\s\cdot \Rmom}, \quad%
\Qmom_\s:=\bdx_\s - \rayTO(\r_\s) = \Psum_\s - \r_\s\Rmom.
\end{equation}
See \figref{fig:ABC}(right). 
 When the denominator $\Psum_\s\cdot \Rmom$ is zero, we set $\r_\s:=\infty$.
We let $\line_\s$ be the perpendicular bisector between $\bdxT_\s$ and $\bdxO_\s$.%

\begin{remark}\label{rmk:coincides_BCFW}
The above formula for $\r_\s$ coincides with the formula for the pole $z_I$ of the deformed scattering amplitude $A(z)/z$ in~\cite{BCFW}; see e.g.~\cite[Equation~(3.3)]{ElHu}.
\end{remark}

\begin{lemma}\label{lemma:bisector}
The bisector $\line_\s$ intersects the line containing $\ray$ at the point $\raypoint_\s$. The vector $\Qmom_\s\in\Rdd$ is null. (If $\r_\s=\infty$ then $\line_\s$ is parallel to $\ray$ and $\Qmom_\s$ is undefined.)
\end{lemma}
\begin{proof}
Since $\Rmom^2=0$, we see that $\Qmom_\s^2 = \Psum_\s^2 - 2\r_\s\Psum_\s\cdot \Rmom = 0$, so $|\QmomT_\s| = |\QmomO_\s|$. By~\eqref{eq:Psum_r_dfn}, $\bdxT_\s - \raypoint_\s = \QmomT_\s$ and $\bdxO_\s - \raypointO_\s = \QmomO_\s$. 
 Thus, $\raypoint_\s=\raypointO_\s$ is located at equal distance from $\bdxT_\s$ and $\bdxO_\s$. 
\end{proof}

\begin{lemma}\label{lemma:BCFW:rs_not_in_interval}
Suppose that $\lalat$ is \Mdash positive and let $\r\in(0,\infty)$. For each $\io+2\leq\s\leq\io+n-2$, we have $(\bdx_\s - \rayTO(\r))^2 > 0$ if and only if $\r_\s \notin [0,\r]$.
\end{lemma}
\begin{proof}
We have $(\bdx_\s-\rayTO(\r))^2 = (\Psum_\s - \r\Rmom)^2 = \Psum_\s^2 - 2\r\Psum_\s\cdot \Rmom$. 
 Since $\lalat$ is \Mdash positive, $\Psum_\s^2>0$, and from~\eqref{eq:Psum_r_dfn}, we see that the sign of $\r_\s$ equals the sign of $\Psum_\s\cdot \Rmom$. Thus, if $\r_\s<0$ then $(\bdx_\s-\rayTO(\r))^2>0$. If $\r_\s\in[0,\infty)$ then $\Psum_\s^2 - 2\r\Psum_\s\cdot \Rmom>0$ if and only if $\r_\s>\r$. If $\r_\s=\infty$ then $\Psum_\s\cdot \Rmom=0$ so $(\bdx_\s-\rayTO(\r))^2>0$.
\end{proof}

\subsection{T-embeddings of BCFW graphs}\label{sec:BCFW:proof_a_b}
Our goal is to prove that the tiles $\{\lalappf\mid f\in\BCFWfkn\}$ in \cref{thm:f_triang} satisfy parts~\itemref{Rtiling1}--\itemref{Rtiling2} of \cref{dfn:BCFW:tiling_amb}. 

Let $\G\in\BCFWGkn$, $f=\fG\in\BCFWfkn$, $\lalat\in\lalappf$, and $C\in\Ptp_f$ be such that $\la\subset C \subset\latp$. Let $\Tll:\SuppGD\to\C$ be the associated t-immersion. 
Since $\lalat$ is \Mdash positive, $\Tll$ is a t-embedding by 
\cref{lemma:Mbd=>simple,lemma:Jordan_curve}. 
 By \cref{rmk:BCFW_reduced}, $\helmin(\G)\geq2$.

We continue to assume that $\G$ has a white-black bridge at $(\bdv_{\io},\bdv_{\io+1})$ and let $\jo$ be as in \cref{rmk:jo}. Recall that $\fo$ is the interior face of $\G$ adjacent to $\bdf_\io$ and $\bdf_\jo$. 
 Thus, $\xdoutll(\fo)$ lies on the ray $\rayTO$.
\begin{proposition}\label{prop:BCFW:from_lalat_to_r}
Let $\rcrit\in(0,\infty)$ be such that $\xToutll(\fo) = \ray(\rcrit)$. Then $\rcrit = \r_\jo\in(0,\infty)$, and for all $\s\neq\jo$ such that $\io+2\leq \s\leq \io+n-2$, we have $\r_\s\notin[0,\r_\jo]$.
\end{proposition}
\begin{proof}
Since $\fo$ is adjacent to $\bdf_\jo$ and $\xdoutll(\fo) = \rayTO(\rcrit)$, we see from~\eqref{eq:intro:O_preserves_edges} that $|\ray(\rcrit) - \bdxT_\jo| = |\rayO(\rcrit) - \bdxO_\jo|$. By \cref{lemma:bisector}, we get $\rcrit = \r_{\jo}$. 
 Suppose that $\s\neq\jo$. 
 By \cref{lemma:BCFW:rs_not_in_interval}, in order to show $\r_\s\notin[0,\r_\jo]$, it suffices to show that $(\bdx_\s-\rayTO(\rcrit))^2=(\bdx_\s-\xdoutll(\fo))^2 > 0$. By \cref{cor:BCFW:share_face}, $\fo$ and $\bdf_\s$ do not share a face of $\GD$, so by \cref{lemma:Mbd=>face_Mnn}, we indeed have $(\bdx_\s-\xdoutll(\fo))^2 > 0$.
\end{proof}

\begin{proof}[{Proof of \cref{thm:f_triang} (injectivity)}]
We prove that the tiles $\{\lalappf\mid f\in\BCFWfkn\}$ satisfy \crefi{dfn:BCFW:tiling_amb}{Rtiling1}. 
Recall from \cref{rmk:BCFW:T_emb_unique} that it is enough to show that the t-embedding $\Tll$ of $\G$ (with unspecified edge weights) is uniquely recovered from $\lalat$. First, by construction, $\G$ determines $\jo$. By \cref{prop:BCFW:from_lalat_to_r}, we can recover $\xdoutll(\fo) = \rayTO(\r_{\jo})$. Let $(\G_L,\G_R)\in\BCFWGknL\times\BCFWGknR$ be such that $\G = \G_L\otimes\G_R$ as in \cref{dfn:BCFW}. The t-embedding $\xll$ restricts to t-embeddings $\xd_L,\xd_R$ of $\G_L,\G_R$, respectively, and each of $\xd_L,\xd_R$ \hasMbd by \cref{lemma:BCFW:rs_not_in_interval,prop:BCFW:from_lalat_to_r}.
By \cref{rmk:BCFW:little_group}, the boundary null polygons of 
$\xd_L(\GD_L)$ and $\xd_R(\GD_R)$ 
 determine the decorations $(\la_L,\lat_L)\in\lalakx_{k_L,n_L}$ and $(\la_R,\lat_R)\in\lalakx_{k_R,n_R}$ up to $\LGpm$-action. (Here, $(\la_L,\lat_L)\in\lalakx_{k_L,n_L}$ and $(\la_R,\lat_R)\in\lalakx_{k_R,n_R}$ by \cref{prop:t_imm=>TRIPLES} since we know that $\xd_L,\xd_R$ are t-immersions.) We recover the t-embeddings $\xd_L,\xd_R$ inductively from $(\la_L,\lat_L)$ and $(\la_R,\lat_R)$.
\end{proof}

\begin{proof}[{Proof of \cref{thm:f_triang} (disjointness)}]
We verify \crefi{dfn:BCFW:tiling_amb}{Rtiling2}. 
Suppose that $f,g\in\BCFWfkn$ are distinct and let $\G,\G'\in\BCFWGkn$ be the corresponding reduced graphs. Suppose for contradiction that $\lalat\in\lalappf\cap\lalappg$. Let $(k_L,n_L,k_R,n_R)$ (resp., $(k_L',n_L',k_R',n_R')$) and $\G = \G_L\otimes\G_R$ (resp., $\G' = \G'_L\otimes\G'_R$) be as in \cref{dfn:BCFW}.
 If $(k_L,n_L,k_R,n_R)=(k_L',n_L',k_R',n_R')$ then we must have either $\G_L\neq \G'_L$ or $\G_R\neq\G'_R$, in which case we continue our analysis inductively. Thus, we may assume that $(k_L,n_L,k_R,n_R)\neq(k_L',n_L',k_R',n_R')$. Suppose first that $n_L\neq n_L'$, and let $\jo:=\io+n_L-1$ and $\jop:=\io+n_L'-1$ as in \cref{rmk:jo}. Applying \cref{prop:BCFW:from_lalat_to_r} to $\jo$ and $\jop$, we get $\r_\jo,\r_\jop\in(0,\infty)$ with $\r_\jop\notin[0,\r_\jo]$ and $\r_\jo\notin[0,\r_\jop]$, a contradiction. Thus, $n_L = n_L'$ and $\jo = \jop$. By \cref{prop:BCFW:from_lalat_to_r}, the unique t-embeddings $\Tll,\Tllp$ of $\G,\G'$ with boundary data $\lalat$ satisfy $\xdoutll(\fo) = \xdoutllp(\fo)$. Let $\xd_L,\xd'_L$ be the restrictions of $\xll,\xllp$ to $\G_L,\G'_L$, respectively. Similarly to our proof of injectivity above, 
 we see that $\xd_L,\xd'_L$ have the same (\Mdash positive by \cref{lemma:BCFW:rs_not_in_interval}) boundary null polygons, and by \cref{lemma:BCFW:null}, the parameters $k_L=k_L'$ are uniquely determined by these boundary polygons. 
 Since $k_L+k_R = k_L'+k_R' = k+1$, we have $k_R = k_R'$. This contradicts our assumption that $(k_L,n_L,k_R,n_R)\neq(k_L',n_L',k_R',n_R')$.
\end{proof}

\subsection{\ORATITLE}\label{sec:BCFW:from_lalat_to_t_emb}
Fix $\lalat\in\lalapp$. Our goal is to show that if $\lalat$ is \emph{generic} in a certain sense (\cref{dfn:BCFW:generic,dfn:msgen}) then there exists a (necessarily unique) $f\in\BCFWfkn$ such that $\lalat\in\lalappf$. In other words, we want to find a t-embedding $\xd$ of some (necessarily unique) graph $\G\in\BCFWGkn$ such that $\Pll$ is the boundary polygon of $\xd(\GD)$. This will complete the proof of \cref{thm:f_triang} by verifying \crefi{dfn:BCFW:tiling_amb}{Rtiling3}. 

The index $\io$ was fixed during the construction of $\BCFWGkn$. Our first goal is to find the index~$\jo$. Recall from~\eqref{eq:Psum_r_dfn} that we have defined $\r_\s\in\R\cup\{\infty\}$ for each $\io+2\leq\s\leq\io+n-2$. 

\begin{lemma}\label{lemma:ray_inside_polygon}
There exists an index $\io+2\leq\s'\leq\io+n-2$ such that $\r_{\s'}\in(0,\infty)$ and such that $\rayT(\r_{\s'})$ is located inside the polygon $\PllT$.
\end{lemma}
\begin{proof}
By \cref{lemma:non_self_intersecting_Sigma}, the polygon $\PllT$ is simple. Since $\wind(\la) = (k-1)\pi$ and $\wind(\lat) = (k+1)\pi$, we have 
$\TURN(\PllT)=-2\pi$ 
by~\eqref{eq:wind_Pcurve_vs_wind_lalat}. 
 It follows from the definition~\eqref{eq:Rmom_dfn} of $\Rmom$ that for small $\r>0$, the point $\ray(\r)$ is located to the right of the portion $\bdxT_{\io-1}\to\bdxT_{\io}\to\bdxT_{\io+1}$ of $\PllT$. In other words, $\ray(\r)$ is inside $\PllT$ for small $\r>0$. Let $\r'\in(0,\infty)$ be the minimal value for which $\rayT(\r')$ belongs to the boundary of $\PllT$. Let $\io+2\leq \s\leq \io+n-1$ be such that 
 the line segment $[\bdxT_{\s-1},\bdxT_\s]$ contains the point $\rayT(\r')$. 
 Let $t'\in[0,1]$ be such that $(1-t')\bdxT_{\s-1}+t'\bdxT_\s = \rayT(\r')$.

For $t\in[0,1]$, let $\Pmom(t):=\Psum_{\s-1} + t\Pmom_\s$. Thus, $\PmomT(t') = \r'\RmomT$. Similarly to the proof of \cref{lemma:weakMbd=>boundary>0_or_=0}, $h_0(t):=\Pmom(t)^2$ is an affine linear function of $t$ which takes nonnegative values at $t\in\{0,1\}$, and at least one of these values must be positive since $\lalat$ is \Mdash positive. We claim that $h_0(t')>0$. Indeed, if $\io+3\leq\s\leq\io+n-2$ then $h_0(0)>0$ and $h_0(1)>0$ so $h_0(t)>0$ for all $t\in[0,1]$. 
 Otherwise, if $\s = \io+2$ (resp., $\s = \io+n-1$) then we get $h_0(0) = 0$ (resp., $h_0(1)=0$) but $t'>0$ (resp., $t'<1$), so we still have $h_0(t')>0$. 

Consider another affine linear function $h_{\r'}(t):=(\Pmom(t) - \r'\Rmom)^2$. 
 Since $\PmomT(t') = \r'\RmomT$ and $\Rmom$ is null, $|\r'\RmomO|=|\r'\RmomT|=|\PmomT(t')|>|\PmomO(t')|$ because $\Pmom(t')^2=h_0(t')>0$ as we showed above. Thus, $\PmomO(t') - \r'\RmomO\neq0$, and so $h_{\r'}(t')<0$. Since $h_{\r'}(t)$ is affine linear in $t$, we must have $h_{\r'}(t'')<0$ for some $t''\in\{0,1\}$. Let $\s':=\s-1$ if $t''=0$ and $\s':=\s$ if $t''=1$, so that $(\Psum_{\s'}-\r'\Rmom)^2<0$. We must have $\io+2\leq \s'\leq \io+n-2$ because $(\Psum_{\io+1}-\r'\Rmom)^2$ and $(\Psum_{\io+n-1}-\r'\Rmom)^2$ are identically zero as functions of $\r'$. 

Since $\Rmom$ is null, $(\Psum_{\s'}-\r'\Rmom)^2 = \Psum_{\s'}^2 - 2\r'\Psum_{\s'}\cdot \Rmom$. Thus, $\Psum_{\s'}^2 < 2\r'\Psum_{\s'}\cdot \Rmom$. Since $\lalat$ is \Mdash positive and $\io+2\leq \s'\leq \io+n-2$, we have $\Psum_{\s'}^2>0$. It follows that $\Psum_{\s'}\cdot \Rmom>0$. By~\eqref{eq:Psum_r_dfn}, we get $0<\r_{\s'}<\r'$. Thus, we have found $\io+2\leq \s'\leq \io+n-2$ such that $\r_{\s'}\in(0,\infty)$, and since $\r_{\s'}<\r'$, the point $\rayT(\r_{\s'})$ is inside the polygon $\PllT$.
\end{proof}

From now on, we let $\jo\in\Z$ be such that $\io+2\leq\jo\leq\io+n-2$ and $\r_\jo\in(0,\infty)$ is minimal possible. 
\begin{definition}\label{dfn:BCFW:generic}
We say that the pair $\lalat$ (resp., the polygon $\Pll$) is \emph{\ssgen} if $\r_\s\neq \r_\jo$ for all $\io+2\leq\s\leq\io+n-2$ with $\s\neq\jo$.
\end{definition}
We assume that $\lalat$ is \ssgen. Since $\r_\s\neq0$ because $\lalat\in\lalapp$, we have $\r_\s\notin[0,\r_\jo]$ for all $\s$ as in \cref{dfn:BCFW:generic}. We set $\xdout(\fo):=\raypointTO_\jo$. %
By \cref{lemma:BCFW:rs_not_in_interval}, we get the following.
\begin{corollary}\label{cor:BCFW:Mand_pos}
For all $\io+2\leq\s\leq\io+n-2$ with $\s\neq\jo$, we have $(\bdx_\s-\xdout(\fo))^2>0$. 
\end{corollary}

Let $n_L:=\jo-\io+1$, $n_R:=\io-\jo+n+1$, and $\r:=\r_\jo$. Let $\Phmom_\io:=\xdout(\fo)-\bdx_{\io-1}$, $\Phmom_{\io+1}:=\bdx_{\io+1}-\xdout(\fo)$, and $\Qhmom:=\Qmom_\jo$. All three of these vectors are null and nonzero. We would like to apply the induction hypothesis to construct a t-embedding $\xd_L(\GD_L)$ (resp., $\xd_R(\GD_R)$) with \Mdash positive boundary null polygon $\AcurveL=(\bdx_{\io+1},\dots,\bdx_{\jo},\xdout(\fo))$
 (resp., $\AcurveR=(\bdx_{\jo},\dots,\bdx_{\io-1},\xdout(\fo))$); see \cref{fig:ABC}. 
We make the following inductive definition.
\begin{definition}\label{dfn:msgen}
We say that $\lalat$ (resp., $\Pll$) is \emph{\msgen} if it is \ssgen and the null polygons $\AcurveL,\AcurveR$ are also \msgen. 
For the base case $(k,n)\in\{(1,3),(2,3)\}$, any $\lalat$ is considered \msgen.
\end{definition}
\noindent In other words, $\lalat$ is \msgen if one never encounters an equality $\r_{\jo}=\r_\s$ for some $\s\neq\jo$ throughout the entire branch of the \ora starting at $\lalat$. 
The notion of being \ssgen depends on the choice ($\io$ and a white-black vs black-white bridge) 
that was made in \crefi{dfn:BCFW}{step_BCFW1} during the corresponding step of the BCFW recursion. 
 Similarly, the notion of being \msgen depends on the choices in \crefi{dfn:BCFW}{step_BCFW1} 
throughout the entire branch of the BCFW recursion terminating in $\G\in\BCFWGkn$. 
All of these choices have been fixed once and for all when we fixed the collection $\BCFWGkn$. 

\begin{remark}\label{rmk:msgen}
By~\eqref{eq:Psum_r_dfn}, for $\io+2\leq\jo,\s\leq\io+n-2$, the conditions $\r_{\jo}\in(0,\infty)$ and $\r_\s\not\in[0,\r_{\jo}]$ translate into
\begin{equation}\label{eq:ssgen>0}
 \Psum_{\jo}\cdot \Rmom > 0 \quad\text{and}\quad
 \Psum_{\s}^2 \cdot \Psum_{\jo}\cdot \Rmom - \Psum_{\jo}^2 \cdot \Psum_\s\cdot \Rmom >0 
\quad\text{for $\s\neq\jo$}.
\end{equation}
The left-hand sides of these inequalities are regular functions on the partial flag variety $\lalats\cong\Fln(2,n-2)$ which are clearly not identically zero. 
 The open subset of \ssgen elements $\lalat\in\MPkntree$ contains (the intersection of $\MPkntree$ with) the non-vanishing locus of these functions. Thus, this subset is nonempty 
 (and therefore dense). Similarly, the subset of \msgen elements $\lalat\in\MPkntree$ contains the non-vanishing locus of a collection of rational functions defined recursively. The fact that the resulting Zariski open subset is nonempty (and therefore dense) follows because for each $\fap\in\BCFWfkn$, $C\in\Ptp_\fap$, and $\LaLat\in\LaLaimmnn$, the pair $\lalat:=\PhiLL(C)$ is \msgen by \cref{prop:BCFW:from_lalat_to_r}. 
See also \cref{rmk:msgen_exist}.
\end{remark}

\begin{proof}[{Proof of \cref{thm:f_triang} (surjectivity)}]
We check \crefi{dfn:BCFW:tiling_amb}{Rtiling3}. Assume that $\lalat\in\lalapp$ is \msgen; by \cref{rmk:msgen}, the set of such elements is dense in $\lalapp$.

Let $\Sig$ be a graph with vertices $\{\bdf_1,\bdf_2,\dots,\bdf_n,\fo\}$ and edges $\{\bdf_{\s},\bdf_{\s-1}\}$ for $\s\in\brn$ together with $\{\fo,\bdf_{\io-1}\},\{\fo,\bdf_\io\},\{\fo,\bdf_{\io+1}\},\{\fo,\bdf_\jo\}$. Recall that we set $\xdout(\fo):=\raypointTO_\jo$. Thus, we have a map $\xd=(\Tcal,\Ocal):\Sig\to\Rdd$. 
By \cref{cor:BCFW:Mand_pos}, this map satisfies the assumptions of \cref{lemma:non_self_intersecting_Sigma}. (Note that $\QhmomT\neq0$ since $\Mxxll\io\jo>0$.) Thus, $\Tcal$ is an embedding of $\Sig$ into $\C$. By \cref{lemma:ray_inside_polygon}, the point $\xTout(\fo)$ is located inside the polygon $\PllT$. 
 It follows that the polygons $\AcurveTL$ and $\AcurveTR$ also satisfy~\eqref{eq:sum_arg_pmom_-2pi}. 

\begin{figure}
 \includegraphics[scale=0.97]{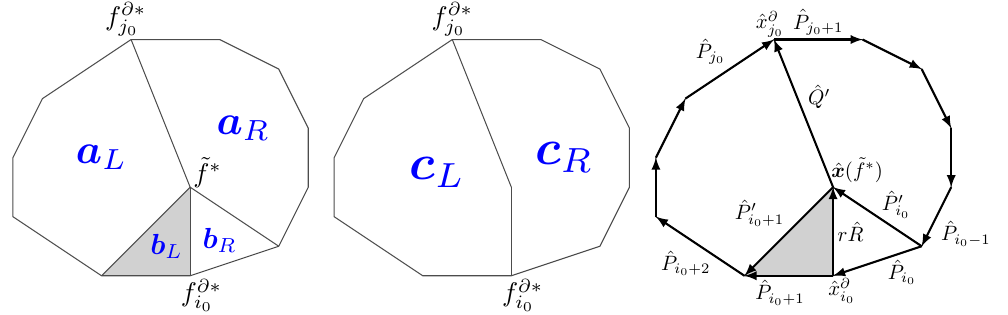}
 \caption{\label{fig:ABC} The null polygons $\AcurveL,\AcurveR,\BcurveL,\BcurveR,\CcurveL,\CcurveR$ projected to the kami plane.}
\end{figure}

By \cref{cor:BCFW:Mand_pos}, the null polygons $\AcurveL,\AcurveR$ satisfy the assumptions of \cref{lemma:BCFW:null} (or \cref{rmk:BCFW:null_triangle} in the case $n_L=3$ or $n_R=3$). Thus, we obtain integers $k_L,k_R$ and decorations $\lalatL\in\lalakx_{k_L,n_L},\lalatR\in\lalakx_{k_R,n_R}$ such that $\AcurveL$ (resp., $\AcurveR$) is the null polygon associated to $\lalatL$ (resp., $\lalatR$). 

We check that $(k_L,n_L,k_R,n_R)$ satisfies conditions~\itemref{BCFW1}--\itemref{BCFW3} in \cref{dfn:BCFW}. Introduce null polygons 
$\BcurveL=(\bdx_{\io},\bdx_{\io+1},\xdout(\fo))$, 
 $\BcurveR=(\bdx_{\io-1},\bdx_{\io},\xdout(\fo))$, 
 $\CcurveL=(\bdx_{\io},\bdx_{\io+1},\dots,\bdx_{\jo},\xdout(\fo))$, and
 $\CcurveR=(\bdx_{\jo},\dots,\bdx_{\io-1},\bdx_{\io},\xdout(\fo))$;
 cf. \cref{fig:ABC}. 
 Thus, we have defined seven null polygons $\AcurveL,\AcurveR,\BcurveL,\BcurveR,\CcurveL,\CcurveR,\Pll$ in $\Rdd$ whose projections to the kami plane are all oriented clockwise, and satisfy the homological identities 
\begin{equation}\label{eq:BCFW:curve_H1}
 \AcurveL+\BcurveL=\CcurveL,\quad
 \AcurveR+\BcurveR=\CcurveR, \quad\text{and}\quad 
\CcurveL+\CcurveR=\Pll
\quad\text{in $H_1(\Sig,\Z)$.}
\end{equation}

For a nonzero complex number $z\in\Cast$, define $\argp(z)\in[0,\pi)$ by setting $\argp(z):=\arg(z)$ if $\arg(z)\in[0,\pi)$ and $\argp(z):=\arg(-z)$ otherwise, where $\arg=\arg_{(-\pi,\pi]}$ as before.
 Let $\Dcurve$ be one of the above seven null polygons, and let $\bdf$ be a boundary vertex of $\Dcurve$ incident to null vectors $\Pmom,\Qmom$ such that $\PmomT,\QmomT$ are oriented clockwise around the boundary of $\Dcurve$. Choose some decorations of $\Pmom,\Qmom$ and define %
\begin{equation*}%
 \sumbD(\bdf):=\argp(\Qmomy/\Pmomy), \quad%
 \sumwD(\bdf):=\argp(\Qmomyt/\Pmomyt), \quad\text{and}\quad
 \sumD(\bdf):=\arg(\QmomT/\PmomT)+\pi.
\end{equation*}
By \cref{rmk:BCFW:little_group}, these quantities do not depend on the decorations of $\Pmom,\Qmom$. We claim that if $\Dcurve$ is one of the above seven null polygons and $\bdf$ is any one of its boundary vertices then 
\begin{equation}\label{eq:BCFW:sumbD+sumwD=sumD}
 \sumbD(\bdf) + \sumwD(\bdf) = \sumD(\bdf).
\end{equation}
When $\Pmom\cdot \Qmom>0$,~\eqref{eq:BCFW:sumbD+sumwD=sumD} follows from~\eqref{eq:wind_Pcurve_vs_wind_lalat}. Suppose now that $\Pmom\cdot \Qmom=0$ (which is the only other possibility). Recall that $\Sig$ satisfies the assumptions of \cref{lemma:non_self_intersecting_Sigma}. Let $\bdf_-$ (resp., $\bdf_+$) be the boundary vertex of $\Dcurve$ immediately before (resp., after) $\bdf$ in clockwise order. Then $\{\bdf_-,\bdf_+\}$ must be an edge of $\Sig$. Thus, the vertices $\{\bdf_-,\bdf,\bdf_+\}$ form a triangle of $\Sig$ embedded by $\T$ into the plane. Since $\DcurveT$ is embedded clockwise, we get $\sumD(\bdf)\in(0,\pi)$. 
 We have $\QmomT/\PmomT = (\Qmomy\Qmomyt)/(\Pmomy\Pmomyt)$ by definition, so the difference $\sumbD(\bdf) + \sumwD(\bdf) - \sumD(\bdf)$ is divisible by $\pi$. Note also that in order to have $\Pmom\cdot \Qmom=0$, we must have either $\sumbD(\bdf)=0$ or $\sumwD(\bdf)=0$, so $\sumbD(\bdf)+\sumwD(\bdf)\in[0,\pi)$. Because $\sumD(\bdf)\in(0,\pi)$, we get that $\sumbD(\bdf) + \sumwD(\bdf) - \sumD(\bdf)=0$. This shows~\eqref{eq:BCFW:sumbD+sumwD=sumD}.

Let $\Dcurve_1+\Dcurve_2=\Dcurve$ be one of the homological identities listed in~\eqref{eq:BCFW:curve_H1}. Let $\bdf$ be a common boundary vertex of $\Dcurve_1,\Dcurve_2,\Dcurve$. (Thus, $\bdf\in\{\bdf_{\io-1},\bdf_\io,\bdf_{\io+1},\bdf_\jo,\fo\}$). We claim that
\begin{equation}\label{eq:BCFW:D1_D2_D}
 \sumbDa(\bdf) + \sumbDb(\bdf) = \sumbD(\bdf) \quad\text{and}\quad
 \sumwDa(\bdf) + \sumwDb(\bdf) = \sumwD(\bdf).
\end{equation}
Indeed, we have $\sumDa(\bdf)+\sumDb(\bdf) = \sumD(\bdf)$ because $\Sig$ is embedded. 
The equations in~\eqref{eq:BCFW:D1_D2_D} hold up to an integer multiple of $\pi$. More precisely, since each term in~\eqref{eq:BCFW:D1_D2_D} belongs to $[0,\pi)$, each of the differences $\sumbDa(\bdf) + \sumbDb(\bdf) - \sumbD(\bdf)$ and $\sumwDa(\bdf) + \sumwDb(\bdf) - \sumwD(\bdf)$ belongs to $\{0,\pi\}$. By~\eqref{eq:BCFW:sumbD+sumwD=sumD}, the sum of these differences is $\sumDa(\bdf)+\sumDb(\bdf) - \sumD(\bdf)=0$. Thus, each difference must be zero. 
Applying a similar argument to the vertex $\fo$, we get 
\begin{equation}\label{eq:BCFW:sum_fo_pi}
 \sumbCL(\fo) + \sumbCR(\fo) = \pi \quad\text{and}\quad
 \sumwCL(\fo) + \sumwCR(\fo) = \pi.
\end{equation}

For each null polygon $\Dcurve$, define $k_{\Dcurve}$ to be the integer determined by \cref{lemma:BCFW:null}, so that the sum of $\sumbD(\bdf)$ over all boundary vertices $\bdf$ of $\Dcurve$ equals $(k_{\Dcurve}-1)\pi$. We see that $k_{\BcurveL} = 2$, $k_{\BcurveR} = 1$, $k_{\AcurveL} = k_L$, $k_{\AcurveR} = k_R$, and $k_{\Pll} = k$. If $n_L=3$ (i.e., if $\AcurveL$ is a triangle) then it must be of different color than $\BcurveL$ because $\Mxxll\io\jo>0$. Thus, $\AcurveL$ must be white, and so $k_{\AcurveL}=1$. 
 Similarly, if $n_R=3$ then $\Mxxll\io\jo>0$ implies that $k_{\AcurveR}=2$. By~\eqref{eq:BCFW:D1_D2_D}, $k_{\AcurveL}+k_{\BcurveL}=k_{\CcurveL}+1$, $k_{\AcurveR}+k_{\BcurveR}=k_{\CcurveR}+1$, and (taking~\eqref{eq:BCFW:sum_fo_pi} into account) $k_{\CcurveL}+k_{\CcurveR}=k_{\Pll}+2$. This implies that $k_L+k_R = k + 1$.

We have checked that the tuple $(k_L,n_L,k_R,n_R)$ satisfies~\itemref{BCFW1}--\itemref{BCFW3}. Since $\lalatL\in\lalakx_{k_L,n_L}$ and $\lalatR\in\lalakx_{k_R,n_R}$ are \Mdash positive by \cref{cor:BCFW:Mand_pos}, we get $\lalatL\in\lalappx_{k_L,n_L}$ and $\lalatR\in\lalappx_{k_R,n_R}$. By the induction hypothesis, there exists a graph $\G_L\in\BCFWGknL$ (resp., $\G_R\in\BCFWGknR$) and a t-embedding $\T_L$ (resp., $\T_R$) of $\G_L$ (resp., $\G_R$) with boundary data $\lalatL$ (resp., $\lalatR$). 
Inserting these t-embeddings inside $\AcurveTL$ and $\AcurveTR$, we obtain an embedding $\T:\SuppGD\to\C$, where $\G:=\G_L\otimes\G_R\in\BCFWGkn$. To check that $\T$ is a t-embedding, we verify conditions~\itemref{intro:t_imm_straight_convex}--\itemref{intro:t_imm_angles_bdry} in \cref{dfn:intro:t_imm}. \itemref{intro:t_imm_straight_convex} is satisfied by construction. \itemref{intro:t_imm_gauge} is vacuous since we define $\wt:=\wtT$. The \Kawangle condition~\itemref{intro:t_imm_angles_int} holds for $\fo$ by~\eqref{eq:BCFW:D1_D2_D}--\eqref{eq:BCFW:sum_fo_pi}, and holds for all other interior vertices of $\xT(\GD)$ by induction. The boundary angle condition~\itemref{intro:t_imm_angles_bdry} holds for $\T$ by~\eqref{eq:angle_sum_arg_rat} since $\lalat\in\lalak$.
 To show~\itemref{intro:t_imm_orientation}, observe that $\T$ is clearly orientation-preserving on the triangular faces $\BcurveTL,\BcurveTR$ of $\xT(\GD)$; all other triangular faces of $\xT(\GD)$ belong to either $\T_L(\GD_L)$ or $\T_R(\GD_R)$, so $\T$ is also orientation-preserving on them by induction. However, to complete the check of~\itemref{intro:t_imm_orientation} for $\T$, 
 we need to omit degenerate bigonal faces from $\xT(\GD)$; cf. \cref{rmk:BCFW_reduced}. Since $\T$ is orientation-preserving on each triangular face of $\GD$, it will be orientation-preserving on each polygonal face $\v\in\Vint$ of the contracted graph $\G$. The fact that the image $\xT(\v)$ is a convex $d$-gon (where $d=\degG(\v)$) follows from the angle conditions~\itemref{intro:t_imm_angles_int}--\itemref{intro:t_imm_angles_bdry}. By~\itemref{intro:t_imm_angles_bdry}, the angle of $\xT(\v)$ at any boundary vertex $\bdxT_i$ belongs to $(0,\pi)$. 
By construction, $\sumbBL(\fo),\sumwAL(\fo),\sumbAR(\fo),\sumwBR(\fo)>0$, so by~\itemref{intro:t_imm_angles_int}, the angle of $\xT(\v)$ at $\xT(\fo)$ cannot be equal to $\pi$ and so belongs to $(0,\pi)$. By induction, the angle of $\xT(\v)$ at every other interior vertex $\xT(\f)$ incident to $\xT(\v)$ also belongs to $(0,\pi)$.
Thus, $\xT(\v)$ is indeed a convex $d$-gon. 
This finishes the proof of \crefi{dfn:BCFW:tiling_amb}{Rtiling3}, completing the proof of \cref{thm:f_triang}. %
\end{proof}

\section{T-duality}\label{sec:TREE}
The goal of this section is to introduce the \emph{T-duality map} $\AAshift:\MPkn\to\AAkn$ from the ambient momentum amplituhedron $\MPkn$ defined in~\eqref{eq:intro:M_amb_dfn} to the \emph{ambient \mta} $\AAkn$ defined in~\eqref{eq:TREE:AAkn_dfn} below. It is based on the magic projector $\Qla$ (\cref{ssec:MOM:magic_Qla_mom_ampl}). While we focus on the $m=4$ amplituhedron in this paper, we note that an $m=2$ version of T-duality was studied in~\cite{LPW,PSBW}, and some properties of the $m=4$ T-duality map $C\mapsto C\cdot \Qla$ on positroid cells 
(e.g., \cref{lemma:la_vs_Bounda_Boundb}) 
 have also been explored in~\cite[Section~5.2]{LPW}.%

In~\justpaptwo, we upgrade T-duality studied in this section to the level of (not necessarily reduced) weighted planar bipartite graphs and loop amplituhedra.

\subsection{T-duality for ambient amplituhedra}

Recall the definition of the ambient momentum amplituhedron $\MPkn$ and its tiles $\MPf$ from~\eqref{eq:intro:M_amb_dfn} and~\eqref{eq:lalappf_dfn}. 
 By \cref{lemma:la_vs_Bounda_Boundb}, $\MPf$ is empty unless $\fap\in\BND_{2,0}(k,n)$. 
Next, we give the ``ambient'' version of the sign flip definition~\cite{AHTT} of the \emph{\mta} of~\cite{AHT};
see also~\cite[Definition~11.9]{PSBW} for an $m=2$ version.
\begin{definition}
For a nonzero vector $v=(v_1,v_2,\dots,v_d)\in\R^d\setminus\{0\}$, let $\var(v)$ be the number of times the sequence $v_1,v_2,\dots,v_d$ changes sign (omitting all zero entries in $v$). 
\end{definition}
\noindent For example, $\var(0,1,2,0,-1,-3,2,0,5,-1) = \var(1,2,-1,-3,2,5,-1) = 3$. It is easy to see that if $\la\in\Gr(2,n)$ satisfies $\brla<i,i+1> >0$ for all $i\in\brn$ then 
\begin{equation}\label{eq:TREE:wind_vs_varx}
 \wind(\la) = (\varx(\la) +1)\pi, \quad\text{where}\quad \varx(\la):=\var\left(\brla<1,2>,\brla<1,3>,\dots,\brla<1,n>\right).
\end{equation}

Let $V=\mat[V_1|V_2|\cdots|V_n]\in\Gr(4,n)$. 
 Following \cref{sec:backgr:cyc}, we set $V_{i+n}=(-1)^{k-1}V_i$ for all $i\in\Z$. For $a,b,c,d\in\brn$, we denote $\brV[a,b,c,d]:=\det\mat[V_a|V_b|V_c|V_d]$. Let
\begin{equation*}%
 \varxxx(V):=\var\left(\brV[1,2,3,4],\brV[1,2,3,5],\dots,\brV[1,2,3,n]\right).
\end{equation*}

\begin{definition}\label{dfn:TREE:AAkn}
The \emph{($m=4$ tree) ambient \mta} is given by
\begin{equation}\label{eq:TREE:AAkn_dfn}
 \AAkn := \left\{V\in\Gr(4,n)\middle| 
\text{
\begin{tabular}{cc}
$\brV[i,i+1,j,j+1] >0$ for all $i+2\leq j\leq i+n-2$, and\\
$\varxxx(V)=k-2$
\end{tabular}
} \right\}.
\end{equation}
For $\ddfperm\in\Boundxx(k-2,n)$, we set 
\begin{equation}\label{eq:TREE:AAddf_dfn}
 \AAddf := \{V\in\AAkn\mid V\subset \ddC^\perp \text{ for some $\ddC\in\Ptp_{\ddfperm}$}\}.
\end{equation}
\end{definition}

We have $\dim\MPkn = 4n-12$ and $\dim\AAkn = 4n-16$. Consider maps
\begin{align}
\label{eq:AAshift_dfn}
\AAshift:\MPkn &\to \Gr(4,n),%
& \lalat&\mapsto \Qlapp(\lat);\\
\label{eq:AAiso_dfn}
\AAiso:\MPkn&\to\Fln(2,4), & \lalat
&\mapsto (\la,\Qlapp(\lat));
\end{align}
\noindent cf. \cref{dfn:TREE:Qlapp}. 
Thus, $\AAshift = \AAforg\circ\AAiso$ factors through $\AAiso$, where $\AAforg:\Fln(2,4)\to\Gr(4,n)$ is the forgetful map sending $(\la,V)\mapsto V$.

Recall that the subset $\lak\subset\Gr(2,n)$ was defined in~\eqref{eq:lak_latk}. 
Consider intermediate subsets
\begin{equation*}%
 \MAkn:=\{(\la,V)\in \lak\times \AAkn\mid \la\subset V\} \quad\text{and}\quad
 \MAddf:=\{(\la,V)\in \lak\times \AAddf\mid \la\subset V\}.
\end{equation*}

Over the course of the next several subsections, we will show the following result.
\begin{theorem}[T-duality for ambient amplituhedra] \label{thm:AAshift}
Let $2\leq k\leq n-2$. 
\begin{enumerate}[label=(\arabic*)]
\item\label{AAshift_iso} The map $\AAiso:\MPkn\xrasim\MAkn$ is a homeomorphism. It restricts to a homeomorphism $\MPf\xrasim \MAddf$ for all $\fap\in\Bounda$. 
\item\label{AAshift_forg} The map $\AAforg:\MAkn\to\AAkn$ is open and surjective. 
It restricts to a surjective open map $\MAddf\to\AAddf$ for each $\ddfperm\in\BND(k-2,n)$.
\end{enumerate}
\end{theorem}

\subsection{\MdashTITLE positivity vs Pl\"ucker-positivity}\label{ssec:TREE:Mand_vs_Pluck}
Fix $\la\in\lak$ and introduce auxiliary subsets
\begin{align*}%
 \MPnla &:= \{\lat\in\Gr(2,\lap)\mid \lalat\text{ is \Mdash positive}\};\\
 \AAnla &:= \{V\in\Grsupla(4,n)\mid \brV[i,i+1,j,j+1] >0\text{ for all $i+2\leq j\leq i+n-2$}\}.
\end{align*}
\noindent Note that we do not yet impose any sign variation conditions on either $\lat$ or $V$.

\begin{proposition}\label{prop:MPnla_AAnla_homeo}
The operators $\Qlapp$ and $\Qla$ 
restrict to 
 mutually inverse homeomorphisms between $\MPnla$ and $\AAnla$.
\end{proposition}

Before we give a proof, we explain how to identify the Minkowski space $\R^{2,2}$ with the space $\Matddr$ of real $2\times 2$ matrices. 
Given $\xs=(\xsT,\xsO)\in\R^{2,2}$, we define a matrix $\xM_{\xs}$ by 
\begin{equation}\label{eq:x_vs_ap_vs_am}
\xM_{\xs}:= \begin{pmatrix}
\Re(\ap) & \Im(\am) \\ -\Im(\ap) & \Re(\am)
\end{pmatrix},\quad\text{where}\quad
\ap:=\frac12(\xsT + \xsO), \quad \am:=\frac12(\xsT-\xsO); \quad\text{conversely,}
\end{equation}
\begin{equation}\label{eq:xM_to_x}
 \xM_{\xs}=\begin{pmatrix}
a & b\\c & d
\end{pmatrix} 
\quad\Longrightarrow\quad \xsT = (a+d) + \I(b-c) \quad\text{and}\quad \xsO = (a-d) - \I(b+c).
\end{equation}
An immediate consequence of this definition is that
\begin{equation}\label{eq:det_xM_vs_Pmom^2}
 \det \xM_{\xs} = \Re(\ap)\Re(\am) + \Im(\ap)\Im(\am) = \frac14 \left(|\xsT|^2 - |\xsO|^2\right) = \frac14 \xs^2.
\end{equation}

\noindent For a decorated polygon $\Pll=(\bdx_1,\bdx_2,\dots,\bdx_n)$ with $\Pmom_i=\bdx_i - \bdx_{i-1}$ as in \cref{dfn:Pll}, we denote $\bdxM_i:=\xM_{\bdx_i}$. Thus,~\eqref{eq:Pmom_vs_yyt} translates into 
\begin{equation}\label{eq:TE:decor_dfn}
 \bdxM_i-\bdxM_{i-1} = \xM_{\Pmom_i} = \lat_i \cdot \la_i^T \quad\text{for all $i\in\brn$}.
\end{equation}

We denote $\CMI := \begin{pmatrix}
0 & 1 \\ -1 & 0
\end{pmatrix}$.

\begin{proof}[Proof of \cref{prop:MPnla_AAnla_homeo}]
By \cref{lemma:TREE:magic_homeo}, it 
 remains to relate the Mandelstam variables $\Mijll$ to the Pl\"ucker coordinates $\brV[i,i+1,j,j+1]$. 
Given $\lat\in\MPnla$, we temporarily denote the matrices $\bdxM_i$ in~\eqref{eq:TE:decor_dfn} by $\bdxM_i(\lat)$. 
We assume that they are in \emph{normal form}: $\bdxM_1(\lat)=\bzero_{2\times2}$. 
By~\eqref{eq:det_xM_vs_Pmom^2}, 
we get $\frac14\Mijll = \det(\bdxM_{i}(\lat) - \bdxM_{j}(\lat))$. 

Given $V\in\Grsupla(4,n)$, we choose an arbitrary $2$-plane $\mu\in\Gr(2,V)$ complementary to $\la$ so that $V = \la\oplus \mu$. Thus, $\mu$ is defined up to adding a multiple $M\cdot \la$ of $\la$ for $M\in\Matddr$. 
For $i\in\brn$, let 
$\bdxM_i(V,\mu):=\mat[\mu_i|\mu_{i+1}] \cdot (\CMI\cdot \mat[\la_i|\la_{i+1}])^{-1}$.
 These matrices are defined up to an additive constant: changing $\mu\mapsto\mu+M\cdot \la$ corresponds to $\bdxM_i(V,\mu)\mapsto\bdxM_i(V,\mu) + M\cdot \CMI^{-1}$. We assume that they are in \emph{normal form}: $\mat[\mu_1|\mu_2] = \bdxM_1(V,\mu)=\bzero_{2\times2}$. In this case, $\mu=\LVVp\cdot V$ is uniquely determined by $V$ so we write $\bdxM_i(V):=\bdxM_i(V,\mu)$. 

We claim that if $V = \Qlapp(\lat)$ then $\bdxM_i(\lat) = \bdxM_i(V)$ for all $i\in\brn$. Indeed, because they are both in normal form, $\bdxM_1(\lat) = \bdxM_1(V) = \bzero_{2\times2}$. It remains to check that $\bdxM_i(\lat) - \bdxM_{i-1}(\lat) = \bdxM_i(V) - \bdxM_{i-1}(V)$ for all $i\in\brn$. We have
\begin{align*}
\bdxM_i(V) - \bdxM_{i-1}(V) 
&= \mat[\mu_i|\mu_{i+1}]\cdot (\CMI\cdot \mat[\la_i|\la_{i+1}])^{-1} 
 - \mat[\mu_{i-1}|\mu_{i}]\cdot (\CMI\cdot \mat[\la_{i-1}|\la_{i}])^{-1} \\
&= \frac1{\brla<i,i+1>}\mat[\mu_i|\mu_{i+1}]\cdot \mat[-\la_{i+1}|\la_i]^T 
 -\frac1{\brla<i-1,i>}\mat[\mu_{i-1}|\mu_i]\cdot \mat[-\la_{i}|\la_{i-1}]^T\\
 &=\frac{\left(\brla<i-1,i>\cdot \mu_{i+1} + \brla<i,i+1>\cdot \mu_{i-1}\right)\cdot \la_i^T 
- \mu_i \cdot \left(\brla<i-1,i>\cdot \la_{i+1}^T + \brla<i,i+1>\cdot \la_{i-1}^T\right)}{\brla<i-1,i>\brla<i,i+1>}.
\end{align*}
We claim that this equals $\lat_i\cdot \la_i^T$. Indeed, let us add
 $\brla<i+1,i-1>\cdot \mu_i\cdot \la_i^T$ to both terms in the numerator. After this, the first term becomes $\left(\brla<i-1,i>\cdot \mu_{i+1} + \brla<i,i+1>\cdot \mu_{i-1} + \brla<i+1,i-1>\cdot \mu_i\right)\cdot \la_i^T = \brla<i-1,i>\cdot \brla<i,i+1>\cdot \lat_i\cdot \la_i^T$ since $\lat = V\cdot \Qla = \mu\cdot \Qla$ for $V=\la\oplus\mu$ with $\la=\Ker\Qla$; cf.~\eqref{eq:intro:Qla}. The second term becomes 
$\mu_i \cdot \left(\brla<i-1,i>\cdot \la_{i+1}^T + \brla<i,i+1>\cdot \la_{i-1}^T +\brla<i+1,i-1>\cdot \la_i^T\right)$; this expression is zero by Cramer's rule. 
 Thus, $\bdxM_i(V) - \bdxM_{i-1}(V)$ agrees with $\bdxM_i(\lat) - \bdxM_{i-1}(\lat)$ given by~\eqref{eq:TE:decor_dfn}. 
From now on, we denote $\bdxM_i(\lat)=\bdxM_i(V)$ simply by $\bdxM_i$. 

Next, we claim that
\begin{equation}\label{eq:brV_vs_Mand}
 \brV[i,i+1,j,j+1] = \frac14\brla<i,i+1> \brla<j,j+1> \Mijll \quad\text{for all $i+2\leq j\leq i+n-2$.}
\end{equation}
Indeed, since $\mat[\mu_i|\mu_{i+1}] = \bdxM_i\cdot \CMI\cdot \mat[\la_i|\la_{i+1}]$,
\begin{equation}\label{eq:TREE:brV_calc_brla}
 \brV[i,i+1,j,j+1] = \det \begin{pmatrix}
\la_i & \la_{i+1} & \la_j & \la_{j+1} \\
\mu_i & \mu_{i+1} & \mu_j & \mu_{j+1} 
 \end{pmatrix} = \det \begin{pmatrix}
\Id_2 & \Id_2 \\
\bdxM_i\CMI & \bdxM_j\CMI
\end{pmatrix} \cdot \brla<i,i+1> \brla<j,j+1>.
\end{equation}
We have $\det \begin{pmatrix}
\Id_2 & \Id_2 \\
\bdxM_i\CMI & \bdxM_j\CMI
\end{pmatrix} 
 = \det(\bdxM_j-\bdxM_i)$ since $\det(\CMI) = 1$. By~\eqref{eq:det_xM_vs_Pmom^2}, $\det(\bdxM_j-\bdxM_i) = \frac14\Mijll$. This shows~\eqref{eq:brV_vs_Mand}. 
By assumption, $\la$ satisfies $\brla<i,i+1>,\brla<j,j+1> > 0$ for all $i,j\in\Z$. Thus, by~\eqref{eq:brV_vs_Mand}, $\Qlapp:\Gr(2,\lap)\xrasim\Grsupla(4,n)$ indeed restricts to a homeomorphism $\Qlapp:\MPnla\xrasim\AAnla$.%
\end{proof}

\begin{remark}\label{rmk:Mand=>brlat>0}
Since $\brla<i,i+1> >0$ and $\frac14\Mxxll{i-1}{i+1} = \brla<i,i+1> \brlat[i,i+1]$, 
 we have 
$\brlat[i,i+1]>0$ for all $\lat\in\MPnla$ and $i\in\Z$.
\end{remark}

\begin{notation}
From now on, we assume that $\la,\lat,\mu,V$ are related by 
\begin{equation}\label{eq:V_la_mu_lat}
 \lat=V\cdot \Qla=\mu\cdot \Qla,\quad V=\Qlapp(\lat)=\begin{pmatrix}
\la\\ \mu
 \end{pmatrix},\quad
\mu = \LVVp\cdot V,\quad
\mat[\mu_1|\mu_2]=\bzero_{2\times2}.
\end{equation}
\end{notation}
\noindent As explained in the proof of \cref{prop:MPnla_AAnla_homeo}, the matrices $\bdxM_i$ may be recovered from $\la$ and $\mu$ via
\begin{equation}\label{eq:bdxM_CMI_laii=muii}
 \bdxM_i \cdot \CMI \cdot \mat[\la_i|\la_{i+1}] = \mat[\mu_i|\mu_{i+1}].
\end{equation}
Next, given a $2\times2$ matrix $M$, we record the \emph{Schur complement} identity:
\begin{equation}\label{eq:Schur_comp}
\text{for}\quad 
\Lcal=\begin{pmatrix}
 \Id_2 \\ M
 \end{pmatrix} \quad\text{and}\quad
 \Lcal^\perp=\mat[-M|\Id_2],\quad\text{we have}\quad
 \det\mat[V_i|V_j|\Lcal] = \det(\Lcal^\perp\cdot \mat[V_i|V_j]).
\end{equation}
Finally, setting $\Arg:=\Arg_{(-\pi,\pi]}$ as in~\eqref{eq:intro:wind} and $\mu_0:=(-1)^{k-1}\mu_n$, we let
\begin{equation}\label{eq:windskip_dfn}
 \windskip(\mu):=\wind\mat[\mu_3|\mu_4|\cdots|\mu_n] = 
\Arg(\mu_0,\mu_3) + \Arg(\mu_3,\mu_4)+\cdots+\Arg(\mu_{n-1},\mu_n).
\end{equation}
Applying~\eqref{eq:Schur_comp} for 
 $\Lcal=\mat[V_1|V_2]$ and using~\eqref{eq:TREE:wind_vs_varx}, we get
\begin{equation}\label{eq:windskip_vs_varxxx}
 \windskip(\mu) = (\varxxx(V) + 1)\pi.
\end{equation}

Recall that when $\lalat$ is \Mdash positive, $\PllT$ is simple by \cref{lemma:Mbd=>simple}, so~\eqref{eq:wind_Pcurve_vs_wind_lalat} gives
\begin{equation}\label{eq:wind_la_lat_vs_TURN}
 \wind(\la) - \wind(\lat) = \TURN(\PllT)=\pm2\pi
 \quad\text{for $\lat\in\MPnla$.}
\end{equation}

\subsection{$\GLp$-\winding numbers}\label{ssec:GLp_winding}
Let $\GLp\subset\GL_2(\R)$ be the subgroup of matrices with positive determinant. By the polar decomposition, $\pi_1(\GLp)\cong\pi_1(\SO(2))\cong\Z$. For a closed loop $\gl:[0,2\pi]\to\GLp$, we let $\windGL(\gl)\in 2\pi\Z$ be its class in $\pi_1(\GLp)\cong\Z$ multiplied by $2\pi$. Thus, for the standard generator $\glrot(t)=\begin{pmatrix}
\cos(t) & -\sin(t)\\ \sin(t) & \cos(t)
\end{pmatrix}$ of $\pi_1(\GLp)$, we have $\windGL(\glrot)=2\pi$. For a curve $\gp:[0,\pi]\to\GLp$ such that $\gp(0)=(-1)^{k-1}\gp(\pi)$, we define $\windGL(\gp)\in\pi\Z$ to be $\frac12\windGL(\gl)$, where $\gl:[0,2\pi]\to\GLp$ is a loop obtained by concatenating $\gp(t)$ with $(-1)^{k-1}\gp(t)$. Thus, $\windGL(\gp)\equiv(k-1)\pi\pmod{2\pi}$ for any such curve $\gp$.

\begin{remark}\label{rmk:turnGL_act}
It follows from the polar decomposition that one can compute $\windGL(\gl)$ or $\windGL(\gp)$ by choosing a nonzero vector $v\in\R^2\setminus\{0\}$ and recording the counterclockwise turning angle $\windRd(v(t))$ of the loop $v(t):=\gl(t)\cdot v$ in $\R^2\setminus\{0\}$, where $\windRd(\cdot)$ is defined similarly to~\eqref{eq:intro:wind}.
\end{remark}

\begin{definition}%
\label{dfn:GLPath}
To every $2\times n$ matrix $\la$ such that $\brla<i,i+1> >0$ for $i\in\brn$ we associate a curve $\GLPath_\la:[0,n\pi]\to\GLp$ such that $\GLPath_\la((i-1)\pi) = \mat[\la_i|\la_{i+1}]$ for all $i=1,2,\dots,n+1$. Explicitly, to go from, say, $\mat[\la_1|\la_2]$ to $\mat[\la_2|\la_3]$, we concatenate a curve $\gp_1(t)=\mat[(1-t)\la_1-t\la_3|\la_2]$, $t\in[0,1]$, from $\mat[\la_1|\la_2]$ to $\mat[-\la_3|\la_2]$ with the curve $\mat[-\la_3|\la_2]\cdot \glrot(t)$, $t\in[0,\pi/2]$, from $\mat[-\la_3|\la_2]$ to $\mat[\la_2|\la_3]$. We then reparameterize the resulting curve so that $\GLPath_\la(0)=\mat[\la_1|\la_2]$ and $\GLPath_\la(\pi)=\mat[\la_2|\la_3]$. 
\end{definition}
\noindent By \cref{rmk:turnGL_act}, $\windGL(\GLPath_\la) = \wind(\la)$. 

\begin{definition}\label{dfn:matrix_null_poly}
A \emph{matrix null polygon} is a sequence $\gmpol=(\gm_1,\gm_2,\dots,\gm_n)$ of elements of $\GLp$ such that $\det(\gm_{i+1} - \gm_{i})=0$ for all $i\in\brn$. Here, we set $\gm_{i+n}:=\gm_i$ for $i\in\Z$. 
For $\gm_i(t):=(1-t)\gm_{i} + t\gm_{i+1}$, observe that $\det(\gm_i(t))$ is an affine linear function in $t$ that is strictly positive for $t\in\{0,1\}$. Thus, $\gm_i(t)$ belongs to $\GLp$ for all $t\in[0,1]$ and $i\in\brn$. Concatenating these curves, we obtain a loop $\gmpoll$ in $\GLp$ starting and ending at $\gm_1=\gm_{n+1}$. We define $\windGL(\gmpol):=\windGL(\gmpoll)$. 
\end{definition}

For $M:=\begin{pmatrix}
a & b \\ c& d
\end{pmatrix}$, define $\trT(M):=\begin{pmatrix}
a + d \\ b -c
\end{pmatrix}$; cf.~\eqref{eq:xM_to_x}. Observe that $\trT(M)$ is never zero on $\GLp$. Thus, it gives a continuous map $\trT:\GLp\to\R^2\setminus\{0\}$.
\begin{lemma}\label{lemma:turnGL_trT}
The induced homomorphism $\trT_\ast:\pi_1(\GLp)\to\pi_1(\R^2\setminus\{0\})$ satisfies $\windGL(\gl) = -\windRd(\trT_\ast(\gl))$ for each loop $\gl$ in $\GLp$. Furthermore, given a matrix null polygon $\gmpol=(\gm_1,\gm_2,\dots,\gm_n)$ in $\GLp$, 
we have $\windGL(\gmpol)=-\windRd(\trT(\gmpol))$ 
for the $2\times n$ matrix $\trT(\gmpol):=\mat[\trT(\gm_1)|\trT(\gm_2)|\cdots|\trT(\gm_n)]$.
\end{lemma}
\begin{proof}
It suffices to check the identity $\windGL(\gl) = -\windRd(\trT_\ast(\gl))$ on a single generator $\gl(t)=\glrot(t)$, which is straightforward. Observe that the polygon $\trT(\gmpoll)$ avoids the origin since $\gm_i(t)\in\GLp$ as we discussed above. Thus, we indeed get $\windGL(\gmpol) = -\windRd(\trT(\gmpol))$.
\end{proof}

For $\ys\in\Rdd$ and $\xM_{\ys}$ given by~\eqref{eq:x_vs_ap_vs_am}, let 
\begin{equation}\label{eq:muof(ys)_dfn}
 \muof(\ys):=\mu - \xM_{\ys}\cdot \CMI\cdot \la, \quad\text{so that}\quad
\mat[\muiof_i(\ys)|\muiof_{i+1}(\ys)] = (\bdxM_i - \xM_{\ys})\cdot \CMI\cdot \mat[\la_i|\la_{i+1}]
\quad\text{for all $i\in\brn$}
\end{equation}
 by~\eqref{eq:bdxM_CMI_laii=muii}. In particular, by~\eqref{eq:det_xM_vs_Pmom^2}, when $(\bdx_i-\ys)^2>0$ for all $i\in\brn$, the $2\times n$ matrix $\muof(\ys)$ satisfies $\det\mat[\muiof_i(\ys)|\muiof_{i+1}(\ys)]>0$ for all $i\in\brn$, and thus gives rise to a curve $\GLPath_{\muof(\ys)}$ in $\GLp$ via \cref{dfn:GLPath}. Furthermore, recall from \cref{lemma:Mbd=>simple} that when $\la\in\lak$, $\lat\in\MPnla$, and $(\bdx_i-\ys)^2>0$ for all $i\in\brn$,
 the point $\ysT$ is disjoint from the (simple) polygon $\PllT$. 
\begin{lemma}
Assume that $\la\in\lak,\lat\in\MPnla,\mu,V$ are related by~\eqref{eq:V_la_mu_lat} and that $\ys\in\Rdd$ satisfies $(\bdx_i-\ys)^2>0$ for all $i\in\brn$. Then 
\begin{equation}\label{eq:windaround_muof}
 \wind(\muof(\ys)) = \wind(\la) - \windaround(\PllT,\ysT), \quad \text{where}\quad
 \windaround(\PllT,\ysT):=\windRd\mat[\bdxT_1 - \ysT|\cdots|\bdxT_n - \ysT] %
\end{equation}
is the winding number of the polygon $\PllT$ around the point $\ysT$. 
\end{lemma}
\begin{proof}
By the well-known \emph{Eckmann--Hilton argument}~\cite{Eckmann_Hilton}, loop concatenation in $\pi_1(\GLp)$ (or in $\pi_1(X)$ for any topological group $X$) coincides with the operation of pointwise loop multiplication (sending loops $\gl_1,\gl_2:[0,2\pi]\to\GLp$ to a loop $\gl:[0,2\pi]\to\GLp$ given by $\gl(t):=\gl_1(t)\cdot \gl_2(t)$ for $t\in[0,2\pi]$).
Letting $\gmpol=(\bdxM_1 - \xM_{\ys},\dots,\bdxM_n - \xM_{\ys})$ be the associated matrix null polygon, we see by~\eqref{eq:muof(ys)_dfn} that after a suitable reparameterization of each curve, we get $\GLPath_{\muof(\ys)}(t) = \gmpoll(t)\cdot \CMI\cdot \GLPath_{\la}(t)$ for all $t$. It follows that $\wind(\muof(\ys)) = \wind(\la) + \windGL(\gmpol)$. By \cref{lemma:turnGL_trT}, we get $\windGL(\gmpol) = -\windaround(\PllT,\ysT)$.
\end{proof}

For the following result, given $\Lcal\in\Mat_{4,2}(\R)$, let $\varx\mat[V|\Lcal]:=\var(\brVLnopar[1,2,\Lcal],\brVLnopar[1,3,\Lcal],\dots,\brVLnopar[1,n,\Lcal])$, where we set $\brVLnopar[i,j,\Lcal]:=\det\mat[V_i|V_j|\Lcal]$ for all $i,j$.
\begin{lemma}
Let $\la\in\lak,\mu,V$ be related by~\eqref{eq:V_la_mu_lat}, and let $\ys\in\Rdd$
 and $\Liney:=\begin{pmatrix}
\Id_2\\ \xM_{\ys}\cdot \CMI
\end{pmatrix}$.
 Then 
\begin{equation}\label{eq:varx_muof=varx_brVLy}
 \varx\mat[V|\Liney] = \varx(\muof(\ys)).
\end{equation}
\end{lemma}
\begin{proof}
By~\eqref{eq:muof(ys)_dfn}, we have $\muof(\ys) = \Liney^\perp\cdot V$ for $\Liney^\perp:=\mat[-\xM_{\ys}\cdot \CMI|\Id_2]$. It follows by~\eqref{eq:Schur_comp} that for all $i,j$, $\brVLy[i,j]=\det\mat[\muiof_i(\ys)|\muiof_j(\ys)]$.
\end{proof}

Combining~\eqref{eq:windaround_muof}--\eqref{eq:varx_muof=varx_brVLy} with~\eqref{eq:TREE:wind_vs_varx} and~\eqref{eq:wind_la_lat_vs_TURN}, we obtain the following result.
\begin{corollary}\label{cor:var_vs_inside}
Assume that $\la\in\lak,\lat\in\MPnla,\mu,V$ are related by~\eqref{eq:V_la_mu_lat} and that $\ys\in\Rdd$ satisfies $(\bdx_i-\ys)^2>0$ for all $i\in\brn$. Then 
\begin{equation}\label{eq:var_vs_inside}
 \varx\mat[V|\Liney] = \varx(\muof(\ys)) =
 \begin{cases}
 \varx(\la), &\text{if $\ysT$ is located outside the polygon $\PllT$;}\\
 \varx(\lat), &\text{if $\ysT$ is located inside the polygon $\PllT$.}\\
 \end{cases}
\end{equation}
\end{corollary}

\subsection{Proof of Theorem~\ref{thm:AAshift}}\label{ssec:TREE:proof_thm_AAshift}

\itemref{AAshift_iso}: 
We start by checking that the map $\AAiso$ sending $\lalat\mapsto (\la,V=\Qlapp(\lat))$ is a homeomorphism between $\MPkn$ and $\MAkn$. 
It is clear that $\AAiso$ is continuous, and that the inverse map $\AAiso^{-1}$ sending $(\la,V) \mapsto (\la,\lat=V\cdot\Qla)$ is also continuous. It remains to show that $\AAiso(\MPkn)\subset\MAkn$ and $\AAiso^{-1}(\MAkn)\subset\MPkn$. 
Let $\la\in\lak$, $\lat\in\Gr(2,\lap)$, and $V\in\Grsupla(4,n)$ be such that $\AAiso\lalat = (\la,V)$. By \cref{prop:MPnla_AAnla_homeo}, we have $\lat\in\MPnla$ if and only if $V\in\AAnla$. 

Assume that $\lat\in\MPnla$ and $V\in\AAnla$. 
By \cref{rmk:Mand=>brlat>0}, $\brlat[i,i+1]>0$ for all $i\in\brn$. The polygon $\PllT$ is simple by \cref{lemma:Mbd=>simple}. By~\eqref{eq:wind_la_lat_vs_TURN}, it turns clockwise when $\wind(\lat)=(k+1)\pi$ and counterclockwise when $\wind(\lat)=(k-3)\pi$. Let $\mu$ be as in~\eqref{eq:V_la_mu_lat} and assume that $\mat[\la_1|\la_2]=\Id_2$. Set $\muot:=\mat[\mu_0|\mu_3]$ and let $\ys\in\Rdd$ be such that $\xM_\ys=-\muot\cdot \CMI^{-1}=\mat[-\mu_3|\mu_0]$, so that we have $-\xM_{\ys}\cdot \CMI\cdot \mat[\la_1|\la_2] = \muot$. 
For $\eps\geq0$, set $\Veps:=\begin{pmatrix}
\Id_2 & \bzero_{2\times2}\\
-\yesM\cdot \CMI & \Id_2
\end{pmatrix}\cdot V = \begin{pmatrix}
\la \\ \muepsy
\end{pmatrix}$, where $\muepsy = \mu - \yesM\cdot \CMI\cdot \la$ as in~\eqref{eq:muof(ys)_dfn}. The first two columns of $\muepsy$ are given by $\mat[\muepsyi_1|\muepsyi_2] = \eps\muot$. 

Next, we check $(\bdx_i - \yes)^2>0$ for all $i\in\brn$ and all small $\eps>0$. We have $\mat[\muepsyi_i|\muepsyi_{i+1}]=(\bdxM_i-\yesM)\cdot \CMI\cdot \mat[\la_i|\la_{i+1}]$, so by~\eqref{eq:det_xM_vs_Pmom^2}, it suffices to check $\det\mat[\muepsyi_i|\muepsyi_{i+1}]>0$ for all $i\in\brn$. For $3\leq i\leq n-1$, we have $\det\mat[\mu_i|\mu_{i+1}]>0$ so for small $\eps>0$, we get $\det\mat[\muepsyi_i|\muepsyi_{i+1}]>0$. Observe also that the result holds for $i=1$ since $\mat[\muepsyi_1|\muepsyi_2] = \eps\muot$ and $\det\mat[\mu_0|\mu_3]=\brV[0,1,2,3]>0$. 
Writing $\mat[\la_0|\la_1|\la_2|\la_3] = \begin{pmatrix}
a_0 & 1 & 0 & a_3\\ b_0 & 0 & 1 & b_3
\end{pmatrix}$, we see that $\brla<0,1>=-b_0>0$ and $\brla<2,3> = -a_3>0$. Using $-\xM_{\ys}\cdot \CMI = \muot$ and $\mat[\muepsyi_1|\muepsyi_2] = \eps\muot$, we find 
$\det\mat[\muepsyi_0|\muepsyi_1] = \det\mat[\mu_0 + \eps\muot \la_0 | \eps\mu_0] = \det\mat[\eps b_0 \mu_3|\eps\mu_0]=-\eps^2 b_0\det\mat[\mu_0|\mu_3]>0$. Similarly, we get
$\det\mat[\muepsyi_2|\muepsyi_3] = \det\mat[\eps\mu_3|\eps a_3 \mu_0]=-\eps^2a_3\det\mat[\mu_0|\mu_3]>0$. Thus, $(\bdx_i - \yes)^2>0$ for all $i\in\brn$. 

Our goal is to apply~\eqref{eq:var_vs_inside} to $\muepsy$. We calculate $\ysT$. 
 Since $\lat = \mu\cdot \Qla$ and $\mat[\mu_1|\mu_2]=\bzero_{2\times2}$, we see from~\eqref{eq:intro:Qla} that $\lat_1 = \frac1{\brla<0,1>}\mu_0$ and $\lat_2 = \frac1{\brla<2,3>}\mu_3$. Using~\eqref{eq:intro:y_to_lalat} and~\eqref{eq:xM_to_x}, we find $\ysT = \brla<0,1>\I\yt_1 - \brla<2,3>\yt_2$. Since $\mat[\la_1|\la_2]=\Id_2$, we have $\y_1 = 1$ and $\y_2 = \I$, and thus 
by~\eqref{eq:yy=Tcal_Ocal}, 
 we get $\bdxT_1-\bdxT_0 = \y_1\yt_1 = \yt_1$ and $\bdxT_2 - \bdxT_1 = \y_2\yt_2 = \I\yt_2$. Since $\bdxT_1=0$, we have $\ysT - \bdxT_1 = \brla<0,1>\I\yt_1 - \brla<2,3>\yt_2 = -\brla<0,1>\I(\bdxT_0 - \bdxT_1) + \I\brla<2,3>(\bdxT_2 - \bdxT_1)$. Recall from~\eqref{eq:angle_sum_arg_rat} that the angle of $\PllT$ at $\bdxT_1$ is given by 
$\arg_{[0,2\pi)}((\bdxT_2-\bdxT_1)/(\bdxT_0-\bdxT_1))=
\sumbT_1+\sumwT_1=\Arg(\la_1,\la_2)+\pi-\Arg(\lat_1,\lat_2)$. This angle belongs to $(\pi/2,3\pi/2)$ since $\Arg(\la_1,\la_2)=\pi/2$ and $\Arg(\lat_1,\lat_2)\in(0,\pi)$.
 Thus, the point $\eps\ysT$ lies strictly to the left of the $\bdxT_0\to\bdxT_1\to\bdxT_2$ portion of $\PllT$. 

We conclude that for small $\eps>0$, $\eps\ysT$ is located outside $\PllT$ when $\PllT$ turns clockwise (i.e., $\wind(\lat)=(k+1)\pi$) and is located inside $\PllT$ when $\PllT$ turns counterclockwise (i.e., $\wind(\lat) = (k-3)\pi$). By~\eqref{eq:var_vs_inside}, 
$\varx(\muepsy) = \varx(\la) = k-2$ when $\wind(\lat)=(k+1)\pi$ and 
$\varx(\muepsy) = \varx(\lat) = k-4$ when $\wind(\lat)=(k-3)\pi$. By \cref{rmk:Mand=>brlat>0}, $\wind(\lat)=(k+1)\pi$ if and only if $\lat\in\latk$. On the other hand, by construction, $\wind(\muepsy) = \windskip(\mu)$ for small $\eps$. We conclude that 
\begin{equation*}%
 \windskip(\mu) = (k-1)\pi 
 \quad\Longleftrightarrow\quad 
\lat\in\latk.
\end{equation*}
By~\eqref{eq:windskip_vs_varxxx}, we get $\AAiso(\MPkn)\subset\MAkn$ and $\AAiso^{-1}(\MAkn)\subset\MPkn$.

We now consider the restriction of $\AAiso$ to $\MPf$ for some $f\in\Boundkn$. By \cref{lemma:la_vs_Bounda_Boundb}, we may assume that $f\in\Bounda$. Let $\lalat\in\MPkn$ and $\AAiso\lalat=(\la,V)\in\MAkn$. 
If $\lalat\in\MPf$ then $\la\subset C\subset \latp$ for some $C\in\Grsupla(k,n)\cap\Ptp_f$. Let $\ddC:=C\cdot \Qla$. By \cref{prop:magic_homeo}, $\ddC\in\Gr(k-2,\lap)\cap\Ptp_{\ddfperm}$. 
By~\eqref{eq:Qlapp_duality_proof}, $V := \Qlapp(\lat) = (\latp \Qla)^\perp$, so 
$V\cdot \ddC^T = (\latp \Qla)^\perp\cdot (C\Qla)^T=\bzero_{4\times(k-2)}$
 because $C\subset\latp$ and thus $C\Qla\subset\latp\Qla$. Therefore, $V\subset\ddC^\perp$, so $(\la,V)\in\MAddf$. 
Conversely, suppose that $(\la,V)\in\MAddf$ for some $\ddfperm\in\BND(k-2,n)$, and let $\ddC\in\Ptp_{\ddfperm}$ be such that $V\subset \ddC^\perp$. 
In particular, $\ddC\in\Gr(k-2,\lap)\cap\Ptp_{\ddfperm}$, and by \cref{lemma:la_vs_Bounda_Boundb}, $\ddfperm\in\Bounddb$. Set $C:=\Qlapp(\ddC)$. By \cref{prop:magic_homeo}, $C\in\Grsupla(k,n)\cap\Ptp_f$. We claim that $C\subset\latp$. Indeed, since $\lat = V\cdot \Qla=\mu\cdot \Qla$, the matrix $C\cdot \lat^T = C\cdot \Qla V^T = \ddC\cdot V^T$ is zero. 
 Thus, $\lalat\in\MPf$. This finishes the proof of part~\itemref{AAshift_iso}.

\itemref{AAshift_forg}:
We first show that the map $\AAforg:\MAkn\to\AAkn$ is surjective. 
Let $V\in\AAkn$. We would like to find $\la\in\lak$ such that $(\la,V)\in\MAkn$. Let $\mu=\mat[0|0|\mu_3|\cdots|\mu_n]:=\LVVp\cdot V$ be as in~\eqref{eq:V_la_mu_lat}. Using $\GL_4(\R)$-action, we can represent $V=\begin{pmatrix}
 \la' \\ \mu
\end{pmatrix}$ for some $\la'$ satisfying $\mat[\la'_1|\la'_2] = \Id_2$, so that $\LVVp=\mat[\bzero_{2\times2}|\Id_2]$. 
For $\eps>0$, let 
$\laieps_0=\mu_0$, $\laieps_1 = 2\eps\mu_0+\eps\mu_3$, $\laieps_2 = \eps\mu_0+2\eps\mu_3$, and $\laieps_3 = \mu_3$, where $\mu_0=(-1)^{k-1}\mu_n$ as before. 
Since $\brV[0,1,2,3]>0$, the matrix $\mat[V_0|V_1|V_2|V_3]$ is invertible. Thus, for each $\eps>0$, we can choose a $2\times 4$ matrix $A_\eps$ such that $\mat[\laieps_0|\laieps_1|\laieps_2|\laieps_3]=A_\eps\cdot \mat[V_0|V_1|V_2|V_3]$. 

By construction, since $\brlae<0,3>=\det\mat[\mu_0|\mu_3]=\brV[0,1,2,3]>0$, we have
\begin{equation}\label{eq:brlae_Arg_lae}
 \brlae<0,1>,\brlae<1,2>,\brlae<2,3> >0 \quad\text{and}\quad
 \Arg(\laieps_0,\laieps_1)+\Arg(\laieps_1,\laieps_2)+\Arg(\laieps_2,\laieps_3)=\Arg(\mu_0,\mu_3)
\end{equation} 
for all $\eps>0$, where $\Arg:=\Arg_{(-\pi,\pi]}$ as in~\eqref{eq:intro:wind}. For $3\leq i\leq n-1$, 
 we get $\det\mat[\mu_i|\mu_{i+1}] = \brV[1,2,i,i+1] > 0$. Since $\lim_{\eps\to0}A_\eps = \LVVp = \mat[\bzero_{2\times2}|\Id_2]$, we get $\lim_{\eps\to0}\brlae<i,i+1> = \det\mat[\mu_i|\mu_{i+1}]>0$. Thus, $\brlae<i,i+1> >0$ for small $\eps>0$. Since $\wind(\laeps)$ is locally constant and equals $\windskip(\mu)$ in the $\eps\to0$ limit in view of~\eqref{eq:brlae_Arg_lae} and~\eqref{eq:windskip_dfn}, 
we have $\wind(\laeps)=\windskip(\mu)=(k-1)\pi$. Thus, for small $\eps>0$, we have
$\laeps\in\lak$. Since $\laeps\subset V$, we get $(\laeps,V)\in\MAkn$. This shows that $\AAforg:\MAkn\to\AAkn$ is surjective. 

To show that $\AAforg$ is an open map, let $(\la,V)\in\MAkn$. Let $A$ be a $2\times 4$ matrix such that $\la = A \cdot V$. Let $U\subset\MAkn$ be an open neighborhood of $(\la,V)$. Suppose for contradiction that $\AAforg(U)$ does not contain an open neighborhood of $V$ in $\AAkn$. Thus, there exists a sequence $\Veps\in\AAkn\setminus\AAforg(U)$ converging to $V$ as $\eps\to0$. Let $\laeps:=A\cdot \Veps$. For small $\eps>0$, we have $\laeps\in\Gr(2,\Veps)$ and $\laeps\to\la$ as $\eps\to0$. 
For each $i\in\brn$, since $\brla<i,i+1> >0$, we have $\brlae<i,i+1> >0$ for small $\eps>0$. Since $\wind(\laeps)$ is locally constant and is well defined for $\laeps$ and for $\la$ (i.e., none of the columns of $\laeps$ and $\la$ are zero and the consecutive columns are not antiparallel), $\wind(\laeps)=\wind(\la)=(k-1)\pi$. It follows that $(\laeps,\Veps)\in\MAkn$ converges to $(\la,V)$ as $\eps\to0$, so $\Veps=\AAforg(\laeps,\Veps)\in\AAforg(U)$ for small $\eps>0$, a contradiction. Thus, $\AAforg$ is an open map.

The statement that $\AAforg$ restricts to a surjective open map $\MAddf\to\AAddf$
 for all $\ddfperm\in\BND(k-2,n)$ follows trivially since the condition ($V\subset\ddC^\perp$ for some $\ddC\in\Ptp_{\ddfperm}$) that cuts out the respective subsets of $\MAkn$ and $\AAkn$ only involves $V$, and so for any given $V\in\AAddf$, any choice of $\la\in\lak$ such that $\la\subset V$ automatically gives rise to a point $(\la,V)\in\MAddf$. See also the proof of \cref{lemma:RX_triple} below.
\qed

\subsection{Properties of \mtilings}
Our next goal is to deduce some further BCFW tiling results from \cref{thm:f_triang,thm:AAshift}. 
We start by proving two abstract topological results concerning \mtilings. 

\begin{proposition}\label{lemma:RX_triple}
For $\s=2,3$ and $\RG\in\RGbf$, let $\RelS\subset\RXS\times\RYS$, $\RXoS_\RG\subset\RXS$, $\ReloS_\RG=\RelS\cap(\RXoS_\RG\times\RYS)$, and $\RYoS_\RG=\RprojS(\ReloS_\RG)$ be as in \cref{dfn:BCFW:tiling_amb}.
Assume that $\RXB=\RXC$ and $\RXoB_\RG=\RXoC_\RG$ for all $\RG\in\RGbf$.
 Suppose that we have a commutative diagram
\begin{equation}\label{eq:TREE:comm_diag}
\begin{tikzcd}[column sep=3em]
\RelB \arrow[r,"{(\id_{\RXB},\Rpi)}"] \arrow[d,"\RprojB"] 
\drar[phantom, "\ \square",pos=0.47]
& \RelC \arrow[d,"\RprojC"] 
\\
\RYB\arrow[r,"\Rpi",twoheadrightarrow] 
& \RYC.
\end{tikzcd} 
\end{equation}
Suppose in addition that the map $\Rpi:\RYB\to\RYC$ is continuous, surjective, and open, and that the diagram is \emph{Cartesian}, i.e., 
$\RelB = \{(x,y)\in \RXB\times\RYB\mid (x,\Rpi(y))\in\RelC\}$. 
Then 
\begin{equation}\label{eq:RtilingS}
 \text{$\{\RYoB_\RG\mid \RG\in\RGbf\}$ is \amtilingB of $\RYB$}
 \quad\Longleftrightarrow\quad
 \text{$\{\RYoC_\RG\mid \RG\in\RGbf\}$ is \amtilingC of $\RYC$}.
\end{equation}
\end{proposition}

\begin{proof}
Restricting the Cartesian diagram condition $\RelB = (\id_{\RXB},\Rpi)^{-1}(\RelC)$ to each $\RG\in\RGbf$, we get
\begin{equation}\label{eq:ReloB_vs_ReloC_Cartesian}
 \ReloB_\RG %
 = \{(x,y)\in \RXoB_\RG\times\RYB\mid (x,\Rpi(y))\in\ReloC_\RG\},
 \quad\text{and}\quad
 \RYoB_\RG = \Rpi^{-1}(\RYoC_\RG)
 \quad\text{for all $\RG\in\RGbf$}.
\end{equation}
We claim that for each $\RG\in\RGbf$, the diagram~\eqref{eq:TREE:comm_diag} restricts to
\begin{equation}\label{eq:TREE:comm_diag_res}
\begin{tikzcd}[column sep=3em]
 \ReloB_\RG \arrow[r,"{(\id_{\RXB},\Rpi)}",twoheadrightarrow] \arrow[d,"\RprojoB_\RG",twoheadrightarrow]
\drar[phantom, "\ \square"]
& \ReloC_\RG \arrow[d,"\RprojoC_\RG",twoheadrightarrow] 
\\
 \RYoB_\RG\arrow[r,"\Rpi",twoheadrightarrow] 
& \RYoC_\RG,
\end{tikzcd} 
\end{equation}
with all four maps surjective. Indeed, the maps $\RprojoS_\RG:\ReloS_\RG\to\RYoS_\RG$ are surjective for $\s=2,3$ by definition. By~\eqref{eq:ReloB_vs_ReloC_Cartesian}, $(\id_{\RXB},\Rpi)(\ReloB_\RG)\subset \ReloC_\RG$, and since $\Rpi:\RYB\to\RYC$ is surjective, we get $(\id_{\RXB},\Rpi)(\ReloB_\RG) = \ReloC_\RG$. Similarly, we obtain 
$\Rpi(\RYoB_\RG)=\RYoC_\RG$. 

Let us denote parts~\itemref{Rtiling1}, \itemref{Rtiling2}, \itemref{Rtiling3} of \cref{dfn:BCFW:tiling_amb} when applied to tiles $\{\RYoB_\RG\mid \RG\in\RGbf\}$ (resp., $\{\RYoC_\RG\mid \RG\in\RGbf\}$) by \Rtila2, \Rtilb2, \Rtilc2 (resp., \Rtila3, \Rtilb3, \Rtilc3). 

\Rtila3 $\Longrightarrow$ \Rtila2: This follows since homeomorphisms are stable under pullback. Explicitly, suppose that $\RprojoC_\RG:\ReloC_\RG\xrasim\RYoC_\RG$ is a homeomorphism. 
Given $y_2\in\RYoB_\RG$, set $y_3:=\Rpi(y_2)\in\RYoC_\RG$ and $(x,y_3) := \RprojoC_\RG^{-1}(y_3)\in\ReloC_\RG$, 
and let $q_{\RYoB_\RG}:\RYoB_\RG\to\ReloB_\RG$ be the map sending $y_2\mapsto (x,y_2)$. 
By~\eqref{eq:ReloB_vs_ReloC_Cartesian}, since $(x,y_3)\in\ReloC_\RG$, we have $(x,y_2)\in\ReloB_\RG$. 
Thus,
$q_{\RYoB_\RG}$ is a continuous inverse of $\RprojoB_\RG:\ReloB_\RG\to\RYoB_\RG$. 

\Rtila2 $\Longrightarrow$ \Rtila3: 
Since $\Rpi:\RYoB_\RG\to\RYoC_\RG$ is surjective by~\eqref{eq:TREE:comm_diag_res}, given $y_3\in\RYoC_\RG$, we can pick an arbitrary preimage $y_2\in\Rpi^{-1}(y_3)$. By~\eqref{eq:ReloB_vs_ReloC_Cartesian}, $y_2\in\RYoB_\RG$, so let $(x,y_2):=\RprojoB_\RG^{-1}(y_2)$. If we had picked a different preimage $y_2'\in\Rpi^{-1}(y_3)$ and set $(x',y_2'):=\RprojoB_\RG^{-1}(y_2')\in\ReloB_\RG$ then by~\eqref{eq:ReloB_vs_ReloC_Cartesian}, we would have $(x',y_2)\in\ReloB_\RG$, and since $y_2$ has a unique preimage under $\RprojoB_\RG$, we must have $x=x'$. Letting $q_{\RYoC_\RG}(y_3):=(x,y_3)\in\ReloC_\RG$, we see that $q_{\RYoC_\RG}:\RYoC_\RG\to\ReloC_\RG$ is a set-theoretic inverse of $\RprojoC_\RG:\ReloC_\RG\to\RYoC_\RG$.
To check that $q_{\RYoC_\RG}$ is continuous, let $U_3\subset \ReloC_\RG$ be open, $U_2:=(\id_{\RXB},\Rpi)^{-1}(U_3)$, $V_2:=\RprojoB_\RG(U_2)$, and $V_3:=\Rpi(V_2)$. 
Since $\Rpi$ is continuous, $U_2\subset\ReloB_\RG$ is open. 
Since $\RprojoB_\RG:\ReloB_\RG\xrasim\RYoB_\RG$ is a homeomorphism and $\Rpi$ is open, $V_2\subset\RYoB_\RG$ and $V_3\subset\RYoC_\RG$ are open. 
Using the surjectivity of $(\id_{\RXB},\Rpi)$ in~\eqref{eq:TREE:comm_diag_res}
and the uniqueness of lifts under $\RprojoB_\RG$, we get
$V_3=q_{\RYoC_\RG}^{-1}(U_3)$.
It follows that $q_{\RYoC_\RG}$ is continuous. 

\Rtilb3 $\Longrightarrow$ \Rtilb2: If $(\RYoC_\RG)_{\RG\in\RGbf}$ are disjoint then by~\eqref{eq:ReloB_vs_ReloC_Cartesian}, 
so are their preimages $\RYoB_\RG=\Rpi^{-1}(\RYoC_\RG)$.

\Rtilb2 $\Longrightarrow$ \Rtilb3: Suppose that $y_3\in\RYoC_\RG\cap\RYoC_\RGg$ for distinct $\RG,\RGg\in\RGbf$. Since $\Rpi$ is surjective, $\Rpi^{-1}(y_3)\neq\emptyset$, and by~\eqref{eq:ReloB_vs_ReloC_Cartesian}, 
$\Rpi^{-1}(y_3)\subset\RYoB_\RG\cap\RYoB_\RGg$, contradicting \Rtilb2.

\Rtilc3 $\Longrightarrow$ \Rtilc2: If $\bigcup_{\RG\in\RGbf} \RYoB_\RG$ is not dense in $\RYB$ then there exists a nonempty open subset $U_2\subset\RYB\setminus(\bigcup_{\RG\in\RGbf} \RYoB_\RG)$. Since $\Rpi$ is open, $\Rpi(U_2)\subset \RYC$ is open, and by~\eqref{eq:ReloB_vs_ReloC_Cartesian}, $\Rpi(U_2)\subset\RYC\setminus(\bigcup_{\RG\in\RGbf} \RYoC_\RG)$.

\Rtilc2 $\Longrightarrow$ \Rtilc3: Given nonempty open $U_3\subset \RYC\setminus(\bigcup_{\RG\in\RGbf} \RYoC_\RG)$, by~\eqref{eq:ReloB_vs_ReloC_Cartesian}, 
$\Rpi^{-1}(U_3)\subset\RYB\setminus\bigcup_{\RG\in\RGbf} \RYoB_\RG$. %
 Since $\Rpi$ is continuous and surjective, $\Rpi^{-1}(U_3)\subset\RYB$ is open and nonempty.
\end{proof}

Next, we formulate an abstract version of \cref{dfn:intro:triang}. 
\begin{definition}[Tiling]\label{dfn:BCFW:tiling}
Let $\RPhi: \RX\to \W$ be a continuous map. %
Let $\RGbf$ be a finite set and let $\{\RXo_\RG\mid\RG\in\RGbf\}$ be a collection of subsets of $\RX$. 
For $\RG\in\RGbf$, define a \emph{tile} $\Wo_\RG:=\RPhi(\RXo_\RG)$. 
We say that the tiles $\{\Wo_\RG\mid \RG\in\RGbf\}$ form a \emph{tiling} of $\W$ if the following conditions are satisfied.
\begin{enumerate}[label=(\alph*)]
\item\label{tiling1} \emph{Injectivity:} For each $\RG\in\RGbf$, the map $\RPhi$ restricts to a homeomorphism 
$\RXo_\RG\xrasim\Wo_\RG$.%
\item\label{tiling2} \emph{Disjointness:} The tiles $\{\Wo_\RG\mid \RG\in\RGbf\}$ are pairwise disjoint.
\item\label{tiling3} \emph{Surjectivity:} The union $\bigsqcup_{\RG\in\RGbf} \Wo_\RG$ is dense in $\W$. 
\end{enumerate}
\end{definition}

\begin{proposition}\label{lemma:Rtiling_to_tiling}
Let $\RX\subset\RXcl$, $\RY\subset\RYcl$, $\W=\Wcl\cap\RY$, and $\Rel=\Relcl\cap(\RX\times\RY)$ for some $\Wcl\subset\RYcl$ and $\Relcl\subset\RXcl\times\RYcl$. 
Suppose that the tiles $\{\RYo_\RG\mid \RG\in\RGbf\}$ form an \emph{\mtiling} of $\RY$.
Suppose in addition that $\RPhi: \RX\to \W$ extends to a continuous map $\RPhicl:\RXcl\to\Wcl$ 
 whose graph 
\begin{equation}\label{eq:RGraph}
\RGraph_{\RPhicl}:=\{(x,\RPhicl(x))\mid x\in\RXcl\} 
\quad\text{satisfies}\quad
 \RGraph_{\RPhicl} = \Relcl\cap (\RXcl\times \Wcl). 
\end{equation}
Assume that $\Relcl\subset\RXcl\times\RYcl$ is a closed subset and that the closure $\RXocl_\RG$ of $\RXo_\RG$ in $\RXcl$ is compact for each $\RG\in\RGbf$. 
Then the tiles $\{\Wo_\RG\mid \RG\in\RGbf\}$ form a tiling of $\W$.
\end{proposition}
\begin{proof}
For $\RG\in\RGbf$, let $\RGraph_{\RPhi|_{\RXo_\RG}}:=\{(x,\RPhi(x))\mid x\in\RXo_\RG\}$. We claim that
\begin{equation}\label{eq:RWo=RYo_cap_RW}
\RGraph_{\RPhi|_{\RXo_\RG}} = \Relo_\RG\cap (\RX\times \W)
\quad\text{and}\quad
 \Wo_\RG = \RYo_\RG\cap \W.
\end{equation}
Indeed, using $\Relo_\RG = \Rel\cap(\RXo_\RG\times\RY)$, $\Rel = \Relcl\cap(\RX\times\RY)$,~\eqref{eq:RGraph}, and $\RPhi = \RPhicl|_{\RX}$, we get
\begin{equation*}%
\Relo_\RG\cap (\RX\times \W) 
= \Rel\cap(\RXo_\RG\times\W) 
= \Relcl\cap(\RXo_\RG\times\W)
= \RGraph_{\RPhicl}\cap(\RXo_\RG\times\W)
= \RGraph_{\RPhi|_{\RXo_\RG}}, 
\end{equation*} 
which gives the first identity in~\eqref{eq:RWo=RYo_cap_RW}. The second identity follows from the first since
\begin{equation*}%
 \Wo_\RG 
= \RPhi(\RXo_\RG) 
= \RprojY(\RGraph_{\RPhi|_{\RXo_\RG}})
= \RprojY(\Relo_\RG\cap (\RX\times \W))
= \RprojY(\Relo_\RG)\cap \W
= \RYo_\RG\cap \W.
\end{equation*}

By \crefi{dfn:BCFW:tiling_amb}{Rtiling1}, $\RprojoY_\RG:\Relo_\RG\xrasim\RYo_\RG$ is a homeomorphism, so restricting $\RprojoY_\RG^{-1}$ to $\Wo_\RG = \RYo_\RG\cap \W$, we obtain a continuous map $\Wo_\RG\to \RGraph_{\RPhi|_{\RXo_\RG}}=\Relo_\RG\cap (\RX\times \W)$. Composing this map with the projection $\RprojX$ gives a continuous inverse of $\RPhi|_{\RXo_\RG}$. This shows part~\itemref{tiling1} of \cref{dfn:BCFW:tiling}. Part~\itemref{tiling2} follows trivially: since the tiles $\RYo_\RG$ are disjoint, by~\eqref{eq:RWo=RYo_cap_RW}, so are the tiles $\Wo_\RG$. 

We show part~\itemref{tiling3}. Let $w\in\W$. We have $\W=\RY\cap\Wcl$ and the union $\bigsqcup_{\RG\in\RGbf}\RYo_\RG$ is dense in $\RY$. Thus, there exists $\RG\in\RGbf$ and a sequence $y_1,y_2,\dots$ of elements of $\RYo_\RG$ converging to $w$. Since $\RYo_\RG=\RprojY(\Relo_\RG)$, for each $i=1,2,\dots$, there exists $x_i\in\RXo_\RG$ such that $(x_i,y_i)\in\Rel$. Since $\RXocl_\RG$ is compact, after passing to a subsequence, we may assume that the sequence $x_1,x_2,\dots$ converges to some limit $x\in\RXocl_\RG$. Since $\Relcl\subset\RXcl\times\RYcl$ is closed, we get $(x,w)\in\Relcl$. By~\eqref{eq:RGraph}, we find $w=\RPhicl(x)$. Letting $w_i:=\RPhi(x_i)$, since $\RPhicl$ is continuous, we see that $\lim_{i\to\infty}w_i = w$. Since $x_i\in\RXo_\RG$, we get $w_i\in\Wo_\RG$. Thus, $\bigsqcup_{\RG\in\RGbf}\Wo_\RG$ is dense in $\W$.
\end{proof}

\subsection{T-duality for ambient amplituhedron tilings} 
We specialize the results in the previous subsection. In the notation of \cref{dfn:Grfsep,dfn:Grfind}
and \cref{lemma:RX_triple}, we set 
\begin{align}
\label{eq:RXA_dfn}
\RXA&:=\Grfsep(k,n), & \RYA&:=\MPkntree, & \RelA&:=\{(C,\la,\lat)\in\RXA\times\RYA\mid \la\subset C\subset\latp\};\\
\label{eq:RXB_dfn}
\RXB&:=\Grddfsep(k-2,n), & \RYB&:=\MAkn, & \RelB&:=\{(\ddC,\la,V)\in\RXB\times\RYB\mid V\subset \ddC^\perp\};\\
\label{eq:RXC_dfn}
\RXC&:=\Grddfsep(k-2,n), & \RYC&:=\AAkntree, & \RelC&:=\{(\ddC,V)\in\RXC\times\RYC\mid V\subset \ddC^\perp\}.
\end{align}
For $\fap\in\BNDfsep(k,n)$,
 we set
\begin{equation*}%
\RXoA_{\fap}:=\Ptp_{\fap},\quad \RXoB_{\ddfperm}=\RXoC_{\ddfperm}:=\Ptp_{\ddfperm};\quad
\RYoA_{\fap}:=\MPf, \quad \RYoB_{\ddfperm}:=\MAddf, \quad \RYoC_{\ddfperm}:=\AAddf. 
\end{equation*}

\begin{remark}
Similarly to \cref{lemma:lalat_MP=>C_fsep}, 
it follows from~\eqref{eq:altp_Delta} that if $V\in\AAkntree$ and $\ddC\in\Grtnn(k-2,n)$ satisfy $V\subset\ddC^\perp$ then we must have $\ddC\in\Grddfsep(k-2,n)$.
\end{remark}

\begin{corollary}[T-duality for \mtilings of ambient amplituhedra]\label{lemma:TREE:amb_tiling_equivalence}
Let $\BBfkn\subset\BNDfsep(k,n)$ and $\ddBBfkn:=\{\ddfperm\mid\fap\in\BBfkn\}$. 
The tiles $\{\MPf\mid \fap\in\BBfkn\}$ form \amtilingA of $\MPkn$ if and only if the tiles $\{\AAddf\mid \ddfperm\in\ddBBfkn\}$ form \amtilingC of $\AAkn$. 
\end{corollary}
\begin{proof}
We have $\RXoA_{\fap}\subset\RXA$ and $\RXoB_{\ddfperm}=\RXoC_{\ddfperm}\subset\RXB=\RXC$ for all $\fap\in\BBfkn\subset\BNDfsep(k,n)$. 
Let $\Rphibot$ and $\Rpi$ be as in \cref{thm:AAshift}. 
Let $\Rphitop:\RelA\to\RelB$ be given by $\Rphitop(C,\la,\lat):=(C\cdot \Qla,\la,\Qlapp(\lat))$. As explained in the proof of \cref{thm:AAshift}, $\Rphitop$ is a homeomorphism with inverse $(\ddC,\la,V)\mapsto(\Qlapp(\ddC),\la,V\cdot \Qla)$. By \cref{prop:magic_homeo}, $\Rphitop$ restricts to a homeomorphism $\ReloA_{\fap}\xrasim\ReloB_{\ddfperm}$ for each $\fap\in\BNDfsep(k,n)$. 
Thus, $\{\MPf\mid \fap\in\BBfkn\}$ is \amtilingA of $\MPkn$ if and only if $\{\MAddf\mid \ddfperm\in\ddBBfkn\}$ is \amtilingB of $\MAkn$. 

To relate \mtilingsB of $\MAkn$ to \mtilingsC of $\AAkn$, we check the assumptions of \cref{lemma:RX_triple}. 
By construction, the diagram~\eqref{eq:TREE:comm_diag} commutes. 
By \cref{thm:AAshift}, $\Rpi$ is open and surjective. 
Finally, the Cartesian square condition $\RelB = \{(x,y)\in \RXB\times\RYB\mid (x,\Rpi(y))\in\RelC\}$ follows from the fact that the relation $V\subset\ddC^\perp$ in~\eqref{eq:RXB_dfn}--\eqref{eq:RXC_dfn} does not involve $\la$. 
\end{proof}

Recall from \cref{rmk:BCFW:T_emb_unique} that $\BCFWfkn\subset\BNDfsep(k,n)$. 
Combining \cref{thm:f_triang} with \cref{lemma:TREE:amb_tiling_equivalence}, we obtain the following result.
\begin{corollary}\label{lemma:ddf_triang}
 The tiles $\{\AAddf\mid \fap\in\BCFWfkn\}$ form \amtilingC of $\AAkn$. 
\end{corollary}

\subsection{Sign flip and linear projection amplituhedra}
We would like to compare ambient amplituhedra $\MPkn,\AAkn$ to \emph{sign flip} amplituhedra $\MomLLamb,\AZamb$ of~\cite{AHTT,HZ_notes,DFLP} and to \emph{linear projection} amplituhedra $\MomLLproj,\AZproj$ of~\cite{AHT,DFLP}. %
We will deduce the BCFW tiling results for these amplituhedra from \cref{lemma:Rtiling_to_tiling}. 

We start by defining the amplituhedra $\MomLLamb,\AZamb,\MomLLproj,\AZproj$. 
Let $\LaLat\in\LaLaimmnn$ (\cref{dfn:Ttauknij}). Following~\cite{HZ_notes,DFLP}, we set
\begin{equation*}%
 \MomLLproj:=\PhiLL(\Grtnn(k,n)) 
 \quad\text{and}\quad
\MomLLamb:=\{\lalat\in\MPkn\mid \la\subset\La\text{ and }\lat\subset\Lat\}.
\end{equation*}
Similarly, for 
 $Z\in\Grtp(k+2,n)$, the \emph{\mta map} is given by
\begin{equation}\label{eq:PsiZ_dfn}
 \PsiZ:\Grtnn(k-2,n)\to\Gr(4,n),\quad \ddC\mapsto \ddC^\perp\cap Z.
\end{equation}
\begin{remark}\label{rmk:dim_ddC_cap_Z=4}
We have $\dim(\ddC^\perp\cap Z)=4$ for all $\ddC\in\Grtnn(k-2,n)$; see e.g.~\cite[Lemma~3.10]{KW}.
\end{remark}
We define
\begin{equation*}%
 \AZproj:=\PsiZ(\Grtnn(k-2,n)) \quad\text{and}\quad
 \AZamb:=\{V\in\AAkn\mid V\subset Z\}.
\end{equation*}
 For $\fap\in\Bounda$, we let
\begin{equation*}%
 \MomLLprojf:=\PhiLL(\Ptp_{\fap}) \quad\text{and}\quad \AZprojddf:=\PsiZ(\Ptp_{\ddfperm}).
\end{equation*}

\begin{lemma}%
\label{lemma:TREE:proj_subset_amb}
We have inclusions
\begin{equation}\label{eq:TREE:proj_subset_cl_amb}
\MomLLproj\subset\clMomLLamb
 \quad\text{and}\quad 
 \AZproj\subset\clAZamb
\end{equation}
for all $\LaLat\in\LaLaimmnn$ and $Z\in\Grtp(k+2,n)$, 
where $\overline\cdot$ denotes closure in the Hausdorff topology on $\lalats$ and $\Gr(4,n)$, respectively.
\end{lemma}
\begin{proof}
Given $C\in\Grtnn(k,n)$ and $\LaLat\in\LaLaimmnn$, we may approximate $C$ by a sequence of $C'\in\Grtp(k,n)$. 
Since $2\leq k\leq n-2$, it follows from \cref{lemma:TREE:fullysep_vs_rhoij} that 
 $\fap_{C'}=\fap_{k,n}$ is \fullysep. By 
\cref{lemma:fullysep=>simple_and_Mbd}, 
 $(\la',\lat'):=\PhiLL(C')$ is \Mdash positive. By \cref{prop:momLL_basic}, $(\la',\lat')\in\lalak$. Thus, $(\la',\lat')\in\MPkn$, so $(\la',\lat')\in\MomLLamb$. Since $C'$ converges to $C$, $(\la',\lat')$ converges to $(\la,\lat):=\PhiLL(C)$ because $\PhiLL$ is continuous on $\Grtnn(k,n)$, so $\lalat\in\clMomLLamb$. This shows the first inclusion in~\eqref{eq:TREE:proj_subset_cl_amb}. The second inclusion is well known; see e.g.~\cite[Section~5.4]{AHTT}.
\end{proof}
Ever since the work of~\cite{AHTT,DFLP}, it has been expected that both inclusions in~\eqref{eq:TREE:proj_subset_cl_amb} are in fact equalities. We deduce this from the BCFW tiling results for $\MomLLproj$ and $\AZproj$.

\begin{theorem}[BCFW tilings of amplituhedra]\ \label{lemma:proj_tiling}
\begin{enumerate}[label=(\arabic*)]
\item\label{proj_tiling1} The tiles $\{\MomLLprojf\mid \fap\in\BCFWfkn\}$ form a tiling of $\MomLLproj$ for all $\LaLat\in\LaLaimmnn$.
\item\label{proj_tiling2} The tiles $\{\AZprojddf\mid \fap\in\BCFWfkn\}$ form a tiling of $\AZproj$ for all $Z\in\Grtp(k+2,n)$. 
\item\label{proj_tiling3} The linear projection and sign flip definitions of the amplituhedron agree, i.e., 
\begin{equation}\label{eq:TREE:proj=cl_amb}
\MomLLproj=\clMomLLamb
 \quad\text{and}\quad 
 \AZproj=\clAZamb
\end{equation}
for all $\LaLat\in\LaLaimmnn$ and $Z\in\Grtp(k+2,n)$. 
\end{enumerate}
\end{theorem}
\begin{proof}
Let $\LaLat\in\LaLaimmnn$ and $Z\in\Grtp(k+2,n)$. 
Let $\RXA,\RYA,\RelA$ and $\RXC,\RYC,\RelC$ be given by~\eqref{eq:RXA_dfn} and~\eqref{eq:RXC_dfn}. 
By \cref{thm:f_triang,lemma:ddf_triang}, the tiles $\{\RYoA_{\fap}\mid\fap\in\BCFWfkn\}$ (resp., $\{\RYoC_{\ddfperm}\mid \fap\in\BCFWfkn\}$) form \amtilingA of $\RYA$ (resp., \amtilingC of $\RYC$). 
Our goal is to apply \cref{lemma:Rtiling_to_tiling} to 
$\WA:=\MomLLamb$ and $\WC:=\AZamb=\{V\in\RYC\mid V\subset Z\}$.

Let $\RXAcl:=\Grtnn(k,n)$, $\RYAcl:=\lalats$, $\WAcl:=\{\lalat\in\RYAcl\mid \la\subset\La\text{ and }\lat\subset\Lat\}$, and 
$\RelAcl:=\{(C,\la,\lat)\in\RXAcl\times\RYAcl\mid \la\subset C\subset\latp\}$. Thus, $\WA=\WAcl\cap\RYA$, $\RelA=\RelAcl\cap(\RXA\times\RYA)$, and the subset $\RelAcl\subset \RXAcl\times\RYAcl$ is closed. 
The map $\RPhiAcl:=\PhiLL:\RXAcl\to\WAcl$ is continuous by \crefi{prop:momLL_basic}{mom1}. 
Since $\LaLat\in\LaLaimmnn$, by \cref{lemma:fullysep=>simple_and_Mbd}, $\RPhiA(\RXA)\subset\WA$, where $\RPhiA:=\RPhiAcl|_{\RXA}$. 
To check~\eqref{eq:RGraph}, observe that by construction, if $\lalat=\RPhiAcl(C)$ for $C\in\RXAcl$ then $\la\subset C\subset\latp$, i.e., $(C,\la,\lat)\in\RelAcl$. Thus, $\RGraph_{\RPhiAcl} \subset \RelAcl\cap (\RXAcl\times \WAcl)$. Conversely, if $(C,\la,\lat)\in \RelAcl\cap (\RXAcl\times \WAcl)$ then we have $\la\subset C \subset\latp$, $\la\subset\La$, and $\lat\subset\Lat$. By \crefi{prop:momLL_basic}{mom1}, both intersections $C\cap \La$ and $C^\perp\cap\Lat$ are $2$-dimensional, so $\lalat=\PhiLL(C)$, and thus $(C,\la,\lat)\in\RGraph_{\RPhiAcl}$.
For $\fap\in\BCFWfkn$, the closure $\RXoAcl_{\fap} = \Povtnn_{\fap}$ is compact. 
Thus, part~\itemref{proj_tiling1} of the theorem follows from \cref{lemma:Rtiling_to_tiling}.

We show part~\itemref{proj_tiling2}. Denote $\RXCcl:=\Grtnn(k-2,n)$, $\RYCcl:=\Gr(4,n)$, $\WCcl:=\{V\in\RYCcl\mid V\subset Z\}$, $\RelCcl:=\{(\ddC,V)\in\RXCcl\times\RYCcl\mid V\subset \ddC^\perp\}$. Let $\RPhiCcl=\PsiZ$ be given by~\eqref{eq:PsiZ_dfn} and let $\RPhiC:=\RPhiCcl|_{\RXC}$. %
 To see that $\RPhiC(\RXC)\subset\WC$, observe that for $i+2\leq j\leq i+n-2$, $\Jij:=\{i,i+1,j,j+1\}\subset\brn$, $\ddC\in\Grtnn(k-2,n)$, $Z\in\Grtp(k+2,n)$, and $V:=\PsiZ(\ddC)$, by~\cite[Equation~(3.11)]{KW},
\begin{equation}\label{eq:KW_eq}
 \brV[i,i+1,j,j+1] = \sum_{L\in{\Jijc\choose k-2}}\Delta_L(\ddC) \Delta_{L\sqcup \Jij}(Z).
\end{equation}
 In particular, $\brV[i,i+1,j,j+1]>0$ if and only if $\rank\ddC_{\Jijc}=k-2$. By~\eqref{eq:TREE:proj_subset_cl_amb}, $\varxxx(V)=k-2$. Thus, when $\ddC\in\RXC=\Grddfsep(k-2,n)$, we have $\RPhiC(\ddC)\in\RYC=\AAkntree$ and thus $\RPhiC(\RXC)\subset\WC$.

By \cref{rmk:dim_ddC_cap_Z=4}, $\RPhiCcl$ is continuous. 
We check~\eqref{eq:RGraph}. It follows from the definitions that 
$\RGraph_{\RPhiCcl}\subset\RelCcl\cap(\RXCcl\times\WCcl)$. Conversely, for $(\ddC,V)\in\RelCcl\cap(\RXCcl\times\WCcl)$, we have $V\subset Z\cap \ddC^\perp$, so by \cref{rmk:dim_ddC_cap_Z=4}, $V=\PsiZ(\ddC)$. Thus, $\RGraph_{\RPhiCcl}=\RelCcl\cap(\RXCcl\times\WCcl)$, so
 part~\itemref{proj_tiling2} of the theorem follows from \cref{lemma:Rtiling_to_tiling}.

By parts~\itemref{proj_tiling1}--\itemref{proj_tiling2} of the theorem, the subsets $\RPhiA(\RXA)\subset\WA$ and $\RPhiC(\RXC)\subset\WC$ are dense. By~\eqref{eq:TREE:proj_subset_cl_amb}, 
$\RPhiA(\RXA)=\PhiLL(\Grfsep(k,n))\subset \MomLLproj\subset\clMomLLamb=\overline{\WA}$. Since $\RPhiA(\RXA)$ is dense in $\WA$, $\MomLLproj$ is dense in $\clMomLLamb$. Since $\MomLLproj=\PhiLL(\Grtnn(k,n))$ is a continuous image of a compact set, it is closed, so we get $\MomLLproj=\clMomLLamb$. The proof of the second equality $\AZproj=\clAZamb$ in~\eqref{eq:TREE:proj=cl_amb} is similar.
\end{proof}

\subsection{Discussion} We compare \cref{lemma:proj_tiling} to some of the previous results in a series of remarks.

\begin{remark}\label{rmk:msgen_exist}
By~\eqref{eq:RWo=RYo_cap_RW}, $\WoA_{\fap}:=\RPhiA(\RXoA_{\fap})\subset\RYoA_\fap$ and $\WoC_{\ddfperm}:=\RPhiC(\RXoC_{\ddfperm})\subset\RYoC_\ddfperm$. By~\eqref{eq:lalappf_dfn} and~\eqref{eq:TREE:AAddf_dfn}, we get 
$\MomLLprojf\subset\MPf$ and $\AZprojddf\subset\AAddf$
for all $\fap\in\BCFWfkn$. 
In particular, 
 $\MPf\neq\emptyset$ for each $\fap\in\BCFWfkn$. Thus, the set of \msgen elements $\lalat\in\lalats$ is nonempty; cf. \cref{rmk:msgen}.
\end{remark}

\begin{remark}\label{rmk:converse_main_diff}
It is not clear to us how to show conversely that any tiling of $\MomLLproj$ or $\AZproj$ induces \amtiling of $\MPkn$ or $\AAkn$. One difficulty lies in finding a positive or negative answer to \cref{que:TREE:extend} below. 
 In particular, we do not know how to prove the equivalence of BCFW tiling conjectures for $\MomLLproj$ and $\AZproj$ directly; we can only deduce it from the equivalence of BCFW tiling conjectures for ambient amplituhedra $\MPkn$ and $\AAkn$. 
\end{remark}
\begin{question}\ \label{que:TREE:extend}
\begin{enumerate}[label=(\arabic*)]
\item Let $\lalat\in\MPkn$. Does there always exist $\LaLat\in\LaLaimmnn$ such that $\la\subset\La$ and $\lat\subset\Lat$?
\item Let $V\in\AAkn$. Does there always exist $Z\in\Grtp(k+2,n)$ such that $V\subset Z$?
\end{enumerate}
\end{question}
\begin{remark}\label{rmk:ELT_ELPTSBW}
\Cref{lemma:proj_tiling} implies the BCFW tiling results of~\cite{ELT,ELPTSBW}. The machinery developed in those papers relies essentially on the topological properties of the amplituhedron map $\PsiZ$.\footnote{For example,~\eqref{eq:TREE:proj=cl_amb} does not follow from the results of~\cite{ELT,ELPTSBW}. We thank Tsviqa Lakrec and Lauren Williams for discussions regarding this question.} In view of \cref{rmk:converse_main_diff}, it remains unclear whether the techniques of~\cite{ELT,ELPTSBW} can be extended to yield a tiling of the ambient amplituhedron $\AAkn$. Such an extension, if possible, would allow one to recover our \cref{thm:f_triang} from their results via \cref{lemma:TREE:amb_tiling_equivalence}.
\end{remark}

\begin{remark}%
In a follow-up paper~\justpaptwo, we extend the T-duality map $(C,\la,\lat)\mapsto(\ddC=C\cdot\Qla,V=\Qlapp(\lat))$ to a local operation on weighted planar bipartite graphs. 
We expect that applying this operation to the glued graph $(\Gf,\wtf)$ (\cref{cor:Ann_split_Delta}) results in the following correspondence. Let $C\in\Grtnn(k,n)$, 
 $\LaLat\in\LaLaArep$, 
and $\lalat := \PhiLL(C)$. Since $\LaLat\in\LaLaArep\subset\LaLafl$ by~\eqref{eq:OCA:Cbotm_subset_Cbotp}--\eqref{eq:OAC:Arep}, there exists a $k$-plane $\Lamid:=\lat\oplus\Lap = \lap\cap\Lat\in\Gr(k,n)$. 
(We have $\Lap\subset\Lamid\subset\Lat$ but $\Lamid\notin\Grtnn(k,n)$.) 
 Let 
\begin{equation}\label{eq:SHIFT:Muni-Auni}
\ddC=C\cdot\Qla,\quad 
V=\Qlapp(\lat),
\quad\text{and}\quad
Z=\Qlapp(\Lamid).
\end{equation}
Then $V = \PsiZ(\ddC)$ and $\ddC\in\Grtnn(k-2,n)$. Moreover, one can show using the machinery of~\justpaptwo that $Z\in\Grtp(k+2,n)$. Conversely, when $\la\subset V$ is known, one can recover $\lat:=V\cdot \Qla$, $C:=\Qlapp(\ddC)$, and $\Lamid:=Z\cdot \Qla$. We have $\lat=\Lamid\cap C^\perp$ and $\la = \Lamid^\perp\cap C$; however, more information is needed to recover $\LaLat$. 
\end{remark}

\begin{remark}[Parity duality]
We explain the relationship between the \emph{parity duality} of~\cite{GL_parity} 
 and the natural parity duality $(\la,\lat)\mapsto(\alt(\lat),\alt(\la))$ between $\MPkntree$ and $\MPop_{n-k,n}$
(which sends $(C,\La,\Lat)\mapsto (\altp(C),\alt(\Lat),\alt(\La))$). The transformation of~\cite{GL_parity} consists of applying the \emph{right twist} $\rtw$ of~\cite{MuSp} to the $(n-4)\times n$ matrix $V^\perp=\begin{pmatrix}
\ddC\\ Z^\perp
\end{pmatrix}$ and then applying $\alt$. It is well known that $\rtw(V^\perp) = \ltw(V)^\perp$, where $\ltw$ is the \emph{left twist} of~\cite{MuSp}. Thus, the parity duality of~\cite{GL_parity} sends $V=\begin{pmatrix}
\la \\ \mu
\end{pmatrix}\in\AAkn$ to $V^\ast:=\alt(\ltw(V)) = \begin{pmatrix}
\mu^\ast \\ \la^\ast
\end{pmatrix}\in\AA_{n-k-2,n}$. One can check that up to applying cyclic symmetry and column rescaling,\footnote{To be precise, letting $D:=\diag\left(\frac{\brV[i-1,i,i+1,i+2]}{\brla<i-1,i>\brla<i,i+1>}\right)_{i\in\brn}$, we set $V^\ast = \alt(\ltw(V)\cdot \Shift^{-2}\cdot D)$ for $\Shift$ given by~\eqref{eq:Shift_dfn}.} we have $\la^\ast=\alt(\lat)$ (where $\lat := \mu\cdot \Qla$) and $\lat^\ast:=\mu^\ast\cdot Q_{\la^\ast} = \alt(\la)$. 
\end{remark}

\section{Perfect t-embeddings}\label{sec:perf}
The following is a slight modification (cf. \cref{rmk:CLR_outer_face}) of the main definition of~\cite{CLR2}.
\begin{definition}\label{dfn:perf}
A null polygon $\Pcurve=(\bdx_1,\bdx_2,\dots,\bdx_n)$ in $\Rdd$ is called \emph{perfect} if 
\begin{enumerate}[label=(\roman*)]
\item\label{perf1} each edge of $\PcurveT$ lies on a line tangent to the unit circle $\Tcirc:=\{z\in\C:|z|=1\}$, and the bisector of the angle at each $\bdxT_i$ passes through $0$, and
\item\label{perf2} for each $i\in\brn$, we have $\sumbT_i = \sumwT_i$.
\end{enumerate}
A t-immersion $\xd:\Faces\to\Rdd$ is called \emph{perfect} if its boundary polygon $\Pbdx$ is perfect.
\end{definition}
\noindent See \cref{fig:hex}. It is clear that if $\xd$ is perfect then the region enclosed by $\PbdxT$ is star-shaped.
Thus, by \cref{lemma:Jordan_curve}, every perfect t-immersion is a t-embedding. 

\begin{remark}\label{rmk:n_even}
By \crefi{dfn:perf}{perf2}, the polygon $\PllO$ is contained in a line, with alternating edge directions $(\bdxO_1-\bdxO_2,\bdxO_2-\bdxO_3,\dots,\bdxO_n-\bdxO_1,\bdxO_1-\bdxO_2)$. In particular, $n$ must be even, and we assume that $n$ is even for the rest of this section. 
\end{remark}
\noindent We will see in \cref{cor:perf:n=2k} that 
we must have $n=2k$ in order for a graph $\G$ of type $(k,n)$ to admit perfect t-embeddings.

It was conjectured in~\cite[Section~4.2]{CLR2} that each ``sufficiently nondegenerate'' graph $(\G,\wt)$ of type $(k,2k)$ admits a perfect t-embedding, and that this perfect t-embedding is unique (provided it exists). The goal of this section is to interpret perfect t-embeddings in terms of a certain involution on $\Gr(2,n)$.
 Using this interpretation, we disprove the uniqueness part in \cref{ex:perf:two_solutions}. We also show in \cref{lemma:perf:no_solutions} that perfect t-embeddings do not exist when $\G$ is a BCFW graph.

On the other hand, we introduce the notion of a \emph{T-dual perfect t-embedding} in \cref{ssec:T_dual_perfect} and show that such objects always exist and are unique up to conformal equivalence.

\subsection{An involution}
Let
$\Gro(2,n):=\{\la\in\Gr(2,n)\mid \brla<i,i+1>,\brla<i+1,i-1>\neq0\text{ for all $i\in\brn$}\}$ and
\begin{equation}\label{eq:pinv(la)_dfn}
\pinv:\Gro(2,n)\to\Gro(2,n),\quad \la\mapsto \alt(\la) \cdot \Qla.
\end{equation}
In other words, for $\la\in\Gro(2,n)$, let 
\begin{equation}\label{eq:perf:tlai_dfn}
 \btla:=\diag(\tla_1,\tla_2,\dots,\tla_n), \quad\text{where}\quad
 \tlai := \frac{\brla<i+1,i-1>}{\brla<i-1,i>\brla<i,i+1>}
\quad\text{for $i\in\brn$.}
\end{equation}
Then since $\la\cdot \Qla = 0$ and since $n$ is even, we get $\pinv(\la) = \alt(\la) \cdot \btla = \la \cdot \alt(\btla)$ (up to a factor of $2$).
We lift $\pinv$ to a map on $2\times n$ matrices by negating the second row (cf.~\eqref{eq:intro:y_to_lalat} and \cref{lemma:wind_alt}) so that 
\begin{equation}\label{eq:perf:pinv_dfn_mat}
\pinv(\la) := \diag(1,-1)\cdot \alt(\la) \cdot \btla;
\quad\text{explicitly,}\quad
\pinv(\la)_i = (-1)^{i-1}\tlai \begin{pmatrix}
\la_{1,i}\\
-\la_{2,i}
\end{pmatrix} \quad\text{for $i\in\brn$}.
\end{equation}
\begin{lemma}\label{lemma:perf:pinv}
The map $\pinv:\Gro(2,n)\to\Gro(2,n)$ is an involution satisfying $\la\perp\pinv(\la)$.
\end{lemma}
\begin{proof}
The $i$-th column of $\pinv(\la)$ is given by~\eqref{eq:perf:pinv_dfn_mat}. By~\eqref{eq:perf:tlai_dfn}, rescaling the $i$-th column of $\la$ by $t$ corresponds to dividing the $i$-th column of $\pinv(\la)$ by $t$. That is, for any $\tdiag\in\LG$, we have $\pinv(\la\cdot \tdiag) = \pinv(\la)\cdot \tdiag^{-1}$. We find 
$\pinv(\pinv(\la)) = \pinv(\la\cdot \alt(\btla)) = \pinv(\la)\cdot \alt(\btla)^{-1} = \la$. Thus, $\pinv$ is an involution.

Since $\Qla$ is self-adjoint and satisfies $\la\cdot \Qla = 0$, we have $\la\cdot \pinv(\la)^T = \la\cdot \Qla \cdot \alt(\la)^T = 0$. 
\end{proof}

\begin{lemma}\label{lemma:perfect=>MP_and_n=2k}
If $\lalat\in\lalak$ satisfies $\lat=\pinv(\la)$ then $n=2k$, 
$\lalat\in\MPkn$,
 and $\tlai<0$ for all $i\in\brn$.
\end{lemma}
\begin{proof}
Since $\lalat\in\lalak$, we have $\brla<i,i+1> >0$ and $\brlat[i,i+1]>0$. Since $\lat = \pinv(\la)$, by~\eqref{eq:perf:pinv_dfn_mat}, $\brlat[i,i+1] = \tlai\tla_{i+1} \brla<i,i+1>$, so $\tlai\tla_{i+1}>0$ for all $i\in\brx{n-1}$. Thus, all $\tlai$'s are nonzero and have the same sign. 
Suppose for contradiction that $\tlai>0$ for all $i\in\brn$. Equivalently, 
 $\brla<i-1,i+1><0$ for all $i\in\brn$. 
We find $\Arg(\la_{2i-1},\la_{2i}) + \Arg(\la_{2i},\la_{2i+1})>\pi$ for all $i\in\brn$. Summing up these angles for $i=1,\dots,\lfloor n/2\rfloor$, we find $\wind(\la)> \lfloor n/2\rfloor \pi$. Since all $\tlai>0$, by~\eqref{eq:wind_alt}, we get $\wind(\lat) = n\pi - \wind(\la) < \lceil n/2 \rceil\pi$. But since $\lalat\in\lalak$, we have $\wind(\lat)=\wind(\la)+2\pi$, a contradiction. We have shown that $\tlai<0$ for all $i\in\brn$. Since $(k+1)\pi = \wind(\lat) = n\pi - \wind(\la) = (n-k+1)\pi$, we get $n=2k$. 

We check that $\lalat$ is \Mdash positive. For integers $p<q$, we denote 
$S_0(p,q):=\frac{\brla<p,q>\brla<p+1,q+1>}{\brla<p,p+1>\brla<q,q+1>}$ and 
$S_1(p,q):=\frac{\brla<p,q+1>\brla<p+1,q>}{\brla<p,p+1>\brla<q,q+1>}$. 
We let $S(p,q):=S_0(p,q)$ if $p\equiv q$ and $S(p,q):=-S_1(p,q)$ if $p\not\equiv q$ modulo~$2$. 
We have the Pl\"ucker relation, $S_0(p,q) = 1 + S_1(p,q)$. 
Applying~\eqref{eq:perf:tlai_dfn}--\eqref{eq:perf:pinv_dfn_mat}, we find 
\begin{equation*}%
 \brla<p,q>\brlat[p,q] = S(p-1,q) + S(p,q-1) - S(p-1,q-1) - S(p,q). 
\end{equation*}
Applying this to each term in~\eqref{eq:intro:Mand_dfn} yields a telescoping sum from which we find 
$\Mand_{i,j}(\la,\lat) = S(i,j)$ for all $i+2\leq j\leq i+n-2$. Explicitly, if $i\equiv j$ (resp., $i\not\equiv j$) modulo~$2$ then 
\begin{equation}\label{eq:Mand_perfect_formula}
 \Mand_{i,j}(\la,\lat) = \frac{\brla<i,j>\cdot \brla<i+1,j+1>}{\brla<i,i+1>\cdot \brla<j,j+1>},
 \quad\text{resp.,}\quad
 \Mand_{i,j}(\la,\lat) = -\frac{\brla<i,j+1>\cdot \brla<i+1,j>}{\brla<i,i+1>\cdot \brla<j,j+1>}.
\end{equation}
Since $\tlai<0$, we get $\brla<i-1,i+1> >0$ for all $i\in\Z$. Denoting $\wind(\la_i\to\la_j):=\sum_{\s=i}^{j-1}\Arg(\la_\s,\la_{\s+1})$, we find $\wind(\la_i\to\la_{i+2d}) < d\pi$ for all $d\in\brk$. 
Thus, $\wind(\la_{i-2(k-d)}\to\la_{i}) < (k-d)\pi$, so $\wind(\la_{i}\to\la_{i+2d}) = \wind(\la) - \wind(\la_{i-2(k-d)}\to\la_{i}) > (d-1)\pi$. 
Since $(d-1)\pi < \wind(\la_i\to\la_{i+2d}) < d\pi$, we get $(-1)^{d-1}\brla<i,i+2d> >0$ for all $d\in\brx{k-1}$. Since all terms $\brla<i,j>$, $\brla<i+1,j+1>$, $\brla<i,j+1>$, $\brla<i+1,j>$ in~\eqref{eq:Mand_perfect_formula} involve indices of the same parity, we see that $\Mand_{i,j}(\la,\lat)>0$ for all $i+2\leq j\leq i+n-2$. 
\end{proof}

\begin{proposition}\label{prop:perfect}
For $\lalat\in\lalak$, the null polygon $\Pll$ is perfect if and only if $\lat = \pinv(\la)$. In particular, a t-immersion $\xll:\Faces\to\Rdd$ is perfect if and only if $\lat = \pinv(\la)$.
\end{proposition}
\begin{proof}
Recall that part~\itemref{perf2} of \cref{dfn:perf} is equivalent to the condition that the polygon $\PllO$ is contained in a line. After rotating it as in \cref{rmk:bdry:alpha}, we may assume that $\bdxO_{i} - \bdxO_{i-1}$ is real for all $i\in\brn$, and that $\bdxO_1 - \bdxO_0<0$. By~\eqref{eq:yy=Tcal_Ocal}, we get that $\ovl{\y_i}\yt_i$ is real, so $\yt_i=(-1)^{i-1}t_i\y_i$ for some $t_i\in\Rast$. Applying~\eqref{eq:intro:y_to_lalat}, we get $\lat_i = (-1)^{i-1}t_i \begin{pmatrix}
\la_{1,i}\\
-\la_{2,i}
\end{pmatrix}$.

 It remains to show that $t_i=\tlai$ for all $i\in\brn$ if and only if part~\itemref{perf1} of \cref{dfn:perf} holds for $\Pll$. Since $\lalat\in\lalak$, similarly to \cref{lemma:perfect=>MP_and_n=2k}, we see that all $t_i$'s have the same sign. Since $t_1|\y_1|^2 = \bdxO_1 - \bdxO_0<0$, we have $t_1<0$ and thus $t_i<0$ for all $i\in\brn$. 

Let $\mupt_i\in\Tcirc$ be the tangent point to the line through $\bdxT_{i-1}$ and $\bdxT_i$. Since $\bdxT_i-\bdxT_{i-1}=(-1)^{i-1}t_i\y_i^2$ by~\eqref{eq:yy=Tcal_Ocal}, $\mupt_i = \pm \I\y_i^2/|\y_i|^2$.
 In fact, since the boundary polygon $(\bdxT_1,\bdxT_2,\dots,\bdxT_n)$ winds clockwise around the origin and since $t_i<0$, we find that $\mupt_i = (-1)^{i}\I\y_i^2/|\y_i|^2 = (-1)^{i}\I\y_i / \ovl{\y_i}$ for all $i\in\brn$.

The point $\bdxT_i$ lies at the intersection of the two tangent lines to $\Tcirc$ at points $\mupt_i$ and $\mupt_{i+1}$, and therefore $\bdxT_i = \frac{2\mupt_i\mupt_{i+1}}{\mupt_i+\mupt_{i+1}}$. We calculate
\begin{equation*}%
\frac{\bdxT_{i}-\bdxT_{i-1}}{\mupt_i} 
= 2\mupt_i\cdot \frac{\mupt_{i+1} - \mupt_{i-1}}{(\mupt_{i-1}+\mupt_i)(\mupt_i+\mupt_{i+1})}
= \I|\y_i|^2 \frac{\brla<i+1,i-1>}{\brla<i-1,i>\brla<i,i+1>}
= \I|\y_i|^2\tlai.
\end{equation*}
Here, we have used that e.g. $\mupt_i + \mupt_{i+1} = (-1)^{i}\I(\y_i\ovl{\y_{i+1}} - \y_{i+1}\ovl{\y_i}) /\ovl{\y_i\y_{i+1}}$ and $\y_i\ovl{\y_{i+1}} - \y_{i+1}\ovl{\y_i} = -2\I\det\mat[\y_i|\y_{i+1}] = -2\I\brla<i,i+1>$. On the other hand, $\bdxT_{i}-\bdxT_{i-1} = (-1)^{i-1}t_i\y_i^2$. Dividing this by $\mupt_i = (-1)^{i}\I\y_i^2/|\y_i|^2$ and equating the result to $\I|\y_i|^2\tlai$, 
 we find $t_i = \tlai$.
\end{proof}

\begin{remark}\label{rmk:perf_Lorentz}
 Let $\perfY,\perfZ\in\Rdd$ be given by $\perfYT=\perfYO=0$ and $\perfZT=0$, $\perfZO = \I$. 
Then a null polygon $\Pcurve=(\bdx_1,\dots,\bdx_n)$ in $\Rdd$ is perfect in the sense of \cref{dfn:perf} if (after some shift and rotation of the line containing $\PcurveO$) we have 
 $(\bdx_i-\perfY)^2 = 1$ and $(\bdx_i - \perfZ)^2 = 0$ for all $i\in\brn$. 
Thus, a translation/rescaling/Lorentz-invariant way to define perfect null polygons is to fix a pair $\perfY,\perfZ\in\Rdd$ satisfying $c:=(\perfY - \perfZ)^2<0$ and then consider null polygons $\Pcurve=(\bdx_1,\dots,\bdx_n)$ satisfying 
$(\bdx_i-\perfY)^2 = -c$ and $(\bdx_i - \perfZ)^2 = 0$ for all $i\in\brn$. 
In the above proof, we have used scaled Lorentz transformations and translations to fix $c=-1$, $\perfY=(0,0)$, and $\perfZ=(0,\I)$. 
\end{remark}

Combining \cref{lemma:perfect=>MP_and_n=2k} with \cref{prop:perfect}, we obtain the following. 
\begin{corollary}\label{cor:perf:n=2k}
Any perfect t-immersion $\xll$ satisfies $\lalat\in\MPkdk$.
\end{corollary}

\begin{figure}
\begin{tabular}{ccc}
 \includegraphics[scale=1.2]{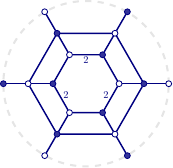}
&
 \includegraphics[scale=1.5]{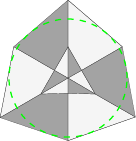}
&
 \includegraphics[scale=1.5]{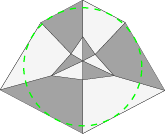}
\end{tabular}
 \caption{\label{fig:hex} A weighted hexagonal prism $(\G,\wt)$ with unmarked edges having weight $1$ and two perfect t-embeddings of $(\G,\wt)$; see \cref{rmk:CLR_outer_face}.}
\end{figure}

\subsection{T-dual perfect t-embeddings}\label{ssec:T_dual_perfect}
Recall from \cref{sec:TREE} that T-duality provides a map $\AAshift:\MPkn\to\AAkn$ sending a pair $\lalat\in\MPkn$ of orthogonal $2$-planes to a $4$-plane $V=\Qlapp(\lat)\in\AAkn$. It turns out that the \emph{perfect condition} $\lat = \pinv(\la)$ (\cref{prop:perfect}) on $\lalat$ becomes especially natural after applying T-duality.

\begin{proposition}\label{prop:T_dual_perfect_vs_perfect}
Let $\lalat\in\MPkn$ and $V:=\Qlapp(\lat)\in\AAkn$. Then $\lalat$ satisfies the
perfect condition $\lat = \pinv(\la)$ if and only if $V$ satisfies the \emph{T-dual perfect condition}
\begin{equation}\label{eq:perfect_V}
 V = \alt(V) \quad\text{as elements of $\Gr(4,n)$.}
\end{equation}
\end{proposition}
\begin{proof}
First, if $\lalat\in\MPkn$ satisfies $\lat = \pinv(\la)$ then we must have $n=2k$ by \cref{cor:perf:n=2k}. Similarly, for $V\in\AAkn$ and $\mu:=\LVVp\cdot V$ as in~\eqref{eq:V_la_mu_lat}, \eqref{eq:wind_alt} yields $\windskip(\mu')=(n-2)\pi - \windskip(\mu)$, where $\mu':=\diag(1,-1)\cdot \alt(\mu)$.
 Having $V=\alt(V)\in\AAkn$ forces $\mu=\mu'$, and since $\windskip(\mu)=(k-1)\pi$ by~\eqref{eq:windskip_vs_varxxx}, we also must have $n=2k$. 

Next, we claim that for all $\la\in\lakdk$, the sum $\la+\alt(\la)$ of $2$-planes in $\R^{2k}$ satisfies
\begin{equation}\label{eq:perfect_alt_oplus_dim4}
 \dim(\la+\alt(\la))=4 \quad\text{for all $\la\in\lakdk$.}
\end{equation}
Indeed, the matrices $\begin{pmatrix}
\la \\ \alt(\la)
 \end{pmatrix}$ and $\begin{pmatrix}
\la_1 & 0 & \la_3 & 0 & \cdots & \la_{2k-1} & 0\\
0 & \la_2 & 0 & \la_4 & \cdots & 0 & \la_{2k}
 \end{pmatrix}$ are related by invertible row operations. 
The latter matrix does not have rank $4$ if and only if either all odd or all even columns of $\la$ are proportional. In this case, we would have $\wind(\la)=k\pi$, but since $\la\in\lakdk$, we have $\wind(\la)=(k-1)\pi$, a contradiction.

Suppose that $\lalat\in\MPkdk$ satisfies the perfect condition $\lat = \pinv(\la)$ and let $V:=\Qlapp(\lat)\in\AAkdk$. Let $\mu:=\alt(\la)$. By~\eqref{eq:pinv(la)_dfn}, $\lat = \mu\cdot \Qla$, so $\la,\mu\subset V$. By~\eqref{eq:perfect_alt_oplus_dim4}, $V = \la\oplus \mu$, so $V = \alt(V)$ as elements of $\Gr(4,2k)$. Conversely, suppose that $V\in\AAkdk$ satisfies $V=\alt(V)$. By \cref{thm:AAshift}, there exists a $2$-plane $\la\in\Gr(2,V)\cap \lakdk$. Choose any such $2$-plane $\la$ and set $\lat:=V\cdot \Qla$. Since $V=\alt(V)$, we have $\alt(\la)\subset V$. By~\eqref{eq:perfect_alt_oplus_dim4}, $V = \la\oplus\alt(\la)$. Thus, $\lat = \alt(\la)\cdot \Qla = \pinv(\la)$. 
\end{proof}

\begin{definition}
A \emph{T-dual perfect t-datum} for $\ddC\in\Grtnn(k-2,2k)$ is a $4$-plane $V\in\AAkdk$ satisfying $V\subset\ddC^\perp$ and $V=\alt(V)$.
\end{definition}
\noindent 
As we show in \cref{thm:T_dual_perfect} below, any sufficiently nondegenerate $\ddC\in\Grtnn(k-2,2k)$ admits a unique T-dual perfect t-datum $V$, given explicitly by~\eqref{eq:V_vs_ddC_cap}. 
The name \emph{T-dual perfect t-datum} is justified as follows. 
\begin{corollary}\label{rmk:from_T_dual_perfect_to_perfect}
Let $\G$ be a graph of type $(k,2k)$ and let $\fap:=\fG\in\Boundkdk$. 
The following algorithm always outputs a perfect t-embedding of $\G$, and any perfect t-embedding of $\G$ arises in this way. 
\begin{enumerate}[label=(\arabic*)]
\item Choose any $\ddC\in\Ptp_{\ddfperm}$ and consider a T-dual perfect t-datum $V$ for $\ddC$. (It exists and is unique for nondegenerate $\ddC$ by \cref{thm:T_dual_perfect} below.)
\item Choose any $2$-plane $\la\in\lakdk$ satisfying $\la\subset V$. (It exists by \cref{thm:AAshift}.)
\item Set $C:=\Qlapp(\ddC)$ and $\lat:=\pinv(\la)$. %
\item Pick $\wt\in\Rtpgauge$ such that $C = \Meas(\G,\wt)$ and let $\xll$ be given by \cref{notn:Tll_xll}. 
 Then $\xll$ is a perfect t-embedding of $(\G,\wt)$. 
\end{enumerate}
\end{corollary}
\begin{proof}
Choosing any $2$-plane $\la\in\lakdk$ such that $\la\subset V$ 
and letting $\lat:=V\cdot \Qla$ and $C:=\Qlapp(\ddC)$, we obtain a triple $\la\subset C\subset\latp$ with $C\in\Grnd(k,2k)$ and $\lalat\in\MPkdk$. By \cref{thm:intro:t_imm_vs_triples}, this triple gives rise to a t-immersion $\xll$ of any weighted graph $(\G,\wt)$ satisfying $C=\Meas(\G,\wt)$. By \cref{prop:perfect,prop:T_dual_perfect_vs_perfect}, $\xll$ is a perfect t-embedding of $(\G,\wt)$. 
Conversely, given a perfect t-embedding $\xll$ of $(\G,\wt)$ with $C:=\Meas(\G,\wt)$, by \cref{prop:perfect}, we have $\lat=\pinv(\la)$, and so by \cref{prop:T_dual_perfect_vs_perfect}, $V:=\Qlapp(\lat)$ is a T-dual perfect t-datum for $\ddC:=C\cdot \Qla$. Thus, any perfect t-embedding arises via the above algorithm. 
\end{proof}

\begin{remark}
Given a T-dual perfect t-datum $V$ for $\ddC$, any other choice of $\la'\in\Gr(2,V)\cap \lakdk$ would give rise to another perfect t-embedding $\xd_{\la',\lat'}$ of the same graph $\G$ but with different edge weights $\wt'$. However, any two such perfect t-embeddings are \emph{conformally equivalent}, meaning they are related by $\GL_4(\R)$-action on the T-dual perfect t-datum $V\in\Gr(4,2k)$.
\end{remark}

\begin{remark}
One can also think of T-dual perfect t-embeddings geometrically as follows. While a t-embedding is an assignment $\xd:\Faces\to\Rdd$ of points in $\Rdd$ to the vertices of $\GD$, a \emph{T-dual t-embedding} associates an oriented line $\mtL_{\ff}$ in $\RP^3$ to each vertex $\ff\in\Faces$ of $\GD$ such that whenever several vertices of $\GD$ share a common white (resp., black) face, the corresponding lines intersect at a common point (resp., are contained in a common plane). Under this correspondence, \Mdash positivity conditions $(\xd(\ff)-\xd(\f))^2>0$ translate into the oriented lines $\mtL_{\ff},\mtL_{\f}$ being skew and positively oriented (via the right-hand rule). The \emph{perfect} condition translates into all boundary lines $\mtL_{\bdf_1},\mtL_{\bdf_2},\cdots,\mtL_{\bdf_{2k}}$ intersecting a fixed pair $\Lcal_0,\Lcal_{\infty}$ of skew lines in $\RP^3$. 
See \Nref{ssec:APP:RP3} for further details.\footnote{We thank Dmitry Chelkak for explaining an equivalent Pl\"ucker-quadric viewpoint to us in 2024, including the above interpretation of the perfect boundary condition.}
\end{remark}

Next, we discuss existence and uniqueness of T-dual perfect t-data. 
In the following definition, we set $\brnodd:=\{1,3,\dots,2k-1\}$ and $\brnev:=\{2,4,\dots,2k\}$.
Throughout, we assume that $0\leq d\leq k$. 
\begin{definition}
We say that $X\in\Grtnn(d,2k)$ is \emph{\oend} and write $X\in\Groe(d,2k)$ if for all 
$J\in{\brnodd\choose d}\sqcup{\brnev\choose d}$, we have $\Delta_J(X)>0$. 
\end{definition}

The next result is a direct consequence of Skandera's inequalities~\cite{Skandera}.
\begin{lemma}[{\cite[Theorem~6.3]{Farber_Postnikov}}]\label{lemma:oend_Skandera}
An element $X\in\Grtnn(k,2k)$ is \oend if and only if there exists $J\in{\brx{2k}\choose k}$ such that $\Delta_J(X)>0$ and $\Delta_{\comp{J}}(X)>0$. \qed
\end{lemma}

\begin{definition}\label{dfn:Catalan}
Let $X\in\Grtnn(d,2k)$, $\fap:=\fX$, and let $\Icalr_\fap=(\Ir_i)_{i\in\Z}$ be the Grassmann necklace of $X$. For each $i\in\brx{2k}$ and $\s\in\brx{2d}$, let $h_i(\s):=|[i,i+\s)\cap\Ir_i| - |[i,i+\s)\setminus\Ir_i|$.\footnote{Construct a path $\DP(\Ir_i)$ such that for each $t=i,i+1,\dots,i+2d-1$, $\DP(\Ir_i)$ has an up-right step if $t\in \Ir_i$ and a down-right step if $t\notin\Ir_i$. Then $\Ir_i$ satisfies the Catalan condition if and only if $\DP(\Ir_i)$ is a Dyck path (never passing below the $x$ axis).} We say that $\Ir_i$ satisfies the \emph{Catalan condition} if $h_i(\s)\geq0$ for all $\s\in\brx{2d}$. We say that $X$ (resp., $\fap$) satisfies the \emph{Catalan condition} if each element $\Ir_i$ of its Grassmann necklace satisfies the Catalan condition.
\end{definition}

\begin{lemma}\label{lemma:oend_Catalan}
An element $X\in\Grtnn(d,2k)$ is \oend if and only if it satisfies the Catalan condition.
\end{lemma}
\begin{proof}
We have $\Delta_J(X)>0$ if and only if $J$ belongs to the positroid $\Matroid_X$ of $X$. By definition, this is equivalent to having $\Ir_i\preceq_i J$ in the cyclically shifted Gale order $\preceq_i$ for all $i\in\Z$. For $J_i:=\{i,i+2,\dots,i+2(d-1)\}$, 
$\Ir_i\preceq_i J_i$ translates precisely into the Catalan condition above. Since $J_i$ is the $\preceq_i$-minimal element of ${\brnodd\choose d}\sqcup{\brnev\choose d}$, having $\Ir_i\preceq_i J_i$ for all $i\in\brx{2k}$ is necessary and sufficient for $X$ to be \oend.
\end{proof}

\noindent In view of \cref{lemma:oend_Catalan}, we say that $\fap\in\BND(d,2k)$ is \emph{\oend} if it satisfies the Catalan condition of \cref{dfn:Catalan}.

\begin{theorem}[Existence and uniqueness of T-dual perfect t-embeddings]\label{thm:T_dual_perfect}
Each \oend $\ddC\in\Groe(k-2,2k)$ admits a unique T-dual perfect t-datum $V\in\AAkdk$. 
It is given by
\begin{equation}\label{eq:V_vs_ddC_cap}
 V = \ddC^\perp\cap\altp(\ddC).
\end{equation} 
\end{theorem}
\begin{proof}
Similarly to the proof of~\eqref{eq:perfect_alt_oplus_dim4}, we see that $\dim(\ddC+\alt(\ddC))=2k-4$ because $\ddC$ is \oend. 
 If $V\subset\ddC^\perp$ satisfies $V=\alt(V)$ then $V\subset\altp(\ddC)$, so $V\subset(\ddC^\perp\cap\altp(\ddC))=(\ddC+\alt(\ddC))^\perp$. Since $\dim(\ddC+\alt(\ddC))=2k-4$, we have $\dim(\ddC+\alt(\ddC))^\perp = 4$, so indeed $V$ must be given by~\eqref{eq:V_vs_ddC_cap}.

Next, we show that the $4$-plane $V$ given by~\eqref{eq:V_vs_ddC_cap} satisfies $V\in\AAkdk$. 
Similarly to~\eqref{eq:KW_eq} (see also~\cite[Lemma~10.5]{GL_parity}), for each $i+2\leq j\leq i+2k-2$ and $\Jij:=\{i,i+1,j,j+1\}\subset\brx{2k}$, we have 
$\brV[i,i+1,j,j+1] = \sum_{L,R\in{\brx{2k}\choose k-2}:\ L\sqcup R = \Jijc} \Delta_L(\ddC) \Delta_R(\ddC)\geq0$. It follows that $\brV[i,i+1,j,j+1]>0$ if and only if the submatrix $\ddC|_{\Jijc}$ 
 contains complementary sets of columns $L,R$ with both Pl\"ucker coordinates $\Delta_L(\ddC),\Delta_R(\ddC)$ nonzero. By \cref{lemma:oend_Skandera}, this is true if and only if $\ddC|_{\Jijc}$ is \oend for each $\Jij=\{i,i+1,j,j+1\}$. Since $\Jij$ can contain any two odd elements (or any two even elements), these conditions are equivalent to $\ddC$ itself being \oend.

We have shown that for each \oend $\ddC\in\Grtnn(k-2,2k)$, the $4$-plane $V$ given by~\eqref{eq:V_vs_ddC_cap} satisfies $\brV[i,i+1,j,j+1]>0$ for all $i+2\leq j\leq i+2k-2$. The set $\Groe(k-2,2k)$ of \oend elements contains the top cell $\Grtp(k-2,2k)$. In particular, $\Groe(k-2,2k)$ is connected. Since the map $\ddC\mapsto V(\ddC):=\ddC^\perp\cap\altp(\ddC)$ is continuous on $\Groe(k-2,2k)$ and since $\varxxx(V)$ is locally constant on the set of $V$ satisfying $\brV[i,i+1,j,j+1]>0$ for all $i+2\leq j\leq i+2k-2$, it suffices to check that $\varxxx(V(\ddC))=k-2$ for a single point $\ddC\in\Groe(k-2,2k)$. Taking $\ddC:=\Xcsx{k-2}\in\Grtp(k-2,2k)$ to be the unique cyclically symmetric point in $\Grtnn(k-2,2k)$ defined in \cref{sec:cs_mom_ampl}, we check that $\varxxx(V(\ddC))=k-2$ similarly to the proof of \cref{prop:Qla_Ptp}. Thus, $V(\ddC)\in\AAkdk$ for all $\ddC\in\Groe(k-2,2k)$.
\end{proof}

Combining \cref{rmk:from_T_dual_perfect_to_perfect,thm:T_dual_perfect}, we obtain the following result.
\begin{corollary}[Existence of perfect t-embeddings]\label{lemma:perfect_exist_G}
Suppose that $\fap:=\fG$ is such that $\ddfperm$ is \oend (cf. \cref{dfn:Catalan,lemma:oend_Catalan}). Then $\G$ admits perfect t-embeddings. 
\end{corollary}

\begin{remark}\label{rmk:perfect_diagram}
\Cref{rmk:from_T_dual_perfect_to_perfect,thm:T_dual_perfect} do not imply existence or uniqueness of perfect t-embeddings of a given \emph{weighted} graph $(\G,\wt)$. 
The main complication arises in the choice of $\la$: for $C:=\Meas(\G,\wt)\in\Grtnn(k,2k)$, finding a perfect t-embedding of $(\G,\wt)$ amounts to finding 
a $2$-plane $\la\in\Gro(2,C)$ such that $\pinv(\la)\perp C$. This corresponds to solving a system~\eqref{eq:perf:perp} of high-degree polynomial equations as we discuss below. On the other hand, given $\ddC:=\Meas(\ddG,\ddwt)\in\Grtnn(k-2,2k)$, finding a $4$-plane $V\subset\ddC^\perp$ satisfying the perfect condition $\alt(V)=V$ amounts to solving a system of \emph{linear} equations which has a unique solution given by~\eqref{eq:V_vs_ddC_cap}. By \cref{thm:T_dual_perfect}, this unique solution $V$ automatically gives rise to a T-dual t-embedding, i.e., satisfies $V\in\AAkdk$. 

Similarly to~\eqref{eq:TREE:comm_diag}, we may consider the following commutative diagram. Let
\begin{equation}\label{eq:MPerf_dfn}
 \MPerfkdk:=\{\lalat\in\MPkdk\mid \lat = \pinv(\la)\}. 
\end{equation} 
 Then the moduli space of perfect t-embeddings of $\G$ with $\fap:=\fG$ is given by
\begin{equation*}%
 \MPerff:=\{(C,\la,\lat)\in\Ptp_{\fap}\times\MPerfkdk\mid \la\subset C\subset\latp\}.
\end{equation*} 
 Similarly, let
\begin{equation}\label{eq:Aperf_dfn}
 \APerfkdk:=\{V\in\AAkdk\mid V = \alt(V)\},
 \quad\text{and let}\quad
\APerff:=\{(\ddC,V)\in\Ptp_{\ddfperm}\times\APerfkdk\mid V\subset\ddC^\perp\}
\end{equation}
 be the moduli space of T-dual perfect t-embeddings of $\G$. Finally, we denote
\begin{equation*}%
 \flAPerff:=\{(\ddC,\la,V)\in\Ptp_{\ddfperm}\times\lakdk\times\APerfkdk\mid \la\subset V\subset\ddC^\perp\}.
\end{equation*}
 In this notation, we have maps
\begin{equation}\label{eq:perf_proj}
 \begin{tikzcd}[column sep=2.5em]
 \Ptp_{\fap} 
& \MPerff \arrow[r,"\Reliso","\sim"'] \arrow[l,"\pperf"'] 
& \flAPerff \arrow[r,"\RelAAforg",twoheadrightarrow] 
& \APerff \arrow[r,"\ddpperf","\sim"'] 
& \Ptp_{\ddfperm},
\end{tikzcd} 
\end{equation}
where $\RelAAforg$ is surjective by \cref{thm:AAshift} while $\Reliso$ and $\ddpperf$ are homeomorphisms by \cref{prop:T_dual_perfect_vs_perfect,thm:T_dual_perfect}, respectively. 
 The algorithm in \cref{rmk:from_T_dual_perfect_to_perfect} follows the diagram~\eqref{eq:perf_proj} from right to left. Existence (resp., uniqueness) of perfect t-embeddings of $(\G,\wt)$ for all $\wt\in\Rtpgauge$ is equivalent to the projection map $\pperf:\MPerff\to\Ptp_{\fap}$ being surjective (resp., injective).
\end{remark}

We leave the following questions for future work.
\begin{question}
Does existence and uniqueness of T-dual perfect t-embeddings (guaranteed by \cref{thm:T_dual_perfect}) allow one to produce uniform bounds on the inverse Kasteleyn matrix $\wtK^{-1}(\w,\b)$ 
 analogously to~\cite[Theorem~3.1]{CLR2}?
\end{question}

\begin{question}
Do the ``perfect amplituhedra'' $\MPerfkdk$ and $\APerfkdk$ given by~\eqref{eq:MPerf_dfn}--\eqref{eq:Aperf_dfn} correspond to some quantum field theory? Do they admit BCFW-like tilings? (See also \cref{rmk:ABJM_vs_perf} below.)
\end{question}

\subsection{Negative results on existence and uniqueness of perfect t-embeddings}\label{ssec:perfect_neg}
\Cref{prop:perfect} allows one to find perfect t-embeddings by solving systems of polynomial equations. Namely, when $C\in\Grnd(k,2k)$ is fixed, one can parameterize $\la\in\Gr(2,C)$ by $2k-4$ variables. Next, one solves the $2k$ equations
\begin{equation}\label{eq:perf:perp}
 \pinv(\la)\cdot C^T = \bzero_{2\times k},
\end{equation}
four of which follow from the rest since $\la\perp\pinv(\la)$ is automatic by \cref{lemma:perf:pinv}. Once the (finitely many) solutions to~\eqref{eq:perf:perp} are found, one checks the condition $(\la,\pinv(\la))\in\lalakdk$ which guarantees that $\Tll$ is a perfect t-embedding by \cref{thm:intro:t_imm_vs_triples} and \cref{prop:perfect}.

\begin{example}\label{ex:perf:two_solutions}
For $k=3$, $n=6$, a weighted top cell graph $(\G,\wt)$ shown in \figref{fig:hex}(left) admits four distinct perfect t-embeddings. Two of them are shown in \figref{fig:hex}(middle,right). The remaining two are obtained from the one in \figref{fig:hex}(right) by $120$- and $240$-degree rotations. None of these four perfect t-embeddings are related by Lorentz transformations. This gives a negative answer to~\cite[Open question~4.8]{CLR2}.
\end{example}

\begin{remark}\label{rmk:CLR_outer_face}
The setup of~\cite{CLR2} involves planar bipartite graphs with a bipartite ``marked outer face'' of degree $n=2d$. Translated into our language, this corresponds to graphs $\G$ such that $n$ is even, the colors of the boundary vertices $\bdv_1,\bdv_2,\dots,\bdv_n$ alternate between black and white, and every boundary face $\bdf_i$ is incident to three edges. An example where $\G$ is a (non-reduced) hexagonal prism is shown in \figref{fig:hex}(left). It follows that any reduced graph $(\G',\wt')$ satisfying $\Meas(\G',\wt')=\Meas(\G,\wt)$, where $(\G,\wt)$ is as in \cref{ex:perf:two_solutions}, also has four distinct perfect t-embeddings. 
\end{remark}

\begin{lemma}\label{lemma:perf:no_solutions}
Suppose that $\G\in\BCFWGkdk$ for $k\geq3$. Then $\G$ (with any choice of positive edge weights) does not admit a perfect t-embedding.
\end{lemma}
\begin{proof}
Suppose otherwise that $\xd$ is a perfect t-embedding of $(\G,\wt)$. Let $\io,\jo$ be as in \cref{dfn:BCFW} and \cref{rmk:jo}. 
The unique point $\xd(\ffout)$ satisfying $(\xd(\ffout) - \xd(\f))^2 = 0$ for $\f=\bdf_{\io-1},\bdf_{\io},\bdf_{\io+1},\bdf_{\jo}$ is the point $\perfZ$ in the notation of \cref{rmk:perf_Lorentz}, with $\perfZT=0$. 
 In other words, for each $\s\in\brn$, the perpendicular bisector $\line_\s$ between $\bdxT_\s$ and $\bdxO_\s$ intersects the ray $\ray$ at the origin. This implies that the parameter $\r_\s$ defined in~\eqref{eq:Psum_r_dfn} is the same for all $\io+2\leq \s\leq \io+2k-2$. For $k\geq3$, the interval $[\io+2,\io+2k-2]$ contains more than one integer, contradicting \cref{prop:BCFW:from_lalat_to_r}.\footnote{On the other hand, for $k=2$, $\BCFWGkdk$ consists of a single graph $\G$ shown in \cref{fig:holom}; this graph admits a (unique) perfect t-embedding for every choice of positive edge weights.}
\end{proof}
\noindent In the notation of \cref{sec:BCFW:triangulation_MP} and~\eqref{eq:MPerf_dfn}, the subvariety 
$\MPerfkdk\subset\MPkdk$ 
 is contained inside the complement of the dense subset $\bigsqcup_{f\in\BCFWfkdk} \lalappf\subset\MPkdk$.

\begin{remark}
In view of \cref{lemma:perfect_exist_G}, we conclude that for all $\fap\in\BCFWfkdk$, $\ddfperm$ is not \oend. 
 It is still possible that perfect t-embeddings exist for all ``sufficiently nondegenerate'' weighted graphs $(\G,\wt)$, where the nondegeneracy condition is the one in \cref{lemma:perfect_exist_G}. 
 This gives a more precise version of~\cite[Open question~4.7]{CLR2}. 
\end{remark}

\appendix

\section{ABJM amplituhedra, s-embeddings, and the problem of Apollonius}\label{sec:ABJM}
We explain how to extend our results to the case of the \emph{ABJM momentum amplituhedron} $\MBJk\LaLat$ and show existence of \emph{s-embeddings} conjectured in~\cite{Chelkak_s_emb}. 

\begin{remark}
After the first version of this paper (which did not contain any results on the ABJM amplituhedron) appeared in 2024, the BCFW tiling conjecture for $\MBJk\LaLat$ was confirmed in~\cite{OPT}. Their proof relied on our results concerning immanant-positive pairs $\LaLat$ studied in \cref{sec:imm}. Below, we give a different argument for the stronger (cf. \cref{rmk:ELT_ELPTSBW}) statement that the \emph{ambient} ABJM momentum amplituhedron is tiled by the corresponding BCFW-like tiles, along the lines of our proof of \cref{thm:intro:A}. We briefly sketch the arguments, leaving the fully rigorous treatment for future work.
\end{remark}

\subsection{ABJM momentum amplituhedron and Ising model}
Following~\cite{Huang_Wen,Huang_Wen_Xie}, let $\OGtnn(k,2k):=\{C\in\Grtnn(k,2k)\mid \altp(C)=C\}$ be the \emph{totally nonnegative orthogonal Grassmannian}. This space was studied rigorously in~\cite{GP_Ising}, where it was related to the \emph{planar Ising model}.
\begin{definition}%
The \emph{ambient ABJM momentum amplituhedron} is given by 
\begin{equation}\label{eq:MBJk_dfn}
 \MBJk:=\{\lalat\in\MPkdk\mid \lat=\alt(\la)\}.
\end{equation}
\end{definition}
\noindent Here, we treat $\lalat$ as a pair of $2$-planes as opposed to a pair of $2\times 2k$ matrices. 
For matrices, we instead impose the condition $\lat = \diag(1,-1)\cdot \alt(\la)$ similarly to~\eqref{eq:perf:pinv_dfn_mat} and \cref{lemma:wind_alt}. 

Thus, for $C\in\OGtnn(k,2k)$ and $\lalat\in\MBJk$, the conditions $\la\subset C$ and $\lat\subset C^\perp$ become equivalent. 

\begin{remark}\label{rmk:ABJM_vs_perf}
Comparing~\eqref{eq:MBJk_dfn} to~\eqref{eq:pinv(la)_dfn} and~\eqref{eq:MPerf_dfn}, we see that the points in $\MBJk$ are closely related to perfect t-embeddings. In fact, under the \oac, 
the condition $\lat = \diag(1,-1)\cdot \alt(\la)$ implies $\sumwT_i = \sumbT_i$. 
This condition coincides with \crefi{dfn:perf}{perf2}. 
Thus, for $\lalat\in\MBJk$, the origami boundary polygon $\PllO=(\bdxO_1,\bdxO_2,\dots,\bdxO_{2k})$ is $1$-dimensional (contained in a line). Compare with \cref{rmk:n_even}. 
\end{remark}

We consider a ``self-dual'' subset $\RelBJ\subset\Rel$ of the correspondence $\Rel$ in~\eqref{eq:TREE:Rel_dfn} defined by
\begin{equation*}%
 \RelBJ:=\{(C,\la,\lat)\in\OGtnn(k,2k)\times\MBJk\mid \la\subset C\subset\latp\}.
\end{equation*}

Given a weighted planar (non-bipartite) graph $(\GIS,J)$ in a disk with $k$ boundary vertices, $(\GIS,J)$ may be converted into a planar bipartite graph $(\G,\wt)$ of type $(k,2k)$ with a bipartite square located in the middle of each edge of $\GIS$; see~\cite{Dubedat,Huang_Wen,GP_Ising} for examples. We call such graphs $(\G,\wt)$ \emph{weighted Ising-bipartite graphs}. For each edge $\e$ of $\GIS$, the weights of the four edges of the corresponding bipartite square in $\G$ are $c,s,c,s$ in clockwise order, where $c,s>0$ are real parameters satisfying $c^2+s^2=1$. In the case of the (ferromagnetic) Ising model, they are given by $c:=\tanh(2J_\e)$ and $s:=\sech(2J_\e)$, where $J_\e\in\Rtp$ is the Ising interaction constant associated to $\e$. When the graph $\GIS$ is ``minimal'' (in the sense that any two strands of its medial graph intersect at most once), the graph $\G$ is reduced. Furthermore, the bounded affine permutation $\fap=\fG$ satisfies $\fap(\fap(i))=i+2k$, i.e., is an involution modulo $2k$. 

\begin{remark}[Existence of s-embeddings]
A special class of t-embeddings called \emph{s-embeddings} was introduced in~\cite{Chelkak_s_emb} in relation to the planar Ising model. An s-embedding of an Ising-bipartite graph $(\G,\wt)$ is a t-embedding $\xd$ of $(\G,\wt)$ such that for each bipartite square in $\G$ as above with face $\ff$ inside incident to faces $\f_1,\f_2,\f_3,\f_4$ in clockwise order, the quadrilateral $(\xT(\f_1),\xT(\f_2),\xT(\f_3),\xT(\f_4))$ is tangent to a circle centered at $\xT(\ff)$; see e.g. the red dashed circle in \figref{fig:ABJM}(left). 
The relationship between s-embeddings of $(\GIS,J)$ and t-embeddings of the weighted Ising-bipartite graph $(\G,\wt)$ was explained in~\cite[Section~8.2]{CLR1}. One can deduce existence of s-embeddings from our results similarly to our proof of \cref{thm:intro:B}. Namely, given an arbitrary weighted Ising-bipartite graph $(\G,\wt)$, we have $C:=\Meas(\G,\wt)\in\OGtnn(k,2k)$. Choose an immanant-nonnegative pair $\LaLat\in\LaLaimmnndk$ satisfying $\Lat = \alt(\La)$.\footnote{One can show existence of such pairs by adapting the arguments in the proof of \cref{lemma:Ttaucoef_nonzero_on_Fl,thm:imm_pos_Fltp}, replacing the group $\Gtp$ with the totally positive part $\Otp$ of the orthogonal group. Similarly to $\Gtnn$, $\Otnn$ is generated as a monoid by hyperbolic rotation matrices $g(r)$ with nonnegative parameters $r\geq0$. See~\cite[Section~7.3]{OPT} for further details.} Applying the momentum amplituhedron map $\PhiLL$ to $C$, we obtain a pair $\lalat:=\PhiLL(C)\in\MBJk$ satisfying $\la\subset C\subset\latp$. Applying \cref{thm:intro:t_imm_vs_triples}, we obtain a t-embedding $\xll$ of $(\G,\wt)$ which is automatically an s-embedding of $(\GIS,J)$. 
\end{remark}

\begin{figure}
 \includegraphics[width=1.0\textwidth]{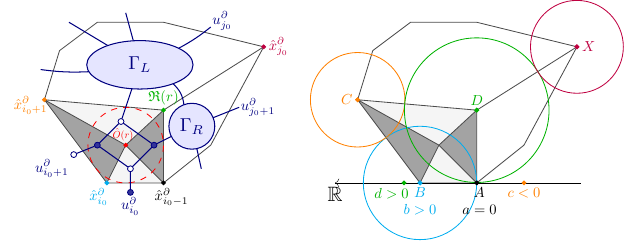}

\vspace{-0.1in}
 \caption{\label{fig:ABJM} An ABJM \ORst. Compare with \cref{fig:BCFW}.}
\end{figure}

\subsection{BCFW recursion and problem of Apollonius}
The ABJM analog of the BCFW recursion was introduced in~\cite[Section~3.2]{Huang_Wen}: instead of adding a single (white-black or black-white) bridge at boundary vertices $\bdv_{\io},\bdv_{\io+1}$, one adds a \emph{double bridge}, i.e., a pair of bridges of opposite colors so that they form a bipartite square adjacent to $\bdv_{\io},\bdv_{\io+1}$; see \figref{fig:ABJM}(left). The weights of the four edges of the square are $c,s,c,s$ as before with $c:=\tanh(2\JIS)$ and $s:=\sech(2\JIS)$ for $\JIS\in\Rtp$. 
The effect of this double bridge on the matrices $C$ and $\la$ amounts to multiplication on the right by a \emph{hyperbolic rotation matrix} $g(\JIS)$ that has a $2\times2$ block $\begin{pmatrix}
1/c & s/c \\ s/c & 1/c
\end{pmatrix}$ in rows and columns $\io,\io+1$; see e.g.~\cite[Remark~5.4]{GP_Ising}. Denoting $C(\JIS):=C\cdot g(\JIS)$, $\la(\JIS):=\la\cdot g(\JIS)$, $\lat(\JIS):=\lat\cdot g(\JIS)^{-1}$, we find that 
\begin{equation}\label{eq:C_JIS}
 C(\JIS)\in\OGtnn(k,2k),\quad \la(\JIS) \subset C(\JIS)\subset\lat(\JIS)^\perp, \quad\text{and}\quad
 \lat(\JIS) = \alt(\la(\JIS)) \quad\text{for all $\JIS>0$}. 
\end{equation}

We denote the collection of Ising-bipartite graphs produced by any branch of this \emph{\ABJW recursion} by $\BJGs$ and let $\BJfs:=\{\fG\mid\G\in\BJGs\}$. 
 For each $\fap\in\BJfs$, we consider a subset 
$\ReloBJ_\fap:=\{(C,\la,\lat)\in\RelBJ\mid C\in \Ptp_{\fap}\cap\OGtnn(k,2k)\}$ and let
\begin{equation*}%
 \MBJ_\fap:=\{\lalat\in\MBJk\mid \la\subset C\subset\latp\text{ for some $C\in\Ptp_\fap\cap\OGtnn(k,2k)$}\}
\end{equation*}
be the projection of $\ReloBJ_\fap$ to the second factor. 

We sketch an argument for how our proof of \cref{thm:f_triang} can be extended to show that the tiles $\{\MBJ_\fap\mid\fap\in\BJfs\}$ form \amtilingBJ of $\MBJk$. First, let us summarize our strategy from \cref{sec:BCFW} in the case of the ordinary momentum amplituhedron $\MPkdk$.
\begingroup
\setlength{\leftmargini}{20pt}
\begin{enumerate}[label=(\arabic*)]
\item\label{ABJM_summ1} Start with an \Mdash positive pair $\lalat\in\MPkdk$ and recover the boundary angle sums $\sumbT_i,\sumwT_i$ from the \Mdash positive null polygon $\Pll=(\bdx_1,\bdx_2,\dots,\bdx_{2k})$;
\item\label{ABJM_summ2} Deform the point $\rayTO(0):=\bdx_{\io}$ along a ray $\rayTO(\r)=\bdx_{\io} + r\Rmom$, where $\Rmom$ is given by~\eqref{eq:Rmom_dfn}.
\item\label{ABJM_summ3} Find the smallest value of $\r = \r_\jo\in(0,\infty)$ for which some Mandelstam variable $(\bdx_{\jo}-\rayTO(\r))^2$ becomes zero. Designate $\xd(\ffout)=\rayTO(\r_\jo)$ as the location of the new face and proceed recursively.
\end{enumerate}
\endgroup
\noindent In step~\itemref{ABJM_summ2}, the point $\rayT(\r)$ moves along a ray in the kami plane, and the value of $\r$ in step~\itemref{ABJM_summ3} is found using the perpendicular bisector construction shown in \cref{fig:BCFW}. 

In the ABJM case, the boundary angle sums in step~\itemref{ABJM_summ1} satisfy $\sumbT_i=\sumwT_i$ for all $i\in\brkk$. Denote the newly created face locations by $\BJO(\r)$ and $\rayT(\r)$ as in \figref{fig:ABJM}(left). It follows from~\eqref{eq:C_JIS} that for each $\r>0$, the deformed boundary polygon $\PllT(\r)$ still has equal boundary angle sums $\sumbT_i(\r)=\sumwT_i(\r)$ for each $i\in\brkk$. Thus, for each $\r>0$, the quadrilateral with vertices $\QuadT(\r):=(\bdxT_{\io-1},\bdxT_\io,\bdxT_{\io+1},\rayT(\r))$ is \emph{tangential}, i.e., tangent to a circle, with center $\BJO(\r)$. See \figref{fig:ABJM}(left).

This leads to the following geometric problem in step~\itemref{ABJM_summ3}. Given a growing tangential quadrilateral $\QuadT(\r)$ as above, find the smallest value $\r=\r_\jo>0$ such that for some index $\jo$ satisfying $\io+2\leq \jo\leq \io+2k-2$, the Mandelstam variable $(\bdx_{\jo}-\rayTO(\r))^2$ becomes zero. 

For each fixed index $j$, the problem of finding $\r>0$ such that $(\bdx_{j}-\rayTO(\r))^2=0$ turns out to be more than 2000 years old. Namely, it is equivalent to the \emph{problem of Apollonius} which consists of finding a circle $\Circ$ tangent to three other fixed circles in the plane. For us, one of the three circles will be a single point, in which case the tangency condition translates into the requirement that the circle $\Circ$ must pass through this point.

Let us denote $\BJA:=\bdxT_{\io-1}$, $\BJB:=\bdxT_{\io}$, $\BJC:=\bdxT_{\io+1}$, $\BJD:=\rayT(\r)$, and $\BJX:=\bdxT_j$. Out of the two possible conventions for the double bridge, we choose the one where $\bdv_{\io}$ is adjacent to the white vertex of the square as in \figref{fig:ABJM}(left). After a global shift and rotation, we may assume that the entire origami boundary polygon $\PllO$ is contained in the line $\BJA\BJB$ which we identify with the real line $\R\subset\C$ with the point $\BJA$ at the origin. (In \figref{fig:ABJM}(right), the positive direction of the real line is to the left of $\BJA$.) 
Because $\QuadT(\r)$ is tangential, the point $\BJDO=\rayO(\r)$ also lies on this line. For the origami map images of $\BJA,\BJB,\BJC,\BJD,\BJX$, we denote the corresponding real numbers by $\BJa=0,\BJb,\BJc,\BJd,\BJx\in\R$. We may assume that $\BJb,\BJd>0$ as in \figref{fig:ABJM}(right). We thus have $\BJc<\BJb$. Because $|\BJA\BJB|<|\BJA\BJD|+|\BJB\BJC|$, 
 we also must have $\BJc<\BJd$. 
The side lengths of the quadrilateral $\QuadT(\r)$ satisfy $|\BJA\BJB|+|\BJC\BJD| = |\BJA\BJD| + |\BJB\BJC|$ and are therefore given by $|\BJA\BJB| = \BJb$, $|\BJB\BJC| = \BJb-\BJc$, $|\BJC\BJD| = \BJd-\BJc$, and $|\BJA\BJD| = \BJd$. The point $\BJD=\BJDofr$ lies on a hyperbola given by $|\BJD\BJA|-|\BJD\BJC|=\BJc$. 

 Consider circles $\CircB,\CircC,\CircD,\CircX$ of radii $\BJb,|\BJc|,\BJd,|\BJx|$ centered at $\BJB,\BJC,\BJD,\BJX$, respectively. Thus, $\CircB$ and $\CircD$ pass through $\BJA$ and $\CircC$ is tangent to both $\CircB$ and $\CircD$. (Here, both tangencies are external if $\BJc<0$ and internal if $\BJc>0$.) 
The distance between origami images of $\BJD$ and $\BJX$ is $|\BJd-\BJx|$. The corresponding Mandelstam variable is zero if and only if $|\BJD\BJX| = |\BJd-\BJx|$, which is equivalent to the circles $\CircD$ and $\CircX$ being tangent (externally when $\BJx<0$ and internally when $\BJx>0$). Thus, the problem of finding the point $\BJD=\rayT(\r)$ reduces to the \emph{problem of Apollonius}: find a circle $\CircD$ passing through $\BJA$ and tangent to given circles $\CircC$ and $\CircX$. 

This problem can be solved explicitly by performing an \emph{inversion} $z\mapsto z^{-1}$ on $\Cast$ centered at $\BJA$. Denote by $\invCX,\invCB,\invCC,\invCD$ the images of the circles $\CircX,\CircB,\CircC,\CircD$ under inversion. Thus, $\invCB$ and $\invCD$ are lines not passing through $\BJA$, and the image of the interior of $\CircB$ (resp., $\CircD$) under inversion is the half-plane on the side of $\invCB$ (resp., $\invCD$) not containing $\BJA$. 

We assume that $k\geq3$; the base case $k=2$ is straightforward. 
By an argument similar to \cref{rmk:n_even}, one can show in this case that it suffices to consider the indices $j$ such that $j\not\equiv\io\pmod2$, so that both $\G_L$ and $\G_R$ in \cref{fig:ABJM} have an even number of external legs. In particular, $j$ satisfies $\io+3\leq j\leq \io+2k-3$, so the Mandelstam variables $(\bdx_i-\bdx_j)^2$ are positive for $i=\io-1,\io,\io+1$, and we also have $(\bdx_{\io+1}-\bdx_{\io-1})^2>0$. These conditions translate into $|\BJX\BJA|>|\BJx|$, $|\BJX\BJB|>|\BJb-\BJx|$, $|\BJX\BJC|>|\BJc-\BJx|$, and $|\BJC\BJA|>|\BJc|$, respectively. It follows that $\BJA$ lies outside the circles $\CircC,\CircX$, and therefore it also lies outside the inverted circles $\invCC,\invCX$. The conditions $|\BJX\BJB|>|\BJb-\BJx|$ and $|\BJX\BJC|>|\BJc-\BJx|$ ensure that the circle $\CircX$ neither contains nor is contained inside $\CircB$ or $\CircC$. 

Since $\CircC$ is tangent to both $\CircB,\CircD$, the lines $\invCB,\invCD$ are both tangent to the circle $\invCC$. 
When $\BJc<0$ (resp., $\BJc>0$), $\CircC$ lies outside (resp., inside) $\CircB,\CircD$, and thus $\invCC$ lies on the same side (resp., on the opposite side) of the lines $\invCB,\invCD$ as $\BJA$. A similar argument applies to the circle $\invCX$ lying on one of the two sides of the line $\invCD$. 

Summarizing, the dynamics of the point $\BJD=\BJDofr$ moving along a branch of a hyperbola $|\BJD\BJA|-|\BJD\BJC|=\BJc$ starting from $\ray(0)=\BJB$ is translated under inversion into the following dynamics. 
The tangent line $\invCD(\r)$ to a fixed circle $\invCC$ rotates counterclockwise starting with $\invCD(0)=\invCB$.
Let $\rinf\in(0,\infty]$ be such that the line $\invCD(\rinf)$ passes through the point $\BJA$. This value corresponds to the point $\BJD$ escaping to infinity. 
Similarly to \cref{ssec:BCFW_boundary}, for $0<\r<\rinf$, the Mandelstam variable $(\bdx_j-\rayTO(\r))^2$ is zero if and only if the line $\invCD(\r)$ is tangent to the fixed circle $\invCX$ (with the interior of $\invCX$ on the same side of $\invCD(\r)$ as $\BJA$ if $\BJx<0$ and on the opposite side if $\BJx>0$). If such a value $0<\r<\rinf$ exists, we denote it by $\r_j$. 
Inverting $\invCD(\r_j)$ gives the circle $\CircD$, whose center is the desired point $\BJD=\BJDofx(\r_j)$.

When $j$ is not fixed, we still have a fixed circle $\invCC$, a tangent line $\invCB$ to it, and a point $\BJA$ not on $\invCB$ and outside $\invCC$. For each $\io+3\leq j\leq \io+2k-3$ with $j\not\equiv\io\pmod2$, we are given a circle $\invCXj$ that does not contain $\BJA$ and neither contains nor is fully contained in $\invCC$. Step~\itemref{ABJM_summ3} of the \ABJW recursion amounts to rotating the tangent line $\invCD(\r)$ to $\invCC$ counterclockwise starting with $\invCD(0)=\invCB$ until it becomes tangent (with the prescribed side condition) to one of the circles $\invCXj$. Once this happens, we set $\jo:=j$ and use $\r_{\jo}$ to define the locations of the new faces. The remainder of the proof consists of adapting the arguments in \crefrange{ssec:BCFW_boundary}{sec:BCFW:from_lalat_to_t_emb} to show that this procedure results in a valid t-embedding of the corresponding graph $\G\in\BJGs$.

\section{T-embeddings~with~prescribed~boundary and positroid hyperplane arrangements}\label{sec:Varc}
Given a weighted planar bipartite graph $(\G,\wt)$ and a fixed kami plane projection $\PcurveT:=(\bdxT_1,\bdxT_2,\dots,\bdxT_n)$ of the boundary polygon, we study how many t-embeddings $\Tcal$ of $(\G,\wt)$ there are with this kami boundary polygon. (We do \emph{not} fix the origami boundary polygon $\PcurveO$.) A related question was asked in~\cite[Section~3]{KLRR}.

This again reduces to a system of polynomial equations, and the number of complex solutions turns out to be given by the number of bounded regions of a \emph{positroid hyperplane arrangement}. We prove this by applying Varchenko's conjecture~\cite{Varchenko}, proved by Orlik--Terao~\cite{OrTe}. This connection is intriguing in view of the recent developments relating the \emph{CHY scattering equations}~\cite{CHY1} to Varchenko's conjecture; see~\cite[Section~3.4]{Lam_moduli}.

Fix $C\in\Grnd(k,n)$. Let $\HArr_C:=\{\Hyp_1,\Hyp_2,\dots,\Hyp_n\}$ be the associated \emph{positroid hyperplane arrangement} in $\R^k$: $\Hyp_i$ is the hyperplane orthogonal to the column vector $C_i\in\R^k$ for $i\in\brn$.

We apply row operations to $C$ so that the first column $C_1=(1,0,\dots,0)^T$ is a basis vector. Define another hyperplane arrangement $\bHArrC\subset\R^{k-1}$ of $n-1$ affine hyperplanes $\bH_2,\bH_3,\dots,\bH_n$, where $\bH_i:=\{(a_2,\dots,a_k)\in\R^{k-1}\mid C_{1,i} + C_{2,i}a_2+\dots+C_{k,i}a_k = 0\}$. (The arrangement $\bHArrC$ is called the \emph{projectivization} of $\HArrC$; see~\cite{Stanley_hyp}.)
We let $\bbar(C)$ be the number of bounded regions of $\bHArrC$, i.e., the number of bounded connected components of $\R^{k-1}\setminus\bHArrC$. 

\begin{figure}
\begin{tabular}{cc}
 \includegraphics[scale=0.45]{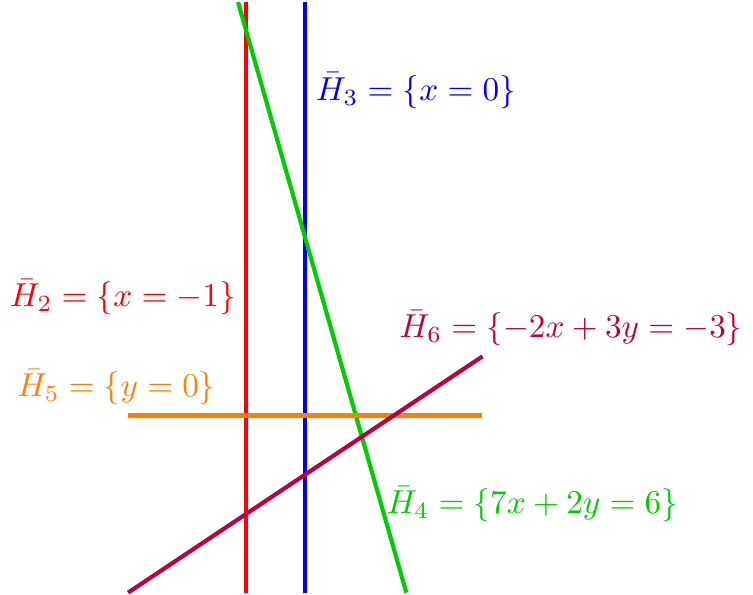}
&
 \includegraphics[scale=0.7]{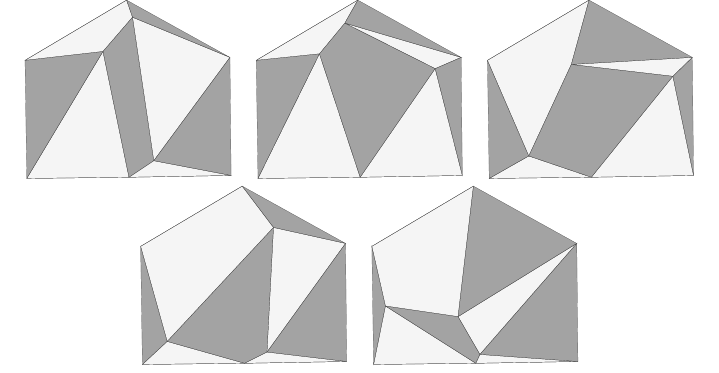}
\end{tabular}
 \caption{\label{fig:Varc} The hyperplane arrangement $\bHArrC$ for $C$ as in \cref{ex:Varc1} (left). The $\bbar(C)=5$ t-embeddings with prescribed boundary (right).}
\end{figure}

\begin{example}\label{ex:Varc1}
Let $k=3$, $n=6$, $f=[3,5,6,7,8,10]\in\BCFWfkn$, and $C = \left(\begin{smallmatrix}
1 & 1 & 0 & -6 & 0 & 3 \\
0 & 1 & 1 & 7 & 0 & -2 \\
0 & 0 & 0 & 2 & 1 & 3
 \end{smallmatrix}\right) \in\Ptp_f$. Thus, $\HArrC$ is an arrangement of six planes in $\R^3$ and $\bHArrC$ is an arrangement of five affine lines in $\R^2$. 
The five lines in $\bHArrC$ are given by $\bH_i = \{(x,y)\in\R^2\mid C_{1,i} + C_{2,i}x + C_{3,i}y = 0\}$ for $i=2,3,\dots,n$. The arrangement $\bHArrC$ is shown in \figref{fig:Varc}(left). It has $\bbar(C) = 5$ bounded regions.
\end{example}

In the following result, we set $\lalatsMAT:=\{\lalat\in\Matdnr^2\mid \rank\la=\rank\lat=2\text{ and }\lat\cdot \la^T = \bzerodd\}$. We say that $\lalat,(\la',\lat')\in\lalatsMAT$ are \emph{$\Cast$-equivalent} if there exists $c\in\Cast$ satisfying $(\y_i',\yt_i')=(c\y_i,c^{-1}\yt_i)$ for all $i\in\brn$ in the notation of~\eqref{eq:intro:y_to_lalat}. 

\begin{proposition}\label{prop:Varchenko}
Fix a generic kami boundary polygon $\PcurveT:=(\bdxT_1,\bdxT_2,\dots,\bdxT_n)$ and let $C\in\Grnd(k,n)$. The number of pairs $(\la,\lat)\in\lalatsMAT$ satisfying $\la\subset C \subset\latp$ and $\PllT = \PcurveT$, counted up to $\Cast$-equivalence, 
 is equal to the number $\bbar(C)$ of bounded regions of $\bHArrC$.
\end{proposition}

\begin{example}\label{ex:Varc2}
Continuing \cref{ex:Varc1}, the $\bbar(C)=5$ possible t-embeddings of some $(\G,\wt)$ satisfying $C=\Meas(\G,\wt)$ with a fixed kami boundary polygon $\PcurveT$ are shown in \figref{fig:Varc}(right).
\end{example}

\begin{remark}
In general, not all pairs $\lalat\in\lalatsMAT$ in \cref{prop:Varchenko} give rise to t-immersions (only the ones satisfying the conditions in~\eqref{eq:intro:LALAK} do). However, experimentally, we see that surprisingly often we indeed get exactly $\bbar(C)$-many t-immersions with given kami boundary polygon $\PcurveT$. For instance, we conjecture that if $\PcurveT$ is a regular $n$-gon then for any generic $C\in\Grtp(k,n)$, all $\bbar(C)={n-2\choose k-1}$ solutions give rise to t-immersions (which are automatically t-embeddings by \cref{lemma:Jordan_curve}).
\end{remark}

\begin{proof}
Finding $\lalat$ as in \cref{prop:Varchenko} is equivalent to finding $\by:=(\y_i)_{i=1}^n\in C\oplus \I C$ and $\byt:=(\yt_i)_{i=1}^n\in C^\perp \oplus\I C^\perp$ such that $\y_i\yt_i = z_i:=\bdxT_{i}-\bdxT_{i-1}$ for all $i\in\brn$; cf. \cref{sec:intro:holomorphic}. 

Let $\alpha_i(\ba):=\ba\cdot C_i$ be the linear function of $\ba\in\C^k$ whose kernel is the (complexified) hyperplane $\Hyp_i$, for $i\in\brn$. We have $\by\in C\oplus\I C$ if and only if there exists $\ba\in\C^k$ such that $\y_i=\alpha_i(\ba)$ for all $i\in\brn$. We set $\yt_i(\ba):=z_i / \alpha_i(\ba)$. We have $\dlog\alpha_i = \frac1{\alpha_i(\ba)}(C_{1,i}\dop a_1 + C_{2,i}\dop a_2 + \cdots + C_{k,i}\dop a_k)$. Writing $\Phi(\ba):=\prod_{i=1}^n \alpha_i(\ba)^{z_i}$ (a complex multivalued function; see~\cite{Varchenko,OrTe}), we have
\begin{equation}\label{eq:Varchenko:dlog_Phi}
 \dlog\Phi 
= \sum_{i=1}^n z_i \dlog\alpha_i
= \sum_{i=1}^n \yt_i(\ba)(C_{1,i}\dop a_1 + C_{2,i}\dop a_2 + \cdots + C_{k,i}\dop a_k).
\end{equation}
A point $\ba\in\C^k$ is a critical point of $\log\Phi$ if and only if the right-hand side of~\eqref{eq:Varchenko:dlog_Phi} is a zero differential form. This is equivalent to the condition that $\byt(\ba):=(\yt_1(\ba),\yt_2(\ba),\dots,\yt_n(\ba))\in\C^n$ is orthogonal to $C$, i.e., $\byt(\ba)\in C^\perp\oplus\I C^\perp$. Thus, the critical points of $\log\Phi$ are in bijection with pairs $(\by,\byt)\in(C\oplus\I C)\times(C^\perp\oplus\I C^\perp)$ such that $\y_i\yt_i = z_i$ for all $i\in\brn$. By \cite[Theorem~1.1 and Proposition~2.4]{OrTe}, when the $z_i$'s are generic complex numbers adding up to zero, the number of critical points of $\log\Phi$ (counted up to rescaling $(\by,\byt)\sim(c\by,c^{-1}\byt)$ for $c\in\Cast$) is given by $\bbar(C)$.
Each solution satisfies $\rank\la=\rank\lat=2$ since the dimension of the space of $\Cast$-equivalence classes of pairs $(\la,\lat)\in\Matdnr^2$ satisfying $\la\subset C \subset\latp$ and, say, $\rank\la\leq1$ is given by $\dim C + 2\dim C^\perp - 1 = k+2(n-k)-1$, which is less than $\dim_{\R}\{(z_1,\dots,z_n)\in\C^n\mid z_1+\dots+z_n=0\}=2n-2$.
\end{proof}

\begin{problem}
Find a closed formula for $\bbar(C)$ or for the characteristic polynomial $\chiC(t)$ for $C\in\Ptp_f$ in terms of the bounded affine permutation $f$.
\end{problem}

\apptocmd{\thebibliography}{\setlength{\itemsep}{-1.3pt}}{}{}
\bibliographystyle{alpha}
\bibliography{loop_ampl}

\end{document}